%% file: main.tex
\documentclass[11pt]{ociamthesis}  
\input{preamble}

\title{Automated Market Making and\\[1ex]
        Decentralized Finance 
	}   

\author{Marcello Monga}             
\college{St Hugh's College}  

\degree{DPhil in Mathematics}     
\degreedate{Trinity 2024}         

\begin{document}
	
\baselineskip=18pt plus1pt

\setcounter{secnumdepth}{3}
\setcounter{tocdepth}{3}

\maketitle                   
\afterpage{\blankpage}

\begin{abstract}
    Automated market makers (AMMs) are a new type of trading venues which are revolutionising the way market participants interact. At present, the majority of AMMs are constant function market makers (CFMMs) where a deterministic trading function determines how markets are cleared. Within CFMMs, we focus on constant product market makers (CPMMs) which implements the concentrated liquidity (CL) feature. In this thesis we formalise and study the trading mechanism of CPMMs with CL, and we develop liquidity provision and liquidity taking strategies. Our models are motivated and tested with market data.\\~
    
    We derive optimal strategies for liquidity takers (LTs) who trade orders of large size and execute statistical arbitrages. First, we consider an LT who trades in a CPMM with CL and uses the dynamics of prices in competing venues as market signals. We use Uniswap v3 data to study price, liquidity, and trading cost dynamics, and to motivate the model. Next, we consider an LT who trades a basket of crypto-currencies whose constituents co-move. We use market data to study lead-lag effects, spillover effects, and causality between trading venues. \\~

    We derive optimal strategies for strategic liquidity providers (LPs) who provide liquidity in CPMM with CL. First, we use stochastic control tools to derive a self-financing and closed-form optimal liquidity provision strategy where the width of the LP's liquidity range is determined by the profitability of the pool, the dynamics of the  LP's position, and concentration risk. Next, we use a model-free approach to solve the problem of an LP who provides liquidity in multiple CPMMs with CL. We do not specify a model for the stochastic processes observed by LPs, and use a long short-term memory (LSTM) neural network to approximate the optimal liquidity provision strategy.

    \thispagestyle{empty}
\end{abstract}

\afterpage{\blankpage}

\begin{acknowledgements}
First and foremost, I would like to thank my supervisor, Álvaro Cartea, who guided me through this journey. The endless hours spent discussing research had significant impact on my professional and personal development. He truly was a unique supervisor; he dedicated plenty of his precious time to discuss our projects, he truly cared about my research progress and he understood how the ups and downs of my personal life were impacting my studies. I will be forever grateful for his supervision.

I would like to thank Fayçal Drissi, who worked closely with me and whose dedication and talent were crucial to the results achieved by our research group. I would also like to thank Phillippe Bergault and Leandro Sánchez-Betancourt for the fruitful discussions during my PhD and my viva. I am grateful to Mihai Cucuringu, Olivier Guéant, Sebastian Jaimungal, and José Penalva, whose papers and books I have avidly read multiple times over the years. I would like to thank the director of the CDT in Random Systems, Rama Cont, and Anthony Ledford from Man AHL, who helped me step out of my comfort zone and provided immensely valuable advice.

I would like to thank Simon Davenport, Renaud Drappier, Gregoire Loeper, and Boris Oumow from BNP Paribas for their insights from the industry and for helping me bridge my research with real-world applications throughout my PhD. Moreover, during my internship at BNP Paribas, I had the privilege of working alongside these exceptional practitioners, who guided and mentored me.

I am grateful to all the people I have met during these years at the Mathematical Institute and the Oxford-Man Institute (OMI). In particular, I would like to thank Jen Desmond and Vanessa Wilkins for their wonderful work at the OMI; Álvaro Arroyo Nuñez and Patrick Chang, who were the first students I met at the Institute; Gabriel Garcia Arenas and Jonathan Plenk for the fruitful discussions we had over various cups of tea; Andrea Clini and László Mikolas for their advice; Harrison Waldon and Fernando Moreno Pino for being role models as postdocs at the OMI; Milena Vuletic for her energy; and all the participants of the Oxford Victoria seminar, including Sergio Calvo Ordo\~nez, Gerardo Durán-Martín, Qi Jin, Donggeun Kim, Yannick Limmer, Nicolas Petit, Daniel Poh, and Danni Shi.

I will never stop saying how grateful I am to my family for their endless support. Finally, I would like to thank all my friends in Italy and Switzerland who have always been there for me during these years. In particular, I would like to thank my friends from my hometown and my friends in Switzerland from the bQm analytics club.

\end{acknowledgements}

\afterpage{\blankpage}

\begin{romanpages}          
    \tableofcontents            
\end{romanpages}            

\chapter{Introduction}
\section{Motivation}
Decentralised Finance (DeFi) is a collective term for blockchain-based financial services that do not rely on intermediaries such as central authorities or banks. Over the past few years, DeFi has gained massive traction and the total value locked in DeFi platforms is as of July 2024 around 91.66 billion US dollars.\footnote{\url{https://www.stelareum.io/en/defi-tvl.html}} New powerful technologies are the engine behind the remarkable growth of DeFi, which is rapidly changing the financial landscape and is in direct competition with many traditional stakeholders. Within DeFi, automated market makers (AMMs) are a new paradigm in the design of trading venues and are revolutionising the way market participants provide and take liquidity. As of July 2024, the daily traded volume in AMMs is around 6.44 billions US dollars,\footnote{\url{https://www.coingecko.com/en/exchanges/decentralized}} which accounts for more than 8.0\% of the total spot cryptocurrency trading volume.\footnote{ \url{https://www.theblock.co/data/decentralized-finance/dex-non-custodial}} Currently, AMMs are mainly known as exchanges for cryptocurrencies; however, the core concepts of AMMs go beyond the cryptocurrency sector and they are poised to challenge traditional electronic exchanges in all asset classes.

Despite the remarkable growth of AMMs, there is little understanding of their microstructure properties. AMMs are a technology with interesting features and high potential. However, more research is needed to understand, regulate, and further develop AMMs. Indeed, little is known of what happens at a microstructural level in AMMs and it is important to provide academics and practitioners the right tools and to study in detail the advantages and disadvantages of this new type of trading venue.

\section{Contribution}
In this thesis we formalise and study the trading mechanism of AMMs, and we develop liquidity provision and liquidity taking strategies. Our models are motivated and tested with market data. For the sake of chapters being self-contained, the necessary terminology is reintroduced in each chapter.

\subsection{Optimal execution}
A classical problem in the market microstructure literature is how an agent can execute a large order while minimising intrinsic and extrinsic execution costs. This problem is extensively studied in the context of limit order book (LOB) markets. In the extant literature, there is a large number of execution models that are specifically designed for traditional LOBs; however, there are no rigorous models to trade and execute assets in AMMs.

Chapter \ref{ch:paper} presents the work in \cite{cartea2022decentralised}, which solves the problem of a liquidity taker (LT) who wishes to trade in a particular class of AMMs. The LT wishes to execute a large position in an asset and to execute statistical arbitrages based on market signals. The LT trades in the AMM and uses the exchange rate from a more liquid venue to anticipate future price movements.

We formulate the trading problem as a stochastic control problem where the LT controls the speed at which  she sends liquidity taking orders. Key to the performance of the LT's strategies, is to balance price risk and execution costs. To showcase the performance of our strategies, we use in-sample market data to estimate model parameters and  out-of-sample market data to execute the strategies in `real time' as an LT would have done. Furthermore, we use market data to derive important properties about execution costs and market impact in AMMs.

Chapter \ref{ch:ieee} presents the work in \cite{cartea2023execution}, which solves the problem of an LT who executes large orders  and executes statistical arbitrages in a basket of crypto-currencies whose constituents co-move. We use statistical tools to study lead-lag effects, spillover effects, and causality between trading venues and we study co-movements between different crypto-currencies.

We use stochastic control tools and derive a closed-form strategy that can be computed and implemented by the LT in real time. The LT uses market signals and exchange rate information from relevant AMMs and traditional venues to enhance the performance of her strategy. Finally, we use market data from two pools of Uniswap v3 and from the LOB-based exchange Binance to showcase the performance of the strategy.

\subsection{Optimal liquidity provision}
In traditional trading venues, liquidity providers quote buy and sell prices and profit from the spread between the two. This problem is extensively studied in the context of LOB markets and quote-driven markets; however, little research has been conducted on strategic liquidity provision in AMMs.

AMMs introduce a new liquidity provision mechanism where rules to clear demand and supply of liquidity are considerably different. The most popular type of AMMs are constant product market makers (CPMMs) with concentrated liquidity (CL). In CPMMs with CL, LPs specify the rate intervals (i.e., tick ranges) over which they deposit their assets, and this liquidity is counterparty to trades of LTs when the marginal exchange rate of the pool is within the liquidity range of the LPs. When LPs deposit liquidity, fees paid by LTs accrue and are paid to LPs when they withdraw their assets from the pool. The amount of fees accrued to LPs is proportional to the share of liquidity they hold in each liquidity range of the pool.

Chapter \ref{ch:pl} presents the work in \cite{cartea2022decentralised2}, which solves the problem of an LP who provides liquidity in a single pool with CL. We use stochastic control to derive a self-financing and closed-form optimal liquidity provision strategy, which the LP uses to dynamically adjust the range where she provides liquidity.

The width of the LP's liquidity range is determined by the profitability of the pool (provision fees minus gas fees), the dynamics of the LP's position, and concentration risk. Concentration risk refers to the decrease in fee revenue if the marginal exchange rate (akin to the midprice in an LOB) in the pool exits the LP's range of liquidity. When the drift in the marginal rate is stochastic, we show how to optimally skew the range of liquidity to increase fee revenue and profit from the expected changes in the marginal rate. Finally, we use Uniswap v3 data to show that, on average, LPs have traded at a significant loss, and to show that the out-of-sample performance of our strategy is superior to the historical performance of LPs in the pool we consider. 

Chapter \ref{ch:deep} uses a model-free approach to solve the problem of an LP who provides liquidity in multiple liquidity pools with CL. We do not specify a model for the stochastic processes observed by the LP, and we characterise the discrete-time dynamics of the LP's wealth, which results in fee income and the value of her holdings in the pool.

At each time step, the LP specifies the proportion of wealth to deposit in each pool and the ranges where she want to provide liquidity. The LP evaluates the performance of her strategy with a performance criterion and uses LSTM neural networks to approximate the optimal strategy. We train the LSTM network on simulated data and on market data. 

We test the model-free strategy on simulated data to showcase the performance of the strategy under different market conditions. First, we consider an LP who provides liquidity in a single liquidity pool with CL and we recover established results from Chapter \ref{ch:pl}. Then, we consider an LP who provides liquidity in two liquidity pools with CL.

Finally, we train the LSTM network on market data from Uniswap v3 and Binance. The input to the LSTM strategy is a set of features which we extract from the data in both venues. Overall, we find that gas fees, a flat fee paid by blockchain users, have a great impact on the performance of the strategy. 

\cleardoublepage
\chapter{Background \label{ch:background}}

\section{Optimal execution}

Whenever an LT has to liquidate or to acquire a large position, she may incur adversarial price movements, as a consequence of her own trades. Indeed, if she exits or enters her position using a single order, she might consume all the available liquidity at the best price  and receive worse prices to complete the order. Typically, to reduce execution costs, LTs split the parent order in smaller child orders which are executed over a longer time window. However, if the LT executes the child orders too slowly, then she is exposed to market price risk and she might face a loss.

The literature on optimal execution in traditional electronic markets is vast. It is built on the seminal works \cite{bertsimas1998optimal} and \cite{almgren2000optimal}. In their work, the authors study the problem of an LT who uses market orders (MOs) to liquidate a large position in a limit order book (LOB).  More specifically, the Almgren--Chriss model is a discrete-time model where the LT maximises a mean-variance objective function. Later on, many variations to the model framework were proposed. \cite{forsyth2012optimal} use a mean-quadratic variation approach, which can be approximated with a mean-variance type model for short trading windows. \cite{schied2009risk} consider an infinite time-horizon problem for a Von Neumann--Morgenstern LT and characterises the optimal strategy using stochastic control tools. \cite{cartea2015book} introduce various frameworks where the LT use an objective function equal to the LT's terminal wealth minus an inventory penalty to solve the optimal execution problem in the context of LOB markets. \cite{cartea2014OptimalEW} extend this framework to allow the LT to both post LOs and send MOs.

Regarding model parameters, \cite{almgren2012optimal} assumes that volatility and temporary price impact are random and compute the optimal strategy in this framework. \cite{cheridito2014optimal} also study the case of random volatility and random temporary price impact but compares three different optimal strategies obtained using three different performance criteria. \cite{gatheral2013dynamical} review dynamical market impact models and studies the regularity of their associated optimal strategies. \cite{graewe2018smooth} introduce a price-dependent market impact and considers a liquidation problem where the LT can unwind her position both in the lit market and in a dark pool. \cite{fouque2021optimal} introduce fast-mean reverting stochastic price impact to reproduce exactly how price impact moves during the day.  

Numerous works in the optimal execution literature explore different models for price impact and in LOBs. \cite{lorenz2013drift} use linear transient price impact with exponential decay and employs singular stochastic control to deal with the non-Markovian framework. \cite{bank2019optimal} introduces a price impact model which accounts for finite market depth, tightness, and resilience. \cite{obizhaeva2013optimal} introduce a dynamic model for LOBs and shows the importance of supply/demand dynamics in determining optimal trading strategies. Later on, the Obizhaeva-Wang model was extended and in \cite{alfonsi2010general} the authors consider an optimal liquidation problem where the LOB has a general shape. \cite{alfonsi2010optimal} derive explicit optimal execution strategies in a discrete-time LOB model with general shape functions and an exponentially decaying price impact. \cite{gatheral2010transient} study the problem of an LT subject to transient price impact who wishes to liquidate a large position while minimising execution costs. \cite{hey2023cost} quantify the cost that LTs face when they misspecify price impact. \cite{muhle-karbe2024stochastic} use a stochastic liquidity parameter to approximate nonlinear price impact.

The Almgren-Chriss model was further extended to allow for the use of market signals to improve performance of strategies. \cite{cartea2016incorporating} use order flow as a predictive market signal and derives closed-form solutions. \cite{bechler2015optimal} study an optimal execution framework where the LT accounts for order flow imbalance and endogenise the trading horizon to react to it. \cite{cartea2018enhancing} also employs order flow imbalance and uses real NASDAQ data to calibrate a Markov chain-modulated pure jump model of price, spread, limit orders (LOs) and MOs arrivals, and volume imbalance. In \cite{lehalle2019incorporating} the authors consider a Von Neumann-Morgenstern LT who uses a Markovian signal to trade and \cite{neuman2020optimal} extend this framework to include transient price impact. \cite{forde2022optimal} consider an LT subject to power-law resilience and zero temporary price impact who employs a Gaussian signal. \cite{cartea2018trading} and \cite{bergault2022multi} consider a multi-asset framework where the LT trades in a basket of assets which co-move. \cite{cartea2022double} use the signature of the path of a stochastic process to derive an optimal double-execution strategy.

Another strain of the optimal execution literature focuses on trying to reduce latency. Specifically, \cite{cartea2021shadow} derive the price that LTs would be willing to pay to reduce their latency in the marketplace. \cite{cartea2021latency} show how LTs can leverage marketable limit orders (MLOs) to reduce latency and  \cite{cartea2022optimal} derive an optimal trading strategy which uses MLOs to liquidate a position over a trading window when there is latency in the marketplace. 

Recent works in the literature use reinforcement learning (RL) to solve the optimal execution problem; see \cite{gueant2023reinforcement}. One of the first example of the use of RL in the algorithmic trading literature was the work of \cite{nevmyvaka2006reinforcement}, where the authors use tabular Q learning to optimise trade execution. \cite{hendricks2014reinforcement} use Q learning to outperform the model by Almgren and Chriss, and they test their strategy on market data. \cite{ning2021double} solve the optimal execution problem using a model free double deep Q learning algorithm where the input of the LT strategy are features of the LOB. \cite{cartea2023reinforcement} use double deep Q networks to derive the optimal strategies for an LT who executes statistical arbitrage in a foreign exchange (FX) triplet. \cite{cartea2023bandits} use contextual bandits and Gaussian processes to exploit trading signals and improve the Almgren--Chriss strategy. \cite{waldon2024dare} introduce a framework to trade in continuous time under uncertainty.


\section{Optimal liquidity provision}
In traditional trading venues, LPs quote buy and sell prices and profit from the spread between the two. In LOB markets, LPs send LOs that specify price and quantity they would like to buy or sell of a given financial instrument. An LO sent to the market is stored in the LOB until an LT matches it or the LO is cancelled. LTs send market orders (MOs) when they wish to buy or to sell the underlying product immediately, which are matched with LOs resting in the LOB. Most LOB markets employ a price-time-priority mechanism where incoming buy (sell) MOs are matched by sell LOs with lowest (highest) price amongst outstanding sell (buy) LOs. Moreover, if two or more outstanding sell LOs offer the same price, then the one with earliest timestamp is executed first.

In trading venues organised as over-the-counter (OTC) markets, trading happens directly between two parties. Typically, an LT interested in trading sends a formal request to various LPs specifying the size of her order. Then, LPs reply to the LT and propose bid and ask prices, and the LT decides whether to trade on the proposed quotes or not. In OTC markets, LPs propose quotes based on their own inventory, the size of the LT's order, and market conditions. 

The literature on optimal liquidity provision in traditional trading venues is vast and it is built on the seminal work of \cite{ho1981optimal}. In their work, the authors use dynamic programming to derive optimal bid and ask prices in feedback form that maximise the expected utility of the LP's final wealth. Later, their work was revived by \cite{avellaneda2008high} who use econophysics results to model arrival rates of MOs to the LOB, and derive an optimal strategy to post LOs. \cite{gueant2012dealing} provide a detailed analysis of the model proposed by Avellaneda and Stoikov, and solve the problem under inventory constraints.

Later on, variations of the model by Avellaneda and Stoikov were proposed. \cite{cartea2014buy} use self-exciting processes to model order flow, and solve the optimisation problem using an objective function equal to the LP's terminal wealth minus an inventory penalty. \cite{cartea2017algorithmic} consider the ambiguity in the specification of fill probabilities and the dynamics of intensities and prices. In the work by \cite{cartea2020market}, the LP incorporates alpha signals into her optimal posting strategy. \cite{jusselin2021optimal} uses generalised Hawkes processes to model order flow and solves the problem of and LP who provides liquidity in an LOB market. \cite{bergault2021size} consider an LP who provides liquidity in multiple assets and uses a factor model to reduce the dimension of the optimisation problem. Moreover, the authors assumes that the the LP receives liquidity taking orders of various sizes, and show how the size of the order impacts the optimal trading strategy. \cite{bergault2021closed} derive a closed-form approximation for the bid and ask quotes offered by an LP who provides liquidity in multiple assets. \cite{barzykin2021market,barzykin2022dealing,barzykin2023algorithmic} study liquidity provision in the FX market. Specifically, they consider an LP who provides liquidity in various currencies and hedges inventory risk by trading with other LPs. \cite{barzykin2024algorithmic} studies an LP who provides liquidity in spot precious metals, where liquidity is mainly provided by futures contracts.

Inspired by the seminal works of \cite{glosten1985bid} and \cite{kyle1985continuous}, various works in the optimal liquidity provision literature consider LPs who trades with informed LTs. \cite{campi2020optimal} considers an LP who provides liquidity in a financial markets where arrival rates and price dynamics depend on unobservable factors. \cite{herdegen2023liquidity} consider an informed LT who trades with multiple LPs who are balancing inventory risk and adverse selection.  In the work from \cite{cartea2022brokers}, the LP trades at a loss with an informed LT to extract trading signals. Their work was later extended by \cite{bergault2024mean} who frame the problem as a mean field game.

Recent works in the literature use RL to solve the optimal liquidity provision problem; see \cite{gasperov2021reinforcement}. One of the first example of the use of RL for liquidity provision was the work of \cite{chan2001adaptive}, where the authors use various RL algorithms to provide liquidity in a model similar to the one proposed by Glosten and Milgrom. Other notable examples in the literature include the work of \cite{spooner2018market}, where the authors use temporal difference RL to solve the optimal liquidity provision problem, and \cite{spooner2021robust} and \cite{gasperov2021market} who use adversarial RL. \cite{gueant2019deep} use a model-based approach to solve the optimal liquidity provision problem, and \cite{jerome2023mbt} introduce a Python module to solve model-based optimal liquidity provision problems in an LOB market.

\section{Automated market makers}
AMMs are a new paradigm in the design of trading venues and are revolutionising the way market participants provide and take liquidity. Instead of relying on a matching mechanism, they rely on hard-coded and immutable programs running on peer-to-peer networks such as Ethereum (\cite{vitalik2014ethereum}), Tezos (\cite{goodman2014Tezos}) and Solana (\cite{yakovenko2014solana}).

At present, the majority of AMMs are CFMMs. In CFMMs, a trading function and a set of rules determine how LTs and LPs interact, and how markets are cleared. The trading function is deterministic and known to all market participants.  CFMMs display pools of liquidity for different assets, where the relative prices assets is determined by their quantities in the pool as prescribed by the trading function. The trading function establishes the link between liquidity and prices, so liquidity takers can compute the execution costs of their trades as a function of the trade size. A key difference between CFMMs and a LOB is that execution costs in CFMMs are given by a closed-form formula, where the convexity of the trading function plays a key role. As in traditional markets that operate an LOB, the larger the size of an order, the higher are the execution costs. 

The most widely used CFMMs are CPMMs, such as Uniswap v3 \cite{uniswap2021core}, PancakeSwap,\footnote{\url{https://docs.pancakeswap.finance/}} SushiSwap,\footnote{\url{https://docs.sushi.com/}} Balancer \cite{martinelli2019balancer}, DODO,\footnote{\url{https://dodoex.github.io/docs/docs/whitepaper/}} and Bancor \cite{bancor2017}. Another important class of CFMMs is that of constant sum market makers (CSMMs) such as mStable\footnote{\url{https://docs.mstable.org/}} and StableSwap \cite{stableswap2019}. We refer the interested reader to \cite{xu2022sok} for an overview of the different types of AMMs.

AMMs are a relatively new concept and thus academic literature is still restricted. Early works on AMMs focus on their market microstructure properties and include the work of \cite{angeris2021analysis}, where the authors formalise the mathematics of Uniswap v2 and show how to take advantage of the spread between the protocol and another reference market. This work is generalised to CFMMs in \cite{angeris2020improved}. \cite{cartea2023predictable} study the losses faced by LPs in CFMMs and then introduce predictable loss, which measures the unhedgeable losses of LPs stemming from the depreciation of their holdings in the pool and from the opportunity costs from locking their assets in the pool. Predictable loss is similar to loss-versus-rebalancing in  \cite{milionis2022automated} which describes the unhedgeable losses of LPs in traditional CFMMs due to losses to arbitrageurs.\footnote{\cite{milionis2022automated} introduced loss-versus-rebalancing in August 2022 as a measure that quantifies the unhedgeable losses of LPs in CFMMs to arbitrageurs. Contemporaneously, \cite{cartea2023predictable} introduced predictable loss in November 2022 as a component in the wealth of strategic LPs that quantifies predictable loss due to the convexity of the trading function and due to opportunity costs in CFMMs and in CL markets.}  \cite{capponi2021adoption} show that rates at which users provide liquidity are negatively correlated with volatility and positively correlated with traded volume. This feature is also observed in traditional LOB markets. \cite{berg2022empirical} empirically studies inefficiencies in AMMs with particular focus on Uniswap and SushiSwap. \cite{milionis2023automated} study the adverse selection incurred by LPs in the presence of fees, and \cite{fukasawa2023model} study the hedging of the impermanent losses of LPs. 

Another literature strain on strategic liquidity provision include \cite{heimbach2022risks} which discusses the tradeoff between risks and returns that LPs face in Uniswap v3. \cite{cartea2022decentralised2} introduces a continuous-time model for optimal liquidity provision, and \cite{li2023yield} study the economics of liquidity provision. The models in \cite{fan2021strategic} and \cite{fan2022differential}  focus on fee revenue and use approximation techniques to obtain dynamic strategies. 

Finally, there is a growing literature on AMM design for fair competition between LPs and LTs. \cite{goyal2023finding} study an AMM with dynamic trading functions that incorporate beliefs of LPs, \cite{lommers2023:case} study AMMs where the LP's strategy adjusts dynamically to market information,  \cite{cartea2023automated} generalise CFMMs and propose AMM designs where LPs express their beliefs and risk preferences, and \cite{bergault2023automated} introduce an AMM design inspired by OTC markets.

Other works on AMMs focus on arbitrage opportunities arising from blockchain inefficiencies. \cite{daian2020flash} formalises front-running in AMMs and introduces priority gas auctions, which consist in competitively bidding up transaction fees in order to obtain priority ordering. \cite{park2021conceptual} discusses front-running arbitrages in AMMs, and proposes a pricing rule to prevent them. \cite{zhou2021} provides a detailed analysis of sandwich attacks and quantifies how much a LT performing sandwich attacks can earn on a daily basis.

\chapter{Execution and Speculation \label{ch:paper}}
\section{Introduction}
Decentralised Finance (DeFi) is a collective term for blockchain-based financial services that do not rely on intermediaries such as brokers or banks.  New powerful technologies are the engine behind the remarkable growth of DeFi, which is changing the financial landscape and is in direct competition with many traditional stakeholders. Within DeFi,  Automated Market Makers (AMMs) are a new paradigm in the design of trading venues and are revolutionising the way market participants provide and take liquidity. Currently, AMMs are mainly known as exchanges for cryptocurrencies; however, the core concepts of AMMs go beyond the cryptocurrency sector and they are poised to challenge traditional electronic  exchanges in all asset classes.


At present, the majority of AMMs are constant function market makers (CFMMs). In CFMMs, a trading function and a set of rules determine how liquidity takers (LTs) and liquidity providers (LPs) interact, and how markets are cleared.  The trading function is deterministic and known to all market participants.  CFMMs display pools of liquidity for pairs of assets, where the relative prices between the two assets is determined by their quantities in the pool as prescribed by the trading function. The trading function establishes the link between liquidity and prices, so LTs can compute the execution costs of their trades as a function of the trade size.  A key difference between CFMMs and limit order books (LOBs) is that execution costs in CFMMs are given by a closed-form formula, where the convexity of the trading function plays a key role. As in traditional markets that operate an LOB, the larger the size of an order, the higher are the execution costs. 

Within CFMMs, we focus on constant product market makers (CPMMs), which are the most popular type of CFMM and where the trading function uses the product of the quantities of each asset in the pool to determine clearing prices. In this chapter, we present the work of \cite{cartea2022decentralised}, who solve the problem of an LT who wishes to trade in a CPMM to execute a large position in an asset and to execute statistical arbitrages based on market signals. We formulate the trading problem as a stochastic control problem where the LT controls the speed at which  she sends liquidity taking orders. Key to the performance of the LT's strategies, is to balance price risk and execution costs. In CPMMs,  the execution costs given by the trading function are inversely proportional to the depth of the pool and proportional to a non-linear transformation of the relative prices of the two assets in the pool. 


Despite very high levels of activity for many of the pairs traded in CPMMs, price formation currently occurs in the LOBs of alternative electronic markets. In one version of our model, we assume that LTs in a CPMM inform their decisions with the prices in the CPMM and those from an alternative venue, and assume that the depth of the pool does not vary during the LT's trading horizon. In this setup, we derive a versatile trading strategy which can be used to focus on the execution of a large order or on statistical arbitrages.



When the focus is to exchange a large position in one asset for another asset, both of which are provided as a pair in the CPMM, the strategy relies on two components. One component is as in the traditional execution strategies (e.g., TWAP-like or Almgren-Chriss), and the second component is an arbitrage that takes advantage of short-lived discrepancies in the prices of the CPMM and those in the alternative venue. Instead, if the objective is speculation, the strategy relies more on the statistical arbitrage component to take advantage of the lead-follow relationship between the prices in the CPMM and the alternative venue.

In anticipation of the growth of AMMs, another version of our model assumes that prices in the CFMMs are efficient, so the discrepancies between the CFMM and LOB prices are not economically significant. The increase in the efficiency of prices in CFMMs  will be a result of an increase in the activity of LPs and LTs, which will also result in more changes in  the depth of the pool of the CPMM. Thus, we propose another model where the depth of the pool is stochastic and we solve the LT's execution problem for large orders. 




We use Uniswap v3 data for CPMMs that trade pairs of cryptocurrencies to study the empirical properties of this particular AMM, and to illustrate the performance of the proposed liquidation and speculative strategies. The efficient prices are those from Binance where LPs and LTs interact through a traditional price-time priority LOB. To showcase the performance of our strategies, we use in-sample data to estimate model parameters and  out-of-sample data to execute the strategies in `real time' as an LT would have done. In our analysis, we use rolling time windows of a few hours starting 1 July 2021 and ending 5 May 2022 to obtain the distribution of the financial performance of the strategies. We look at two pairs of assets, one that is heavily traded and one that is not as frequently traded. We show the superior performance of our execution strategy over TWAP and over a strategy that would have executed the whole inventory in one trade at the beginning of the trading horizon. Finally, we show that there are profitable opportunities to execute statistical arbitrages in Uniswap v3  when the strategy is informed by Binance prices.

The remainder of this chapter is organised as follows. Section \ref{sec:AMM} discusses how CFMMs operate and uses Uniswap v3 data to study price, liquidity, and trading cost dynamics in CPMMs. Section \ref{sec:Model} solves the optimal execution problem when the pool has constant depth during the execution window and price formation is in an alternative trading venue. Section \ref{sec:Model2} solves the optimal execution problem when the pool depth is stochastic and price formation takes place in the AMM. Finally, Section \ref{sec:Performance} showcases the performance of liquidation and statistical arbitrage strategies.
\section{Automated Market Making \label{sec:AMM}}

In this section, we discuss how CFMMs operate and how they differ from electronic markets where traders interact through an LOB. In particular, we describe the interactions of market participants with a CFMM that is in charge of a pair of assets.  We use  transaction data from Uniswap v3 to study the activity of market participants, the dynamics of liquidity, and implicit transaction costs.

\subsection{Description \label{sec:AMM1}}

AMMs are hard-coded and immutable programs running on a peer-to-peer network. They provide a venue to trade pairs of assets $X$ and $Y$, where the liquidity of the pool consists of $x$ units of $X$ and $y$ units of $Y$. The exchange rate of the pool is the price of $Y$ in terms of the price of $X$, and it is determined by the quantities $x$ and $y$.\footnote{Some AMMs also display pools with more than two assets.} Two types of market participants interact in an AMM: LPs deposit their assets in the pool and LTs trade directly with the pool.

Here, we consider a CFMM that trades the pair of assets $X$ and $Y$.  CFMMs are characterised by a deterministic trading function $f(x,y)$ that determines the rules of engagement among participants in the pool. For instance, the trading function of the CPMM is $f(x,y)=x\times y$. Other types of CFMMs are the  constant sum market maker with $f(x,y) = x+y$; the constant mean market maker with $f(x,y) = w_x \, x + w_y \, y $, where $w_x, w_y > 0 $ and $w_x + w_y = 1$; and the hybrid function market maker, which uses combinations of trading functions.

In peer-to-peer networks, participants invoke the code of the AMM smart contract to instruct market operations. LPs send messages with instructions to deposit or withdraw liquidity, and LTs send messages to exchange one asset for the other. To provide liquidity, an LP instructs the AMM with the quantities in assets $X$ and $Y$ to be deposited in a specific pool. On the other hand, LTs indicate to the AMM the pool and the quantity of the asset to be exchanged. The available liquidity in the pool and the trading function of the AMM determine the exchange rate received by the LT. For each trade, LTs pay the AMM a transaction fee, which is distributed amongst LPs in the same proportion as their contributions to the pool.\footnote{See \cite{heimbach2021behavior} and \cite{cartea2022decentralised} for an analysis on how LPs profit from their activity.}

The trading function $f\left(x, y\right)$ is increasing in $x$ and $y,$ and it ties the state of the pool before and after an LT transaction is executed. Throughout, the signs of $\Delta x$ and $\Delta y$ are the same. If $\Delta y>0,$ the LT sells asset $Y$, and if $\Delta y<0,$ the LT buys asset $Y$; for simplicity, we assume zero fees.\footnote{To take into account the fee for an LT transaction, one applies a discount to the quantity $\Delta y$ before calculations are carried out.} The condition
\begin{equation}\label{eq:trade}
    f(x-\Delta x,y+\Delta y)=f(x,y)=\kappa^2\ 
\end{equation}
determines the quantity $\Delta x$ that the agent receives (pays) when exchanging $\Delta y >0$ ($\Delta y <0$). The trading function keeps the quantity $\kappa^2$ constant before and after a trade is executed. We write $f(x,y)=\kappa^2$ as $x=\varphi(y)$ for an appropriate function $\varphi$ that depends on $\kappa$;  we refer to $\varphi$ as the \emph{level function}, and assume it is convex.\footnote{The convexity of the level function is by design. One can show that a no-arbitrage condition leads to the necessary convexity of the level function; see \cite{cartea2022decentralised}.}

If an LT wishes to sell $\Delta y$ units of asset $Y$, she receives  $\Delta x = \Delta y \times \tilde{Z}(\Delta y)$ units of asset $X$  in exchange. Here, $\tilde Z(\Delta y)$, with units $X/Y,$ is the exchange rate received by the agent when trading a quantity $\Delta y$ of asset $Y.$ Therefore,
	\begin{align*} 
		x - \Delta x = \varphi(y + \Delta y) \implies
		\varphi(y) - \Delta y \,\tilde Z(\Delta y) = \varphi(y + \Delta y) \,,
	\end{align*}
	so
	\begin{align}
		\label{eq:execprice}
		\tilde Z(\Delta y) = \frac{\varphi(y)-\varphi(y+\Delta y)}{\Delta y}\,,
	\end{align}
	and  for an infinitesimal quantity $\Delta y$ we write
	\begin{align}
		\label{eq:instantaneousprice}
		Z =-\varphi'(y)\, .
	\end{align}
	
	We refer to $Z$  as the \emph{instantaneous rate} of the AMM, which is equivalent to the midprice in an LOB. The instantaneous rate $Z$ is a reference exchange rate -- the difference between its value and the execution rate is similar to the difference between the LOB midprice and the average price obtained by a liquidity taking order that crosses the spread and walks the book when it is filled.
	

	The trading function $f(x,y)$ is increasing in the pool quantities $x$ and $y.$ Thus, when LP activity increases (decreases) the size of the pool, the value of $\kappa$ increases (decreases). We refer to $\kappa$ as the depth of the pool. A distinctive characteristic of AMMs is that liquidity provision changes the depth of the pool, but it does not change the instantaneous rate. For example, in a CPMM, the instantaneous rate is the ratio of the quantities supplied in the pool, i.e.,
	\begin{equation}
		Z=\frac xy\, ,
	\end{equation}
	and when an LP deposits quantities $\Delta x$ and $\Delta y$ in the pool, the pair $(\Delta x,\Delta y)$ must satisfy
	\begin{equation}
		\label{eq:LPconditionCPMM}
		\frac{x}{y} = \frac{x+\Delta x}{y+\Delta y} = Z\, ,
	\end{equation}
	and the value of $\kappa$ changes from $\sqrt{x \times y}\ $ to $\sqrt{\left(x+\Delta x\right)\left(y + \Delta y\right)}$.  For \eqref{eq:LPconditionCPMM} to hold, there exists $\rho$ such that $\Delta x = \rho \ x$ and $\Delta y= \rho \ y,$ i.e., liquidity provision and removal by LPs in a CPMM is performed in fractions of the pool quantities $x$ and $y,$ and the depth changes from $\kappa$ to $(1+\rho) \ \kappa.$

	The level function $\varphi$ is decreasing, thus the rate in  \eqref{eq:execprice} received by an LT deteriorates as the size of the trade increases.\footnote{The trading function $f$ is increasing in $x$ and $y$ and $\partial_y f(x,y) = \partial_y f\left(\varphi(y),y\right) = 0 \text{ so } \varphi'(y) = - \frac{\partial_y f(x,y) }{\partial_x f(x,y)} < 0.$}  The formulas \eqref{eq:execprice} and \eqref{eq:instantaneousprice} encode all the information needed by an LT to interact with an AMM. For a trade of size $\Delta y$,  the distance between the instantaneous rate $Z$ and the execution rate $\tilde Z\left(\Delta y\right)$ is the \textit{unitary execution cost} of the AMM, specifically
	\begin{align}
		\label{eq:executionCost}
		\textrm{Unitary execution cost} \ = \left|Z - \tilde Z(\Delta y)\right|\,.
	\end{align}
    The unitary execution cost is defined as the average execution cost of one unit of asset $Y$, in units of asset $X$.
	
	Another distinctive characteristic of AMMs is that liquidity taking activity may change the instantaneous rate $Z$, but does not change the depth of the pool. Furthermore, AMMs and LOBs differ in a number of other aspects including accessibility and the way LPs are compensated.  By design, AMMs are permissionless, so anyone can participate in the market.\footnote{This is true for AMMs running on permissionless peer-to-peer networks. If the same AMMs were to be implemented by a central authority for fiat currencies or stocks, then one expects stricter participation rules.} LPs in AMMs are compensated in two ways. One, they are rewarded with the proceeds from fees that LTs pay for every trade. Two, the execution costs in \eqref{eq:executionCost} incurred by LTs due to the convexity of the level function  stay in the pool. In some AMMs, the fee proceeds are put back in the pool, so the number of assets owned by LPs increases. In other AMMs, in particular those with a CL feature, fees are accumulated in a separate account and are earned  by LPs when they withdraw their liquidity from the pool.

	\subsection{Data analysis \label{sec:AMM2}}
	\subsubsection{Data description\label{sec:AMM2_0}}
	
	Uniswap v3 is considered the cornerstone of DeFi and is currently the most liquid AMM. It is a CPMM with trading function
	\begin{align}
		\label{eq:CPMMtradingfunction}
		f(x,y)=x \times y=\kappa^2,
	\end{align}
	so the level function is 
	\begin{align}
		\label{eq:CPMMlevelfunction}
		\varphi(y)= \frac{\kappa^2}{y}\,.
	\end{align}
	When an LT trades $\Delta y$, the execution rate \eqref{eq:execprice} is 
	\begin{equation*}
		\tilde Z(\Delta y) = \frac{1}{\Delta y} \left(\frac{\kappa^2}{y}-\frac{\kappa^2}{y+\Delta y}\right) \ ,
	\end{equation*} 
	and the instantaneous rate \eqref{eq:instantaneousprice} is 
	\begin{equation}\label{eq:insta_price}
		Z = -\varphi'(y) = \frac{\kappa^2}{y^2}\ .
	\end{equation} 
	
	To further study the characteristics of CPMMs and motivate our framework, we use transaction data from Uniswap v3 and the traditional LOB-based exchange Binance. Specifically, we look at the two pairs ETH/USDC and ETH/DAI. The ticker ETH represents the cryptocurrency \emph{Ether}, which is the native cryptocurrency of the Ethereum blockchain. The ticker USDC represents \emph{USD coin}, a cryptocurrency fully backed by U.S. Dollars (USD); and DAI represents the cryptocurrency \emph{Dai}, which tracks parity with the U.S. Dollar. For the pool ETH/USDC, the unit of depth $\kappa$ is $\sqrt{\textrm{ETH} \cdot \textrm{USDC}},$ of $x$ is USDC, of $y$ is ETH, and the instantaneous rate, the execution rate, and the unitary execution cost are all in $\textrm{USDC} / \textrm{ETH}$; similarly for the pool ETH/DAI. For ease of reading, in the remainder of this work we omit the units of $\kappa$.
	
	We analyse transaction data of the most liquid pools for the pairs ETH/USDC and ETH/DAI. Transaction information from decentralised exchanges is public. Between 5 May 2021 and 5 May 2022, there are 1,757,181 LT transactions and 42,403 LP transactions for the pair ETH/USDC, and 101,538 LT transactions and 12,142 LP transactions  for the pair ETH/DAI. Figures \ref{fig:volumeETHUSDC01} and \ref{fig:volumeETHDAI01} show the historical daily volumes of transactions in terms of USD value and in terms of the number of transactions for ETH/USDC and ETH/DAI, respectively. These pools are considered an alternative to LOB-based trading venues such as Binance, which is the most liquid and active venue for both pairs.

	\begin{figure}[H]\centering
		\includegraphics{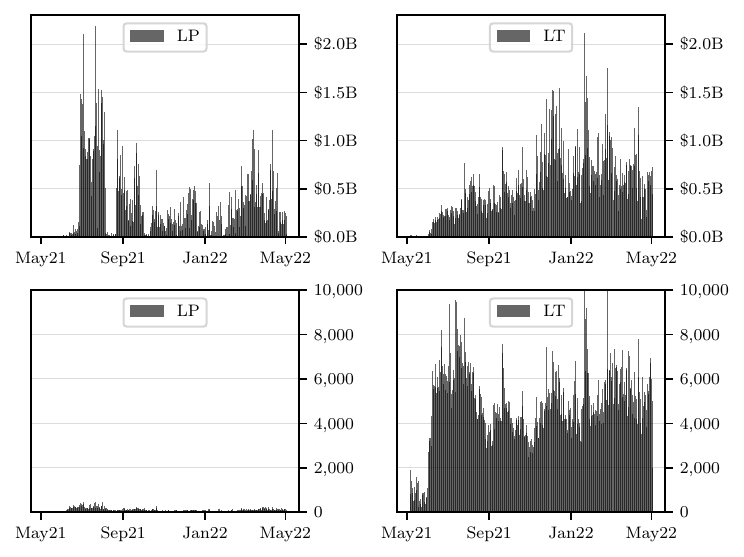}\\
		\caption{Daily transactions for the ETH/USDC pool (B is shorthand for billions), between 5 May 2021 to 5 May 2022. \textbf{Top left}: LP transaction volume in USD. \textbf{Top right}: LT transaction volume in USD. 
			\textbf{Bottom left}: number of LP transactions. \textbf{Bottom right}: number of LT transactions. }
		\label{fig:volumeETHUSDC01}
	\end{figure}
	
		
		

	\begin{figure}[H]\centering
		\includegraphics{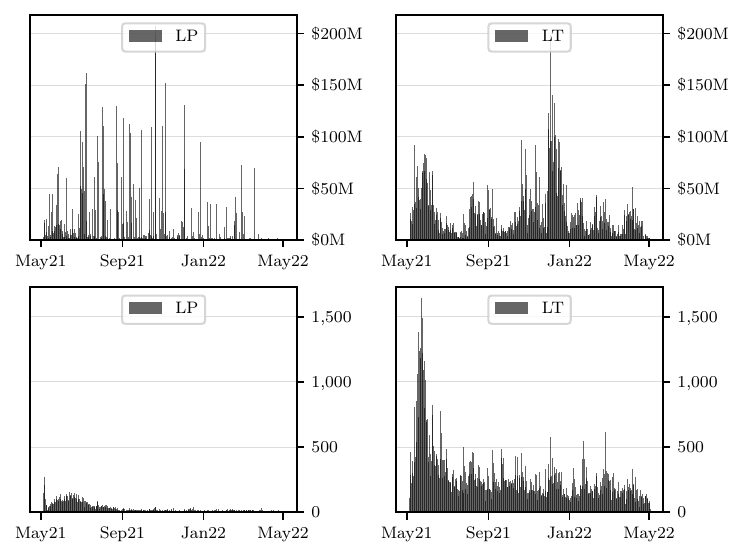}\\
		\caption{Daily transactions for the ETH/DAI pool (M is shorthand for millions), between 5 May 2021 to 5 May 2022. \textbf{Top left}: LP transaction volume in USD. \textbf{Top right}: LT transaction volume in USD. 
			\textbf{Bottom left}: number of LP transactions. \textbf{Bottom right}: number of LT transactions.}
		\label{fig:volumeETHDAI01}
	\end{figure}
	
	
	Next, we study liquidity provision and taking activity within the pools and examine their key characteristics, all of which we use to frame the agent's execution problem.
	
	
	\subsubsection{Rates and liquidity dynamics\label{sec:AMM2_1}}
	Compared with most CPMMs, Uniswap v3 operates with the concentrated liquidity (CL) feature. In CPMMs without this feature, each LP owns a percentage of the pool, and the fees paid by LTs are distributed to LPs in the same proportion; thus, LPs provide liquidity at all feasible rates. In contrast, LPs in Uniswap v3 specify the range of rates where they supply liquidity. For example, LPs can target a range around the instantaneous rate to earn more fees than those who provide liquidity at rates far from the instantaneous rate. In practice, the continuous space of possible rates is discretised and subdivided in rate intervals whose boundaries are called \emph{ticks}.\footnote{In Uniswap v3, ticks are specific rates that are used as the boundaries of an LP transaction. In contrast, ticks in LOBs represent the smallest price increment.} Two consecutive ticks define a \emph{tick range} and the rate can take values in this range with increments set by the AMM.  LPs designate two ticks between which they wish to provide liquidity. Therefore,  with CL, the pool is characterised by the available quantities in tick ranges. Finally, in Uniswap v3, fees paid by LTs are distributed among the LPs who had provided liquidity in a range that included the rate at which liquidity was taken.
	
	In CPMMs without CL, the value of $\kappa$ is the same for all the tick ranges and recall that the value of $\kappa$ can only change when LPs deposit or withdraw liquidity from the pool.   On the other hand, due to the CL feature of Uniswap v3, the depth of the pool may be different when the instantaneous rate crosses the boundary of a tick because liquidity is distributed unevenly across rates in the pool. Thus, to estimate the execution cost incurred when trading, one must track the distribution of liquidity across ranges of rates. Finally, in the extreme case where all liquidity is withdrawn from a range of rates around the current rate $Z$, this effectively constitutes a change in the instantaneous rate of the pool -- the next LT transaction will start at a rate different from $Z$.
	
	
	To analyse liquidity provision and consumption activity, we reconstruct the supply and demand of liquidity since the inception of Uniswap v3 in May 2021. For example, Figure \ref{fig:uniswap01} shows the amount of liquidity, given by the depth $\kappa_i,$ available at each rate range $i$ at 02:00, 12:00, and 18:00, on 14 April 2022 and 15 April 2022 for the ETH/USDC pool.
	
	

	\begin{figure}[H]
		\centering
		\subfloat[\textbf{Top}: ETH/USDC rates on 14 April 2022 sampled every 60 seconds. \textbf{Other plots}: Pool depth $\kappa$ at 02:00 (top), 10:00 (middle), and 18:00 (bottom). The red bar corresponds to the instantaneous rate.]{{\includegraphics{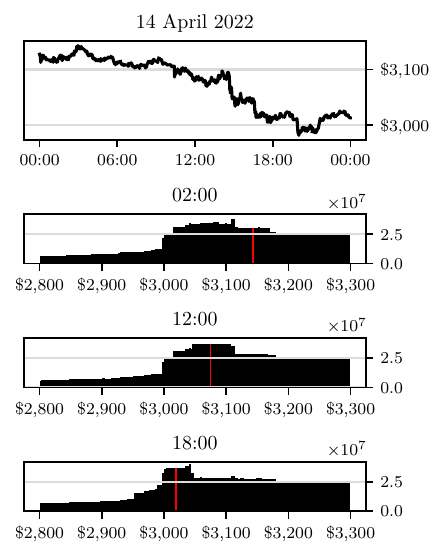} }}%
		\qquad
		\subfloat[\textbf{Top}: ETH/USDC rates on 15 April 2022 sampled every 60 seconds. \textbf{Other plots}: Pool depth $\kappa$ at 02:00 (top), 10:00 (middle), and 18:00 (bottom). The red bar corresponds to the instantaneous rate.]{{\includegraphics{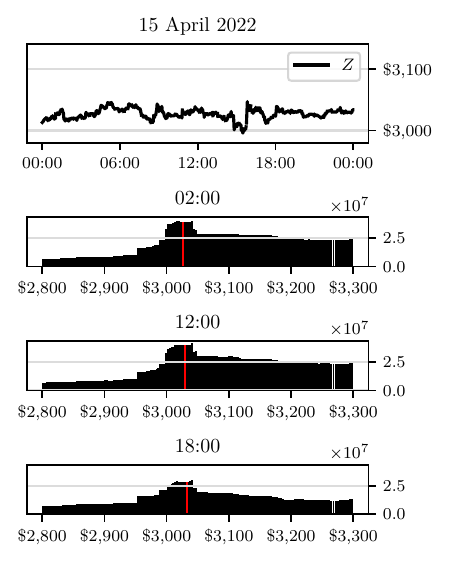} }}%
		\caption{Rate and LP dynamics in AMMs.}%
		\label{fig:uniswap01}%
	\end{figure}
	
	The bars in the panels of Figure \ref{fig:uniswap01} represent a tick range. Whenever the instantaneous rate is in a given range, the liquidity of that range defines the depth $\kappa$ used by the trading function, and hence the execution rate \eqref{eq:execprice}. When the volume of an LT transaction is large enough to make the instantaneous rate cross a tick where the level of liquidity changes, the AMM treats it as multiple trades, each with a different $\kappa_i$; similar to a market order walking the levels of an LOB. Tracking the instantaneous rate, and the liquidity around it, is critical for an agent interacting with an AMM. In Figure \ref{fig:uniswap01}, most of the liquidity is concentrated around the instantaneous rate and the depth is the same over a large range around the rate. The behaviour of LPs is different when volatility is high or low. When $Z$ is volatile, such as on 14 April on the left panel of Figure \ref{fig:uniswap01}, LPs provide liquidity over a wider range  around the rate  to earn fees that anticipate  large swings in $Z$. On the other hand, when markets are less volatile, such as on 15 April in the right panel of Figure \ref{fig:uniswap01}, LPs concentrate their liquidity provision more tightly around the rate $Z$. 
	
	
	
	Next, we use transaction data from the LOB-based exchange Binance for the same pairs to compare the dynamics in the two trading venues. Figure \ref{fig:uniswapbinance01} shows the Binance quoted rate and the AMM rate, and liquidity provision activity for both pools between 09:00 and 12:00 on 13 April 2022.
	
	\begin{figure}[H]
		\centering
		\subfloat[\textbf{Top}: ETH/DAI instantaneous rate $Z$ and oracle $S$ between April 13 00:00 and April 13, 2022 12:00. \textbf{Middle}: LT Transaction size. \textbf{Bottom}: LP transaction size.]{{\includegraphics{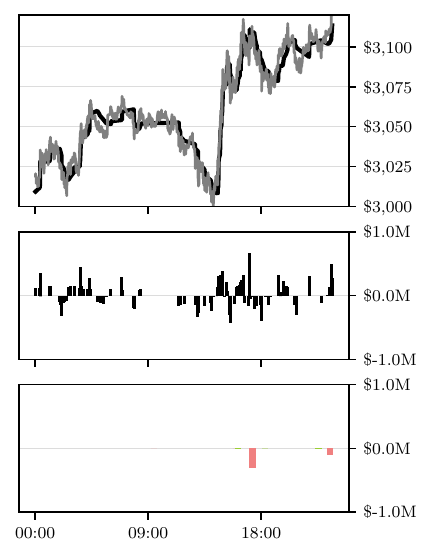} }}%
		\qquad
		\subfloat[\textbf{Top}: ETH/USDC Instantaneous rate $Z$ and oracle $S$ between April 13 09:00 and April 13, 2022 12:00. \textbf{Middle}: LT Transaction size. \textbf{Bottom}: LP transaction size.]{{\includegraphics{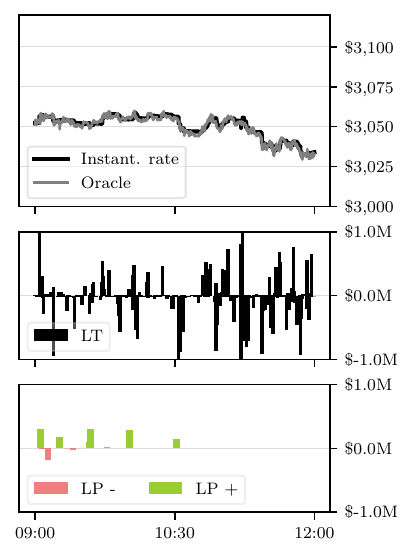} }}%
		\caption{Rate and LP dynamics in AMMs.}%
		\label{fig:uniswapbinance01}%
	\end{figure}

	Figure \ref{fig:uniswapbinance01} shows that Binance rates are more volatile and lead the rates in the AMM. This is not by design, it is a consequence of the higher liquidity in Binance. Currently, it is crucial to consider the rate from a more liquid venue when trading in an AMM. In what follows, we call the leading exchange rate from another trading venue the \emph{oracle}, which in this case is the Binance quoted rate. Similar to Figures \ref{fig:volumeETHUSDC01} and \ref{fig:volumeETHDAI01}, Figure \ref{fig:uniswapbinance01} shows that there is more LT than LP activity measured by the frequency of instructions and the size of the orders. During the periods with little trading activity in the ETH/DAI pool,  the oracle rate plays a central role to attract LT  activity in the AMM whenever the difference between the instantaneous rate and the oracle rate is significant (recall that only liquidity taking trades can change the rate of the pool). The widening of the difference between the two rates triggers LT activity which drives the two exchange rates to converge; i.e., arbitrageurs keep markets in check.

	\subsubsection{Convexity and execution costs\label{sec:AMM2_2}}
	
	Here, we analyse execution costs implied by the pool reserves and the trading function \eqref{eq:CPMMtradingfunction}.   Figure \ref{fig:convexity0} shows the historical unitary execution costs as a function of the transaction size $\Delta y$ and the liquidity for all transactions between 1 March and 31 March 2022 in the ETH/USDC pool. The figure  shows that for the same transaction size, it is cheaper to trade in a more liquid pool; clearly, as liquidity increases (higher value of $\kappa$), execution costs decrease.   The figure shows two unitary execution cost curves, i.e., $\Delta y \mapsto \left|Z - \left(\varphi(y+\Delta y) - \varphi(y)\right) \right|,$ corresponding to the same depth $\kappa \,=\, 3 \times 10^7.$ One curve assumes $Z=\,$2,975 or $y=\,$500,000, and the other assumes $Z=\,$3,600 or $y=\,$550,000.

	\begin{figure}[H]\centering
		\includegraphics{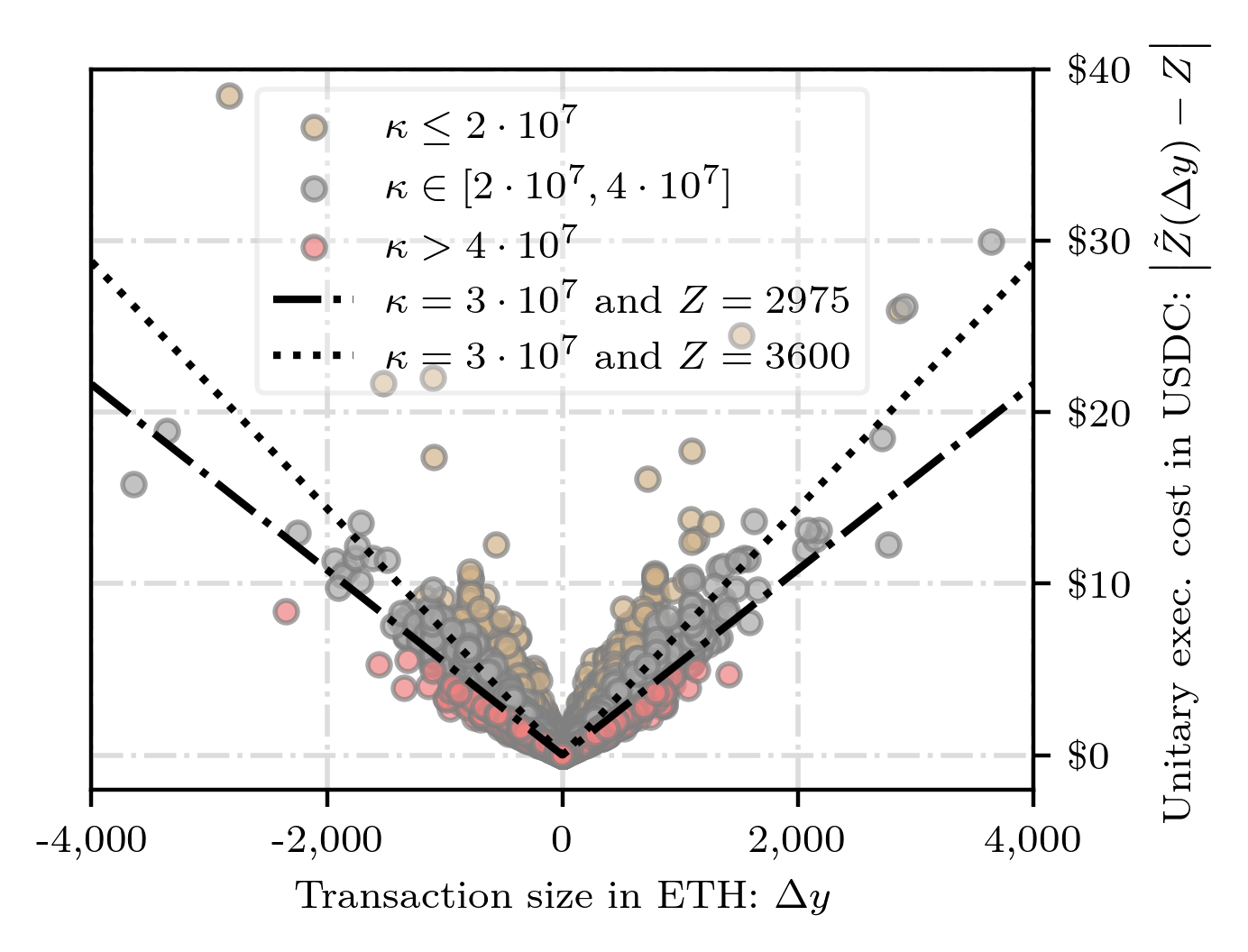}\\
		\caption{Unitary execution costs for all transactions between 1 March and 31 March, 2022 in the ETH/USDC pool. Negative transaction sizes corresponds to buying ETH, and positive sizes to selling ETH. The unitary execution cost is defined in \eqref{eq:executionCost} and is given for the following ranges of the depth $\kappa:$ $[0,2\times10^7], \ [2\times10^7,4\times10^7],$ and $[4\times10^7, +\infty].$}\label{fig:convexity0}
	\end{figure}

	
	Next, we analyse the geometry of the constant product trading function to understand how unitary execution costs relate to the depth of the pool and the instantaneous rate $Z$. Figure \ref{fig:geometry} shows the level function $\varphi$ for $\kappa = \,$2,500,000. Every point on the curve gives possible values for $x$ and $y$ that result in the same pool depth. The point $O$ corresponds to the current pool quantities $x$ and $y$. The slope of the tangent at that point gives the current instantaneous rate $Z = -\varphi'(y).$ 
	

	A change from $O$ to $A$ in Figure \ref{fig:geometry} is the result of an LT selling 2,500 ETH. A change from $O$ to $B$ is the result of an LT buying 2,500 ETH. The new rates, after these transactions, are given by the slopes at the new points $A$ and $B$, respectively. When an LT sells $\Delta y=$2,500 ETH, the unitary execution rate $\Delta x  / \Delta y$ is given by the slope of the line (OA). Similarly, the slope of the line (OB) gives the unitary rate for buying $\Delta y.$ On the other hand, the unitary execution cost is given by the difference between the slope of the lines and the slope of the tangent at point $O;$ the magnitude of this difference depends on the curvature of $\varphi$ in the neighbourhood of $O.$ This curvature is proportional to the convexity of the level function and can be approximated by the second-order Taylor polynomial $\tfrac12 \, \varphi ''(y_\textrm{O})\, \Delta y^2.$ A higher degree of convexity, i.e., more curvature around point $O,$ does not change the slope of the tangent at point $O,$ but changes the slopes of the lines $(OA)$ and $(OB).$ The convexity of the level function is given by
	\begin{equation}\label{eqn: convexity varphi prime prime}
		\varphi ''(y) = \frac{2\,\kappa^2}{y^3} = \frac{2\,Z^{3 / 2}}{\kappa}\, .
	\end{equation}
	Therefore, the execution rate obtained for buying or selling $\Delta y$ is always less advantageous than the instantaneous rate $Z$ because the level function $\varphi$ is convex. Clearly, as the depth $\kappa$ increases, the convexity of the level function is less pronounced. For orders of ``small size''  one can approximate the unitary execution cost in \eqref{eq:executionCost} with 
	\begin{align}
		\label{eq:execcostsApproxCPMM}
		| Z- \tilde Z(\Delta y)|=   \frac{1}{\kappa}\,Z^{3/2}\, |\Delta y|\,,
	\end{align}
	or equivalently approximate the execution rate with 
	\begin{align}
		\label{eq:execratesApproxCPMM}
		\tilde Z(\Delta y) = Z - \frac{1}{\kappa}\,Z^{3/2}\, \Delta y\,.
	\end{align}

	{\footnotesize
		\begin{figure}
			\footnotesize
			\parbox{.45\linewidth}{
				\centering
				\includegraphics[width=0.53\textwidth]{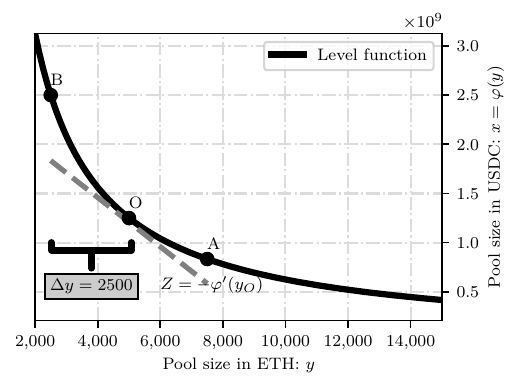}\\
				\caption{Geometry of the constant product trading function. The figure shows the function $\varphi\left(y\right) = x$ where $x$ is the quantity of USDC, and $y$ is the quantity of ETH in the pool.}
				\label{fig:geometry}
			}
			\hfill
			\parbox{.45\linewidth}{
				\centering
				\vspace{2.5mm}
				\includegraphics[width=0.51\textwidth]{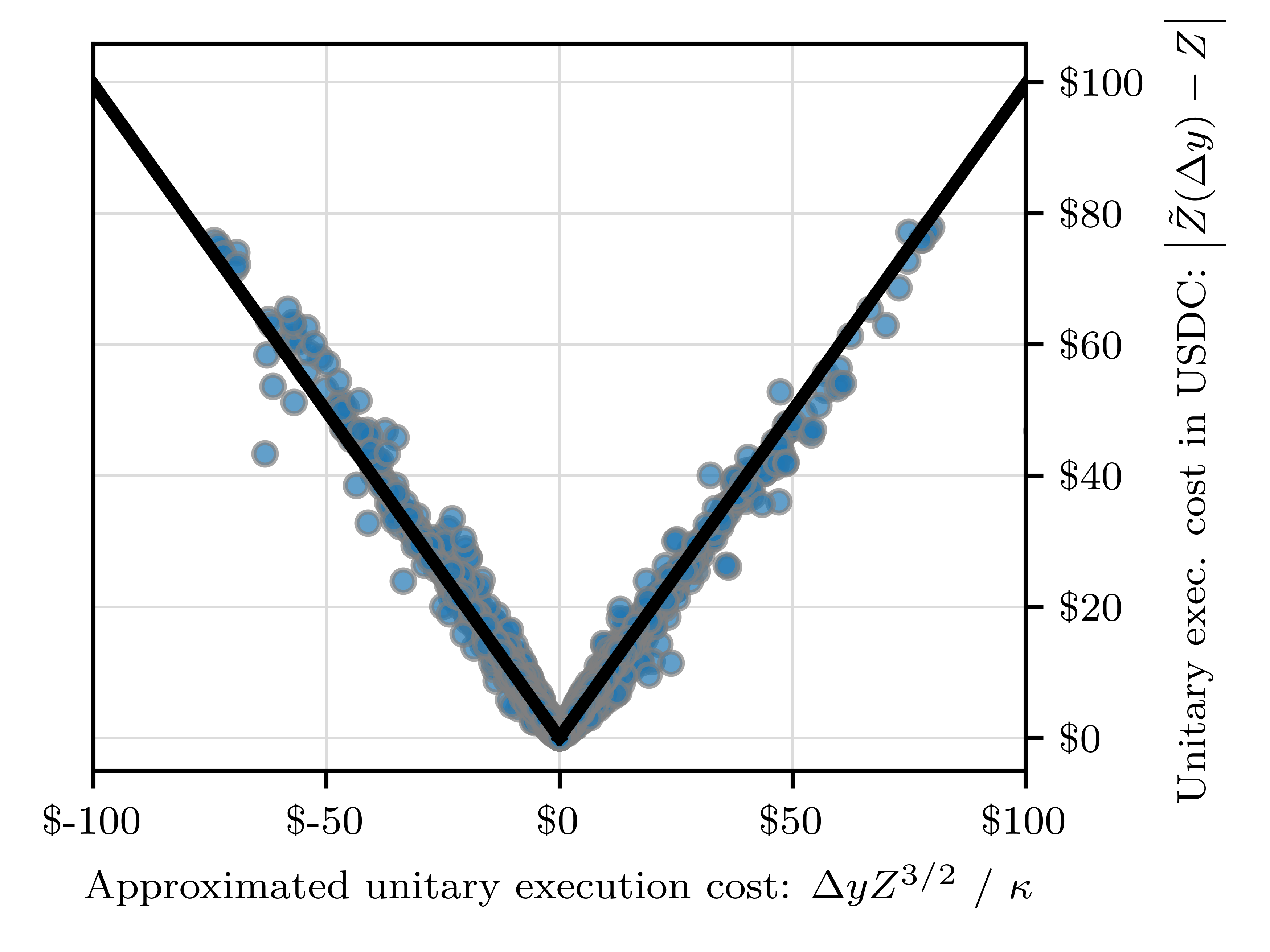}\\ 
				\vspace{-1mm}
				\caption{Scatter plot of transaction costs and the approximation \eqref{eq:execcostsApproxCPMM}
					for all transactions in the ETH/USDC pool between 5 May 2021 and 5 May 2022.}\label{fig:convexity1}
			}
		\end{figure}
	}
	
	To show the extent to which the convexity of the level function captures the execution cost, we use transaction data from the ETH/USDC pool from inception on 5 May 2021 until 5 May 2022, approximately 1,757,000 LT transactions, to compare unitary execution costs incurred by LTs with our approximation in \eqref{eq:execcostsApproxCPMM}. Figure \ref{fig:convexity1} shows a scatter plot of the observed unitary execution costs from transaction data and the approximation $\left(Z^{3/2} /  \kappa\right) \Delta y.$ Recall that negative values of $\Delta y$ are buy orders and positive values are sell orders. The figure shows that the expression in \eqref{eq:execcostsApproxCPMM} is a good approximation, and that the convexity of the level function $\varphi$ can be used to model the unitary execution cost incurred when interacting in CPMMs. 
	

	In  an optimal trading framework, the LT controls the speed $\nu_t$ at which she sends orders to the AMM. Now, assume that the LT trades the quantity $\Delta y_t = \nu_t\, \Delta t$, where  $\Delta t$ is a fixed time-step that determines the LT's frequency of trading and  $\nu_t$ is fixed during the time interval. The execution rate for $\Delta y$ is $\tilde Z(\Delta y) = Z - Z^{3/2}\,\nu_t \, \Delta t\,/\,\kappa$. Thus, to reflect the LT's pace of trading we write the LT's execution rate as 
	\begin{align}
		\label{eq:execratesApproxCPMM_nu}
		\tilde Z_t = Z_t - \frac{\eta}{\kappa}\,Z^{3/2}_t\,\nu_t \, \,,
	\end{align}
	where the parameter $\eta$ scales the execution costs according to the LT's trading frequency $\Delta t$. Note that the continuous-time formulation of the problem in \eqref{eq:execratesApproxCPMM_nu} is used in the algorithmic trading literature to solve optimal execution problems, which are by nature discrete.

	Due to the convexity of the level function, it is sub-optimal to execute large orders in one trade.\footnote{This is a consequence of the trading activity of other market participants.} An optimal trading framework, similar to those developed for traditional LOB-based markets, should balance the trade-off between execution costs and rate risk. The execution cost in CFMMs is similar to the cost of ``walking the book'' when trading in LOBs, sometimes referred to as the temporary price impact. The difference between the temporary price impact and the CFMM execution cost is that in the CFMM we have a deterministic closed-form expression for the execution cost as a function of the depth and the rate, both of which can be employed by LTs to estimate execution costs. On the other hand, in LOBs, traders usually rely on historical data analysis and assumptions to obtain an estimate of the execution costs. In LOBs, it is generally assumed that temporary impact is a linear function of the speed of trading where the slope of the function is assumed to be fixed; see \cite{cartea2015book} and \cite{gueant2016book}.

	As in LOBs, there is uncertainty in execution costs received by LTs.  On blockchains, rate and liquidity updates occur at the block validation frequency. For instance, in the \emph{Ethereum} network hosting the most popular AMMs, a block is validated every 13 seconds, on average. Hence, the most up-to-date information available is the previous block's rate and liquidity data. Transactions are grouped and executed in a block which is externally validated.  Within the block, the transactions form a queue that determines the priority with  which  they are executed. Market participants can pay \emph{tips}  to gain priority in the block.  Finally, every transaction sent on the network pays \emph{gas fees}. Gas is a unit that measures computational effort of network transactions, and all market participants pay gas fees to the network -- fees are proportional to the algorithmic complexity of the functions a trader invokes in the AMM.  Consequently, there is randomness regarding the exact execution costs.

	
	Below, in Section \ref{sec:Model}, we propose a trading model (Model I) in AMMs where execution costs are given by the convexity of the level function, where the depth $\kappa$ does not change over the agent's trading window, and where it is assumed that there is an oracle rate. Model I introduces a method that uses piecewise constant execution costs to obtain a closed-form approximation strategy that adjusts to the convexity of the level function. We use Model I in our performance analysis of Section \ref{sec:Performance}. In the future, an oracle rate  may become less relevant because activity in  AMMs will increase and significant rate discrepancies will very seldom appear. Thus, in Section \ref{sec:Model2}  we propose Model II for liquid AMMs where the depth $\kappa$ of the pool is stochastic and where an oracle rate is redundant.

	\section{Model I: optimal trading with an oracle rate\label{sec:Model}}
	We consider the problem of an LT who trades in a CPMM and wishes  to exchange a large position in asset $Y$ into asset $X$ or to execute a statistical arbitrage in the pair. In both cases, the LT uses rate information from the pool in the CPMM and from another more liquid exchange in which the oracle rate $S$ is the price of $Y$ in terms of that of $X$. The LT acknowledges that rate formation does not take place in the AMMs.
	
	The depth $\kappa$ of the pool is assumed to be constant during the execution window $[0,T]$, where $T>0$\,. Despite intraday changes in the depth of the pool, this assumption does not have a material effect on the LT's execution problem. Currently, the activity and liquidity in AMMs is such that the depth $\kappa$ is constant for periods of time which are longer than the trading horizon $T$ considered by the LT.
	\footnote{Clearly, the LT cannot know when a change in liquidity of the pool occurs. Thus, there is the possibility that within the trading window of the execution programme the value of $\kappa$ changes.}  
	
	In the future, activity in AMMs is expected to increase, and so will the informational content of the rates implied by the pool of AMMs. An increase in activity of the pool would affect our modelling choices in two ways. One, the innovations in depth $\kappa$ will occur more often, so the value of $\kappa$ cannot be regarded as constant throughout the execution window. Two, the  oracle rate becomes redundant because the rates in the AMMs become efficient, so they incorporate all the information available to market participants -- i.e., the discrepancies with rates and prices in other trading venues are negligible and economically insignificant. Therefore, below in  Section \ref{sec:Model2} we propose Model II to solve the LT's execution problem with stochastic dynamics for the depth $\kappa$ and without an external oracle rate. 
	
	We proceed with the setup of Model I, for which  we fix a filtered probability space \newline $\left(\Omega, \mathcal F, \P; \F = (\mathcal F_t)_{t \in [0,T]} \right)$ satisfying the usual conditions, where $\F$ is the natural filtration generated by the collection of observable stochastic processes that we define below.
	
	The LT must liquidate a position $\tilde{y}_0 $ in asset $Y$ over the period of time $[0,T]$, and her wealth is valued in terms of asset $X$. We introduce the progressively measurable oracle rate process $(S_t)_{t \in [0,T]}$ with dynamics 
	\begin{align} \label{eq:SProcess}
		dS_t = \sigma\, S_t\, dW_t\,, 
	\end{align} where the volatility parameter $\sigma$ is a nonnegative constant, and $(W_t)_{t \in [0,T]}$ is a standard Brownian motion.
	
	As discussed above, when an LT sends a trade to the AMM, the rate impact received by the order is encoded in the trading function $f(x,y)$. As in traditional models for optimal execution (see e.g., \cite{cartea2015book} and  \cite{gueant2016book}), the rate impact received by the LT's trade is a function of the trading speed. In AMMs, rate impact is a function of the trading speed and of the rate $Z$ and the depth $\kappa$. Specifically, we write the difference between the execution rate $\tilde Z$  and the instantaneous rate $Z$ as in \eqref{eq:execratesApproxCPMM_nu}, i.e.,
	\begin{equation}\label{eqn:price impact}
		\tilde Z - Z= - \frac{\eta}{\kappa}\, Z^{3/2}\,\nu\,.
	\end{equation}
	
	The term on the right-hand side of \eqref{eqn:price impact} is the rate impact function of the AMM, which is determined by the convexity of the level function $\varphi$, see \eqref{eq:execcostsApproxCPMM}, and depends on the instantaneous rate $Z_t$ of the pool at the time the liquidity taking order is executed. The key difference between the functional form of the execution costs in \eqref{eqn:price impact} and those in the equity LOB literature is that in general, the price impact functions proposed for LOB models do not depend on the price of the asset, see \cite{cartea2015book}. 

    The LT trades at the speed $(\nu_t)_{t \in [0,T]}$, so her inventory $(\tilde y_t)_{t \in [0,T]}$ evolves as 
	\begin{align} \label{eq:ytildeProcess_modelI}
		d\tilde y_t = -\nu_t \,dt\, ,
	\end{align}
	where, for simplicity, trading fees are zero. We do not restrict the speed in \eqref{eq:ytildeProcess_modelI} to be positive; if $\nu>0$ the LT sells the asset and if $\nu<0$ the LT buys the asset. When the initial inventory is $\tilde{y}_0>0$ (resp. $\tilde{y}_0=0$) the LT executes a liquidation (resp. speculation) programme.
	
	When an LT trades at speed $\nu$, the quantity of $Y$ swapped at every instant in time is given by $\nu\,dt$, so the dynamics of her holdings in asset $X$ are given by 
	\begin{equation}\label{eq:wealth}
		\begin{split}
			d\tilde x_t &= \tilde Z_t\, \nu_t\,dt=\left(Z_t-\frac{\eta}{\kappa} \, Z_{t}^{3/2} \,\nu_{t}\right)\,\nu_t\,dt\, .
		\end{split}
	\end{equation}
	Here, the execution cost is stochastic and its dynamics are known; see \cite{barger2019optimal} and \cite{fouque2021optimal}  for similar frameworks where the price impact of trading is stochastic.
	
	In the absence of market frictions, continuous arbitrage between the oracle rate and the rate quoted in the AMM would  make rates converge so that $S_t = Z_t$ at any time $t$. However, exchanges and AMMs are not frictionless, and the oracle rate is the most efficient rate (i.e., reflects all the information available) so we consider the following dynamics for $\left(Z_t\right)_{t\in[0,T]}:$ 
	\begin{align}
		\label{eq:ZProcess}
		dZ_{t}=&\beta\left(S_{t}-Z_{t}\right)dt+\gamma\,Z_{t}\,dB_{t}\ ,
	\end{align}
	where $\beta > 0$ is the mean-reverting parameter, $\gamma > 0$ quantifies the dispersion induced by the trading flow that drives the AMM instantaneous rate $Z_t$ away from the oracle rate $S_t,$ and $(B_t)_{t \in [0,T]}$ is a standard Brownian motion independent of $(W_t)_{t \in [0,T]}.$

	
	\subsection{Optimal trading strategy: execution and statistical arbitrage\label{sec:model:optimal}}

	The LT maximises her expected terminal wealth in units of $X$ while penalising inventory in $Y$. The set of admissible strategies is
	\begin{align}
		\label{def:admissibleset_t_modelI}
		\mathcal A_t = \left\{ (\nu_s)_{s \in [t,T]},\ \R\textrm{-valued},\ \F\textrm{-adapted, and } \int_t^T |\nu_s|^2 \,ds < +\infty, \ \ \P\textrm{-a.s.}  \right\}\,.
	\end{align}
	We write $\mathcal A := \mathcal A_0$ and let $\nu\in\Ac$. The performance criterion of the LT, who trades at speed $\nu$, is the function $u^{\nu}\colon[0,T] \times \R \times \R \times \R_{++} \times\R_{++} \rightarrow \R\ $ given by
	\begin{align}
		\label{eq:perfcriteria_modelI}
		u^{\nu}(t,\tilde{x},\tilde{y},Z,S)=\E_{t,\tilde{x},\tilde{y},Z,S}\left[\tilde{x}_{T}^{\nu}+\tilde{y}_{T}^{\nu}\,Z_T-\alpha\,\left(\tilde{y}_{T}^{\nu}\right)^{2}-\phi\,\int_{t}^{T}\left(\tilde{y}_{s}^{\nu}\right)^{2}\,ds\right]\,,
	\end{align}
	and the LT's value function $u:[0,T] \times \R \times \R \times \R_{++} \times\R_{++} \rightarrow \R$ is
	\begin{align}
		\label{eq:valuefunc_modelI}
		u(t,\tilde{x},\tilde{y},Z,S)=\underset{\nu\in\mathcal{A}}{\sup} \, u^{\nu}\left(t,\tilde{x},\tilde{y},Z,S\right)\,.
	\end{align}
	
	The first term on the right-hand side of \eqref{eq:perfcriteria_modelI} represents the LT's holdings in asset $X$ at the end of the trading window. The second term represents the LT's earnings from liquidating her remaining inventory at the final time $T$ at rate $Z_T$. The third term represents an inventory penalty where the parameter $\alpha > 0$ quantifies the aversion of the LT to holding non-zero inventory at time $T$; the units of $\phi$ are such that the penalty is in units of $X$. Finally, the last term on the right-hand side of \eqref{eq:perfcriteria_modelI} represents a running inventory penalty where the parameter $\phi \geq 0$ quantifies the urgency of the LT to liquidate inventory; the units of $\alpha$ are such that the penalty is in units of $X$. Note that the third term on the right-hand side of \eqref{eq:perfcriteria_modelI} is sometimes interpreted in the algorithmic trading literature as the `cost' of liquidating the final inventory $\tilde y_T$ at time $T.$ Using this interpretation, and in light of the results of the previous section, this term should be proportional to $Z_T^{3/2}$. Here we interpret this term as penalisation term and we observe that for $\alpha$ sufficiently large, $\alpha\,\left(\tilde{y}_{T}^{\nu}\right)^{2}$ is greater than expected execution costs of terminal inventory.
	
	In Appendix \ref{sec:APX:model}, we use the tools of stochastic control to study the optimisation problem in \eqref{eq:perfcriteria_modelI}. The functional form of the execution costs leads to a system of PDEs, one of which is a semilinear PDE which we cannot  solve in closed-form; see \eqref{eq:ANX:system}. The optimal trading speed in feedback form is a function of the solution to the semilinear PDE, see \eqref{eq:ANX:optimalspeed}, and one can compute the optimal trading speed with a numerical scheme. However, in our case, the numerical scheme is computationally expensive because the semilinear PDE requires a thin grid and a linearisation iterative method at each time step to transform the nonlinear problem into a sequence of linear problems. 
	
	We refer to the optimal strategy obtained with a numerical scheme as the numerical approximation strategy. Clearly, the numerical approximation strategy takes too long both to compute and to implement by the LT in real time. In practice, the profitability of execution and statistical arbitrage strategies relies on computing the strategy and instructing the AMM within very short periods of time (e.g., milliseconds).  Thus, because speed is paramount for LTs, below we derive a strategy in closed-form  that can be deployed by the LT in real time. We refer to this strategy as the closed-form approximation strategy. Finally, to assess the precision of the closed-form approximation strategy, Subsection \ref{sec:model:comparison} compares it with the numerical approximation strategy \eqref{eq:ANX:optimalspeed} derived in \ref{sec:APX:model}.
	
	\subsection{Constant impact parameter strategy}
	
	
	
	To obtain a trading scheme that can be implemented by the  LT in real time, we first derive a strategy where the impact parameter of the execution cost is deterministic. We use this  strategy as the building block for the LT's closed-form approximation strategy we derive below. Accordingly, we write the execution cost in \eqref{eqn:price impact} as
	\begin{equation}\label{eq:impact_deterministic}
		\tilde Z - Z= - \eta\,\zeta\,\nu\,,
	\end{equation}
	where $\zeta > 0$ is the impact parameter and recall that the value of $\eta$ depends on the LT's trading frequency. With fixed executions costs,  the LT can derive a closed-form optimal trading strategy $\left(\nu^{\star,\zeta}_t\right)_{t \in [0,T]}$ for a given value of the parameter $\zeta$.
	
	For each $\zeta$ the set of admissible strategies is 
	\begin{align}
		\label{def:admissibleset_t}
		\mathcal A_t^{\zeta} = \left\{ (\nu_s)_{s \in [t,T]},\ \R\textrm{-valued},\ \mathbbm F\textrm{-adapted, and } \int_t^T |\nu_s|^2 \,ds < +\infty, \ \ \mathbbm P\textrm{-a.s.}  \right\},    
	\end{align}
	and we write $\mathcal{A}^{\zeta}:= \mathcal A_0^{\zeta}.$ We consider an optimal trading problem in which the LT trades at speed $\left(\nu_t^{\zeta}\right)_{t \in [0,T]}$, so the inventory $\left(\tilde y_t^{\zeta}\right)_{t \in [0,T]}$ evolves as
	\begin{align}
		\label{eq:ytildeProcessjN}
		d\tilde y^{\zeta}_t = -\nu^{\zeta}_t \,dt\, ,
	\end{align}
	where we do not restrict the speed to be positive, and recall that trading fees are zero.
	
	When an LT trades at speed $\nu^{\zeta}$, the quantity of $Y$ swapped at every instant in time is given by $\nu^{\zeta}\,dt$, so the dynamics of the LT's holdings in asset $X$ are given by
	\begin{equation}\label{eq:wealth_AC}
		\begin{split}
			d\tilde x_t^{\zeta} &= \tilde Z_t\, \nu_t^{\zeta}\,dt 
			=\left(Z_t-\eta \, \zeta \,\nu^{\zeta}_{t}\right)\,\nu^{\zeta}_t\,dt\, .
		\end{split}
	\end{equation}
	
	Let $\nu^{\zeta}\in\Ac^{\zeta}$. The performance criterion of the LT, who trades at speed $\nu^{\zeta}$\,, is a function $u^{\nu^{\zeta}}\colon[0,T] \times \R^4 \rightarrow \R\ $ given by  
	\begin{align}
		\label{eq:perfcriteria}
		u^{\nu^{\zeta}}(t,\tilde{x},\tilde{y},Z,S)=\E_{t,\tilde{x},\tilde{y},Z,S}\left[\tilde{x}_{T}^{\nu^\zeta}+\tilde{y}_{T}^{\nu^\zeta}\,Z_T-\alpha\,\left(\tilde{y}_{T}^{\nu^\zeta}\right)^{2}-\phi\,\int_{t}^{T}\left(\tilde{y}_{s}^{\nu^\zeta}\right)^{2}\,ds\right]\ ,
	\end{align}
	and the value function $u_{\zeta}\colon[0,T] \times \R^4\rightarrow \R$ of the LT is given by 
	\begin{align}
		\label{eq:valuefunc}
		u_{\zeta}(t,\tilde{x},\tilde{y},Z,S)=\underset{\nu^{\zeta}\in\mathcal{A}^{\zeta}}{\sup}\,u^{\nu^{\zeta}}(t,\tilde{x},\tilde{y},Z,S)\ .
	\end{align}
	
	Next, we solve the above dynamic optimisation problem for each value of $\zeta>0$ and derive the optimal liquidation strategy. The value function \eqref{eq:valuefunc} is the unique classical solution to the Hamilton--Jacobi--Bellman (HJB) equation 
	\begin{equation}\label{eq:hjbu}
		\begin{split}
			0=\,&\partial_{t}w_{\zeta}-\phi\,\tilde{y}^{2}+\beta\,\left(S-Z\right)\,\partial_{Z}w_{\zeta}+\frac{1}{2}\,\gamma^{2}\,Z^{2}\,\partial_{ZZ}w_{\zeta}+\frac{1}{2}\,\sigma^{2}\,S^{2}\,\partial_{SS}w_{\zeta}\\
			&+\sup_{\nu\in\R}\left(\left(\nu\, Z-\eta\,\zeta\,\nu^{2}\right)\partial_{\tilde{x}}w_{\zeta}-\nu\,\partial_{\tilde{y}}w_{\zeta}\right)\, ,
		\end{split}
	\end{equation}
	with terminal condition 
	\begin{align}
		\label{eq:termcondu}
		w_{\zeta}(T,\tilde{x},\tilde{y},Z,S)=\tilde{x}+\tilde{y}\,Z-\alpha\,\tilde{y}^{2}\ .
	\end{align}
	
	The form of the terminal condition \eqref{eq:termcondu} suggests the ansatz 
	\begin{align}
		\label{eq:ansatz1}
		w_{\zeta}(t,\tilde{x},\tilde{y},Z,S)=\tilde{x}+\tilde{y}\,Z+\theta_{\zeta}(t,\tilde{y},Z,S)\,,
	\end{align}
	where the first two terms on the right-hand side represent the mark-to-market value of the LT's holdings at the instantaneous rate of the AMM. The last term represents the additional value obtained by the LT when following the optimal strategy. The ansatz  is justified by the following proposition, for which a proof is straightforward.
	\begin{prop}
		Assume there exists a function $\theta_{\zeta}\in C^{1,1,2,2}\left([0,T] \times \R^3\right)$ which solves 
		\begin{equation}\label{eq:hjbtheta1}
			\begin{split}
				0=\,&\partial_{t}\theta_{\zeta}-\phi\,\tilde{y}^{2}+\beta\,\left(S-Z\right)\,\left(\tilde{y}+\partial_{Z}\theta_{\zeta}\right)+\frac{1}{2}\,\gamma^{2}\,Z^{2}\,\partial_{ZZ}\theta_{\zeta}+\frac{1}{2}\,\sigma^{2}\,S^{2}\,\partial_{SS}\theta_{\zeta}\\
				&+\sup_{\nu\in\R}\left(-\eta\,\zeta\,\nu^{2}-\nu\,\partial_{\tilde{y}}\theta_{\zeta}\right)\ ,
			\end{split}
		\end{equation}
		on $[0,T)\times \R^3$, with terminal condition 
		\begin{equation}\label{eq:terminal1}
			\theta_{\zeta}(T,\tilde{y},Z,S)= -\alpha\,\tilde{y}^2\ .
		\end{equation}
		Then, the function $w_\zeta\colon[0,T] \times \R^4 \rightarrow \R $ defined by 
		\begin{align}
			w_{\zeta}(t,\tilde{x},\tilde{y},Z,S)=\tilde{x}+\tilde{y}\,Z+\theta_{\zeta}(t,\tilde{y},Z,S)\ ,
		\end{align}
		is a solution to \eqref{eq:hjbu} on $[0,T)\times \R^4 $, with terminal condition \eqref{eq:termcondu}\ .
	\end{prop}

	Next, solve the first order condition of the supremum term in \eqref{eq:hjbtheta1} to obtain the LT's optimal trading speed in feedback form
	\begin{equation}\label{eq:optimalspeed}
		\nu^{\zeta,\star}=-\frac{1}{2\,\eta\,\zeta}\,\partial_{\tilde{y}}\theta_\zeta\ .
	\end{equation}
	Substitute \eqref{eq:optimalspeed} into \eqref{eq:hjbtheta1} to write
	\begin{equation}\label{eq:hjbtheta2}
		\begin{split}
			\partial_{t}\theta_{\zeta}=&-\phi\,\tilde{y}^{2}+\beta\,\left(S-Z\right)\,\left(\tilde{y}+\partial_{Z}\theta_{\zeta}\right)+\frac{1}{2}\,\gamma^{2}\,Z^{2}\,\partial_{ZZ}\theta_{\zeta}+\frac{1}{2}\,\sigma^{2}\,S^{2}\,\partial_{SS}\theta_{\zeta}\\
			&+\frac{1}{4\,\eta\,\zeta}\,\left(\partial_{\tilde{y}}\theta_{\zeta}\right)^{2}\ .
		\end{split}
	\end{equation}
	Finally, simplify \eqref{eq:hjbtheta2} with the ansatz 
	\begin{align}
		\label{eq:ansatz2}
		\theta_{\zeta}(t,\tilde{y},Z,S)=&\,A_{\zeta}(t)\,\tilde{y}^{2}+B_{\zeta}(t)\,Z\,\tilde{y}+C_{\zeta}(t)\,\tilde{y}\,S+D_{\zeta}(t)\,\tilde{y}+ E_{\zeta}(t)\,Z^{2}\\
		&+F_{\zeta}(t)\,S^{2} +G_{\zeta}(t)\,Z\,S+H_{\zeta}(t)\,Z+I_{\zeta}(t)\,S+J_{\zeta}(t)\ ,
	\end{align}
	which is justified by the following proposition, for which a proof is straightforward.
	\begin{prop}
		Assume there exist functions $A_{\zeta} \in C^{1}([0 ,T]), $ $\ B_{\zeta} \in C^{1}([0 ,T]), $ $\ C_{\zeta} \in C^{1}([0 ,T]), $ $\ D_{\zeta} \in C^{1}([0 ,T]), $ $\ E_{\zeta} \in C^{1}([0 ,T]), $ $\ F_{\zeta} \in C^{1}([0 ,T]), $ $\ G_{\zeta} \in C^{1}([0 ,T]), $ $\ H_{\zeta} \in C^{1}([0 ,T]), $ $\ I_{\zeta} \in C^{1}([0 ,T]),$ and $J_{\zeta} \in C^{1}([0 ,T])$ which solve the system of ODEs
		\begingroup
		\allowdisplaybreaks
		\begin{equation}\label{eq:ODEsystem}
			\left\{
			\begin{aligned}
				A'_{\zeta}(t) & =\phi-\frac{A_{\zeta}(t)^{2}}{\eta\,\zeta}\ ,\\
				B'_{\zeta}(t) & =\beta+\beta\,B_{\zeta}(t)-\frac{A_{\zeta}(t)\,B_{\zeta}(t)}{\eta\,\zeta}\ ,\\
				C'_{\zeta}(t) & =-\beta-\beta\,B_{\zeta}(t)-\frac{A_{\zeta}(t)\,C_{\zeta}(t)}{\eta\,\zeta}\ ,\\
				D'_{\zeta}(t) & =-\frac{A_{\zeta}(t)\,D_{\zeta}(t)}{\eta\,\zeta}\ ,\\
				E_{\zeta}'(t) & =-(\gamma^{2}-2\beta)\,E_{\zeta}(t)-\frac{B_{\zeta}(t)^{2}}{4\,\eta\,\zeta}\,,\\
				F'_{\zeta}(t) & =-\beta\,G_{\zeta}(t)-\sigma^{2}\,F_{\zeta}(t)-\frac{C_{\zeta}(t)^{2}}{4\,\eta\,\zeta}\ ,\\
				G'_{\zeta}(t) & =-2\,\beta\,E_{\zeta}(t)+\beta\,G_{\zeta}(t)-\frac{B_{\zeta}(t)\,C_{\zeta}(t)}{2\,\eta\,\zeta}\,,\\
				H'_{\zeta}(t) & =\beta\, H_{\zeta}(t)-\frac{B_{\zeta}(t)\,D_{\zeta}(t)}{2\,\eta\,\zeta}\,,\\
				I'_{\zeta}(t) & =-\beta\,H_{\zeta}(t)-\frac{C_{\zeta}(t)\,D_{\zeta}(t)}{2\,\eta\,\zeta}\,,\\
				J'_{\zeta}(t) & =-\frac{D_{\zeta}(t)^{2}}{4\,\eta\,\zeta}\,,
			\end{aligned}
			\right.
		\end{equation}%
		\endgroup
		on $[0,T)$, with terminal conditions 
		\begin{align}\label{eq:termcondODEsystem_j}
			A_{\zeta}(T) &= - \alpha\ , \\  B_{\zeta}(T) = C_{\zeta}(T) = D_{\zeta}(T) = E_{\zeta}(T) = F_{\zeta}(&T) = G_{\zeta}(T) = H_{\zeta}(T) = I_{\zeta}(T) = J_{\zeta}(T) =0\ .
		\end{align} 
		Then, the function $\theta_{\zeta}:[0,T]\times\R^3\to \R$ defined by 
		\begin{align}
			\theta_{\zeta}(t,\tilde{y},Z,S)=&\,A_{\zeta}(t)\,\tilde{y}^{2}+B_{\zeta}(t)\,Z\,\tilde{y}+C_{\zeta}(t)\,\tilde{y}\,S+D_{\zeta}(t)\,\tilde{y}+ E_{\zeta}(t)\,Z^{2}\\
			&+F_{\zeta}(t)\,S^{2} +G_{\zeta}(t)\,Z\,S+H_{\zeta}(t)\,Z+I_{\zeta}(t)\,S+J_{\zeta}(t)\ ,
		\end{align}
		solves \eqref{eq:hjbtheta1} over $[0,T)\times\R^3$ with terminal condition \eqref{eq:terminal1}\ .
	\end{prop}
	
	The system of ODEs \eqref{eq:ODEsystem} can be solved sequentially. First, solve the ODE in $A_{\zeta}$, and then solve for $B_{\zeta}\ , \ C_{\zeta}\ , \ D_{\zeta}\ ,\ E_{\zeta}\ , \ G_{\zeta}\ , \ F_{\zeta}\ , \ H_{\zeta}\ ,\ I_{\zeta}\ , \ J_{\zeta}$\ , respectively. The system in \eqref{eq:ODEsystem} admits a unique solution, which is given by
	\begin{equation}\label{eq:solution_system_ode_conve}
		\left\{
		\begin{aligned}
			A_{\zeta}(t) & =\sqrt{\phi\,\eta\,\zeta}\,\tanh\left(\frac{\sqrt{\phi}}{\sqrt{\eta\,\zeta}}\,(T-t)+\arctanh\left(-\frac{\alpha}{\sqrt{\phi\,\eta\,\zeta}}\right)\right)\ ,\\
			B_{\zeta}(t) & =\int_{t}^{T}\beta\exp\left(\int_{t}^{s}\left(\beta-\frac{1}{\eta\,\zeta}A_{\zeta}(u)\right)du\right)ds\, ,\\
			C_{\zeta}(t) & =-B_{\zeta}(t)\, ,\\
			D_{\zeta}(t) & = 0\\
			E_{\zeta}(t) & =-\int_{t}^{T}\exp\left(-(\gamma^{2}-2\,\beta)(t-s)\right)\frac{1}{4\,\eta\,\zeta}\,B_{\zeta}(s)^{2}ds\, ,\\
			F_{\zeta}(t) & =-\int_{t}^{T}\exp\left(-\sigma^{2}(t-s)\right)\,\left(\beta\, G_{\zeta}(s)+\frac{1}{4\,\eta\,\zeta}\,C_{\zeta}(s)^{2}\right)ds\, ,\\
			G_{\zeta}(t) & =-\int_{t}^{T}\exp\left(\beta\,(t-s)\right)\left(2\,\beta\, E_{\zeta}(s)-\frac{1}{2\,\eta\,\zeta}\,B_{\zeta}(s)^{2}\right)\,ds\, ,\\
			H_{\zeta}(t) & =0\,,\\
			I_{\zeta}(t) & =0\,,\\
			J_{\zeta}(t) & =0\,.
		\end{aligned}
		\right.
	\end{equation}
	
	The optimal trading strategy in feedback form \eqref{eq:optimalspeed} over $[0,T]\times\R^3$ is now given by
	\begin{equation}
		\label{eq:optimalspeed_feedback_j2}
		\nu^{\zeta,\star}\left(t,\tilde{y},Z,S\right)=-\frac{1}{\eta\,\zeta}\,A_{\zeta}(t)\,\tilde y+\frac{1}{2\,\eta\,\zeta}\,B_{\zeta}(t)\,(S-Z)\ .
	\end{equation}
	
	The first term on the right-hand side of \eqref{eq:optimalspeed_feedback_j2} is the optimal liquidation rate in the continuous Almgren-Chriss model. The second term is an arbitrage component; it accounts for the spread between the instantaneous rate $Z$ and the oracle rate $S$.
	
	\subsection{The closed-form approximation strategy \label{section:discussion}}
	
	Here,  we use a family of closed-form strategies of the type in \eqref{eq:optimalspeed_feedback_j2} to derive a piecewise-defined trading strategy which approximates the optimal trading speed in feedback form \eqref{eq:ANX:optimalspeed2}. Specifically, we partition the space of the rate $Z$ into strips and define a piecewise strategy which uses a different impact parameter $\zeta$ in each different strip. Finally, we show that as the width of the strip becomes arbitrarily small, the piecewise strategy converges to the closed-form approximation strategy. 
	
	Let $\left\{Z_0^N, \dots,Z_N^{N}\right\}$ be a partition of $\left[\underline{Z},\overline{Z}\right],$ where $0<\underline{Z}<\overline{Z}$. Here $\underline{Z} $ and $ \overline{Z}$ are are such that the LT has high confidence that the rate $Z$ will be within the range $\left[\underline{Z},\overline{Z}\right]$ during the execution window. Next, for each $N\in\N$ and  $j\in\{0,\dots,N\}$ we define
	\begin{equation}\label{eq:partitionedConvexity}
		Z_{j}^N \coloneqq \underline{Z}+\frac{j}{N} \left(\overline{Z}-\underline{Z}\right)\quad\text{and}\quad\zeta_{j}^{N} =  \frac{1}{\kappa}\, \left(Z_{j}^N\right)^{3/2}\,.
	\end{equation}
	
	In the remainder of this section, we simplify the notation and use $\nu^{\star,j,N}$ instead of $\nu^{\zeta_{j}^{N}, \star}$ to denote the optimal trading strategy with impact parameter $\zeta_{j}^{N}$. 
	
	Whenever $Z$ is arbitrarily close to $Z_{j}^N$ the impact parameter  in \eqref{eqn:price impact} can be approximated by $\zeta_{j}^{N}$. Thus, to construct the approximate trading strategy, we first define a strategy $\nu^{\star,N}$ that uses the closed-form optimal trading speed $\nu^{\star,j,N}$ to approximate the optimal trading speed whenever the rate is close to $\zeta_{j}^{N}$. We define the piecewise-defined trading speed $\nu^{\star,N}\colon\left[0,T\right]\times\R^3\to\R$ 
	\begin{equation}\label{eq:piecewise_strategy}
		\begin{split}
			\nu^{\star,N}\left(t,\tilde{y},Z,S\right) =\,& \nu^{\star,0,N}\left(t,\tilde{y},Z,S\right)\mathbbm{1}_{Z<Z^N_1} + \sum_{j=1}^{N-1}\nu^{\star,j,N}\left(t,\tilde{y},Z,S\right)\mathbbm{1}_{Z\in[Z_j^N,Z^N_{j+1})} \\
			&+ \nu^{\star,N,N}\left(t,\tilde{y},Z,S\right)\mathbbm{1}_{Z\geq Z^N_N}\,.
		\end{split}
	\end{equation}
	The strategy $\nu^{\star,N}\left(t,\tilde{y},Z,S\right)$ has first-type discontinuity points; specifically, it is discontinuous over $[0,T]\times\R\times\{Z_j^N\}\times\R\,$ for each $j\in\{1,\dots,N\}$ because for each $\left(t,\tilde{y},Z^N_{j+1},S\right)\in[0,T]\times\R^2$ we have  $\nu^{\star,j,N}\left(t,\tilde{y},Z^N_{j+1},S\right)\neq\nu^{\star,j+1,N}\left(t,\tilde{y},Z^N_{j+1},S\right)$.

	The theorem below shows how to partition  $\left[\underline{Z},\overline{Z}\right]$ to make the discontinuities in $\nu^{\star,N}\left(t,\tilde{y},Z,S\right)$ arbitrarily small. Furthermore, when the distance between points in the partition becomes sufficiently small, the sequence of piecewise-defined optimal strategies $\left\{\nu^{\star,N}\right\}_{N\in\N}$ converges uniformly to a continuous closed-form approximation strategy which we use in our performance study of Section \ref{sec:Performance}.
	
	\begin{thm}\label{thm}
		For each $\varepsilon>0\,$ there exists $N\in\N$ such that 
		\begin{equation}\label{eq:inequality}
			\max_{j=1,\dots,N}\left|\nu^{\star,j,N}\left(t,\tilde{y},Z^N_{j+1},S\right)-\nu^{\star,j+1,N}\left(t,\tilde{y},Z^N_{j+1},S\right)\right|<\varepsilon\,.
		\end{equation}
		Furthermore, for each $N\in\N\,,$ let $\tilde{\nu}^{\star,N}\coloneqq \nu^{\star,N}\left|_{[0,T]\times\R\times[\underline{Z},\overline{Z}]\times\R}\right.$. Then, the sequence $\{\tilde{\nu}^{\star,N}\}$ converges to $\tilde{\nu}^{\star}$ uniformly in $ [0,T]\times\R\times\left[\underline{Z},\overline{Z}\right]\times\R $\,, where 
		\begin{equation}
			\label{eq:pseudooptimal}
			\tilde{\nu}^{\star}\left(t,\tilde{y},Z,S\right)=-\frac{\kappa}{\eta}\,Z^{-3/2} \,A(t,Z)\,\tilde y+\frac{\kappa}{2\,\eta}\,Z^{-3/2}\,B(t,Z)\,(S-Z)\,,
		\end{equation}
		and
		\begin{equation}\label{eq:pseudoAB}
			\begin{split}
				A(t,Z) =&\, \sqrt{\frac{\phi\,\eta\,Z^{3/2}}{\kappa}}\tanh\left(\frac{\sqrt{\phi\,\kappa}}{\sqrt{\eta\,Z^{3/2}}}t+\arctanh\left(-\frac{\alpha\,\sqrt{\kappa}}{\sqrt{\phi\,\eta\,Z^{3/2}}}\right)\right)\ ,\\
				B(t,Z) =& \,\int_{t}^{T}\beta \exp\left(\int_{s}^{t} \left(\beta-\frac{\kappa}{\eta\,Z^{3/2}}A(u,Z)\right)du\right)ds\ .\\
			\end{split}
		\end{equation}
	\end{thm}
	For a proof see \ref{sec:annex_proof}\,.

	\subsection{Comparison with the numerical approximation strategy \label{sec:model:comparison}}
	In this subsection, we use a Euler scheme to compute the numerical approximation strategy in \eqref{eq:ANX:system}, where execution costs are not piecewise constant. We compare this numerical solution with the closed-form approximation strategy in \eqref{eq:pseudooptimal}.
	
	The numerical approximation strategy uses a three-dimensional grid; one dimension is time, one is the rate $Z,$ and one is the oracle rate $S$. We consider a wide interval for the values of $Z$ and $S$, and use a Neumann boundary condition for both space variables (the derivatives are zero at the boundaries) because the value function should not  vary significantly for values of the rates $Z$ and $S$ that are considerably far from the observed rates. For the semilinear PDE in \eqref{eq:ANX:system}, we use the Picard iterative method to linearise the problem at every time level of the grid. More precisely, at each time level, we use Picard iterations to linearise the term that is quadratic in the unknown solution, which is replaced by the product of the unknown solution and the most recent computed solution obtained at the previous time level. The iterative approximation process is carried on until the difference between two consecutive numerical solutions is smaller than a given threshold.
	
	Figure \ref{fig:comparisonnumerical} compares the numerical approximation strategy \eqref{eq:ANX:optimalspeed2} and the closed-form approximation strategy \eqref{eq:pseudooptimal} for different values of the inventory, the rate $Z$, and the oracle rate $S$. Figure   \ref{fig:comparisonnumerical} indicates that the closed-form and the numerical approximation strategies are significantly close and capture the same financial effects. In particular, the strategies clearly depend on the spread $S-Z$ and the inventory $\tilde y.$ When the LT has zero inventory, the strategy is mostly speculative because the strategy buys asset $Y$ when the oracle rate $S$ is above the rate $Z$, and sells otherwise; recall that the LT buys asset $Y$ when the trading speed $\nu$ is positive, and sells otherwise. When the LT has a positive (negative) inventory, the optimal strategy buys (sells) asset $Y$ only when $S$ is significantly higher (lower) than $Z$. 
	
	Figure \ref{fig:comparisonnumerical} shows that the absolute difference between both strategies increases as the difference between the rate $Z$ and $S$ increases, and the difference is minimal when the rates are equal. In practice, by arbitrage, the rates $Z$ and $S$ are aligned by LTs in the pool, so  differences between rates are small. For instance,  with the market data we use in Section \ref{sec:Performance}, the average absolute difference between rates is $2$ USD ($0.07\%$) in the liquid pool and $10$ USD ($0.37\%$) in the illiquid pool. Finally, it is unclear whether the difference between the optimal speed obtained with a numerical scheme and the closed-form approximation strategy shown in the lower panels of Figure \ref{fig:comparisonnumerical} stem from the numerical approximation or from the closed-form approximation strategy, but in either case, the difference has no bearing on the results we discuss in Section \ref{sec:Performance}.
	
	\begin{figure}[H]\centering
		\includegraphics{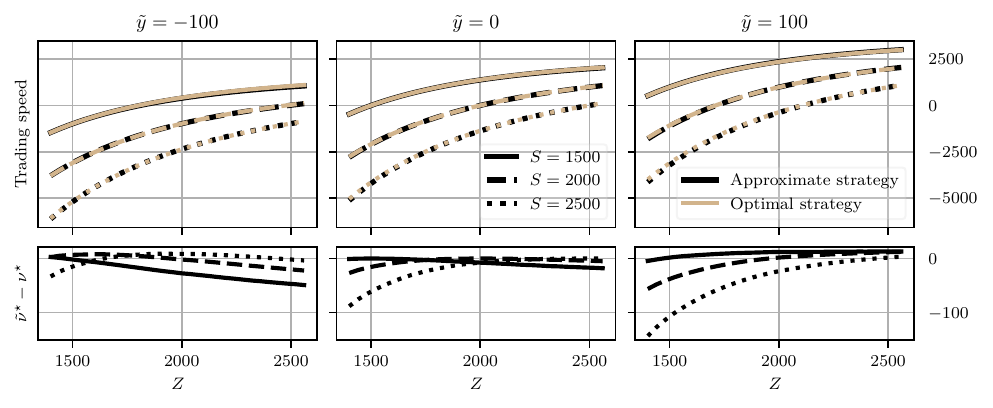}\\
		\caption{Comparison of the optimal speed \eqref{eq:ANX:optimalspeed2} obtained with a numerical scheme and the closed-form approximation strategy in \eqref{eq:pseudoAB}. Model parameters: $T=0.1$, $\sigma = 0.03$, $\gamma = 0.02$, $\beta = 1$, $\alpha = 5$, $\kappa = 10^7$, $\eta=1$, $\phi=10^{-5}$, and $Z_0=S_0=2000$. Inventory is $\tilde y=-100$ (left panel), $\tilde y=0$ (middle panel), and $\tilde y=100$ (right panel).}
		\label{fig:comparisonnumerical}
	\end{figure}

	\section{Model II: optimal execution without oracle rate \label{sec:Model2}}
	
	In this section, we consider the problem of an LT who wants to exchange a large position in asset $Y$ into asset $X$ in a CPMM. The key differences between Model II and Model I of Section \ref{sec:Model} are that in Model II the AMM rates are efficient, so the LT does not use an oracle rate from another venue, and  the depth of the pool is stochastic. 
	
	
	The LT must liquidate a large position in asset $Y$ in a CPMM over a period of time $[0,T]$, where $T>0$. Asset $Y$ and the wealth of the LT are valued in terms of asset $X.$ The CPMM is very active and with concentrated liquidity, so the depth $\kappa$ exhibits very frequent and random updates. The instantaneous rate $(Z_t)_{t\in[0,T]}$ evolves as
	\begin{align}
		\label{eq:priceProcess_model2}
		dZ_{t}=\gamma\,Z_t\,dB_t\ ,
	\end{align}
	and the dynamics of the depth process $(\kappa_t)_{t\in[0,T]}$ are given by
	\begin{align}
		\label{eq:kappaProcess_model2}
		d\kappa_{t}=\varsigma\, \kappa_{t} \,dL_{t}\ .
	\end{align}
	Here, $(B_t)_{t\in[0,T]}$ and $(L_t)_{t\in[0,T]}$ are independent standard Brownian motions and $\varsigma$ is the volatility of the depth $\kappa$. The wealth process $(\tilde x_t)_{t\in[0,T]}$ satisfies
	\begin{align}
		\label{eq:wealthProcess_model2}
		d\tilde{x}_{t} = \left(Z_{t}-\eta\,\frac{Z_{t}^{3/2}}{\kappa_{t}}\,\nu_{t}\right)\,\nu_{t}\,dt\, .
	\end{align}
	\subsection{Performance criterion and value function \label{Model2_1}}
	For each $(t,\tilde{x},\tilde{y},Z,\kappa )\in[0,T]\times\R\times\R\times\R_{++}\times\R_{++},$ and for each admissible control $\nu\in\mathcal{A}$ the performance criterion of the LT is given by
	\begin{align}
		\label{eq:perfcriteria_model2}
		u^{\nu}(t,\tilde{x},\tilde{y},Z,\kappa)=\E_{t,\tilde{x},\tilde{y},Z,\kappa }\left[\tilde{x}_{T}^{\nu}+\tilde{y}_{T}^{\nu}\,Z_T-\alpha\left(\tilde{y}_{T}^{\nu}\right)^{2}-\phi\int_{t}^{T}\left(\tilde{y}_{s}^{\nu}\right)^{2}ds\right]\ ,
	\end{align}
	and the value function is
	\begin{align}
		\label{eq:valuefunc_model2}
		u(t,\tilde{x},\tilde{y},Z,\kappa )=\underset{\nu\in\mathcal{A}}{\sup}\,u^{\nu}(t,\tilde{x},\tilde{y},Z,\kappa )\ .
	\end{align}
	
	\subsection{The dynamic programming equation \label{Model2_2}}
	The value function \eqref{eq:valuefunc_model2} is the unique classical solution to the HJB equation
	\begin{equation}
		\label{eq:HJB_model2}
		\begin{split}
			0=&\,\partial_{t}w-\phi\,\tilde{y}^{2}+\frac{1}{2}\,\gamma^{2}\,Z^{2}\,\partial_{ZZ}w+\frac{1}{2}\,\varsigma^{2}\,\kappa^{2}\,\partial_{\kappa\kappa}w \\
			&+\underset{\nu\in\R}{\sup}\left(\left(\nu\,Z-\eta\,\frac{Z^{3/2}}{\kappa}\,\nu^{2}\right)\partial_{\tilde{x}}w-\nu\,\partial_{\tilde{y}}w\right)\,,
		\end{split}
	\end{equation}
	with terminal condition
	\begin{equation}\label{eq:termcondu_model2}
		w(T,\tilde{x},\tilde{y},Z,\kappa) =\tilde{x}+\tilde{y}\,Z-\alpha\,\tilde{y}^{2}\ .
	\end{equation}
	
	The terminal condition \eqref{eq:termcondu_model2} suggests the ansatz
	\begin{equation}
		\label{eq:ansatz1_Model2}
		w(t,\tilde{x},\tilde{y},Z,\kappa)=\tilde{x}+\tilde{y}\,Z+\theta(t,\tilde{y},Z,\kappa)\,,
	\end{equation}
	which we justify by the following proposition, for which a proof is straightforward.
	\begin{prop}
		Assume there exists a function $\theta\in C^{1,1,2,2}([0,T]\times\R\times\R_{++}\times\R_{++})$ that solves
		\begin{equation}\label{eq:PDE_Model2}
			0=\partial_{t}\theta-\phi\,\tilde{y}^{2}+\frac{1}{2}\,\gamma^{2}\,Z^{2}\,\partial_{ZZ}\theta+\frac{1}{2}\,\varsigma^{2}\,\kappa^{2}\,\partial_{\kappa\kappa}\theta+\underset{\nu\in\R}{\sup}\left(-\eta\,\frac{Z^{3/2}}{\kappa}\,\nu^{2}-\nu\,\partial_{\tilde{y}}\theta\right)\,,
		\end{equation}
		on $[0,T)\times\R\times\R_{++}\times\R_{++}$\ , with terminal condition
		\begin{equation}\label{eq:terminal_liquid_ansatz}
			\theta(T,\tilde{y},Z,\kappa)=-\alpha\,\tilde{y}^{2}\ .
		\end{equation}
		Then, the function $w\colon[0,T] \times \R \times \R \times \R_{++} \times \R_{++} \rightarrow \R$ defined by 
		\begin{equation}
			\label{eq:ansatz_sec4}
			\begin{split}
				w(t,\tilde{x},\tilde{y},Z,\kappa)=\tilde{x}+\tilde{y}\,Z+\theta(t,\tilde{y},Z,\kappa)
			\end{split}
		\end{equation}
		is a solution to \eqref{eq:HJB_model2} on $[0,T)\times\R\times\R_{++}\times\R_{++}$ with terminal condition \eqref{eq:termcondu_model2}\ .
	\end{prop}
	
	Next, solve the first order condition in \eqref{eq:PDE_Model2} to obtain the LT's trading speed in feedback form
	\begin{equation}\label{eq:strat_liquid}
		\nu^{\star}=	-\frac{\kappa}{2\,\eta}\,\partial_{\tilde{y}}\theta\,Z^{-3/2}\,.
	\end{equation}
	Substitute \eqref{eq:strat_liquid} into \eqref{eq:PDE_Model2} to write
	\begin{equation}\label{eq:liquid_pde_theta}
		0=\partial_{t}\theta-\phi\,\tilde{y}^{2}+\frac{1}{2}\,\gamma^{2}\,Z^{2}\,\partial_{ZZ}\theta+\frac{1}{2}\,\varsigma^{2}\,\kappa^{2}\,\partial_{\kappa\kappa}\theta+\frac{\kappa}{4\,\eta}\,\left(\partial_{\tilde{y}}\theta\right)^{2}\,Z^{-3/2}\,.
	\end{equation}
	
	Finally, simplify \eqref{eq:liquid_pde_theta} with the ansatz
	\begin{equation}\label{eq:second_liquid_ansatz}
		\theta(t,\tilde{y},Z,\kappa)=\theta_{0}(t,Z,\kappa)+\theta_{1}(t,Z,\kappa)\,\tilde{y}+\theta_{2}(t,Z,\kappa)\,\tilde{y}^{2}\,,
	\end{equation}
	which is justified by the following proposition, for which a proof is straightforward.
	\begin{prop}
		Assume there exist functions $\theta_0\in C^{1,2,2}([0,T]\times\R_{++}\times\R_{++})$, $\theta_1\in C^{1,2,2}([0,T]\times\R_{++}\times\R_{++})$, and $\theta_2\in C^{1,2,2}([0,T]\times\R_{++}\times\R_{++})$ which solve the system of PDEs
		\begin{equation}\label{eq:liquid_system}
			\left\{
			\begin{aligned}
				&0=\partial_{t}\theta_{2}-\phi+\frac{1}{2}\,\gamma^{2}\,Z^{2}\,\partial_{ZZ}\theta_{2}+\frac{1}{2}\,\varsigma^{2}\,\kappa^{2}\,\partial_{\kappa\kappa}\theta_{2}+\frac{\kappa}{\eta}\,\theta_{2}^{2}\,Z^{-3/2}\,,\\
				&0=\partial_{t}\theta_{1}+\frac{1}{2}\,\gamma^{2}\,Z^{2}\,\partial_{ZZ}\theta_{1}+\frac{1}{2}\,\varsigma^{2}\,\kappa^{2}\,\partial_{\kappa\kappa}\theta_{1}+\frac{\kappa}{\eta}\,\theta_{1}\,\theta_{2}\,Z^{-3/2}\,,\\
				&0=\partial_{t}\theta_{0}+\frac{1}{2}\,\gamma^{2}\,Z^{2}\,\partial_{ZZ}\theta_{0}+\frac{1}{2}\,\varsigma^{2}\,\kappa^{2}\,\partial_{\kappa\kappa}\theta_{0}+\frac{\kappa}{4\,\eta}\,\theta_{1}^{2}\,Z^{-3/2}\,,
			\end{aligned}
			\right.
		\end{equation}
		on $[0,T)\times\R_{++}\times\R_{++}$ with terminal conditions
		\begin{equation}\label{eq:liquid_system_terminal}
			\theta_{2}(T,Z,\kappa)=-\alpha\,,\ \ \theta_{1}(T,Z,\kappa)=0\,,\ \,\text{and}\ \ \theta_{0}(T,Z,\kappa)=0\,.
		\end{equation}
		Then, the function $\theta\colon[0,T]\times\R\times\R_{++}\times\R_{++}\to\R$ defined by
		\begin{equation}
			\theta(t,\tilde{y},Z,\kappa)=\theta_{0}(t,Z,\kappa)+\theta_{1}(t,Z,\kappa)\,\tilde{y}+\theta_{2}(t,Z,\kappa)\,\tilde{y}^{2}\,,
		\end{equation}
		solves \eqref{eq:liquid_pde_theta} over $[0,T)\times\R\times\R_{++}\times\R_{++}$ with terminal condition \eqref{eq:terminal_liquid_ansatz}\,.
	\end{prop}
	
	The optimal strategy in feedback form \eqref{eq:strat_liquid} is now given by
	\begin{equation}
		\label{eq:feedbackformliquid}
		\nu^{\star}=-\frac{\kappa}{2\,\eta}\,\left(2\,\theta_2\,\tilde{y}+\theta_1\right)\,Z^{-3/2}\,.
	\end{equation}
	
	The system of PDEs in \eqref{eq:liquid_system} can be solved sequentially as follows. Solve the first PDE in the system to obtain $\theta_2$. Substitute $\theta_2$ is the second and third equations of the system so  the PDEs in $\theta_1$ and $\theta_0$ become linear. We cannot solve the semilinear PDE in $\theta_2$ in closed-form, and providing an existence result is out of the scope of this work. However, Theorem \ref{thm:bounds} provides \textit{a priori} estimates for $\theta_2$ and suggests the system of PDEs is well-behaved. The proof is omitted as it uses the same arguments as those in the proof of Theorem \ref{thm:ANX:existence}.
	
	\begin{thm}\label{thm:bounds}
		Let $\theta_0\in C^{1,2,2}([0,T]\times\R_{++}\times\R_{++})$, $\theta_1\in C^{1,2,2}([0,T]\times\R_{++}\times\R_{++})$, and $\theta_2\in C^{1,2,2}([0,T]\times\R_{++}\times\R_{++})$ be a solution to the system of PDEs \eqref{eq:liquid_system} with terminal condition \eqref{eq:liquid_system_terminal}. Define $\theta$ in \eqref{eq:second_liquid_ansatz} and define $w$ in \eqref{eq:ansatz_sec4}. Assume that for all  $\left(t,\tilde{x},\tilde{y},Z,\kappa \right) \in [0,T] \times \R \times \R \times \R_{++} \times \R_{++}$ we have $$u^{\nu}(t,\tilde{x},\tilde{y},Z,\kappa ) \leq w\left(t,\tilde{x},\tilde{y},Z,\kappa\right),$$ where $u^\nu$ is defined in \eqref{eq:perfcriteria_model2}, and that equality is obtained for the optimal control $\left(\nu^{\star}\right)_{t\in[0, T]}$ in feedback form in \eqref{eq:feedbackformliquid}. Then $\theta_2$ has the following bounds
		\begin{align}
			-\alpha-\phi\left(T-t\right) \leq \theta_2\left(t,\tilde{x},\tilde{y},Z,\kappa \right) \leq 0, \quad \forall \left(t,\tilde{x},\tilde{y},Z,\kappa \right) \in [0,T] \times \R \times \R \times \R_{++} \times \R_{++}\,.
		\end{align}
		
	\end{thm}

	\section{Performance of strategies \label{sec:Performance}}
	We study the performance of two versions of the closed-form approximation strategy in \eqref{eq:pseudooptimal} and \eqref{eq:pseudoAB} of Model I; see Section \ref{sec:Model}. One version focuses on liquidating a large position in one asset and the other uses the lead-follow relationship between the oracle and AMM rates to execute a statistical arbitrage. We use the Uniswap v3 data of Section \ref{sec:AMM} for the liquid pool ETH/USDC and the illiquid pool ETH/DAI. We account for AMM and gas fees and assume that the orders sent by the LT do not impact the dynamics of the pools.
	
	We use in-sample data to estimate model parameters and use out-of-sample data to execute the strategies. For in-sample data, we use a window of 24 hours prior to the start of the trading programme for both pools. For out-of-sample data, we use windows of $2$ and $12$ hours when the LT trades in the liquid and illiquid pools, respectively. To measure performance, we use rolling time windows for estimation and execution, between 1 July 2021 and 5 May 2022, to carry out this procedure. Specifically, after every execution programme, we shift both windows by $2$ and $12$ hours for the liquid and illiquid pools, respectively, and repeat the same procedure, i.e., estimate parameters with in-sample data and trade with out-of-sample data. We remark that we do not simulate rates, we use those of the AMM and Binance, and execution costs are those the trades would have received. In total, we run  3,635 and 607 execution programmes for ETH/USDC and ETH/DAI, respectively.
	
	We proceed as follows. Subsection \ref{subsec: liquidation stategy} describes how parameter estimates are obtained and showcases the performance of the liquidation strategy.  Subsection \ref{subsec: speculation stategy} discusses the use of model I for statistical arbitrage, and showcases the performance.

	\subsection{Liquidation strategy}\label{subsec: liquidation stategy}
	
	We describe how to estimate the in-sample model parameters for every run of the liquidation strategy. For rate dynamics, the LT performs OLS regressions on the discretised versions of \eqref{eq:SProcess} and \eqref{eq:ZProcess}:  
	\begin{align}
		\label{eq: discrete price dyn}
		\Delta \log S_t =& -\frac{\sigma^2}{2} \, \Delta t + \sigma  \,\sqrt{\Delta t} \, \upsilon_t\,,
		\nonumber \\
		\Delta \log Z_t  =& -\frac{\gamma^2}{2} \, \Delta t + \beta\, \left(\frac{S_t - Z_t}{Z_t}\right) \,\Delta t + \gamma \, \sqrt{\Delta t} \, \epsilon_t\,,
	\end{align}
	where  $\{\epsilon_t\,,\upsilon_t\}$ are error terms. Here, the size of the time-step $\Delta t$ is the frequency of the liquidity taking orders (from all LTs) that arrive in the pool during the estimation period.

	For the liquidation strategy, we target a participation rate of $50\%$ of the observed hourly volume to set the initial inventory, which is liquidated by the LT over the trading window at the same frequency as the observed average trading frequency over the in-sample estimation period. The LT's trading frequency determines the value of the parameter $\eta$ in \eqref{eq:execratesApproxCPMM_nu}, which scales the execution costs.
	
	The value of the other model parameters are as follows. The value of the parameter $\kappa$ in the trading speed \eqref{eq:pseudooptimal} is  the last observed depth of the pool before the start of the execution. The value of the running inventory parameter $\phi$  is kept constant for all runs. The value of the terminal penalty parameter $\alpha$ is arbitrarily large to enforce full liquidation of outstanding inventory by the end of the trading horizon. Table \ref{table:paramsFixed} shows the parameter values we use for all strategy runs. Finally, as a more detailed example, \ref{Example one run liquidation} describes parameter estimation and performance for a specific run of the liquidation strategy.
	
	{\footnotesize
		\begin{table}[H]
			\footnotesize
			\begin{center}
				\begin{tabular}{c  r  r} 
					\hline 
					& ETH/USDC & ETH/DAI \\ [0.5ex] 
					\hline
					$T$ & 0.083$\ \textrm{days}$ & 0.5$\ \textrm{days}$ \\ [0.5ex] 
					$\phi$ & 0.005$\  \textrm{USDC}\cdot\textrm{ETH}^{-2}$ & 0.05$\  \textrm{DAI}\cdot\textrm{ETH}^{-2}$ \\ [0.5ex] 
					$\alpha$ & 10$\  \textrm{USDC}\cdot\textrm{ETH}^{-2}$ & 10$\  \textrm{DAI}\cdot\textrm{ETH}^{-2}$ \\ 
					\hline 
				\end{tabular}
			\end{center}
			\caption {Values of the liquidation model parameters.}
			\label{table:paramsFixed}
		\end{table}
	}
	


	We benchmark the performance of the liquidation strategy with two strategies: TWAP, which consists in trading at a constant rate; and a single order execution strategy, which consists in executing the entire order at the beginning of the execution window.  The market rates at the time of trading are used to compute the execution costs for all strategies. Gas fees are $5$ USD per transaction, regardless of transaction size. On the other hand, AMM fees depend on transaction size, and here we impute a $0.01 \, \%$ fee to the value of every transaction. Figure \ref{fig:backtest_USDC_DAI_PnLs} depicts the distribution of the gross PnL for each execution programme, which is given by $\tilde x_T + \tilde y_T \, Z_T - \tilde y_0 \, Z_0.$ Tables \ref{table:perfsUSDC} and \ref{table:perfsDAI} show the average and standard deviation of the gross PnLs, the number of transactions, the gas fees, and the AMM fees.\footnote{Gross PnL, as opposed to net PnL, is computed without the AMM fees and gas fees paid by the LT.}


	\begin{figure}%
		\centering
		\subfloat[Distribution of gross PnL for all strategies for ETH/USDC (3,635 executions).] {{\includegraphics{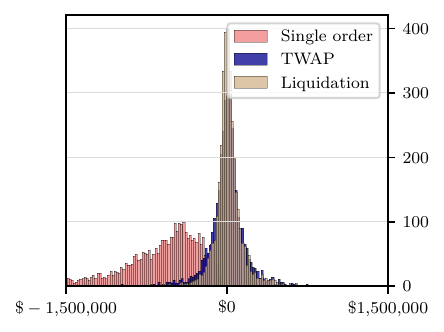} }}%
		\qquad
		\subfloat[Distribution of gross PnL of strategies for ETH/DAI (607 executions).] {{\includegraphics{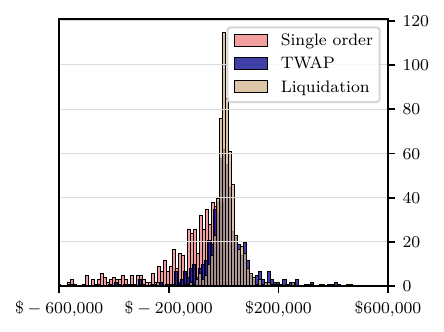} }}%
		\caption{PnL distribution.}%
		\label{fig:backtest_USDC_DAI_PnLs}%
	\end{figure}

	{\footnotesize
		\begin{table}[H]
			\scriptsize
			\parbox{.45\linewidth}{
				\centering
				\begin{tabular}{c  r  r   r  r} 
					\hline 
					& Gross avg.  & Std. & Avg. num. & Avg.   \\
					& PnL  & dev. & trades &  fees  \\ [0.5ex]
					\hline
					Single order & $-$956,298 & 1,963,014 & 1 & 2,538\\ [0.5ex] 
					TWAP &  3998 &  217,001 & 439 & 10,529\\ [0.5ex] 
					Liquidation &  27,185 &  288,518 & 439 & 11,885 \\ 
					\hline 
				\end{tabular}
				\caption {Performance and fees for ETH/USDC (3,635 executions). The Average PnL does not include fees}
				\label{table:perfsUSDC}
			}
			\hfill
			\parbox{.45\linewidth}{
				\centering
				\begin{tabular}{c  r  r  r  r} 
					\hline 
					& Gross avg.  & Std. & Avg. num. & Avg.    \\
					& PnL  & dev. & trades & fees  \\ [0.5ex]
					\hline
					Single order & $-$233,390 &  428,688 & 1 & 634 \\ [0.5ex] 
					TWAP & 1,875 &  170,008 & 108 & 1,217 \\ [0.5ex] 
					Liquidation & 12,240 &  63,605 & 108 & 1,782 \\
					\hline 
				\end{tabular}
				\caption {Performance and fees for ETH/DAI $\ \ \quad $ (607 executions). The Average PnL does not include fees}
				\label{table:perfsDAI}
			}
		\end{table}
	}

	Figure \ref{fig:backtest_USDC_DAI_PnLs} and Tables \ref{table:perfsUSDC} and \ref{table:perfsDAI} show that liquidating all the inventory in one trade is sub-optimal compared with the other strategies due to the high execution costs of the large order. In both cases, our model outperforms TWAP in terms of the ratio between performance, net of fees, and risk measured by the standard deviation. Key to the outperformance is that the liquidation strategy uses the rates in Binance as a trading signal.

	\subsection{Speculative strategy \label{subsec: speculation stategy}}
	We consider the same setup as before, i.e., the in-sample estimation and out-of-sample execution. Here, the LT arbitrages the AMM. To this end, the LT starts with zero inventory in $Y$ and sets the values of $\phi$ to $0.001 \ \textrm{USDC}\cdot\textrm{ETH}^{-2}$ and $0.01 \ \textrm{DAI}\cdot\textrm{ETH}^{-2}$ for the liquid and illiquid pools, respectively. The strategy profits from the oracle rate as a predictive signal. Figure \ref{fig:backtest_USDC_DAI_PnLs_speculation} depicts the distribution of the gross PnL for each run and Tables \ref{table:perfsUSDC2} and \ref{table:perfsDAI2} show the average and standard deviation of the gross PnLs, the number of transactions, and the estimated AMM and gas fees. Note that the speculative strategy is more profitable in the illiquid pool due to a larger discrepancies between the oracle and the instantaneous pool rate, leading to more arbitrage opportunities.

	\begin{figure}%
		\centering
		\subfloat[Distribution of gross PnL for ETH/USDC.] {{\includegraphics{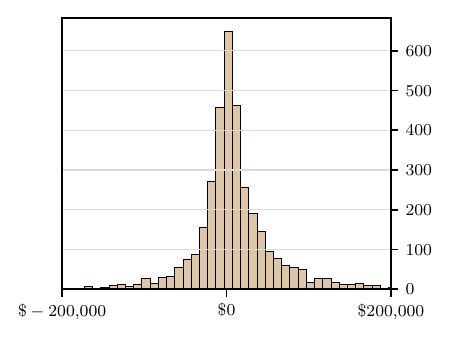} }}%
		\qquad
		\subfloat[Distribution of gross PnL for ETH/DAI.] {{\includegraphics{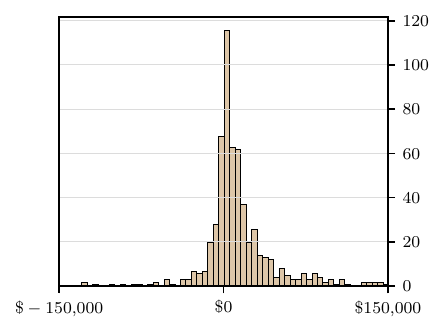} }}%
		\caption{Statistical arbitrage PnL distribution.}%
		\label{fig:backtest_USDC_DAI_PnLs_speculation}%
	\end{figure}
	
	{\footnotesize
		\begin{table}
			\footnotesize
			\parbox{.45\linewidth}{
				\centering
				\begin{tabular}{c  r  r  r  r} 
					\hline 
					& Gross avg.  & Std. & Avg. num. & Avg.    \\
					& PnL  & dev. & trades & fees \\ [0.5ex]
					\hline
					Speculative  &  &  \\ 
					strategy &  22,693 &  190,789 & 439 & 7,111 \\ 
					\hline 
				\end{tabular}
				\caption {Performance and fees for the pair ETH/USDC (3,635 executions). The Average PnL does not include fees}
				\label{table:perfsUSDC2}
			}
			\hfill
			\parbox{.45\linewidth}{
				\centering
				\begin{tabular}{c  r  r  r  r} 
					\hline 
					& Gross avg.  & Std. & Avg. num. & Avg.  \\
					& PnL  & dev. & trades & fees  \\ [0.5ex]
					\hline
					Speculative  &  & & & \\ 
					strategy &  20,886 &  54,043 & 108 & 2,082 \\
					\hline 
				\end{tabular}
				\caption {Performance and fees for the pair ETH/DAI (607 executions). The Average PnL does not include fees}
				\label{table:perfsDAI2}
			}
		\end{table}
	}
	
	\clearpage
	\section{Conclusions}
	In this chapter, we used Uniswap v3 data to analyse rate, liquidity, and execution costs of CPMMs. We proposed two models for optimal trading in CPMMs. In the first model, we assumed that prices are formed in an alternative venue and the liquidity provided in the CPMM remains constant for relevant periods of time. In the second model, price formation takes place in the pool and the liquidity provided in the pool is stochastic. Finally, we used in-sample estimation of model parameters and out-of-sample market data to test the performance of the strategies, so our results do not rely on simulations. We show our strategy considerably outperforms TWAP and a strategy that consists in sending a single large order. We also showed that there are significant arbitrage opportunities between Binance and AMM rates.
	
\chapter{Optimal Trading with Signals in Multiple Automated Market Makers}\label{ch:ieee}

 \section{Introduction}

Automated market makers (AMMs) are open-source immutable programs that run on peer-to-peer networks. These programs encode the trading rules that govern how liquidity takers (LTs) and liquidity providers (LPs) interact. The emergence of these trading venues poses great challenges to traditional financial services because they do not rely on traditional intermediaries and stakeholders such as banks and brokers. At present, the most popular type of AMM is the constant product market maker (CPMM) such as Uniswap v3, and AMMs serve mainly as trading venues for crypto-currencies. CPMMs are a special type of constant function market maker (CFMM), where a convex trading function and a set of rules determine how the exchange operates. 

The literature on the optimal design of these venues and the mathematical tools to optimally take and provide liquidity are new. In \cite{cartea2022decentralised}, the authors propose models for LTs that trade in a CPMM when prices form in the pool and when the prices form in an alternative venue. In  \cite{cartea2023predictable}, the authors characterise the wealth of LPs in CFMMs and in CPMMs with concentrated liquidity, and provide dynamic LP strategies. The work of \cite{bergault2023automated} uses traditional market making models to propose a new design of AMMs. In \cite{fan2021strategic} and \cite{ fan2022differential} the authors show that LPs can strategically exploit their beliefs on future exchange rates, and study the implications on smart contract design. Finally, the work in \cite{angeris2022convex, angeris2022optimal} uses convex optimisation for routing and multi-asset trading in CFMMs.

In practice, LTs usually trade in multiple assets simultaneously. At present, it is key for LTs that trade in CPMMs to consider exchange rates from more liquid alternative venues where prices are formed, and to take into account the significant co-movements that crypto-currencies exhibit. In this work, we solve the problem of an LT who wishes to execute large orders in multiple assets. She has access to different CPMM pools and uses information from  exchange rates from alternative trading venues (oracle). Moreover, we assume that the LT uses a set of predictive signals to improve the performance of his strategy; see \cite{forde2022optimal} and \cite{cartea2018enhancing} for the use of signals in optimal trading.

We formulate the optimal trading problem as a stochastic control problem where the LT controls the speed at which she trades in the different pools. Key to the performance of the trading strategy, is to balance price risk and execution costs. In CPMMs, the execution costs are approximated with the convexity of the trading function, which is a function of the pool's depth and the pool's exchange rate; see \cite{cartea2022decentralised}. In our model, the joint dynamics of the CPMM rates, the oracle rates, an a set of predictive signals follow a multi-dimensional Ornstein-Uhlenbeck process (multi-OU). Multi-OU dynamics are well suited to capture (i) the common stochastic trends that crypto-currencies usually share, (ii) the lead-lag effects between oracle and pool rates, and (iii) to incorporate predictive signals.

The remainder of this chapter is organised as follows. Section \ref{sec:AMM_multi} recalls how LTs and LPs interact with a CPMM. Section \ref{sec:data} uses Uniswap v3 data to study the cointegrated dynamics of oracle and pool rates for ETH and BTC. Section \ref{sec:Model_multi} introduces the optimal trading model and derives a closed-form approximation strategy. Section \ref{sec:performance_multi} implements the model of Section \ref{sec:Model_multi} using oracle rates from Binance and pool rates from Uniswap v3 for an agent in charge of a portfolio composed of the crypto-currencies ETH and BTC. In particular, we show that our strategy uses the cointegration factors in the joint dynamics of the exchange rates, and that it outperforms a classical Almgren-Chriss strategy without cointegration and a TWAP strategy.

\section{Constant function market makers  \label{sec:AMM_multi}}

In this section, we discuss how CFMMs operate and characterise their execution costs and the market impact of liquidity taking activity. We recall the mathematical properties of the trading function and how market participants interact with CFMM pools in charge of a pair of assets. 

In a CFMM, LPs and LTs interact through liquidity pooling. LPs deposit assets in a common pool, and LTs trade directly with the pool. In the remainder of this chapter, we consider $n$ assets $\bm{Y}=\left(Y^1, \dots, Y^n\right)^\intercal$ that are valued in terms of a reference asset $X$. Moreover, we consider a CFMM that offers a pool for every pair $\left(X,Y^i\right),$ and we refer to the instantaneous exchange rate of $Y^i$ in terms of $X$ in the pool by $Z^i\,.$ 

The liquidity pool for a pair $\left(X,Y^i\right)$ consists of a quantity $x>0$ of asset $X$ and a quantity $y>0$ of asset $Y^i.$ The CFMM pool is characterized by a trading function $f\left(x,y\right)$, which is differentiable, has convex level sets, and is increasing in both $x$ and $y$. A CFMM pool imposes an LT trading condition and an LP trading condition both of which define the state of the pool after an LT transaction and after an LP operation is completed. 

\paragraph{\textbf{LT trading condition}} Let $\left(x,y\right)$ be the state of the pool before the arrival of an LT buy order for a quantity $\Delta y>0$ of asset $Y^i$. The quantity $\Delta x>0$ of asset $X$ that the LT pays to the pool is determined by the LT trading condition $f\left(x,y\right) = f\left(x+\Delta x,y-\Delta y\right) = \kappa^2,$ where $\kappa>0$ is the \emph{depth} of the pool. For a sell order, the LT trading condition is $f\left(x,y\right) = f\left(x-\Delta x,y+\Delta y\right) = \kappa^2$. As in \cite{cartea2022decentralised, cartea2022decentralised2}, we define the \emph{level function} $\varphi$ as $x=\varphi\left(y\right)$, which is differentiable, convex, and decreasing.\footnote{To ensure that the pool cannot be depleted by LT trades, the level function should also satisfy $\lim_{y\to0} \varphi\left(y\right) = +\infty$ and $\lim_{y\to+\infty} \varphi\left(y\right) = 0$.}  Moreover, the instantaneous exchange rate is given by
\begin{align}\label{eq:instantaneousprice_multi}
    Z = -\varphi'\left(y\right)\,.
\end{align}

In Section \ref{sec:AMM}, we show that one can use a Taylor expansion to approximate the execution costs for trading a quantity $\Delta y$ with the convexity of $\varphi$, so that
\begin{equation}\label{eq:approx}
    \left|Z-\tilde{Z}\left(\Delta y\right)\right|\approx \frac{1}{2}\,\varphi''\left(y\right)\,|\Delta y| = \frac{1}{\kappa}\,Z^{3/2}\,|\Delta y|\,,
\end{equation}
where $\tilde{Z}\left(\Delta y\right)$ is the execution rate. The level function $\varphi$ is decreasing and convex, thus the execution rate received by an LT deteriorates (i.e., execution costs increase) as the size of the trade increases, and as the depth of the pool $\kappa$ decreases.

Finally, following an LT trade of size $\Delta y$, the \textit{price impact} is the difference between the two rates, i.e., $-\varphi'\left(y+\Delta y\right)+\varphi'\left(y\right)\,$ or $-\varphi'\left(y-\Delta y\right)+\varphi'\left(y\right)\,$. Using the definition in \eqref{eq:instantaneousprice_multi} and a Taylor's expansion we write 
\begin{align*}
    \text{Price impact} \approx\varphi''\left(y\right)\,|\Delta y| = \frac{2\,\,Z^{3/2}}{\kappa}\,|\Delta y|\,.
\end{align*}

\paragraph{\textbf{LP trading condition}} LPs interact with the pool by depositing assets in the pool or withdrawing assets from the pool.  The LP trading condition requires that LP operations do not impact the rate $Z$. However, LP operations do change the depth $\kappa$ of the pool. More precisely, let $\left(x, y\right)$ be the state of the pool before an LP deposits liquidity $\left(\Delta x, \Delta y\right)$,\footnote{Here, $\Delta x\ge0$ and $\Delta y\ge0$.} so $\left(x+\Delta x, y+\Delta y\right)$ is the state of the pool after the deposit, which increases the depth of the pool, because $$f\left(x+\Delta x,y+\Delta y\right) = \overline \kappa^2 > f\left(x,y\right) = \kappa^2\,;$$ recall that the trading function is increasing.

Next, let $\varphi_{\kappa}$ be the level function corresponding to the depth $\kappa$ and let $\varphi_{\overline{\kappa}}$ be the level function corresponding to the depth $\overline \kappa$. The LP trading condition requires that when an LP deposits the quantities $\left(x, y\right)\,$ such that $f\left(x+\Delta x,y+\Delta y\right) = \overline \kappa^2\,,$ then the following equality must hold:
\begin{align}\label{eq:LPtradingCondition_multi}
Z = -\varphi_{\kappa}'\left(y\right) = -\varphi_{\overline{\kappa}}'\left(y+\Delta y\right)\,.
\end{align}

To enforce this property, we require the trading function $f\left(x, y\right)$ to be homothetic. In practice, LPs usually invest an initial wealth $V$ in terms of the reference asset $X,$ so the amounts $\Delta x$ and $\Delta y$ are obtained from the LP trading condition and the equality $V = \Delta x + Z \, \Delta y.$ In contrast, LTs usually specify a quantity $\Delta x$ to trade, and receive or pay a quantity $\Delta y$ which is given by the LT trading condition.

	\section{Data analysis}\label{sec:data}
	In this section, we use data from the Uniswap v3 pools ETH/USDC $0.05\%$ and BTC/USDC $0.3\%\,,$ and from the LOB-based market Binance for the same pairs, to study the joint dynamics of crypto-currencies in both markets between 1 January 2022 and 30 June 2022. Uniswap v3 is a CPMM with concentrated liquidity  running on the Ethereum blockchain, which is a type of distributed ledger technology that validates and stores transactions. At present, it is the most liquid CPMM. In CPMMs, the  trading function is $f\left(x,y\right) = x \times y\,,$ so the level function is $\varphi(y)= \kappa^2 \big{/} y\,$ and the instantaneous exchange rate is $Z = -\varphi'(y) = \kappa^2 \big{/} {y}^2  = x \big{/} y\ .$ Binance is the most liquid LOB-based exchange for crypto-currencies, and the venue where prices are formed.  Table \ref{table:datadescr} shows descriptive statistics of the dataset.
    \begin{table}
        \begin{center}
            \begin{tabular}{c c c c} 
                Pair &  Volume & Volume & Depth  \\ [0.ex] 
                &  Uniswap & Binance & Uniswap\\ [0.5ex] 
                \hline 
                \\ [-1.8ex] 
                ETH/USDC & $\$\, 129.8 \times 10^9$  & $\$ \,267.44 \times 10^9$ & $2.04 \times  10^7 $ \\ [0ex]     
                & (1 txn per $13$ sec) & (1 txn per $0.13$ sec) & \\ [0ex]     
                BTC/USDC & $\$\, 5.1 \times 10^9$  & $\$ \,358.37 \times 10^9$ & $1.32 \times  10^6 $ \\ [0ex]     
                & (1 txn per $381$ sec) & (1 txn per $1.49$ sec) & \\ [0.2ex]   
                \hline
            \end{tabular}
        \end{center}
        \caption {Volume, trading frequency, and depth in Uniswap v3 and Binance for the pairs ETH/USDC and BTC/USDC between 1 January 2022 and 30 June 2022.}
        \label{table:datadescr}
    \end{table}
    
    \subsection{Lead-lag relationships between trading venues}
    
    Figure \ref{fig:LEADLAG} shows the dynamics of the pool and oracle exchange rates for both pairs throughout two hours of trading on 1 April 2023. The figure showcases the lead-lag relationship between both trading venues for a specific day, and suggests that the information from the oracle rates, i.e., those in Binance in our case, can be used to enhance the trading performance of an LT in the CPMM Uniswap v3. 

    The next subsections study the lead-lag relationships between trading venues by looking at cointegration, causality, and spillover effects.
\begin{figure}%
    \centering
    \includegraphics{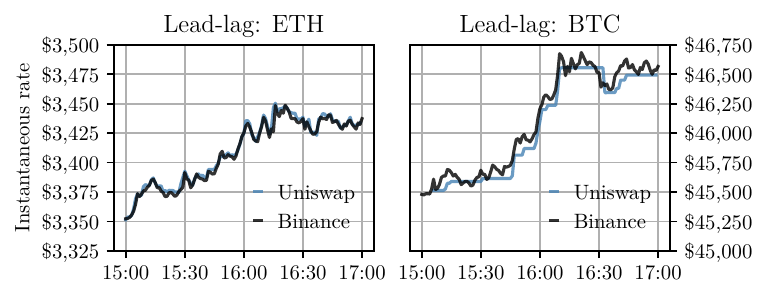}
    \caption{Exchange rates BTC/USDC and ETH/USDC in Uniswap v3 and Binance between 15:00 and 17:00 on 1 April 2022.}%
    \label{fig:LEADLAG}
\end{figure}

\subsubsection{Cointegration}\label{sec:coint}
To further study the relationship between the two trading venues, we use the Johansen's cointegration Trace test to compare the ETH/USDC and BTC/USDC exchange rates in Uniswap v3 and Binance; see \cite{johansen1991estimation}.

To run the Johansen's Trace tests, we first verify the assumption that the time series are integrated of order one. Table \ref{table:adf} shows the result of the augmented Dickey–Fuller stationarity test for each time series; see \cite{dickey1979distribution}. The test result show that all time series are not stationary and are integrated of order 1.
\begin{table}
    \begin{center}
        \begin{tabular}{c c c c c}
            & ETH/USDC & ETH/USDC & BTC/USDC & BTC/USDC\\
            & Binance & Uniswap v3 &  Binance & Uniswap v3\\ \hline 
            p-value &0.886&0.886&0.924&0.940\\
            test statistic &--0.532&--0.532&--0.312&--0.185\\ \hline
        \end{tabular}
    \end{center}
    \caption {{Augmented Dickey–Fuller test on each individual time series. The test concludes that each time series is not stationary and is integrated of order 1.}}
    \label{table:adf}
\end{table}

Table \ref{table:johansenTestResultsTable_venues} shows the results of two Johansen's tests. The first test compares the ETH/USDC exchange rates in Uniswap v3 and Binance, and the second test compares those of BTC/USDC. Both tests reject the hypothesis of cointegration of rank $r=0$, but fail to reject the hypothesis of cointegration of rank $r=1$. Thus, the tests show that there exists common stochastic trends between Uniswap v3 and Binance for the pairs we consider.
\begin{table}
    \begin{center}
        \begin{tabular}{c c || c c | c c}
            \hline &~&~&~&~&~\\[-1.8ex] 
            & &\multicolumn{2}{c|}{ETH/USDC} & \multicolumn{2}{c}{BTC/USDC} 
            \\ \hline ~&~&~&~&~&~\\[-1.8ex] 
            Null & Critical &  Trace & Result & Trace & Result\\
            Hypothesis & Value & statistic & & statistic &\\ [0.5ex] 
            \hline &~&~&~&~&~
            \\ [-1.8ex] 
            $r \leq 0$ & 15.49 & 5,369.79 & Rejected & 4,923.37 &  Rejected\\ [1ex]  
            $r \leq 1$ & 3.84 &  0.29 &  Not Rejected &  0.09 & Not Rejected \\ [1ex]
            \hline 
        \end{tabular}
    \end{center}
    \caption {{Results for two Johansen's cointegration Trace tests. The first test compares the ETH/USDC exchange rates in Uniswap v3 and Binance. The second test compares exchange rates for BTC/USDC in Uniswap v3 and Binance. Critical values are given for a significance level of $95\%$.}}
    \label{table:johansenTestResultsTable_venues}
\end{table}

\subsubsection{Causality}
Section \ref{sec:coint} uses Johansen's test to study the presence of common stochastic trends driving the exchange rates in Binance and Uniswap v3. Here, we investigate the causality between the exchange rates in both venues. 

Specifically, we run four Granger's causality tests; see \cite{granger1969investigating}. For each pair, we test whether historical rates in Binance can be used to forecast future rates in Uniswap v3, or vice versa. We perform Granger's test for $n=1$ to $n=5$ lags, where one lag corresponds to 5 minutes. Table \ref{table:granger} shows the test result. The test shows that for ETH/USDC, historical rates in Binance forecast rates in Uniswap v3 and vice versa. For BTC/USDC, the test shows that historical rates in Binance forecast rates in Uniswap v3,  but rates in Uniswap v3 do not forecast rates in Binance after 5 and 10 minutes.

The results of the Granger's causality tests suggest that for ETH/USDC prices in Uniswap v3 and Binance are aligned. On the other hand, the exchange rates in Uniswap v3 for BTC/USDC are delayed with respect to the exchange rates in Binance.

\begin{table}
    \begin{center}
        \begin{tabular}{|c |c | c |c |c|}
            \hline & \multicolumn{4}{c|}{ETH/USDC}
            \\[0.2ex] \hline ~& \multicolumn{2}{c|}{} &\multicolumn{2}{c|}{}\\[-1.8ex] 
            & \multicolumn{2}{c|}{Uniswap v3 $\not\to$ Binance} & \multicolumn{2}{c|}{Binance $\not\to$ Uniswap v3}
            \\[0.2ex] \hline ~&~&~&~&~\\[-1.8ex]
            Number of lags& p-value & Test statistic $\chi^2$ & p-value & Test statistic $\chi^2$
            \\[0.2ex] \hline ~&~&~&~&~\\[-1.8ex]
            $n=1$ & $<0.001$ & 105.09 & $<0.001$ & 1,317.03\\[0.2ex]
            $n=2$ & $<0.001$ & 182.78 & $<0.001$ & 1,870.48\\[0.2ex]
            $n=3$ & $<0.001$ & 231.26 & $<0.001$ & 2,003.54\\[0.2ex]
            $n=4$ & $<0.001$ & 232.55 & $<0.001$ & 2,123.15\\[0.2ex]
            $n=5$ & $<0.001$ & 243.18 & $<0.001$ & 2,156.65\\[0.2ex]
            
            \hline
        \end{tabular}\\
        \begin{tabular}{cc}
            &  \\
            & 
        \end{tabular}\\
        \begin{tabular}{|c |c | c |c |c|}
            \hline & \multicolumn{4}{c|}{BTC/USDC}
            \\[0.2ex] \hline ~& \multicolumn{2}{c|}{} &\multicolumn{2}{c|}{}\\[-1.8ex] 
            & \multicolumn{2}{c|}{Uniswap v3 $\not\to$ Binance} & \multicolumn{2}{c|}{Binance $\not\to$ Uniswap v3}
            \\[0.2ex] \hline ~&~&~&~&~\\[-1.8ex]
            Number of lags& p-value & Test statistic $\chi^2$ & p-value & Test statistic $\chi^2$
            \\[0.2ex] \hline ~&~&~&~&~\\[-1.8ex]
            $n=1$ & 0.1663 & 1.92 & $<0.001$ & 5,381.31\\[0.2ex]
            $n=2$ & 0.4766 & 1.48 & $<0.001$ & 6,479.78\\[0.2ex]
            $n=3$ & $<0.001$ & 18.59 & $<0.001$ & 6,656.85\\[0.2ex]
            $n=4$ & $<0.001$ & 23.15 & $<0.001$ & 6,883.71\\[0.2ex]
            $n=5$ & $<0.001$ & 26.97 & $<0.001$ & 6,963.54\\[0.2ex]
            
            \hline
        \end{tabular}
    \end{center}
    \caption {{Granger's causality test results. \textbf{Top panel:} The test shows that for ETH/USDC, historical rates in Binance can be used to forecast rates in Uniswap v3 and vice versa. \textbf{Bottom panel:} The test shows that for BTC/USDC, historical rates in Binance can be used to forecast rates in Uniswap v3,  but rates in Uniswap v3 cannot be used to forecast rates in Binance after 5 and 10 minutes.}}
    \label{table:granger}
\end{table}

\subsubsection{Spillover effects}
Next, we study spillover effects between Uniswap v3 and Binance. Spillover effects are the consequences that shocks in a random variable have on other variables. Here, we look at the spillover indices introduced in \cite{diebold2012better} and we refer to \cite{giudici2020vector} for a detailed analysis of spillover effects in BTC/USDC between various centralised trading venues.

First, we recall the definition of generalised impulse response function introduced in \cite{koop1996impulse}. For each pair, let $S$ be the oracle from Binance and let $Z$ be the exchange rate in Uniswap v3. We represent the time series of the exchange rates returns using a VAR($k$) model
\begin{equation}\label{eq:var_rep}
    \begin{pmatrix}
        \Delta Z_t\\ \Delta S_t
    \end{pmatrix} = \sum_{i=1}^{k} \bm\Phi_{i}\begin{pmatrix}
    \Delta Z_{t-i}\\ \Delta S_{t-i}
\end{pmatrix} + \varepsilon_t\,,
\end{equation}
where $\bm\Phi_{1},\dots,\bm\Phi_{k}$ are $2\times2$ regression matrices and $\bm\varepsilon_t$ is the error term, which is distributed as $\mathcal{N}\left(0,\bm\Sigma\right)$, where
\begin{equation}
    \bm\Sigma = 
    \begin{pmatrix}
        \sigma_{ZZ} & \sigma_{ZS}\\
        \sigma_{ZS} & \sigma_{SS}
    \end{pmatrix}
\end{equation}
is the covariance matrix. Next, we write the moving average representation of \eqref{eq:var_rep} as
\begin{equation}\label{eq:ma_rep}
    \begin{pmatrix}
        \Delta Z_t\\ \Delta S_t
    \end{pmatrix} = \sum_{k=1}^{\infty} \bm A_{k}\,\bm\varepsilon_{t-k}\,,
\end{equation}
where the coefficient matrices $\bm A_{k}$ satisfy
\begin{equation}\label{key}
    \bm A_{k} = 
    \begin{cases}
        0 & \text{for }k<0\,, \\
        I_2 & \text{for }k=0\,, \\
        \displaystyle	\sum_{l=1}^{k}\bm\Phi_{l}\,\bm A_{k-l}& \text{for }k>0\,.\\
    \end{cases}
\end{equation}

The generalised impulse response function of $\left(\Delta Z_t, \Delta S_t\right)$ at horizon $n$ is defined as
\begin{equation}\label{eq:gi}
    GI_{n}(\bm{\delta}) = \E_{t-1}\left[\begin{aligned}
        \Delta Z_{t+n}\\ \Delta S_{t+n}
    \end{aligned}\middle|\bm\epsilon_t=\bm{\delta}\right]-
\E_{t-1}\left[\begin{aligned}
    \Delta Z_{t+n}\\ \Delta S_{t+n}
\end{aligned}\right] = \bm A_{n}\,\bm{\delta}\,,
\end{equation}
and it represents the effect that a shock $\bm{\delta}$ at time $t$ has on the expected future return at time $t+n$. Next, use $\bm\epsilon\sim\mathcal{N}\left(0,\bm\Sigma\right)$ to write
\begin{equation}\label{eq:epsilon}
    \E\left[\bm\epsilon_t\middle|\epsilon_t^i=\delta_i\right] = \frac{1}{\sigma_{ii}^2}\,\bm{\Sigma}\,e_{i}\,\delta_i\,,\quad \text{for }i\in\{Z,S\}\,.
\end{equation}
Next, set $\delta_{i} = \sigma_{ii}$ and use \eqref{eq:gi} and \eqref{eq:epsilon} to obtain, for each $n\in\N$, the scaled impulse response function
\begin{equation}\label{eq:scaled}
    \psi^{g}_i(n)=\frac{1}{\sigma_{ii}}\bm A_{n}\,\bm\Sigma\,\bm{e}_i\,, \quad \text{for }i\in\{Z,S\}\,. 
\end{equation}
Here, for $i\in\{Z,S\}$, $\psi^{g}_i(n)$ in \eqref{eq:scaled} measures the effect that one standard error shock in the returns of the trading venue $i$ has on $\left(\Delta Z, \Delta S\right)$ after $n$ time steps.

Next, we use \eqref{eq:scaled} to derive, for each $n\in\N$, the generalised forecast error variance decomposition
\begin{equation}\label{eq:fevd}
    \theta_{ij}^{g}(n)=\frac{1}{\sigma^2_{ii}}\times \frac{\sum_{l=0}^{n}(\bm e^\intercal_i\,\bm A_l\,\bm\Sigma\,\bm e_j)^2}{\sum_{l=0}^{n}(\bm e^\intercal_i\,\bm A_l\,\bm\Sigma\,\bm A_l^\intercal\,\bm e_i)^2}\,, \quad \text{for }i,j\in\{Z,S\}\,,
\end{equation}
and we refer the reader to \cite{lutkepohl1991introduction} and \cite{pesaran1998generalized} for more details. Here, $\theta_{ij}^{g}(n)$ represents the proportion of the $n$-step-ahead forecast error variance of the returns in venue $i$ which is accounted for by shocks in venue $j$. Finally, normalise the generalised forecast error variance decomposition to obtain 
\begin{equation}\label{eq:normalised}
    \tilde \theta_{ij}^{g}(n) = \frac{\theta_{ij}^{g}(n)}{\displaystyle\sum_{j\in\{Z,S\}}\theta_{ij}^{g}(n)}\,.
\end{equation}

Next, define the total spillover index as
\begin{equation}\label{eq:tsi}
    TSI(n) = \frac{\displaystyle\sum_{\substack{i,j\in\{Z,S\} \\ i\neq j}}\tilde \theta_{ij}^{g}(n)}{\displaystyle\sum_{i,j\in\{Z,S\}}\tilde \theta_{ij}^{g}(n)}\times 100\,,
\end{equation}
which represent the percentage of total $n$-step ahead forecast error variance which is caused by spillovers between trading venues. Then, for each $i\in\{Z,S\}$, define the directional spillover indices as
\begin{equation}\label{eq:dsi}
    DSI_{i\to J}(n) = \frac{\displaystyle\sum_{\substack{j\in\{Z,S\} \\ i\neq j}} \tilde \theta_{ji}^{g}(n)}{\displaystyle \sum_{i,j\in\{Z,S\}}\tilde \theta_{ij}^{g}(n)}\times 100\,,\quad DSI_{J\to i}(n) = \frac{\displaystyle\sum_{\substack{j\in\{Z,S\} \\ i\neq j}}\tilde\theta_{ji}^{g}(n)}{\displaystyle\sum_{i,j\in\{Z,S\}}\tilde \theta_{ij}^{g}(n)}\times 100\,.
\end{equation}
Finally, for each $i\in\{Z,S\}$, define the net spillover index as
\begin{equation}\label{eq:nsi}
    NSI_{i}(n) = DSI_{i\to J}(n) - DSI_{J\to i}(n)\,.
\end{equation}
The directional spillover index $DSI_{i\to J}(n)$ represents the percentage of total $n$-step ahead forecast error variance which is caused by spillovers from $i$ to the other trading venues. On the other hand, the directional spillover index $DSI_{J\to i}(n)$ represents the percentage of total $n$-step ahead forecast error variance which is caused by spillovers from other trading venues to $i$. 

\begin{figure}%
    \centering
    \includegraphics{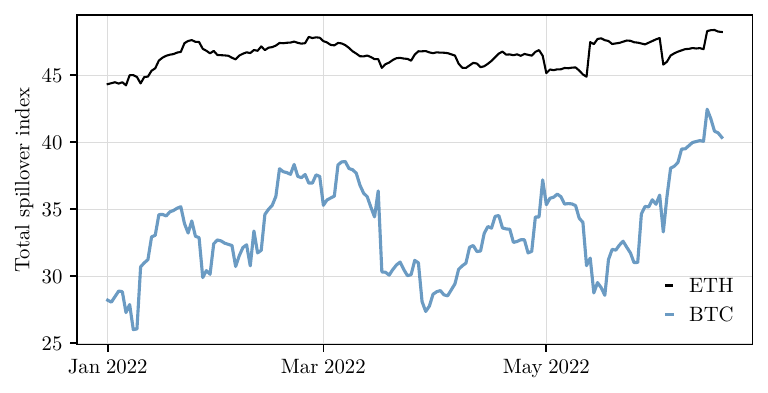}
    \caption{ 10-step-ahead total spillover index between Uniswap v3 and Binance for ETH/USDC and BTC/USDC.}%
    \label{fig:tsi}%
\end{figure}
For ETH/USDC and BTC/USDC, we compute the 10-step-ahead spillover indices between Uniswap v3 and Binance over a one-day period. Then, using rolling windows of two hours, we look at how spillover indices evolve. Figure \ref{fig:tsi} shows how the total spillover index between Uniswap v3 and Binance evolve over time for ETH/USDC and BTC/USDC. The figure shows that for the ETH/USDC pair, the total spillover index is higher than 45\% for almost the entire time window considered. This suggests that both trading venues are equally important in determining the exchange rate. On the other hand, for the BTC/USDC pair, the total spillover index is for almost the whole time window below 40\% and occasionally drops below 30\%. This, indicates that prices are formed in Binance, with returns in Uniswap v3 adding only little information. This claim about the asymmetry in the role of the two trading venues is confirmed by Granger's causality test in Table \ref{table:granger}.

\begin{figure}%
    \centering
    \includegraphics{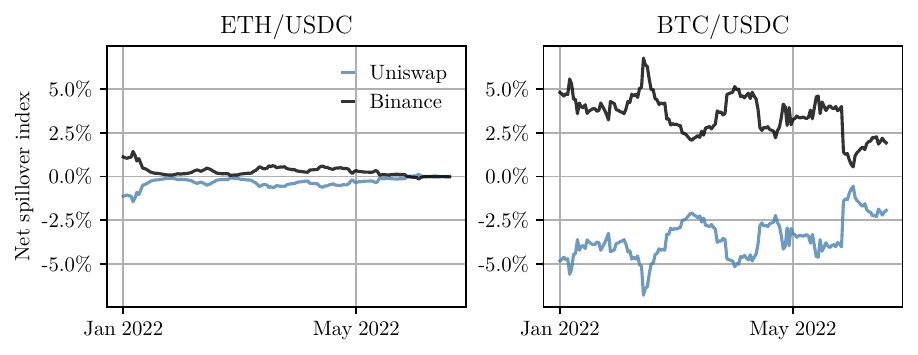}
    \caption{ 10-step ahead net spillover index between Uniswap v3 and Binance for ETH/USDC and BTC/USDC.}%
    \label{fig:nsi}%
\end{figure}
Figure \ref{fig:nsi} shows the evolution of the net spillover indices between Uniswap v3 and Binance over time for ETH/USDC and BTC/USDC. The figure shows that, for both pairs, shocks in rates in Binance have larger impact than shocks in Uniswap v3. This shows that exchange rates in Binance have higher forecasting power than exchange rates in Uniswap v3. Thus, an LT who trades in Uniswap v3 can enhance the performance of her strategy by using the exchange rate in Binance as a predictive signal.  This further motivates the model we introduced in Section \ref{sec:Model}, where we assume that price formation is in the alternative trading venue and not in the CPMM.
	
	\subsection{Cointegration and oracle rates}
	Figure \ref{fig:coint} shows the rates for both BTC and ETH in terms of USDC in the pools we consider. The figure suggests that the exchange rates share common stochastic trends.
	
	\begin{figure}%
		\centering
		\includegraphics{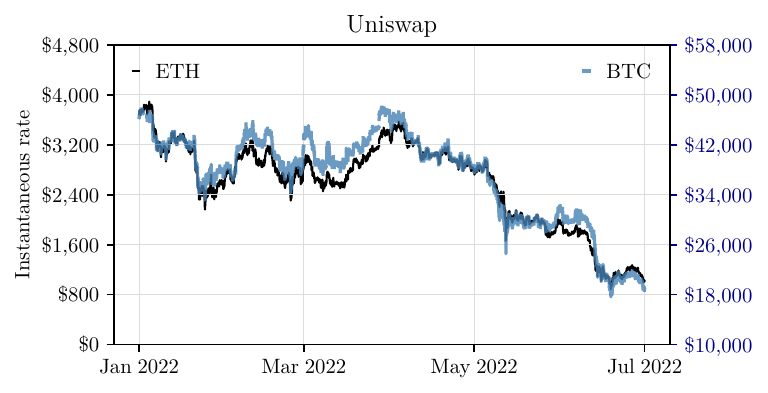}
		\caption{ Exchange rates ETH/USDC and BTC/USDC  in Uniswap v3 between January 2022 and June 2022.}%
		\label{fig:coint}%
	\end{figure}
	
	In our model, we capture the joint dynamics of the pool rates and the oracle rates using a multi-OU process. This process is suited to capture equilibrium relationships as those exhibited between BTC and ETH, and to capture lead-lag effects between different trading venues as those exhibited between Uniswap and Binance. The lead-lad effects are encoded in the cointegration factors in the mean-reversion matrix, which drives the deterministic and long-term component of the process. 
	
	To further study the suitability of the multi-OU process, we use the data described above to perform a Johansen's cointegration Trace test for the vector $\bm R = \left(Z^1, Z^2, S^1, S^2\right)$, where $Z^1$ and $S^1$ are the pool and oracle rates of ETH, and $Z^2$ and $S^2$ are the pool and oracle rates of BTC. Johansen's test rejects the hypothesis of no cointegration and the hypothesis of a cointegration rank $r = 1\,,$ but does not reject a cointegration rank $r = 2$; see \cite{johansen1991estimation}. 
	
	The cointegration vectors are $(1,  -0.022, -1.002, 0.022)$ and $(1, 0.09, -1.004, -0.088)$. The cointegration vectors, and hence the mean-reversion matrix  (i.e., the matrix $\boldsymbol\beta$ in \eqref{eq:dynmultiOU}), capture the joint dynamics of BTC and ETH, but also the lead-lag relation between the pool rates and oracle rates. Table \ref{table:johansenTestResultsTable} shows the results for the Johansen's cointegration Trace test. 
	\begin{table}
		\begin{center}
			\begin{tabular}{c r c c} 
				Null Hypothesis &  Trace statistics & Critical Value & Result\\ [0.5ex] 
				\hline 
				\\ [-1.8ex] 
				$r \leq 0$ & $11,196.97$ & $54.68$ & Rejected \\ [1ex]  
				$r \leq 1$ & $4,821.52$ & $35.46$ & Rejected \\ [1ex] 
				$r \leq 2$ & $11.45$ & $19.93$ & Not rejected \\ 
				\hline
			\end{tabular}
		\end{center}
		\caption {Johansen's cointegration Trace test results (critical values are given for a significance level of $99\%$). The test confirms a cointegration rank $r=2$.}
		\label{table:johansenTestResultsTable}
	\end{table}
	
	Thus, it is key for an LT who trades multiple crypto-currencies to account for external information from similar assets and from other trading venues.
	
	\section{Optimal portfolio trading with oracles and signals \label{sec:Model_multi}}
	In this section, we solve the problem of an LT in charge of liquidating a large position in a basket of  $n$ risky and cointegrated assets $\bm{Y} = \left(Y^1,\dots,Y^n\right)$ in a CPMM. The CPMM offers $n$ pools for each asset against a reference asset $X\,$ and we denote by $(\bm Z_t)_{t \in [0,T]} = \left(Z^1,\dots,Z^n\right)_{t \in [0,T]}$ the exchange rates in these pools. The LT uses $n$ oracle rates $(\bm S_t)_{t \in [0,T]} = \left(S^1,\dots,S^n\right)_{t \in [0,T]}$  from another more liquid exchange and a set of $m$ predictive signals $(\bm I_t)_{t \in [0,T]} = \left(I^1,\dots,I^m \right)_{t \in [0,T]}\,.$ 
	
	We fix a filtered probability space $\left(\Omega, \mathcal F, \P; \F = (\mathcal F_t)_{t \in [0,T]} \right)$ satisfying the usual conditions, where $\F$ is the natural filtration generated by the collection of observable stochastic processes that we define below.
	
	The LT must liquidate her initial inventory $\bm{\tilde{y}}_0\in\R^n$ over a period of time $[0,T]$ and her wealth is valued in terms of asset $X$. The joint dynamics of the pool rates $\bm Z\,,$ the oracle rates $\bm S\,,$ and the signals $\bm I$ are modelled by a $2n+m$-dimensional vector $(\bm R_t)_{t \in [0,T]}$ that follows multi-OU dynamics
	\begin{equation}\label{eq:dynmultiOU}
		d\bm{R}_t = \bm{\beta}\,\left(\bm{\mu}-\bm{R}_t\right)\,dt + \bm{\sigma}^\intercal\, d\bm{W}_t\,,
	\end{equation}
	where we force the first $n$ components of $\bm R$ to coincide with $\bm Z\,,$ $\bm{\beta}$ is a $2n+m\times2n+m$ mean-reversion matrix, $\bm{\mu}$ is a $2n+m$-dimensional vector that models the long-term unconditional mean, $\bm{\sigma}^\intercal$ is the Cholesky decomposition of the asset prices correlation matrix $\bm{\Sigma}$\,, the symbol $^\intercal$ denotes the transpose operator, and $(\bm{W}_t)_{t \in [0,T]}$ is a $2n+m$-dimensional Brownian motion with independent coordinates. For previous work with multi-OU price dynamics see \cite{cartea2018trading}, \cite{bergault2022multi}, and \cite{drissi2022solvability}.
	
	The LT trades at a continuous $n$-dimensional speed $(\bm{\nu}_t)_{t \in [0,T]}$ and the dynamics of her holdings are given by
	\begin{equation}
		d \bm{\tilde y}_t = -\bm{\nu}_t \,dt\, .
	\end{equation}
	We do not restrict the speed to be positive; if for some $i\in\{1,\dots,n\}$ we have $\nu^{i}>0$ the LT sells asset $Y^{i}$ and if $\nu^{i}<0$ the LT buys the asset.
	
	As discussed in Section \ref{sec:AMM2_2}, the rate impact of the LT's activity is encoded in the trading function $f\left(x,y\right)$ and is a function of the trading speed as in traditional execution models (see \cite{cartea2015book}), the pool rate $\bm Z$, and the pool depth $\kappa$; see \cite{cartea2022decentralised}.
	
	
	More precisely, for each asset $Y^i$, we write the difference between the execution rate $\tilde Z^i$  and the instantaneous rate $Z^i$ as
	\begin{equation}\label{eqn:price_impact_multi}
		\tilde Z^i - Z^i= - \frac{\eta}{\kappa^i}\, \left(Z^i\right)^{3/2}\,\nu^i\,,
	\end{equation}
	where $\kappa^i$ is the depth of the pool that offers liquidity for the pair $\left(X,Y^i\right).$ Thus, the dynamics of the LT's holdings of asset $X$ are given by 
	\begin{equation}\label{eq:wealth_multi}
		\begin{split}
			d\tilde x_t &= \bm{\tilde Z}_t^\intercal\, \bm{\nu}_t\,dt=\left(\bm{Z}_t-\mathfrak{D}\left(\frac{\eta}{\bm{\kappa}} \odot \bm{Z}_{t}^{3/2}\right)\,\bm{\nu}_{t}\right)^\intercal\bm{\nu}_t\,dt\, ,
		\end{split}
	\end{equation}
	where $\bm{\kappa}$ is the $n$-dimensional vector of the depth of the pools, $\eta$ is a parameter which depends on the LT's trading frequency, $\mathfrak{D}(\bm{c})$ denotes a diagonal matrix whose diagonal elements are equal to the vector $\bm{c}$\,, and $\odot$ denotes the component-wise product. Compared with previous works in the literature that employ multi-OU price dynamics, here market impact parameter is a deterministic function of the stochastic vector $\bm Z$. 
	
	Let $\bm{\mathcal{X}}$ be a $n\times (2n+m)$ matrix with $\bm{\mathcal{X}}_{ij}=\mathbbm{1}_{\{i=j\}}$ which maps the first $n$ elements of a $2n+m$-dimensional vector into an $n$ dimensional vector, and write \eqref{eq:wealth_multi} in terms of the vector $\bm{R}$ as
	\begin{equation}\label{eq:wealth_new}
		\begin{split}
			d\tilde x_t &=\left(\bm{\mathcal{X}}\,\bm{R}_t-\mathfrak D\left(\frac{\eta}{\bm{\kappa}} \odot \left(\bm{\mathcal{X}}\,\bm{R}_t^{3/2}\right)\right) \,\bm{\nu}_{t}\right)^\intercal\bm{\nu}_t\,dt\, .
		\end{split}
	\end{equation}
	
	\subsection{Performance criterion and value function \label{sec:model:optimal_multi}}
	In our model, the LT maximises her expected terminal wealth in units of $X$ while penalising inventory in the risky assets $\bm{Y}$. The set of admissible strategies is
	\begin{equation}\label{def:admissibleset_t_multi}\mathcal A_t = \left\{
		(\bm{\nu}_s)_{s \in [t,T]} \ \middle\vert \begin{array}{l}
			\R^n\textrm{-valued},\ \F\textrm{-adapted, and } \displaystyle\int_t^T \lVert \bm{\nu}_s \rVert^2 \,ds < +\infty, \ \ \textrm{a.s.}
		\end{array}\right\}\,,
	\end{equation}
	and we write $\mathcal A := \mathcal A_0\,.$ Let $\bm{\nu}\in\Ac$, the performance criterion of the LT who trades at speed $\bm{\nu}$, is the function $u^{\nu}\colon[0,T] \times \R \times \R^n \times \R_{++}^{2n+m} \rightarrow \R\ $ given by
	\begin{align}
		\label{eq:perfcriteria_model_multi}
		u^{\bm{\nu}}(t,\tilde{x},\bm{\tilde{y}},\bm{R})=\E_{t,\tilde{x},\bm{\tilde{y}},\bm{R}}\Big [&\tilde{x}_{T}^{\bm{\nu}}+ \bm{R}_T^\intercal\,\bm{\mathcal{X}}^\intercal\, \bm{\tilde{y}}_{T}^{\bm{\nu}} -\left(\bm{\tilde{y}}_{T}^{\bm{\nu}}\right)^\intercal\bm{\alpha}\,\bm{\tilde{y}}_{T}^{\bm{\nu}}-\phi\,\int_{t}^{T}\left(\bm{\tilde{y}}_{s}^{\bm{\nu}}\right)^\intercal\bm{\tilde{\Sigma}}\,\bm{\tilde{y}}_{s}^{\bm{\nu}}\,ds\Big]\,,
	\end{align}
	and the LT's value function is 
	\begin{align}
		\label{eq:valuefunc_multi}
		u(t,\tilde{x},\bm{\tilde{y}},\bm{R})=\underset{\nu\in\mathcal{A}}{\sup} \, u^{\nu}(t,\tilde{x},\bm{\tilde{y}},\bm{R})\,.
	\end{align}
	
	The first term on the right-hand side of \eqref{eq:perfcriteria_model_multi} is the amount of asset $X$ that the LT holds at the end of the trading window. The second term represents the terminal value of her remaining inventory marked-to-market using the pool rates $\bm{Z}$. The third term represents an inventory penalty where the diagonal positive matrix $\bm{\alpha}$ quantifies the aversion of the LT to holding non-zero inventory at time $T$; the units of $\bm{\alpha}$ are such that the penalty is in units of $X$. Finally, the last term is a running inventory penalty where $\bm{\tilde{\Sigma}}$ is the covariance matrix of the pool rates $\bm{Z}$, and the parameter $\phi \geq 0$ quantifies the urgency of the LT to liquidate inventory; the units of $\phi$ are such that the penalty is in units of $X$. Note that the third term on the right-hand side of \eqref{eq:perfcriteria_model_multi} is sometimes interpreted in the algorithmic trading literature as the `cost' of liquidating the final inventory $\bm{\tilde{y}}_T$ at time $T.$ Using this interpretation, and in light of the results of the previous section, this term should be proportional to $\mathcal{D}\left(\frac{\eta}{\bm{\kappa}} \odot \left(\bm{\mathcal{X}}\,\bm{R}_t^{3/2}\right)\right)$. Here we interpret this term as penalisation term and we observe that for $\bm{\alpha}$ sufficiently large, $\left(\bm{\tilde{y}}_{T}^{\bm{\nu}}\right)^\intercal\bm{\alpha}\,\bm{\tilde{y}}_{T}^{\bm{\nu}}$ is greater than expected execution costs of terminal inventory.
	
	The value function solves the Hamilton--Jacobi--Bellman (HJB) equation
	\begin{equation}\label{eq:hjbu_multi}
		\begin{split}
			0=\,&\partial_{t}w-\phi\,\bm{\tilde{y}}^\intercal\bm{\tilde{\Sigma}}\,\bm{\tilde{y}}+ \,\left(\bm{\mu}-\bm{R}\right)^\intercal\bm{\beta}\,\partial_{\bm{R}}w+\frac{1}{2}\,\Tr\left(\bm{\Sigma}\,D_{\bm{R}\bm{R}}w\right) \\
			&+\sup_{\nu\in\R^n}\Bigg(-\bm{\nu}^\intercal\partial_{\bm{\tilde{y}}}w+\bigg(\bm{\mathcal{X}}\,\bm{R}-\mathfrak D\left(\frac{\eta}{\bm{\kappa}} \odot \left(\bm{\mathcal{X}}\,\bm{R}^{3/2}\right)\right) \,\bm{\nu}\bigg)^\intercal\bm{\nu}\,\partial_{\tilde{x}}w\Bigg)\ ,
		\end{split}
	\end{equation}
	with terminal condition \begin{align} \label{eq:termcondu_multi}
		w(T,\tilde{x},\bm{\tilde{y}},\bm{R})=\tilde{x}+ \bm{R}^\intercal\,\bm{\mathcal{X}}^\intercal\, \bm{\tilde{y}} -\bm{\tilde{y}}^\intercal\,\bm{\alpha}\,\bm{\tilde{y}}\, .
	\end{align}

	Next, substitute the  ansatz 
	\begin{align}
		\label{eq:ansatz1_multi}
		w(t,\tilde{x},\bm{\tilde{y}},\bm{R})\,=\,\tilde{x}+ \bm{R}^\intercal\,\bm{\mathcal{X}}^\intercal\, \bm{\tilde{y}}+\theta(t,\bm{\tilde{y}},\bm{R})\ ,
	\end{align}
	into \eqref{eq:hjbu_multi} and solve the first order condition to obtain the partial differential equation (PDE)
	\begin{align}\label{eq:hjbtheta1_multi}
		0\,=\,&\partial_{t}\theta-\phi\,\bm{\tilde{y}}^\intercal\bm{\tilde{\Sigma}}\,\bm{\tilde{y}}+ \,\left(\bm{\mu}-\bm{R}\right)^\intercal\bm{\beta}\,\left(\bm{\mathcal{X}}^\intercal\bm{\tilde y}+\partial_{\bm{R}}\theta\right)\\
		&+\frac{1}{2}\,\Tr\left(\bm{\Sigma}\,D_{\bm{R}\bm{R}}\theta\right)+\partial_{\bm{\tilde{y}}}\theta^\intercal\,\mathfrak D \left(\frac{\bm{\kappa}}{4\,\eta} \odot \left(\bm{\mathcal{X}}\,\bm{R}^{-3/2}\right)\right)\,\partial_{\bm{\tilde{y}}}\theta\, ,\nonumber
	\end{align}
	with terminal condition
	\begin{equation}\label{eq:terminal1_multi}
		\theta(T,\bm{\tilde{y}},\bm{R})=-\bm{\tilde{y}}^\intercal\,\bm{\alpha}\,\bm{\tilde{y}}\, ,
	\end{equation}
	where the optimal trading speed in feedback form is
	\begin{equation}\label{eq:optimalspeed_multi}
		\nu^{\star}=-\mathfrak D \left(\frac{\bm{\kappa}}{2\,\eta} \odot \left(\bm{\mathcal{X}}\,\bm{R}^{-3/2}\right)\right)\,\partial_{\bm{\tilde{y}}}\theta\ .
	\end{equation}

	The semi-linear PDE \eqref{eq:hjbtheta1_multi} does not admit a closed-form solution. In practice, one uses a numerical scheme to approximate the solution $\theta$ and deduce the optimal speed \eqref{eq:optimalspeed_multi}. In our model, the dimensionality of the problem grows exponentially with the number of assets. As a consequence, numerical schemes suffer from the curse of dimensionality.
	
	In practice, the performance of an execution or statistical arbitrage strategy heavily relies on the ability to compute and send liquidity taking instructions with minimal latency. Thus, it is crucial for an LT to employ a trading strategy that is relatively fast to compute. Hence, we derive a closed-form approximation strategy for \eqref{eq:optimalspeed_multi} which can be implemented and tracked by the LT in real time. Therefore, we extend the method developed in Section \ref{sec:Model} for a single asset framework.
	
	Subsection \ref{sec:approxstrat} solves an execution problem when convexity costs are fixed for a value of the pool rate $Z\,,$ and Subsection \ref{sec:closedform} defines the closed-form approximation strategy as the uniform limit of a family of these strategies.
	
	\subsection{Optimal strategy with deterministic execution costs \label{sec:approxstrat}}

	Here, we derive a trading strategy that assumes a deterministic and fixed cost parameter in \eqref{eqn:price_impact_multi}.  More precisely, consider here that the execution cost in \eqref{eqn:price_impact_multi} is
	\begin{equation}\label{eq:impact_deterministic_multi}
		\bm{\tilde Z} - \bm{Z}= - \eta\,\mathfrak D(\bm{\zeta})\,\bm{\nu}\,.
	\end{equation}
	Here the $n$-dimensional vector $\bm{\zeta} > 0$ is a fixed cost parameter. Our aim is to derive the optimal strategy $\left(\bm{\nu}^{\star,\bm{\zeta}}_t\right)_{t \in [0,T]}$ for a given value of the parameter $\bm{\zeta}>0$. For each value of $\bm{\zeta}$, we consider the set of admissible strategies $\mathcal A_t^{\bm{\zeta}}$ similar to that in \eqref{def:admissibleset_t_multi} and we write $\mathcal{A}^{\bm{\zeta}}:= \mathcal A_0^{\bm{\zeta}}.$ Let $\bm{\nu}^{\bm{\zeta}}\in\Ac^{\bm{\zeta}}$. For an  LT who trades at speed $\big(\bm{\nu}_t^{\bm{\zeta}}\big)_{t \in [0,T]}$, the inventory $\big(\tilde y_t^{\zeta}\big)_{t \in [0,T]}$ evolves as
	\begin{align}
		\label{eq:ytildeProcessjN_multi}
		d \bm{\tilde y}^{\bm{\zeta}}_t = -\bm{\nu}^{\bm{\zeta}}_t \,dt\, ,
	\end{align}
	and her holdings of asset $X$ follow the dynamics
	\begin{equation}\label{eq:wealth_AC_multi}
		\begin{split}
			d\tilde x_t^{\bm{\zeta}} &=\left(\bm{\mathcal{X}}\,\bm{R}_t-\eta\,\mathfrak D\left(\bm{\zeta}\right) \,\bm{\nu}_{t}^{\bm{\zeta}}\right)^\intercal\bm{\nu}_t^{\bm{\zeta}}\,dt\, .
		\end{split}
	\end{equation}
	
	We consider the same performance criterion as that in \eqref{eq:perfcriteria_model_multi} and let 
	$u_{\bm{\zeta}}$ be the value function  of the LT. To solve the new optimisation problem for each value of $\bm{\zeta}$, follow the same steps as above to obtain the PDE verified by $\theta_{\bm{\zeta}}$, where $w_{\bm{\zeta}}(t,\tilde{x},\bm{\tilde{y}},\bm{R})=\tilde{x}+ \bm{R}^\intercal\,\bm{\mathcal{X}}^\intercal\, \bm{\tilde{y}}+\theta_{\bm{\zeta}}(t,\bm{\tilde{y}},\bm{R})\,,$ and write
	\begin{equation}\label{eq:hjbtheta2_multi}
		\begin{split}
			0\,=\,&\partial_{t}\theta_{\bm{\zeta}}-\phi\,\bm{\tilde{y}}^\intercal\,\bm{\tilde{\Sigma}}\,\bm{\tilde{y}}+ \,\left(\bm{\mu}-\bm{R}\right)^\intercal\,\bm{\beta}\,\left(\bm{\mathcal{X}}^\intercal\,\bm{\tilde y}+\partial_{\bm{R}}\theta_{\bm{\zeta}}\right)\\
			&+\frac{1}{2}\,\Tr\left(\bm{\Sigma}\,D_{\bm{R}\bm{R}}\,\theta_{\bm{\zeta}}\right)+\frac{1}{4\,\eta}\,\partial_{\bm{\tilde{y}}}\,\theta_{\bm{\zeta}}^\intercal\,\mathfrak D\left(\bm{\zeta}\right)^{-1}\,\partial_{\bm{\tilde{y}}}\theta_{\bm{\zeta}}\ ,
		\end{split}
	\end{equation}
	with terminal condition
	\begin{equation}\label{eq:terminal2_multi}
		\theta_{\bm{\zeta}}(T,\bm{\tilde{y}},\bm{R})=-\bm{\tilde{y}}^\intercal\bm{\alpha}\,\bm{\tilde{y}}\, ,
	\end{equation}
	where the LT's optimal trading speed in feedback form is
	\begin{equation}\label{eq:optimalspeed_multi_deterministic}
		\bm\nu^{\star,\bm{\zeta}}=-\frac{1}{2\,\eta}\,\mathfrak D\left(\bm{\zeta}\right)^{-1}\,\partial_{\bm{\tilde{y}}}\theta_{\bm{\zeta}}\ .
	\end{equation}

	To further study the problem, substitute the ansatz
	\begin{equation}\label{eq:ansatz2_multi}
		\begin{split}
			\theta_{\bm{\zeta}}(t,\tilde{x},\bm{\tilde{y}},\bm{R})=&\,\bm{\tilde{y}}^{\intercal}A_{\bm{\zeta}}(t)\,\bm{\tilde{y}} +\bm{\tilde{y}}^{\intercal}B_{\bm{\zeta}}(t)\,\bm{R}+C_{\bm{\zeta}}(t)\,\bm{\tilde{y}}+\bm{R}^{\intercal}D_{\bm{\zeta}}(t)\,\bm{R}+E_{\bm{\zeta}}(t)\,\bm{R}+F_{\bm{\zeta}}(t)\,,
		\end{split}
	\end{equation}
	in \eqref{eq:hjbtheta2_multi} to obtain the  system of ODEs
	\begin{equation}\label{eq:system}
		\left\{
		\begin{aligned}
			0=&A'_{\bm{\zeta}}(t)-\phi\,\bm{\tilde{\Sigma}}+\frac{1}{\eta}A_{\bm{\zeta}}(t)^{\intercal}\mathfrak{D}\left(\bm{\zeta}\right)^{-1}A_{\bm{\zeta}}(t)\,,\\
			0=&B_{\bm{\zeta}}'(t)-\bm{\mathcal{X}}\,\bm{\beta}^{\intercal}-B_{\bm{\zeta}}(t)\,\bm{\beta}^{\intercal}+\frac{1}{\eta}\,A_{\bm{\zeta}}(t)^{\intercal}\,\mathfrak{D}\left(\bm{\zeta}\right)^{-1}B_{\bm{\zeta}}(t)\,,\\
			0=&C_{\bm{\zeta}}'(t)+\bm{\mu}^{\intercal}\,\bm{\beta} \,\bm{\mathcal{X}}^{\intercal}+ \bm{\mu}^{\intercal}\,\bm{\beta}B_{\bm{\zeta}}(t)^{\intercal}+\frac{1}{\eta}\,C_{\bm{\zeta}}(t)^{\intercal}\,\mathfrak{D}\left(\bm{\zeta}\right)^{-1}A_{\bm{\zeta}}(t)\,,\\
			0=&D_{\bm{\zeta}}'(t)+\frac{1}{4\,\eta}B_{\bm{\zeta}}(t)^{\intercal}\,\mathfrak{D}\left(\bm{\zeta}\right)^{-1}B_{\bm{\zeta}}(t)\,,\\
			0=&E_{\bm{\zeta}}'(t)-E_{\bm{\zeta}}(t)^{\intercal}\,\bm{\beta}^{\intercal}+\frac{1}{2\,\eta}\,C_{\bm{\zeta}}(t)^{\intercal}\,\mathfrak{D}\left(\bm{\zeta}\right)^{-1}\,B_{\bm{\zeta}}(t)\,,\\
			0=&F_{\bm{\zeta}}'(t)+ \bm{\mu}^{\intercal}\,\bm{\beta}E_{\bm{\zeta}}(t)+Tr\left(\bm{\Sigma}\,D_{\bm{\zeta}}(t)\right)+\frac{1}{4\,\eta}\,C_{\bm{\zeta}}(t)^{\intercal}\,\mathfrak{D}\left(\bm{\zeta}\right)^{-1}C_{\bm{\zeta}}(t)\,,
		\end{aligned}\right.
	\end{equation}
	with terminal condition
	\begin{equation}\label{eq:terminal2}
		\begin{split}
			&A_{\bm{\zeta}}(T)=-\bm{\alpha}\,,\\
			B_{\bm{\zeta}}(T)=C_{\bm{\zeta}}(T)=&D_{\bm{\zeta}}(T)=E_{\bm{\zeta}}(T)=F_{\bm{\zeta}}(T)=0\,.
		\end{split}
	\end{equation}
	
	The ODEs in $D$, $E$, $F$, and the corresponding terminal conditions admit the unique solutions  $D_{\bm{\zeta}}=E_{\bm{\zeta}}=F_{\bm{\zeta}}=0\,.$ Next, use the classical tools for Riccati equations to show that the ODE for $A_{\bm{\zeta}}(t)$ admits the unique solution
	\begin{equation}
		A_{\bm{\zeta}}(t)=\eta\,\mathfrak{D}\left(\bm{\zeta}\right)^{\frac{1}{2}}\left(\bm{\varPsi}+\bm{\xi}\left(t\right)^{-1}\right) \,\mathfrak{D}\left(\bm{\zeta}\right)^{\frac{1}{2}}\,,
	\end{equation}
	where
	\begin{equation}
		\left\{
		\begin{aligned}
			\bm{\varPsi} & =\frac{\phi}{\eta}\,\mathfrak{D}\left(\bm{\zeta}\right)^{-\frac{1}{2}}\bm{\tilde{\Sigma}}\,\mathfrak{D}\left(\bm{\zeta}\right)^{-\frac{1}{2}}\,,\\
			\bm{\xi}\left(t\right) & =-\frac{\bm{\varPsi}^{-1}}{2}\left(I-e^{-2\bm{\varPsi}\left(T-t\right)}\right)-e^{-\bm{\varPsi}\left(T-t\right)}\left(\bm{\Phi}+\bm{\varPsi}\right)^{-1}e^{-\bm{\varPsi}\left(T-t\right)}\text{},\\
			\bm{\Phi} & =\frac{1}{\eta}\mathfrak{D}\left(\bm{\zeta}\right)^{-\frac{1}{2}}A_{\bm{\zeta}}(T)\,\mathfrak{D}\left(\bm{\zeta}\right)^{-\frac{1}{2}}\,.
		\end{aligned}\right.
	\end{equation}
	
	Finally, use Theorem $3$ in \cite{cartea2018trading} to obtain the unique solution for $B_{\bm{\zeta}}$ and $C_{\bm{\zeta}}$\,. More precisely, observe that $C_{\bm{\zeta}} = -\bm{\mu}^\intercal\,B_{\bm{\zeta}}^\intercal$ and write
	\begin{equation}
		B_{\bm{\zeta}}(t) = \int_t^T\bm{\colon e}^{\int_t^u (\bm{\beta}^{\intercal}-\frac{1}{\eta}A_{\bm{\zeta}}(s)^{\intercal}\mathfrak{D}\left(\bm{\zeta}\right)^{-1})\,ds}\bm{\colon}\bm{\mathcal{X}}\bm{\beta}^{\intercal}\,du\,,
	\end{equation}
	where the notation $\bm{\colon e}^{\int_t^u \cdot \,ds}\bm{\colon}$ denotes the time-ordered exponential. 
	
	Thus, we obtain a closed-form solution to the HJB equation and an analytical formula for the optimal trading speed in the case of fixed execution costs. In particular, substitute \eqref{eq:ansatz2_multi} into \eqref{eq:optimalspeed_multi_deterministic} to obtain the optimal strategy in feedback form
	\begin{equation}\label{eq:speed}
		\bm\nu^{\star,\bm{\zeta}}=\frac{1}{2\,\eta}\,\mathfrak D\left(\bm{\zeta}\right)^{-1}\,\left(B_{\bm{\zeta}}(t)^\intercal\,(\bm{\mu}-\bm{R})-2\,\bm{A}_{\bm{\zeta}}(t)\,\bm{\tilde{y}}\right).
	\end{equation}
	
	\subsection{Closed-form approximation strategy \label{sec:closedform}}
	
	Here,  we use a family of closed-form strategies of the type in \eqref{eq:optimalspeed_multi_deterministic} to derive a piecewise-defined trading strategy which approximates the optimal trading speed in feedback form \eqref{eq:optimalspeed_multi}. Specifically, we partition the space of the rate $\bm Z$ into hyperrectangles and define a piecewise strategy which uses a different impact parameter $\bm\zeta$ in each different hyperrectangle. Finally, we show that as the dimension of the hyperrectangles becomes arbitrarily small, the piecewise strategy converges to the closed-form approximation strategy. 
	
	For each $i\in\{1,\dots,n\}$, let $\left\{Z_0^{i,N}, \dots,Z_N^{i,N}\right\}$ be a partition of $\left[\underline{Z}^{i},\overline{Z}^{i}\right],$ where $0<\underline{Z}^{i}<\overline{Z}^{i}$. Here, $\underline{Z}^i $ and $ \overline{Z}^i$ are are such that the LT has high confidence that the rate $Z^i$ will be within the range $\left[\underline{Z}^i,\overline{Z}^i\right]$ during the execution window. We define $\Rk\coloneqq\left[\underline{Z}^1, \overline{Z}^1\right]\times\left[\underline{Z}^n, \overline{Z}^n\right]$. Next, for each $N\in\N$, $i\in\{1,\dots,n\}$, and  $j\in\{0,\dots,N\}$ we define
	\begin{equation}\label{eq:partitionedConvexity_multi}
		Z_{j}^{i,N} \coloneqq \underline{Z}^{i}+\frac{j}{N} \left(\overline{Z}^{i}-\underline{Z}^{i}\right)\quad\text{and}\quad\zeta_{j}^{i,N} \coloneqq  \frac{1}{\kappa^{i}}\, \left(Z_{j}^{i,N}\right)^{3/2}\,,
	\end{equation}
	and for each multi-index $\bm j = \left(j_1,\dots,j_n\right)\in\{1,\dots,N\}^n$ we define
	\begin{equation}
		\begin{split}
			\bm\zeta_{\bm{j}}^N\coloneqq\left(\zeta_{j_1}^{1,N},\dots, \zeta_{j_n}^{n,N}\right),\text{\quad and\quad}
			\bm\nu^{\star,\bm{j},N}=\left(\nu^{\star,1,j_1,N},\dots, \nu^{\star,n,j_n,N}\right)\coloneqq\bm\nu^{\star,\bm\zeta_{\bm{j}}^N}\,.
		\end{split}
	\end{equation}
	
	To construct the approximate trading strategy, we first define a strategy $\bm\nu^{\star,N}$ that uses the closed-form optimal trading speed $\bm\nu^{\star,\bm{j},N}$ to approximate the optimal trading speed whenever the rate is close to $\bm\zeta_{\bm{j}}^N$. For each $i\in\{1,\dots,n\} $, we define the piecewise-defined trading speed $\nu^{\star,i,N}\colon\left[0,T\right]\times\R\times\R^n\times\R^{2n+m}\to\R$ 
	\begin{equation}\label{eq:piecewise_strategy_multi}
		\begin{split}
			\nu^{\star,i,N}\left(t,\bm\tilde{y},\bm R\right) =\,& \nu^{\star,i,0,N}\left(t,\bm\tilde{y},\bm R\right)\mathbbm{1}_{Z^i<Z^{i,N}_1} + \sum_{j=1}^{N-1}\nu^{\star,i,j,N}\left(t,\bm\tilde{y},\bm R\right)\mathbbm{1}_{Z^i\in[Z_j^{i,N},Z^{i,N}_{j+1})} \\
			&+ \nu^{\star,i,N,N}\left(t,\bm\tilde{y},\bm R\right)\mathbbm{1}_{Z^i\geq Z^{i,N}_N}\,.
		\end{split}
	\end{equation}
	The strategy $\nu^{\star,i,N}\left(t,\bm\tilde{y},\bm R\right)$ has first-type discontinuity points; specifically, for each $j\in\{1,\dots,N\}$ we have  $\nu^{\star,i,j,N}|_{Z^{i}=Z^{i,N}_{j+1}} \neq\nu^{\star,i,j+1,N}|_{Z^{i}=Z^{i,N}_{j+1}}$.

	The theorem below shows how to partition $\Rk$ to make the discontinuities in each $\nu^{\star,i,N}$ arbitrarily small simultaneously. Furthermore, when the distance between points in the partition becomes sufficiently small, the sequence of piecewise-defined optimal strategies $\left\{\bm\nu^{\star,N}\right\}_{N\in\N}$ converges uniformly to a continuous closed-form approximation strategy which we use in our performance study of Section \ref{sec:performance_multi}.
	
	\begin{thm}\label{thm_multi}
		For each $\varepsilon>0\,$ there exists $N\in\N$ such that 
		\begin{equation}\label{eq:inequality_multi}
			\underset{\substack{j=1,\dots,N \\ i=0,\dots,n}}{\max} \left\lVert\nu^{\star,i,j,N}|_{Z^{i}=Z^{i,N}_{j+1}}\ - \nu^{\star,i,j+1,N}|_{Z^{i}=Z^{i,N}_{j+1}}\right\rVert<\varepsilon\,.
		\end{equation}
		Furthermore, for each $N\in\N\,,$ let $\bm{\tilde{\nu}}^{\star,N}\coloneqq \bm\nu^{\star,N}\left|_{\left[0,T\right]\times\R\times\R^n\times\R^{n}\times\Rk\times\R^m}\right.$. Then, the sequence $\left\{\bm{\tilde{\nu}}^{\star,N}\right\}$ converges to $\bm\tilde{\nu}^{\star}$ uniformly in $ \left[0,T\right]\times\R\times\R^n\times\R^{n}\times\Rk\times\R^m$\,, where 
		\begin{equation}
			\label{eq:pseudooptimal_multi}
			\bm{\tilde{\nu}}^{\star}\left(t,\bm\tilde{y},\bm R\right)=\frac{1}{2\,\eta}\,\mathfrak D\left(\bm{\zeta}\right)^{-1}\, \left(B\left(t,\bm{\kappa}^{-1}\odot\bm{Z}^{3/2}\right)^\intercal\, (\bm{\mu}-\bm{R})-2\,\bm{A}\left(t,\bm{\kappa}^{-1}\odot\bm{Z}^{3/2}\right)\,\bm{\tilde{y}}\right)\,.
		\end{equation}
		and
		\begin{equation}\label{eq:pseudoAB_multi}
			\begin{split}
				\bm{A}\left(t,\bm{\kappa}^{-1}\odot\bm{Z}^{3/2}\right)& =\eta\,\mathfrak{D}\left(\bm{\kappa}^{-1}\odot\bm{Z}^{3/2}\right)^{\frac{1}{2}}\left(\bm{\varPsi}+\bm{\xi}\left(t\right)^{-1}\right) \,\mathfrak{D}\left(\bm{\kappa}^{-1}\odot\bm{Z}^{3/2}\right)^{\frac{1}{2}}\,,\\
				\bm{B}\left(t,\bm{\kappa}^{-1}\odot\bm{Z}^{3/2}\right) &= \int_t^T\bm{\colon e}^{\int_t^u (\bm{\beta}^{\intercal}-\frac{1}{\eta}\bm{A}\left(s,\bm{\kappa}^{-1}\odot\bm{Z}^{3/2}\right)^{\intercal}\mathfrak{D}\left(\bm{\kappa}^{-1}\odot\bm{Z}^{3/2}\right)^{-1})\,ds}\bm{\colon}\bm{\mathcal{X}}\bm{\beta}^{\intercal}\,du\,,\ .\\
			\end{split}
		\end{equation}
	\end{thm}
	The proof for Theorem \ref{thm_multi} is a trivial generalisation of the proof in \ref{sec:annex_proof}\,.

	\section{Performance analysis \label{sec:performance_multi}}
	
	We use transaction data to show the performance of the closed-form approximation strategy in \eqref{eq:pseudooptimal_multi}. Similar to Section \ref{sec:AMM_multi}, we use transaction data for ETH and BTC in two pools of the CPMM Uniswap v3 and the more liquid exchange Binance. 
	
	To estimate the execution costs incurred by the LT in the pools we consider, one needs to track the depth of the pool at every instant. We use historical LP operations to reconstruct the historical depth $\kappa\,$ of the two pools, which we use to estimate the execution costs during the execution programmes of the LT.
	
	Here, we perform consecutive in-sample estimation of model parameters using $24$ hours of transaction data sampled at the minute frequency, and out-of-sample execution of the closed-form approximation strategy \eqref{eq:pseudooptimal_multi} using $T=$12 hours, $\phi=1\times10^{-6}$, and $\bm{\alpha}=\left(\begin{array}{cc}
		10 & 0\\
		0 & 10
	\end{array}\right)\,.$
	
	For each individual run, we compute the optimal strategy and compute the terminal profit and loss of the execution programme, which is $\bm{\tilde x}_T + \bm{\tilde y}_T^\intercal \, \bm{Z}_T - \bm{\tilde y}_0^\intercal \, \bm{Z}_0$. Next, we shift both the in-sample and the out-of-sample time windows by $12$ hours and repeat the process, thus performing $357$ runs in total. In each run, the agent must (randomly with probability $1/2$) either liquidate $1000$ units of ETH, and $100$ units of BTC, or buy the same amounts.
	
	To estimate the coefficients of the process driving the joint dynamics of the exchange rates, we use the VAR(1) model $$\Delta R_t = a + \Pi \,\Delta R_{t-1} + \epsilon_t\,,$$ which is a time discretisation of our multi-OU process; where $a \in \R^{2n+m}$, $\Pi \in \mathcal M_{2n+m} (\R)$, and $\epsilon$ is white noise. We use in-sample data and a least squares regression to estimate the coefficients of the VAR model, from which one obtains the mean-reversion matrix $\bm{\beta}$, the long-term unconditional mean $\bm{\mu}$, and the covariance matrix $\bm{\Sigma}\,.$
	
	As an example, Figure \ref{fig:Inventory} shows the optimal inventory when the in-sample data is that of 1 February 2022 and the execution programme is over the following $12$ hours. Figure \ref{fig:Inventory} shows that the strategy uses the spread between the oracle rates and the pool rates for both pairs to predict a potential decrease or increase in the rate. The depth of the ETH/USDC $0.05\%$ pool is larger than that of the BTC/USDC $0.3\%$; see Table \ref{table:datadescr}, so trading costs caused by the convexity of the level function are low in the first pool and high in the second. Hence, the speculative component of the closed-form approximation strategy \eqref{eq:pseudooptimal_multi} is less exploited in the less liquid pool even though the spread between the oracle and the pool rate is wide.
	
	\begin{figure}%
		\centering
		\includegraphics{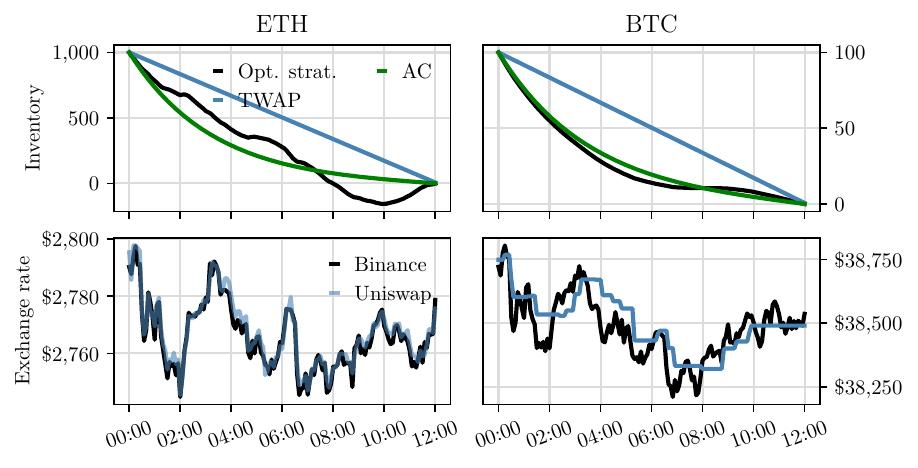}
		\caption{Inventory in ETH and BTC during the execution programme (top) and the pool and oracle exchange rates (bottom). Model parameters are obtained using data between 00:00 on 1 February and 00:00 on 2 February 2022, and the execution programme is between 00:00 on 2 February and 12:00 on 2 February 2022.}%
		\label{fig:Inventory}%
	\end{figure}
	
	In our analysis, we compare the performance of our strategy with that of a classical Almgren-Chriss (AC) strategy and a TWAP strategy; see \cite{gueant2016book}. Moreover, we run the strategy  \eqref{eq:pseudooptimal_multi} with zero initial inventory, which corresponds to an arbitrage strategy that uses the cointegration factors to take speculative positions. Table \ref{table:results} shows the distribution of the terminal P\&L for the three strategies and shows that the closed-form approximation strategy outperforms TWAP and AC on average, and is profitable for statistical arbitrage.
	\begin{table}
		\begin{center}
			\begin{tabular}{c c r} 
				&  Avg. P\&L & Std. Dev.\\ [0.5ex] 
				\hline 
				\\ [-2ex] 
				Optimal & \$\, 2,328 & \$ 54,805  \\ [1ex] 
				Optimal (stat. arb.) & \$\, 2,658 & \$ 5,819  \\ [1ex]   
				AC & \$\;\; -343 & \$ 54,823\\ [1ex]  
				TWAP & \$ -3,705 & \$ 99,946\\
				\hline
			\end{tabular}
		\end{center}
		\caption {Average and standard deviation of the terminal P\&L with strategy \eqref{eq:pseudooptimal_multi} (liquidation and statistical arbitrage) and from executing TWAP and AC for $357$ runs of in-sample estimation of model parameters and out-of-sample execution.}
		\label{table:results}
	\end{table}
	
	\section{Conclusions}
	In this chapter, we derived the optimal strategy for a liquidity taker (LT) who trades in an AMM a basket of crypto-currencies whose constituents co-move. The LT uses market signals and exchange rate information from relevant AMMs and traditional venues to enhance the performance of her strategy. We used stochastic control tools to derive a closed-form strategy that can be computed and implemented by the LT in real time. Finally, we used market data from Uniswap v3 and Binance to study  co-movements between crypto-currencies and lead-lag effects between trading venues, and to showcase the performance of the strategy.

\chapter{Optimal Liquidity Provision in constant product market makers with concentrated liquidity}\label{ch:pl}

\section{Introduction}
Traditional electronic exchanges are organised around limit order books to clear demand and supply of liquidity. In contrast, the takers and providers of liquidity in  constant function market makers (CFMMs) interact in liquidity pools; liquidity providers (LPs) deposit their assets in the liquidity pool and liquidity takers (LTs) exchange assets directly with the pool. At present, constant product market makers (CPMMs) with concentrated liquidity (CL) are the most popular type of CFMM, with  Uniswap v3 as a prime example; see \cite{uniswap2021core}. In CPMMs with CL, LPs specify the rate intervals (i.e., tick ranges) over which they deposit their assets, and this liquidity is counterparty to trades of LTs when the marginal exchange rate of the pool is within the liquidity range of the LPs. When LPs deposit liquidity, fees paid by LTs accrue and are paid to LPs when they withdraw their assets from the pool. The amount of fees accrued to LPs is proportional to the share of liquidity they hold in each liquidity range of the pool.

Existing research characterises the losses of LPs, but does not offer tools for strategic liquidity provision. In this chapter, we study strategic liquidity provision in CPMMs with CL. We derive the continuous-time dynamics of the wealth of strategic LPs which consists of the position they hold in the pool (position value),  fee income, and costs from repositioning their liquidity position. The width of the range where the assets are deposited affects the value of the LP's position in the pool and we show that fee income is subject to a tradeoff between the width of the LP's liquidity range and the volatility of the marginal rate in the pool. More precisely, CL increases fee revenue when the rate is in the range of the LP, but also increases \textit{concentration risk}. Concentration risk refers to the risk that the LP faces when her position is concentrated in narrow ranges; the LP stops collecting fees when the rate exits the range of her position.

We derive an optimal dynamic strategy to provide liquidity in a CPMM with CL. In our model, the LP maximises the expected utility of her terminal wealth, which consists of the accumulated trading fees and the gains and losses from the market making strategy. The dynamic strategy controls the width and the skew of liquidity that targets the marginal exchange rate. For the particular case of log-utility, we obtain the strategy in closed-form and show how the solution balances the  opposing effects between volatility, convexity, and fee collection. When volatility increases, there is an incentive for the LP to widen the range of liquidity provision. In particular, in the extreme case of very high volatility, the LP must withdraw from the pool. Also, when there is an increase in the potential provision fees that the LP may collect because of higher liquidity taking activity, the strategy balances two opposing forces. One, there is an incentive to increase fee collection by concentrating the liquidity of the LP in a tight range around the exchange rate of the pool. Two, there is an incentive to limit the losses due to concentration risk by widening the range of liquidity provision. 
Finally, when the drift of the marginal exchange rate is stochastic (e.g., a predictive signal), the strategy skews the range of liquidity to increase fee revenue, by capturing the LT trading flow, and to increase the position value, by profiting from the expected changes in the marginal rate. 

Finally, we use Uniswap v3 data to motivate our model and to test the performance of the strategy we derive. The LP and LT data are from the pool ETH/USDC (Ethereum and USD coin) between the inception of the pool on $5$ May $2021$ and $18$ August $2022$.  To illustrate the performance of the strategy we use in-sample data to estimate model parameters and out-of-sample data to test the strategy. Our analysis of the historical transactions in Uniswap v3 shows that LPs have traded at a significant loss, on average, in the ETH/USDC pool. We show that the out-of-sample performance of our
strategy is considerably superior to the average LP performance we observe in the ETH/USDC pool.

Early works on AMMs are in \cite{chiu2019blockchain, angeris2021analysis, lipton2021blockchain}. Some works in the literature study strategic liquidity provision in CFMMs and CPMMs with CL. \cite{heimbach2022risks} discuss the tradeoff between risks and returns that LPs face in Uniswap v3, \cite{cartea2023predictable}  study the predictable losses of LPs in a continuous-time setup, \cite{milionis2023automated} study the impact of fees on the profits of arbitrageurs in CFMMs,  \cite{fukasawa2023model} study the hedging of the impermanent losses of LPs, and \cite{li2023yield} study the economics of liquidity provision. Closest to our work are the models in \cite{fan2021strategic} and \cite{fan2022differential} which focus on fee revenue and use approximation techniques to obtain dynamic strategies. Finally, there is a growing literature on AMM design for fair competition between LPs and LTs. \cite{goyal2023finding} study an AMM with dynamic trading functions that incorporate beliefs of LPs, \cite{lommers2023:case} study AMMs where the LP's strategy adjusts dynamically to market information, and  \cite{cartea2023automated} generalise CFMMs and propose AMM designs where LPs express their beliefs and risk preferences.

Our work is related to the algorithmic trading and optimal market making literature.  Early works on liquidity provision in traditional markets are \cite{ho1983dynamics}, \cite{biais1993price}, and  \cite{avellaneda2008high} with extensions in many directions; see \cite{cartea2014buy, cartea2017algorithmic, gueant2017optimal, bergault2021closed, drissi2022solvability}. We refer the reader to \cite{cartea2015book}, \cite{gueant2016book}, and \cite{donnelly2022optimal} for a comprehensive review of algorithmic trading models for takers and makers of liquidity in traditional markets. Also, our work is related to those in  \cite{cartea2018enhancing, barger2019optimal, cartea2020market, donnelly2020optimal, forde2022optimal, bergault2022multi} which implement market signals in algorithmic trading strategies. 

The remainder of the chapter proceeds as follows. Section  \ref{sec:CL} describes CL pools. Section \ref{sec:wealthCL} studies the continuous-time dynamics of the wealth of LPs as a result of the position value, the fee revenue, and rebalancing costs. In particular, we use Uniswap v3 data to study the fee revenue component of the LP's wealth and our results motivate the assumptions in our model.  Section \ref{sec:4} introduces our liquidity provision model and uses stochastic control to derive a closed-form optimal strategy. Next, we study how the strategy controls the width and the skew of the liquidity range as a function of the pool's profitability, PL, concentration risk, and the drift in the marginal rate. Finally, Section  \ref{sec:num1} uses  Uniswap v3 data to test the performance of the strategy and showcases its superior performance.

\section{Constant function market makers and concentrated liquidity \label{sec:CL}}

\subsection{Constant function market makers}

Consider a reference asset $X$ and a risky asset $Y$ which is valued in units of $X.$ Assume there is a pool that makes liquidity for the pair of assets $X$ and $Y$, and denote by $Z$ the marginal exchange rate of asset $Y$ in units of asset $X$ in the pool. In a CFMM that charges a fee $\tau$ proportional to trade size, the trading function $f\left(x, y\right)  = \kappa^2$ links the state of the pool before and after a trade is executed, where $x$ and $y$ are the quantities of asset $X$ and $Y$ that constitute the \textit{reserves} in the pool, $\kappa$ is the depth of the pool, and $f$ is increasing in both arguments. We write $f\left(x, y\right) = \kappa^2$ as $x = \varphi_\kappa\left(y\right)$ for an appropriate decreasing \textit{level function} $\varphi_\kappa$.

We denote the execution rate for a traded quantity $\pm \Delta y$ by $\tilde Z\left(\pm \Delta y\right)$, where $\Delta y\ge 0$. When an LT buys $\Delta y$ of asset $Y$, she pays $\Delta x = \Delta y \times \tilde Z\left(y\right)$ of asset $X$, where 
\begin{equation}\label{eq:CFMM_lt_cond_1}
f\left(x + (1-\tau)\,\Delta x, y - \Delta y\right) = \kappa^2 \quad \implies \quad  \tilde Z(\Delta y) = \frac{\varphi_\kappa\left(y-\Delta y\right) - \varphi_\kappa(y)}{\left(1-\tau\right)\,\Delta y}\,.
\end{equation}
Similarly, when an LT sells $y$ of asset $Y$, she receives $x = y \times \tilde Z\left(-y\right)$ of asset $X$, where
\begin{equation}\label{eq:CFMM_lt_cond_2}
f\left(x + \Delta x, y +  (1-\tau)\,\Delta y\right) = \kappa^2 \quad \implies \quad  \tilde Z(-\Delta y) = \frac{\varphi_\kappa\left(y\right) - \varphi_\kappa\left(y+\left(1-\tau\right)\,\Delta y\right)}{\Delta y}\,.
\end{equation}

In CL markets, the marginal exchange rate  is $Z = -\varphi_\kappa'(y)$, which is the price of an infinitesimal trade when fees are zero, i.e., when $\Delta y\rightarrow 0$ and $\tau = 0$.\footnote{When fees are not zero, the exchange rates for infinitesimal trades are $\lim_{\Delta y\rightarrow 0} \tilde Z(\Delta y) = - \frac{1}{1-\tau}\, \varphi'_\kappa(y)= \frac{1}{1-\tau}\,Z $ and $\lim_{\Delta y\rightarrow 0} \tilde Z(-\Delta y) = - \left(1-\tau\right) \varphi'_\kappa(y) = \left(1-\tau\right)\,Z$; see \eqref{eq:CFMM_lt_cond_1} and \eqref{eq:CFMM_lt_cond_2}.} In traditional CPMMs such as Uniswap v2, the trading function is $f\left(x, y\right) = x\times y$, so the level function is $\varphi_\kappa(y) = \kappa^2/y$ and the marginal exchange rate is $Z = x / y.$ 

Liquidity provision operations in CPMMs do not impact the marginal rate, so when an LP deposits the quantities $x$ and $y$ of assets $X$ and $Y$, the condition $Z = x/y = (x+\Delta  x)/(y+\Delta  y)$ must be satisfied; see \eqref{eq:LPconditionCPMM}.

\subsection{Concentrated liquidity markets}

This chapter focuses on liquidity provision in CPMMs with CL. In pools with CL, the space of rates where an LP provides liquidity is discretised into a set of values $\left\{\Zc(i)\right\}_{i\in\Z}$ called ticks.\footnote{In traditional limit order books, a tick is the smallest price increment.} For instance, in Uniswap v3 each tick $\Zc(i)$ is uniquely identified by an integer $i\in\Z$ according to
\begin{equation}\label{eq:tick}
    \Zc(i) = 1.0001^i\,.
\end{equation}
The boundaries of the LP's position in the pool take values in the set of ticks. The range between two consecutive ticks defines the smallest range available for LPs. Every tick range $\left(\Zc(i), \Zc(i+1)\right]$ has its own depth which determines the execution cost of LT trades when the instantaneous rate $Z$ is in that tick range. 

The LP specifies a lower tick $\Zc(L)$ and an upper tick $\Zc(U)$ between which she provides liquidity that is evenly spread over all the tick ranges in the position range $\left(\Zc(L),\Zc(U)\right]$. The position of an LP who provides liquidity in a range $\left(\Zc(L),\Zc(U)\right]$ is characterised by the position depth $\tilde{\kappa}$. The value of $\tilde{\kappa}$ represents the depth of the LP's liquidity within every tick range in her position range, and $\tilde{\kappa} / \kappa^{i}$ represents the portion of liquidity that the LP holds in the tick range $\left(\Zc(i), \Zc(i+1)\right] \subset \left(\Zc(L),\Zc(U)\right]$. 

The assets that the LP deposits in a range  $\left(\Zc(L),\Zc(U)\right]$ provide the liquidity that supports marginal rate movements between $\Zc(L)$ and $\Zc(U)$. The quantities $\tilde x$ and $\tilde y$ that the LP provides verify the key formulae
\begin{equation}\label{LP:eq:position}
    \begin{cases}
        \tilde x = 0 \qquad\qquad\qquad\qquad\quad\quad\,\,\, \text{and} \qquad \tilde y = \tilde{\kappa}\left(\frac{1}{\sqrt{\Zc(L)}}-\frac{1}{\sqrt{\Zc(U)}}\right) & \text{if}\quad Z<\Zc(L)\,,\\
        \tilde x = \tilde{\kappa}\left(\sqrt{Z}-\sqrt{\Zc(L)}\right) \,\qquad\,\,\, \text{and} \qquad \tilde y = \tilde{\kappa}\left(\frac{1}{\sqrt{Z}}-\frac{1}{\sqrt{\Zc(U)}}\right) & \text{if}\quad \Zc(L) \leq Z<\Zc(U) \,,\\
        \tilde x =  \tilde{\kappa}\left(\sqrt{\Zc(U)}-\sqrt{\Zc(L)}\right)\,\,\,\,\, \text{and} \qquad \tilde y =  0 & \text{if}\quad Z \geq \Zc(U) \,.\\
    \end{cases} 
\end{equation}
When the rate $Z$ changes, the equations in \eqref{LP:eq:position} and the prevailing marginal rate $Z$  determine the holdings of the LP in the pool, in particular, they determine the quantities of each asset received by the LP when she withdraws her liquidity.

Within each tick range, the constant product formulae \eqref{eq:CFMM_lt_cond_1}--\eqref{eq:CFMM_lt_cond_2} determine the dynamics of the marginal rate, where the depth $\kappa$ is the total depth of liquidity in that tick range. To obtain the total depth in a tick range, one sums the depths of the individual liquidity positions in the same tick range; see \cite{drissi2023models}. When a liquidity taking trade is large, so the marginal rate crosses the boundary of a tick range, the pool executes two separate trades with potentially different depths for the constant product formula. 

The proportional fee $\tau$ charged by the pool to LTs is distributed among LPs. More precisely, if an LP's liquidity position with depth $\tilde \kappa$ is in a tick range where the total depth of liquidity is $\kappa$, then for every liquidity taking trade that pays an amount $p$ of fees, the LP with liquidity $\tilde \kappa$ earns the amount
\begin{align}
\label{eq:remuneration_LP}
\tilde  p = \frac{\tilde{\kappa}}{ \kappa}\, p\,\mathbbm{1}_{\{\Zc(L) < Z \leq \Zc(U)\}}\,.
\end{align}
Thus, the larger is the position depth $\tilde \kappa$, the higher is the proportion of fees that the LP earns.\footnote{For instance, if the LP is the only provider of liquidity in the range $(Z^\ell, Z^u]$ then $\kappa = \tilde \kappa$, so the LP collects all the fees in that range.} 

\section{The wealth of liquidity providers in CL pools \label{sec:wealthCL}}

In this section, we consider a strategic LP who dynamically tracks the marginal rate $Z$. In our model, the LP's position is self-financed throughout an investment window $[0, T]$, so the LP repeatedly withdraws her liquidity and collects the accumulated fees, then uses her wealth, i.e., the collected fees and the assets she withdraws, to deposit liquidity in a new range. In the remainder of this work, we work in a filtered probability space $\left(\Omega, \mathcal F, \mathbbm P; \mathbbm F = (\mathcal F_t)_{t \in [0,T]} \right)$ that satisfies the usual conditions, where $\mathbbm F$ is the natural filtration generated by the collection of observable stochastic processes that we define below.

We assume that the marginal exchange rate  in the pool $\left(Z_t\right)_{t\in [0,T]}$ is driven by a stochastic drift $\left(\mu_t\right)_{t\in [0,T]}$ and we write
\begin{align}
\label{eq:dynZ}
   \diff Z_t = \mu_t \, Z_t\, \diff t + \sigma\, Z_t\, \diff W_t\,,  
\end{align}
where the volatility parameter $\sigma$ is a nonnegative constant and $(W_t)_{t \in [0,T]}$ is a standard Brownian motion independent of $\mu$. We assume that $\mu$ is càdlàg with finite fourth moment, i.e., $\mathbbm E\left[\mu_t^4\right]<+\infty$ for $ 0\leq t\leq T$. The LP observes and uses $\mu$ to optimise her liquidity positions and improve trading performance. 

Consider an LP who provides liquidity in the CPMM with CL during an investment horizon $[0, T],$ with $T>0.$ At time $t$, she chooses a range $\left(Z_t^L,Z_t^U\right]$ to provide liquidity. Note that $Z_t^L$, $Z_t^U$ do not necessarily correspond to ticks as defined in \eqref{eq:tick}; in practice, when the LP implements her strategy in real time, she rounds  $Z_t^L$, and $Z_t^U$ to the closest ticks.

The LP's initial wealth is $V_0,$ in units of $X$, and at time $t=0\,,$ she deposits  quantities $\left(\tilde x_0, \tilde y_0\right)$ in the range $\left(Z^L_0, Z^U_0\right]$, so the initial depth of her position is $\tilde \kappa_0,$ and the value of her initial position, marked-to-market in units of $X$, is $V_0 = \tilde x_0 +\tilde  y_0 \, Z_0.$   The dynamics of the LP's wealth consist of the value  of the liquidity position in the pool $\left(\alpha_t\right)_{t \in [0, T]}$, the fee revenue $\left(p_t\right)_{t \in [0, T]}$, and the rebalancing costs $\left(c_t\right)_{t \in [0, T]}$.  We introduce the wealth process $(V_t = \alpha_t +  p_t + c_t)_{t \in [0,T]}$, which we mark-to-market in units of the reference asset $X$, with $V_0 > 0$ known.  At any time $t,$ the LP uses her wealth $V_t$ to provide liquidity. Next, Subsection  \ref{sec:positionvalue} studies the dynamics of the LP's position $\alpha$ in the pool, Subsection \ref{sec:feerevenue} studies the dynamics of the LP's fee revenue $p$, and Subsection \ref{sec:costs} studies the dynamics of the rebalancing costs $c$.

\subsection{Position value \label{sec:positionvalue}}

In this section, we focus our analysis on the \textit{position value} $\alpha$.  Throughout the investment window $[0, T]$, the holdings $\left(\tilde x_t, \tilde y_t\right)_{t\in[0,T]}$ of the LP change because the marginal rate $Z$ changes and because she continuously adjusts her liquidity range around $Z$. More precisely, to make markets optimally, the LP controls the values of $\left(\delta_t^L\right)_{t\in[0,T]}$ and  $\left(\delta_t^U\right)_{t\in[0,T]}$ which determine the dynamic liquidity provision boundaries $\left(Z_t^L\right)_{t\in[0,T]}$ and $\left(Z_t^U\right)_{t\in[0,T]}$ as follows:
\begin{equation}\label{eq:ZlZucontrol}
    \begin{aligned}
        \sqrt{Z_{t}^{U}}= & \sqrt{Z_{t}}\left(1-\frac{\delta_{t}^{U}}{2}\right)^{-1},\\
        \sqrt{Z_{t}^{L}}= & \sqrt{Z_{t}}\left(1-\frac{\delta_{t}^{L}}{2}\right),
    \end{aligned}
\end{equation}
where $\delta^L \in \left(-\infty, 2\right]$, $\delta^U \in \left[-\infty, 2\right)$, and $\delta^L\,\delta^U/2<\delta^L+\delta^U$ because $0 \leq Z^L < Z^U < \infty.$ We define $\delta^L$ and $\delta^U$ in \eqref{eq:ZlZucontrol} in terms of $\sqrt Z$ to simplify and linearise the CL constant product formulae; see \cite{cartea2023predictable} for more details.

In practice, the LP earns fees when the rate $Z_t$ is in the LP's liquidity range $(Z^L_t, Z^U_t]$, so  $\delta^L \in \left(0, 2\right]$, $\delta^U \in \left[0, 2\right)$, and $\delta^L\,\delta^U/2<\delta^L+\delta^U$.
Below, Section \ref{sec:4} considers a problem where the controls are not constrained, and values $\delta^L \notin \left(0, 2\right]$, $\delta^U \notin \left[0, 2\right)$ are those where liquidity provision is unprofitable.

In the remainder of this chapter we define the \textit{spread} $\delta_t$ of the LP's position  as
\begin{align}
\label{eq:pos_spread}
\delta_t = \delta_t^U + \delta_t^L\,,
\end{align}
and we  consider admissible strategies $\delta$ that are $\R\textrm{-valued}$ and such that $\int_{0}^{T}\delta_{t}^{-4}\,\diff t<\infty,$ almost surely; see Section \ref{sec:4} for more details. For small position spreads, we use the first-order Taylor expansion to write the approximation
\begin{equation}
    \frac{Z_t^U - Z^L_t}{ Z_t} =  \left(1-\frac{\delta_{t}^{U}}{2}\right)^{-2}-\left(1-\frac{\delta_{t}^{L}}{2}\right)^{2}\approx \delta_t.
\end{equation}

We  assume  that the marginal rate process $\left(Z_t\right)_{t\in [0,T]}$ follows the dynamics \eqref{eq:dynZ}. \cite{cartea2023predictable} show that the holdings in assets $X$ and $Y$ in the pool for an LP who follows an arbitrary strategy $\left(Z_t^L, Z_t^U\right)$ are given by
\begin{equation}\label{eq:holdingspool}
\tilde x_t =    \frac{\delta_{t}^{L}}{\delta_{t}^{L}+\delta_{t}^{U}}\,\alpha_{t} \quad \text{and} \quad  \tilde y_t = \frac{\delta_{t}^{U}}{Z_t\left( \delta_{t}^{L}+\delta_{t}^{U}\right) }\,\alpha_{t}\,,
\end{equation}
so the value $\left(\alpha_t\right)_{t \in [0, T]}$ of her position follows the dynamics
\begin{equation}\label{eq:dynAlphaFirst}
\diff\alpha_{t}=\,\,V_{t}\,\left(\frac{1}{\delta_{t}^{L}+\delta_{t}^{U\,}}\right)\left(-\frac{\sigma^{2}}{2}\,\diff t+\mu_t\,\delta_{t}^{U}\,\diff t+\sigma\,\delta_{t}^{U}\,\diff W_{t}\right)\,.
\end{equation}
The dynamics in \eqref{eq:dynAlphaFirst} show that a larger position spread $\delta$ reduces the LP's exposure to the volatility of the pool. 

For a fixed value of the spread $\delta_t = \delta^L_t + \delta^U_t$, the dynamics in \eqref{eq:dynAlphaFirst} show that if $\mu_t \ge 0$, then the LP increases her expected wealth by increasing the value $\delta^U$, i.e., by skewing her range of liquidity to the right. However, note that the quadratic variation of the LP's position value is $\diff \langle \alpha, \alpha \rangle_t = V_t^2 \, \sigma^2 \left(\delta_t^U / \delta_t\right)^2\, \diff t\,,$ so skewing the range to the right also increases the variance of the LP's position. On the other hand, if $\mu \le 0$, then the LP reduces her losses by decreasing the value of $\delta^U$ or equivalently increasing the value of $\delta^L$, i.e., by skewing her range of liquidity to the left. Thus, the LP uses the expected changes in the marginal rate to skew the range of liquidity and to increase her expected terminal wealth.

\subsection{Fee income \label{sec:feerevenue}}

\subsubsection{Fee income: pool fee rate}

The dynamics of the fee income in our model of Section \ref{sec:4} uses a fixed depth $\kappa$ and assumes that the pool generates fee income for all LPs at an instantaneous \textit{pool fee rate} $\pi$; clearly, these fees are paid by LTs who interact with the pool, see \eqref{eq:CFMM_lt_cond_1}--\eqref{eq:CFMM_lt_cond_2}.  The value of $\pi$ represents the instantaneous profitability of the pool, akin to the size of market orders and their arrival rate in LOBs. 

To analyse the  dynamics of the pool fee rate $\pi,$ we use historical LT transactions in Uniswap v3 as a measure of activity and to estimate the total fee income generated by the pool; Table \ref{table:datadescr} provides descriptive statistics. Figure \ref{fig:DOLP_feerate} shows the estimated  fee rate $\pi$ in the ETH/USDC pool. For any time $t,$ we use\footnote{Pools have different fee rates in Uniswap v3: $0.01\%$, $0.05\%,$ and $0.1\%.$}
$$\pi_t = 0.05\% \, \frac{\Xi_t}{2\,\kappa\, \sqrt{Z_t}}  \,,$$ 
where $\Xi_t$ is the volume of LT transactions the day before $t,$  $2\,\kappa\, \sqrt{Z_t}$ is the pool value in terms of asset $X$ at time $t,$ and $0.05 \%$ is the fixed fee of the pool.\footnote{The value of a CPMM pool is given by the active pool depth. In particular, the size  of a pool with quantity $x$ of asset $X$ and quantity $y$ of asset $Y$ is $x+Z\,y = 2\,x =  2\,\kappa\, Z^{1/2}$ units of $X$ because $x\, y =\kappa^2 \ \text{and} \ Z= x/ y \ \text{ so } x^2 = \kappa^2 \, Z. $} Figure \ref{fig:DOLP_feerate} suggests that the pool fee rate $\pi$ generated by liquidity taking activity in the pool is stochastic and mean reverting. Here, we assume that $\pi$ is independent of the rate $Z$ over the time scales we consider; Table \ref{table:corrPiZ} shows that the pool fee rate is weakly correlated to the  rate $Z$ at different sampling frequencies, especially for the higher frequencies we consider in our numerical tests. In Section \ref{sec:4}, the pool fee rate $\pi$ follows Cox-Ingersoll-Ross-type dynamics.

\begin{figure}[!htb]\centering
\includegraphics[width=0.45\textwidth]{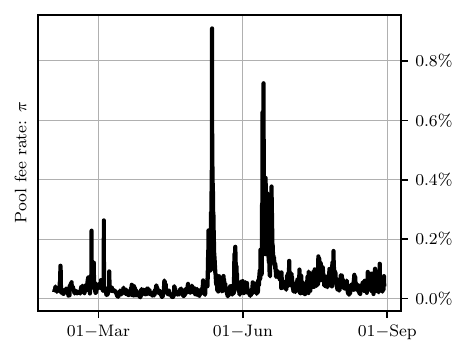}
\includegraphics[width=0.45\textwidth]{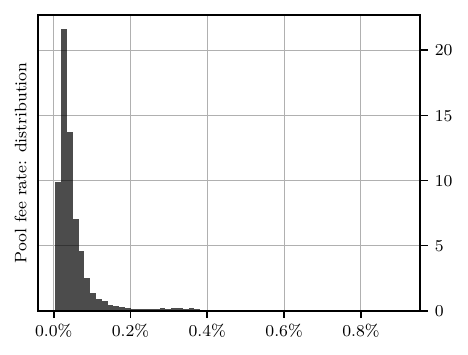}\\
\caption{ 
Estimated pool fee rate from February to August 2022 in the ETH/USDC pool. For any time $t,$ the pool fee rate is the total fee income, as a percentage of the total pool size, paid by LTs on the period $[t-1\text{ day}, t].$ The pool size at time $t$ is $2 \, \kappa \, \sqrt{Z_t}$ in terms of asset $X$, where $Z_t$ is the active rate in the pool at time $t.$ Left panel: historical values of the pool fee rate. Right panel: distribution of the pool fee rate. }
\label{fig:DOLP_feerate}
\end{figure}

\begin{table}[h]
\footnotesize
\begin{center}
\begin{tabular}{c  r  r  r  r} 
\hline 
& $\Delta t = 1$ minute & $\Delta t = 5$ minutes & $\Delta t = 1$ hour & $\Delta t = 1$ day \\ [0.5ex] 
\hline
Correlation & $-2.1 \%\ \ $ & $-2.4 \%$ & $-2.6 \%$ & $-10.9 \%$ \\ 
{Two-tailed p-value} & {$2.0\%\ \ $} & {$8.9\%$} & {$8.2\%$} & {$16.7\%$} \\ 
{$R^2$ regression} & {$0.01\%$} & {$1.1\%$} & {$2.9\%$} & {$5.0\%$} \\ 
\hline 
\end{tabular}
\end{center}
\caption {{First row:} correlation of the returns of the rate $Z$ and the fee rate $\pi$, i.e., $\left(Z_{t+\Delta t} - Z_t\right) /  Z_t$ and $\left(Z_{t+\Delta t} - Z_t\right)  /  Z_t$ for $\Delta t = 1$ minute, five minutes, one hour, and one day, using data of the ETH/USDC pool  between 5 May 2021 and 18 August 2022. {Second row: two-tailed p-value of the t-test. Final row: $R$-squared of regression of the pool fee rate's returns against the marginal exchange rate's returns.} }
\label{table:corrPiZ}
\end{table}


\subsubsection{{Fee income: spread and concentration risk} \label{sec:concentrationcost}}

In the three cases of \eqref{LP:eq:position}, increasing the spread reduces the depth $\tilde \kappa$ of the LP's position.  Recall that the LP fee income is proportional to $\tilde \kappa / \kappa,$ where $\kappa$ is the pool depth. Thus, decreasing the value of $\tilde \kappa$  potentially reduces LP fee income. Figure \ref{fig:SIMPLE1} shows the value of $\tilde \kappa$ as a function of the spread $\delta$.

\begin{figure}[!h]
\centering
\includegraphics{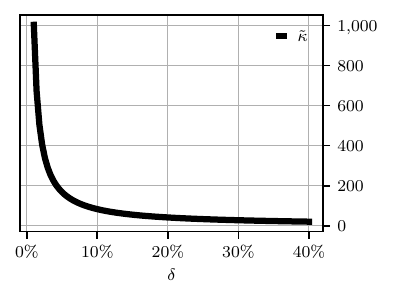}\\
\caption{Value of the depth $\tilde \kappa$ of the LP's position in the pool as a function of the spread $\delta.$ The spread is in percentage of the marginal exchange rate; recall that  $\left(Z_t^U - Z^L_t\right) / Z_t  \approx \delta_t.$}\label{fig:SIMPLE1}
\end{figure}

However, although narrow ranges increase the potential fee income, they also increase concentration risk; a wide spread (i.e., a lower value of the depth $\tilde \kappa$)  decreases fee income per LT transaction but reaps earnings from a larger number of LT transactions because the position is active for longer periods of time (i.e., it takes longer, on average, for $Z$ to exit the LP's range). Thus, the LP must strike a balance between maximising the depth $\tilde \kappa$ around the rate and minimising the concentration risk, which depends on the volatility of the rate $Z.$

In our model, the LP continuously adjusts her position around the current rate $Z,$ so we write the continuous-time dynamics of  \eqref{eq:remuneration_LP}, conditional on the rate not exiting the LP's range, as
\begin{align}
\label{eq:feeDynamics00}
\diff p_{t}= \underbrace{ \left(\tilde{\kappa}_{t} \, / \, \kappa\right)}_{\text{Position depth}} \ \underbrace{\pi_t }_{\text{Fee rate}} \ \underbrace{\,2 \,\kappa \,\sqrt{Z_{t}} }_{\text{Pool size}} \, \diff t\,,
\end{align}
where $(\tilde \kappa_t)_{t \in [0,T]}$ models the depth of the LP's position and $p$ is the LP's fee income for providing  liquidity with depth $\tilde\kappa$ in the pool. The fee income is proportional to the pool size, i.e., proportional to $2\,\kappa\, \sqrt{Z_t}.$ Next, use the second equation in \eqref{LP:eq:position} and  equations \eqref{eq:ZlZucontrol}--\eqref{eq:holdingspool} to write the dynamics of the LP's position depth $\tilde \kappa_t$ as
\begin{align}
\label{eq:dynKappaTildeDis0}
\tilde{\kappa}_{t}=2\,V_{t}\,\left(\frac{1}{\delta_{t}^{L}+\delta_{t}^{U}}\right)\,\frac{1}{\sqrt{Z_{t}}}\,,
\end{align}
so the dynamics in \eqref{eq:feeDynamics00} become
\begin{align}
\label{eq:feeDynamics0}
\diff p_{t} = \left(\frac{4}{\delta_{t}^{L}+\delta_{t}^{U}}\right)\, \pi_t\, V_{t}\, \diff t\,.
\end{align}

In practice, {there is latency in the market and the LP cannot reposition her liquidity  position and rebalance her assets continuously.  Thus,  
the LP faces concentration risk in between the times she repositions her liquidity; narrow spreads generate less fee income because the rate $Z$ may exit the range of the LP's liquidity, especially in volatile markets.}

{
To model the losses due to concentration risk in the continuous-time dynamics  \eqref{eq:feeDynamics0} of the LP's fee revenue, we introduce an instantaneous concentration cost that reduces the fees collected by the LP as a function of the spread;} the concentration cost increases (decreases) when the spread narrows (widens). We modify the dynamics of the fees collected by the LP  in \eqref{eq:feeDynamics0}  as follows
\begin{align}
\label{eq:feedyn_0}
\diff p_{t} = \left(\frac{4}{\delta_{t}^{L}+\delta_{t}^{U}}\right)\, \pi_t\, V_{t}\, \diff t\, - \gamma\,\left(\frac{1}{\delta_t^{ L }+\delta_t^{u}}\right)^{2}\,V_t\,\diff t\,,
\end{align}
where $\gamma>0$ is the concentration cost parameter and $V_t$ is the wealth invested by the LP in the pool at time $t.$ 

{
To justify the form of the concentration cost, we study the realised fee revenue  in the ETH/USDC pool as a function of the spread $\delta$ in \eqref{eq:pos_spread} for rebalancing frequencies $m=1$ minute and $m=5$ minutes. We denote by $\widehat  p_{\delta,m}$ the average realised fee revenue earned by an LP in the ETH/USDC pool who provides liquidity with wealth $V = 1$, in a range with spread $\delta$, around the marginal rate $Z$ throughout a time window of length $m$. The left panel of Figure \ref{fig:gamma_1_5min} shows that the value of $\delta$ which maximises the fee revenue $\widehat  p_{\delta,m}$ is strictly positive for both values of $m$ that we consider. In particular, narrow ranges reduce fee revenue due to concentration risk.

The form of the fee revenue dynamics \eqref{eq:feedyn_0} is a second order Taylor approximation that captures the specific shape of the fee revenue in CL markets; see left panel of Figure \ref{fig:gamma_1_5min}. The LP uses the regression model
\begin{equation}
    \delta^2 \, \widehat p_{\delta,m} = 4\,\pi\,m\,\delta - \gamma \, m\,,
\end{equation} 
which is based on the dynamics \eqref{eq:feedyn_0}, to estimate the concentration cost parameter $\gamma$. 
The right panel of Figure \ref{fig:gamma_1_5min} shows that $\delta^2 \, \widehat p_{\delta,m}$ is affine in $\delta$ and that the estimated concentration cost parameter $\gamma$ depends on the rebalancing frequency of the LP. In particular, the figure shows that the dynamics \eqref{eq:feedyn_0} and the second order approximation are suitable to describe the realised fee revenue in CL markets. The performance study of Section \ref{sec:num1} uses this methodology to set the value of the concentration cost parameter $\gamma$.

\begin{figure}[!h]\centering
\includegraphics{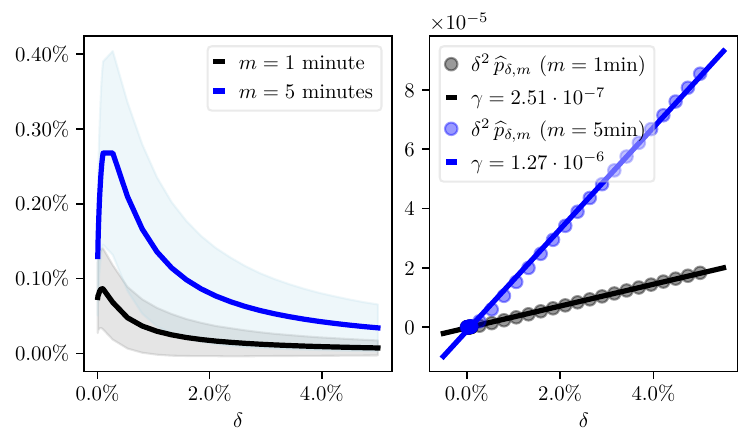}\\
\caption{{Left panel: mean and standard deviation of the fee revenue of hypothetical liquidity positions as a function of the position spread $\delta$ for the rebalancing frequencies $m=1$ minute and $m=5$ minutes. The fee revenue of each position considers a wealth $V=1$ and uses historical LT transactions of the ETH/USDC pool of Uniswap v3. Right panel: $\delta^2\,\widehat p_{\delta,m}$ as a function of the position spread $\delta$ for the rebalancing frequencies $m=1$ minute and $m=5$ minutes. The concentration cost parameter is estimated as $-\iota / m$ where $\iota$ is the intercept of the regression of $\delta^2\,\widehat p_{\delta,m}$ on $\delta.$}  }\label{fig:gamma_1_5min}
\end{figure}

}

\subsubsection{Fee income: drift and asymmetry}

The stochastic drift $\mu$ indicates the future expected changes of the marginal exchange rate in the pool. In practice, the LP may use a predictive signal so that $\mu$ represents the belief that the LP holds over the future marginal rate. For an LP who maximises fee revenue, it is natural to consider asymmetric liquidity positions that capture the liquidity taking flow. We define the \textit{asymmetry} of a position as 
\begin{align}
\label{eq:asymmetryrhot}
\rho_t = \frac{\delta^U_t}{\delta^U_t + \delta^L_t } = \frac{\delta^U_t}{\delta_t}\, ,
\end{align}
where $\delta^U_t$ and $\delta^L_t$ are defined in \eqref{eq:ZlZucontrol}. In one extreme, when the asymmetry $\rho \to 0,$ then $Z^U \to Z$ and  the position consists of only asset $X,$ and in the other extreme, when $\rho \to 1,$ then $Z^L \to Z$ and the position consists of only asset $Y.$ 

{In the remainder of this work, the asymmetry of the LP's position is a function of the observed drift:
\begin{align}
\label{eq:relation_rho_mu}
\rho_t = \rho\left(\delta_t, \mu_t\right) = \frac12 + \frac{\mu_t}{\delta_t} = \frac12 + \frac{\mu_t}{\delta_t^U + \delta_t^L} \,, \quad \forall t \in [0, T]\,.
\end{align}
The asymmetry \eqref{eq:relation_rho_mu} adapts the skew of the position to the expected drift of the marginal rate. When $\mu=0$, the position is symmetric around the marginal rate and $\rho_t = 1/2$, so $\delta_t^L = \delta_t^U.$ When $\mu>0,$ the position is skewed to the right (i.e., $\delta_t^U > \delta^L_t$) to capture more LT trades and fee revenue, and similarly when $\mu<0,$ the position is skewed to the left (i.e., $\delta_t^U < \delta^L_t$). Also, when $\mu>0$ and the position is skewed to the right according to \eqref{eq:relation_rho_mu}, the proportion of asset $Y$ increases and the position profits from rate appreciation, and when the position is skewed to the left, the proportion of asset $Y$ decreases and the position is protected from rate depreciation.

Optimal liquidity provision is the dynamic choice of  $\delta^{U}$ and $\delta^{L}$ over a trading window, or equivalently, the dynamic choice of  $\delta$ and $\rho$. Our model assumes that the LP fixes the asymmetry $\rho$ of her position at time $t$ according to \eqref{eq:relation_rho_mu}, so we reduce the trading problem to a one-dimensional dynamic optimisation problem, which significantly simplifies calculations. 

To further justify the form of the asymmetry }\eqref{eq:relation_rho_mu}, we use Uniswap v3 data to study how the asymmetry and the width of the LP's range of liquidity relate to fee revenue. First, we estimate the realised drift $\mu$ in the pool ETH/USDC over a rolling window of $T=5$ minutes.\footnote{The values of the drift in this section are normalised to reflect daily estimates. In particular, we use $\mu = \tilde \mu \, / \, \Delta t$ where $\tilde{\mu}$ is the average of the observed log returns and  $\Delta t$ is the observed average LT trading frequency.} Next, for any time $t$, the fee income for different positions of the LP's liquidity range is computed for various values of the spread $\delta$ and for various values of the asymmetry $\rho.$ For each value of the realised drift $\mu$ during the investment horizon, and for each fixed value of the spread $\delta,$ we record the asymmetry that maximises fee income. Figure \ref{fig:DOLP5} shows the optimal (on average) asymmetry $\rho$ as a function of the spread  $\delta$ of the position for multiple values of the realised drift $\mu.$ 

\begin{figure}[!h]\centering
\includegraphics{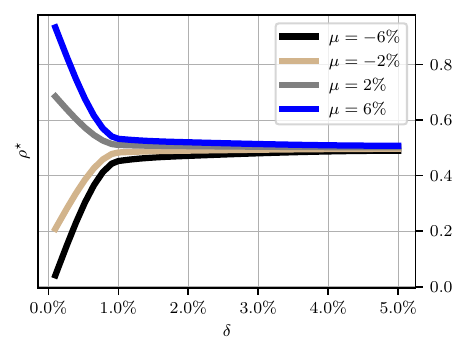}
\caption{Optimal position asymmetry $\rho^\star$ in \eqref{eq:asymmetryrhot} as a function of the spread  $\delta$ of the position, for multiple values of the drift $\mu.$ {For each historical value of the drift in our data, we compute the performance of hypothetical liquidity positions with asymmetry values in the range $[0, 1]$ and position spreads in the range $[0.05\%, 5\%].$} The  asymmetry $\rho^\star$ in the figure is the value of $\rho$ that maximises the average fee income for each value of the pair $\left(\delta, \mu\right)$. {The drift is computed as the mean of $5$-minute returns during one hour of trading, and the performance of the liquidity positions is computed as the total fee revenue if the position is held for an hour. } }\label{fig:DOLP5}
\end{figure}

Figure \ref{fig:DOLP5} suggests that there exists a preferred  asymmetry of the position for a given value of the spread $\delta$ and a given  value of the drift  $\mu.$ First, for all values of the spread $\delta$, the LP skews her position to the right when the drift is positive ($\rho^\star > 0.5$) and she skews her position to the left when the drift is negative ($\rho^\star < 0.5$). Second, for narrow spreads, the liquidity position requires more asymmetry than for large spreads when the drift is not zero. 

In our liquidity provision model of Section \ref{sec:4}, the LP  holds a belief over the future exchange rate throughout the investment window and controls the spread $\delta=\delta^U + \delta^L$  of her position. Thus, she strategically chooses the asymmetry of her position as a function of $\delta$ and $\mu.$  We approximate the relationship exhibited in Figure  \ref{fig:DOLP5} with the asymmetry function \eqref{eq:relation_rho_mu}.


{

\subsection{Rebalancing costs and gas fees \label{sec:costs}}

Our model considers a strategic LP who continuously repositions her liquidity position in the pool which requires rebalancing of the LP's assets.  Specifically, repositioning typically leads to different holdings \eqref{eq:holdingspool} in the pool, in which case the LP trades one asset for the other in the pool or in another trading venue. However, rebalancing assets to reposition liquidity is costly. 

Let $\left(c_t\right)_{t\in[0,T]}$ denote the cost of rebalancing in terms of asset $X$. Similar to \cite{fan2021strategic} and \cite{fan2022differential}, we model rebalancing costs as proportional to the quantity $\tilde y_t$ of asset $Y$ that the LP deposits in the pool.  At any time $t$, we assume that the LP uses all her wealth $V_t$ when repositioning her liquidity position, so we use \eqref{eq:holdingspool} to write
\begin{equation}\label{eq:rebalancing_cost_dyns}
c_t = -\zeta\,\tilde y_t\,Z_t = -\zeta\,\frac{\delta_{t}^{U}}{\delta_{t}^{L}+\delta_{t}^{U} }\,V_{t}\,, \quad c_0=0\,,
\end{equation}
where $\zeta>0$ is a constant that models the execution costs.

Transactions sent to the pool also bear gas fees. Gas fees are a flat fee paid to the blockchain and do not depend on the size of the LP's transaction; see \cite{li2023yield} for more details on the impact of gas fees on liquidity provision. Thus, gas fees scale with the frequency at which the LP sends transactions to the pool. Our model considers continuous trading, so gas fees do not influence the optimisation problem, but should be considered when the strategy is implemented.}

In the next section, we derive an optimal liquidity provision strategy, and prove that the profitability of liquidity provision  is subject to a tradeoff between fee revenue, PL, and concentration risk. 

\section{Optimal liquidity provision in CL pools \label{sec:4}}

\subsection{The problem}

Consider an LP who provides liquidity in a CPMM with CL throughout the investment window $[0, T]$. We work on the filtered probability space $\left(\Omega, \mathcal F, \mathbbm P; \mathbbm F = (\mathcal F_t)_{t \in [0,T]} \right)$ where $\mathcal F_t$ is the natural filtration generated by the collection $\left(Z, \mu, \pi\right)$.

The dynamics of the LP's wealth consist of the fees earned, the position value, and {the rebalancing costs}.  Similar to Section \ref{sec:positionvalue}, we denote the wealth process of the LP by $(V_t = \alpha_t +  p_t + c_t)_{t \in [0,T]},$ with $V_0 > 0$ known, where $\alpha$ is the value of the LP's position, $p$ is the fee revenue{, and $c$ is the rebalancing cost}. At any time $t,$ the LP uses her wealth $V_t$ to provide  liquidity. {Now, use \eqref{eq:dynAlphaFirst}, \eqref{eq:feedyn_0}, and \eqref{eq:rebalancing_cost_dyns} to write the dynamics of the LP's wealth as
\begin{align}
\label{eq:temp_x_dyn}
\diff V_{t}=V_{t}\left(\frac{1}{\delta_{t}^{L}+\delta_{t}^{U}}\right)\left[\left(4\,\pi_{t}-\frac{\sigma^{2}}{2}+\left(\mu_{t}-\zeta\,\right)\,\delta_{t}^{U}\right)\diff t+\sigma\,\delta_{t}^{U}\,\diff W_{t}\right]-\gamma\,\left(\frac{1}{\delta_{t}^{L}+\delta_{t}^{U}}\right)^{2}V_t\,\diff t\,,\ \ 
\end{align}
where $\gamma \geq 0$ is the concentration cost parameter, see Subsection \ref{sec:feerevenue}, and $\pi$ follows the dynamics defined below in \eqref{eq:dyn_pi}. To simplify notation, we set $\zeta=0$; our results hold when replacing $\mu$ by $\mu-\zeta$.}

Next, use $\delta_t = \delta^U_t + \delta_t^L$ and $\delta^U_t \, / \, \delta_t = \rho\left(\mu_t, \delta_t\right)$ in \eqref{eq:temp_x_dyn} to write the dynamics of the LP's wealth as
\begin{align}
\label{eq:dynxtilde}
\diff V_{t}&=\frac{1}{\delta_{t}}\,\left(4\,\pi_{t}-\frac{\sigma^{2}}{2}\right)\, V_{t}\,\diff t+\mu_t\,\rho\left(\delta_{t},\mu_t\right)\, V_{t}\,\diff t+\sigma\,\rho\left(\delta_{t},\mu_t\right)\, V_{t}\,\diff W_{t}-\frac{\gamma}{\delta_{t}^{2}}\, V_{t}\,\diff t\,.
\end{align}

{Next, following the discussions of Section \ref{sec:feerevenue}, we denote by $\eta$ the \textit{profitability threshold} and we assume that the pool fee rate $\pi$ follows the dynamics}
\begin{align}
\label{eq:dyn_pi}
\diff(\pi_{t} - \eta_t)&=\Gamma\,\left(\overline{\pi}+\eta_t-\pi_{t}\right)\diff t+\psi\,\sqrt{\pi_{t}-\eta_t}\,\diff B_{t}\,,
\end{align}
where  $\Gamma>0$ denotes the mean reversion speed, $\overline{\pi} > 0$ is the long-term mean of $\left(\pi_t - \eta_t\right)_{t\in[0,T]}$, $\psi>0$ is a non-negative volatility parameter, $\left(B_t\right)_{t\in[0,T]}$ is a Brownian motion independent of $\left(W_t\right)_{t\in[0,T]},$ and $\pi_0-\eta_0>0$ is known. 
{To solve the LP's optimal liquidity provision problem, we introduce the following assumption.
\begin{assume}\label{assumption:eta}
The profitability {threshold} $\eta$ in the dynamics  \eqref{eq:dyn_pi} is given by
\begin{equation}\label{eq:eta stochastic}
\eta_t=\frac{\sigma^{2}}{8}-\frac{\mu_t}{4}\left(\mu_t-\frac{\sigma^{2}}{2}\right)+\frac{\varepsilon}{4}\,.
\end{equation}
\end{assume}}

From Assumption \ref{assumption:eta} and the CIR dynamics \eqref{eq:dyn_pi} it follows that
\begin{align}\label{assumption:1}
    \pi_{t} - \eta_t \geq 0 \implies 4\,\pi_{t}-\frac{\sigma^{2}}{2}+\mu_t\left(\mu_t-\frac{\sigma^{2}}{2}\right)\ge\varepsilon>0\,, \quad \forall t \in [0, T]\,.
\end{align}
{Assumption \ref{assumption:eta} and condition \eqref{assumption:1}} ensure that the spread $\delta$ of the optimal strategy is {admissible.  Financially, the inequality in \eqref{assumption:1} is a profitability condition $\pi_t \ge \eta_t$ that 
guarantees that the LP's fee income $\pi$ is greater than the PL faced by the LP, adjusted by the drift in the marginal rate.}

We impose specific dynamics for the fee rate $\pi$ such that it satisfies a profitability condition \eqref{assumption:1} that allows us to obtain an admissible strategy. While  we use this constraint to solve the optimal liquidity provision problem, it also represents an adequate and natural measure for LPs to assess the profitability of liquidity provision in different pools before depositing their assets.

Finally, note that the dynamics in \eqref{eq:dynZ} and \eqref{eq:dyn_pi} imply that the LP also observes $W,$ $B,$ and the {profitability threshold} $\eta$ is determined by $\mu$, so the LP observes all the stochastic processes of this problem.

\subsection{The optimal strategy}

The LP controls the spread  $\delta$ of her position to maximise the expected utility of her terminal wealth in units of $X$. {To define the set of admissible strategies $\mathcal{A}$, note that if the LP assumes that the asymmetry function $\rho$ is that in  \eqref{eq:relation_rho_mu}, then for each $\delta$, we need
\begin{equation}\label{eq:rho_condition}
    \int_{0}^{T} \rho\left(\delta_{t},\mu_t\right)^2 \,\diff t<\infty  \quad \mathbbm{P}\textrm{--a.s.}\,.
\end{equation}
Observe that 
\begin{equation}
    \begin{split}
        \int_{0}^{T} \rho\left(\delta_{t},\mu_t\right)^2 \,\diff t \,= \,& \int_{0}^{T} \left(\frac12 + \frac{\mu_t}{\delta_t}\right)^2 \,\diff t\,
        \leq \, \frac T4 + \frac{1}{2}\,\int_{0}^{T} \mu_t^4 \,\diff t+\frac{1}{2}\,\int_{0}^{T} \frac{1}{ \delta_t^4} \,\diff t\,.
    \end{split}
\end{equation}
Thus, to satisfy \eqref{eq:rho_condition} and to ensure that the control problem below is well-posed, we define the set of admissible strategies
\begin{align}
\label{LP:def:admissibleset_t}
\mathcal{A}_{t}=\left\{ (\delta_{s})_{s\in[t,T]},\ \R\textrm{-valued},\ \mathbbm{F}\textrm{-adapted, and }\int_{t}^{T}\frac{1}{\delta_{s}^{4}}\,\diff s<+\infty \ \mathbbm{P}\textrm{--a.s.}\right\} \, ,
\end{align}
where $\mathcal A := \mathcal A_0.$ 
}

Let $\delta\in\Ac$. The performance criterion of the LP is a function $u^{\delta}\colon[0,T] \times \R^4 \rightarrow \R$ given by
\begin{align}
\label{LP:eq:perfcriteria}
    u^{\delta}(t,v,z, \pi,\mu)=\mathbbm{E}_{t,v,z, \pi,\mu}\left[U\left(V_{T}^{\delta}\right)\right]\, ,
\end{align}
where $U$ is a concave utility function, and the value function $u:[0,T] \times \R^4 \rightarrow \R$ of the LP is
\begin{equation}
\label{LP:eq:valuefunc}
    u(t,v,z,\pi,\mu)=\underset{\delta\in\mathcal{A}_t}{\sup} \, u^{\delta}(t,v,z,\pi,\mu)\, .
\end{equation}

The following results solve the optimal liquidity provision model when the LP assumes a general stochastic drift $\mu$ and her preferences are given by a logarithmic utility function.

\begin{prop}
\label{prop:1}
Assume the asymmetry function $\rho$ is as  in  \eqref{eq:relation_rho_mu} and that $U(v) = \log(v).$ Then,
\begin{align}
\label{eq:hjbsol}
\begin{split}
w\left(t,v,z,\pi,\mu\right)=&\log\left(v\right)+\left(\pi-\eta\right)^{2}\int_{t}^{T}\E_{t,\mu}\left[\frac{8}{2\,\gamma+\mu_{s}^{2}\,\sigma^{2}}\right]\exp\left(-2\,\Gamma\left(s-t\right)\right)\,\diff s\\&+\left(\pi-\eta\right)\left(2\,\Gamma\,\overline{\pi}+\psi^{2}\right)\int_{t}^{T}\E_{t,\mu}\left[C\left(s,\mu_{s}\right)\right]\exp\left(-\Gamma\left(s-t\right)\right)\,\diff s\\&-\left(\pi-\eta\right)\int_{t}^{T}\E_{t,\mu}\left[\frac{4\,\varepsilon}{2\,\gamma+\sigma^{2}\,\mu_{s}^{2}}\right]\exp\left(-\Gamma\left(s-t\right)\right)\,\diff s\\&+\int_{t}^{T}\left(\Gamma\,\overline{\pi}\,\E_{t,\mu}\left[E\left(s,\mu_{s}\right)\right]+\psi^{2}\,\E_{t,\mu}\left[\eta_{s}\,C\left(s,\mu_{s}\right)\right]\right)\diff s
\\
&-\frac{1}{2}\,\int_{t}^{T}\left(\E_{t,\mu}\left[\frac{\varepsilon^{2}}{2\,\gamma+\sigma^{2}\,\mu_{s}^{2}}+\mu_{s}\right]\right)\diff s-\pi\frac{\sigma^{2}}{8}\left(T-t\right)\,
\end{split}
\end{align}
solves the HJB equation associated with problem \eqref{LP:eq:valuefunc}. Here, $\eta_s = \frac{\sigma^{2}}{8}-\frac{\mu_s}{4}\left(\mu_s-\frac{\sigma^{2}}{2}\right)+\frac{\varepsilon}{4}$ for $s\ge t\,,$ $\eta_t = \eta$, and  $\E_{t,\mu}$ represents expectation conditioned on $\mu_t = \mu$, and $$C\left(t,\mu\right)=\E_{t,\mu}\left[\,\int_{t}^{T}\frac{8}{2\,\gamma+\mu_{s}^{2}\,\sigma^{2}}\exp\left(-2\,\Gamma\left(s-t\right)\right)\,\diff s\right]\,,$$ and $$E\left(t,\mu\right)=\E_{t,\mu}\left[\,\int_{t}^{T}\left(\left(2\,\Gamma\,\overline{\pi}+\psi^{2}\right)C\left(s,\mu\right)+\frac{4\,\varepsilon}{2\,\gamma+\sigma^{2}\,\mu_{s}^{2}}\right)\exp\left(-\Gamma\left(s-t\right)\right)\,\diff s\right]\,.$$

\end{prop}

\noindent
For a proof, see \ref{sec:proofs:hjb1}.

\begin{thm}
\label{thm:verif}
Let {Assumption \ref{assumption:eta} hold and assume that } the asymmetry function $\rho$ is as in \eqref{eq:relation_rho_mu} and that $U(v)= \log(v).$ Then, the solution in Proposition \ref{prop:1} is the unique solution to the optimal control problem \eqref{LP:eq:valuefunc}, and the optimal spread $\left(\delta_s\right)_{s\in[t,T]} \in \mathcal A_t$ is given by
\begin{equation}
    \label{eq:optimalspeed_lp}
    \begin{split}
        \delta_{s}^{\star}=\frac{2\,\gamma+\mu_s^{2}\,\sigma^{2}}{4\,\pi_{s}-\frac{\sigma^{2}}{2}+\mu_s\left(\mu_s-\frac{\sigma^{2}}{2}\right)}=\frac{2\,\gamma+\mu_s^{2}\,\sigma^{2}}{4\left(\pi_{s}-\eta_s\right)+\varepsilon}\,,
    \end{split}
\end{equation}
where $\eta_s = \frac{\sigma^{2}}{8}-\frac{\mu_s}{4}\left(\mu_s-\frac{\sigma^{2}}{2}\right)+\frac{\varepsilon}{4}\,.$
\end{thm}

\noindent
For a proof, see \ref{sec:proofs:hjb2}.

\subsection{Discussion: profitability, PL, and concentration risk \label{sec:4:discussion}}

{In this section, we study the strategy 
when $\mu = 0\,,$ in which case the position is symmetric, so $\rho = 1/2$ and $\delta^U_t = \delta^L_t = \delta_t / 2$ and the optimal spread \eqref{eq:optimalspeed} becomes}\footnote{The position range is approximately symmetric around the position rate $Z$ because $\delta^U_t = \delta^L_t$ does not imply that $Z-Z^L = Z^U-Z;$ see \eqref{eq:ZlZucontrol}. However, for small values of $\delta^L$ and $\delta^U,$ one can write the approximation $Z-Z^L \approx Z^U-Z$, in which case the position is symmetric around the rate $Z$.}
\begin{align}
\label{eq:optimalspreadsul_symm}
\delta_{t}^{ L \,\star}=\delta_{t}^{U\,\star}=\frac{2\,\gamma}{8\,\pi_{t}-\sigma^{2}}\,\quad \implies \quad \delta_{t}^{\star}=\frac{4\,\gamma}{8\,\pi_{t}-\sigma^{2}},
\end{align}
so the inequality in \eqref{assumption:1} becomes \begin{align}\label{assumption:1_symm}
    4\,\pi_{t}-\frac{\sigma^{2}}{2}\ge\varepsilon>0\,, \quad \forall t \in [0, T]\,.
\end{align}

The inequality in \eqref{assumption:1_symm} guarantees that the optimal control \eqref{eq:optimalspreadsul_symm} does not explode, and ensures that fee income is large enough for LP activity to be profitable. In particular, it ensures that $\pi > \sigma^2/8 + \varepsilon.$ When $\varepsilon \to 0$, i.e., $\sigma^2/4 \to \pi,$ the spread $\delta \to +\infty\,.$ However, we require that the spread $\delta = \delta^U + \delta^L \le 4$, so the conditions $\delta^L \leq 2$ and $\delta^U \leq 2$ become
\begin{align}
\label{eq:cond0}
\frac{\gamma}{4\,\pi-\frac{\sigma^{2}}{2}}\leq2 \implies \pi-\frac{\gamma}{8}\geq\frac{\sigma^{2}}{8} \,. 
\end{align}

When $\delta^L = \delta^U = 2,$ the LP provides liquidity in the maximum range $(0\,, +\infty),$ so the depth of her liquidity position $\tilde \kappa$ is minimal, and the LP's position is equivalent to providing liquidity in CPMMs without CL; see \cite{cartea2023predictable} for more details.

On the other hand, when $\delta\leq4$, the depreciation rate of the LP's position value in \eqref{eq:dynAlphaFirst} is higher {than $\sigma^2/8$}. In particular, if $\delta = \delta^\text{tick}$, where $\delta^\text{tick}$ is the spread of a liquidity position concentrated within a single tick range, then the depth of the LP's liquidity position $\tilde \kappa$ is maximal and the LP's exposure to volatility is maximal.

LPs should track the profitability of the pools they consider before depositing their assets in the pool. When $\mu=0,$ we propose that LPs use $\sigma^2 / 8$ as a {lower bound} rule-of-thumb for the  pool's rate of profitability because $\sigma^2 / 8$ is the lowest rate of depreciation of their wealth in the pool.

Condition \eqref{eq:cond0} ensures that the profitability $\pi - \gamma / 8,$ which is the pool fee rate adjusted by the concentration cost, is higher than the depreciation rate of the LP's assets in the pool. Thus, the condition imposes a  minimum profitability level of the pool, so LP activity is viable. An optimal control $\delta^\star > 4$ indicates non-viable LP activity because fees are not enough to compensate the losses borne by the LP. Figure \ref{fig:DOLP_feerate_vol} shows the estimated pool fee rate and the estimated depreciation rate in the ETH/USDC pool (from January to August 2022).  In particular, the CIR model captures the dynamics of $\pi_t - \sigma^2 / 8.$

\begin{figure}[!h]\centering
\includegraphics{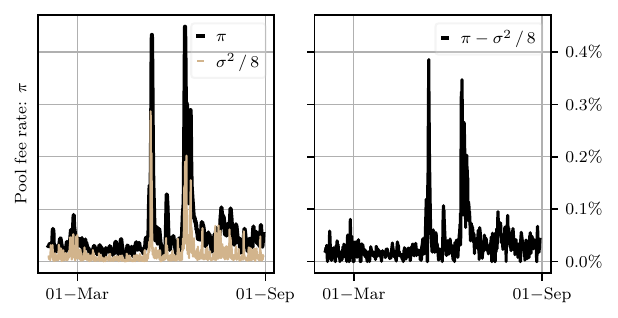}\\
\caption{Estimated pool fee rate from February to August 2022 in the ETH/USDC pool. For any time $t$, the pool fee rate is the total fee income, as a percentage of the total pool size, paid by LTs in the period $[t-1\text{ day}, t].$ The pool size at time $t$ is $2 \, \kappa \, \sqrt{Z_t}$ where $Z_t$ is the active rate in the pool at time $t.$ }\label{fig:DOLP_feerate_vol}
\end{figure}

Next, we study the dependence of the optimal spread on the value of the concentration cost coefficient $\gamma,$ the fee rate $\pi,$ and the volatility $\sigma$. The concentration cost coefficient $\gamma$ scales  the spread linearly in \eqref{eq:optimalspreadsul_symm}. Recall that the cost term penalises small spreads because there is a risk that the rate will exit the LP's range. Thus, large values of $\gamma$ generate large values of the spread. Figure \ref{fig:OLP_dyn0} shows the optimal spread as a function of the pool fee rate $\pi$. Large potential fee income pushes the strategy towards targeting more closely the marginal rate $Z$ to profit from fees. Lastly, Figure \ref{fig:OLP_dyn2} shows that the optimal spread increases as the volatility of the rate $Z$ increases.

\begin{figure}[!h]
\centering 
\includegraphics{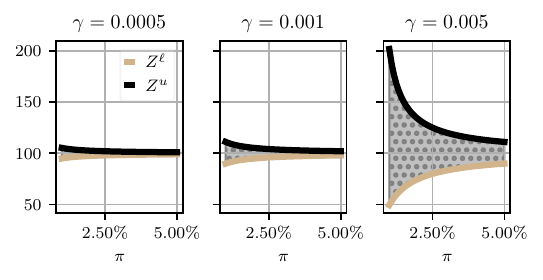}\\
\caption{Optimal LP position range $\left(Z^L, Z^U\right]$ as a function of the pool fee rate $\pi$ for different values of the concentration cost parameter $\gamma,$ when $Z = 100,$ $\sigma = 0.02,$ and $\mu=0.$}\label{fig:OLP_dyn0}
\end{figure}

\begin{figure}[!h]
\centering 
\includegraphics{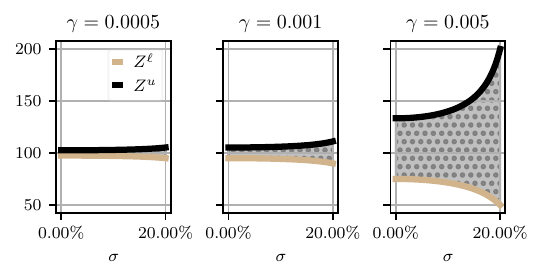}\\
\caption{Optimal LP position range $\left(Z^L, Z^U\right]$ as a function of the volatility $\sigma$ for different values of the concentration cost parameter $\gamma,$ when $Z = 100,$ $\pi = 0.02,$ and $\mu=0.$}\label{fig:OLP_dyn2}
\end{figure}

Finally, the optimal spread does not depend on time or on the terminal date $T.$ The LP marks-to-market her wealth in units of $X,$ but does not penalise her holdings in asset $Y$. In particular, the LP's performance criterion does not include a running penalty or a final liquidation penalty (to turn assets into cash or into the reference asset). For example, if at the end of the trading window the holdings in asset $Y$ must be exchanged for $X$, then the optimal strategy would skew, throughout the trading horizon,  the liquidity range  to convert holdings in $Y$ into $X$ through LT activity.\footnote{In LOBs, one usually assumes that final inventory is liquidated with adverse price impact and that there is a running inventory penalty, thus market making strategies in LOBs depend on the terminal date $T$.  }

\subsection{Discussion: drift and position skew}

In this section, we study how the strategy depends on the stochastic drift $\mu$. Use $\delta_t = \delta^L_t + \delta^U_t$ and $\rho\left(\delta_t, \mu_t\right) = \delta^U_t / \delta_t$ to write the two ends of the optimal spread as
\begin{align}
\label{eq:optimalspreadsul}
\delta_{t}^{U \, \star}=\frac{2\,\gamma+\mu_t^{2}\,\sigma^{2}}{8\,\pi_{t}-\sigma^{2}+2\,\mu_t\left(\mu_t-\frac{\sigma^{2}}{2}\right)}+\mu_t \quad \text{and} \quad \delta_{t}^{ L \, \star}=\frac{2\,\gamma+\mu_t^{2}\,\sigma^{2}}{8\,\pi_{t}-\sigma^{2}+2\,\mu_t\left(\mu_t-\frac{\sigma^{2}}{2}\right)}-\mu_t\,.
\end{align}

The inequality in \eqref{assumption:1} guarantees that the optimal control in \eqref{eq:optimalspreadsul} does not explode and ensures that fee income is large enough for LP activity to be profitable. The profitability condition in \eqref{eq:cond0} becomes
\begin{equation}
\label{eq:cond1}
\pi_{t}-\frac{\gamma}{8}\ge\frac{\sigma^{2}}{8}\left(\frac{\mu_t^{2}\,}{2}+1\right)-\frac{\mu_t}{4}\left(\mu_t-\frac{\sigma^{2}}{2}\right)\,,
\end{equation}
so LPs that assume a stochastic drift in the dynamics of the exchange rate $Z$ should use this simplified measure of the depreciation rate due to PL as a rule-of-thumb before considering depositing their assets in the pool.

Next, we study the dependence of the optimal spread on the value of the drift $\mu$.  First, recall that the controls in \eqref{eq:optimalspreadsul} must obey the inequalities\footnote{The admissible set of controls is not restricted to these ranges. However, values outside these range cannot be implemented in practice.}
$$0<\delta^{L}_t\leq2 \quad \text{and}\quad 0\leq\delta^{U}_t<2\,,$$
because $0 \leq Z^L < Z^U < \infty$ and $Z_t \in (Z^L_t, Z^U_t]$, which together with \eqref{eq:pos_spread} implies $0 \leq\delta_t\leq 4$. Next, the asymmetry function satisfies 
\begin{equation}\label{eq:rho}
    0<\rho\left(\delta_t, \mu\right) = \frac{\delta^U_t}{\delta_t}<1\,,
\end{equation}
which implies
\begin{equation}
\label{eq:double_rho}
0\leq\rho \left(\delta_t, \mu\right)\,\delta_t<2\quad \text{and}\quad 0\leq\left(1-\rho \left(\delta_t, \mu\right)\right)\,\delta_t<2\,.
\end{equation}
Now, use \eqref{eq:relation_rho_mu} and \eqref{eq:double_rho} to write
\begin{equation}
\label{eq:double_rho_two}
0\leq\left(\frac{1}{2}+\frac{\mu}{\delta_t}\right)\,\delta_t<2\quad \text{and}\quad
0\leq\left(\frac{1}{2}-\frac{\mu}{\delta_t}\right)\,\delta_t<2\,.
\end{equation}
Finally, use \eqref{eq:rho}  and \eqref{eq:double_rho_two} to obtain the inequalities
\begin{align}
\label{eq:conditionFin}
2\,\left|\mu\right|\leq\delta_t\leq4-2\,\left|\mu\right|,
\end{align}
so $\mu$ must be in the range $\left[-1, 1\right]$ for the LP to provide liquidity. If $\mu$ is outside this range, concentration risk is too high so the LP must withdraw her holdings from the pool. Recall that the dynamics of $Z$ are geometric and $\mu$ is a percentage drift, so values of $\mu$ outside the range $\left[-1, 1\right]$ are unlikely. Moreover, when $\mu=-1,$ the drift of the exchange rate $Z$ is large and negative, so the optimal range is $(0, Z],$ i.e., the largest possible range to the left of $Z.$ When $\mu=1,$ the drift of the exchange rate $Z$ is large and positive, so the optimal range is $(Z, +\infty),$ which is the largest possible range to the right of $Z.$ Condition \eqref{eq:conditionFin} is always verified when we study the performance of the strategy in the ETH/USDC pool. Figure \ref{fig:OLP_dyn1} shows how the optimal spread adjusts to the value of the drift $\mu$. Finally, note that
\begin{equation}\label{eq:derivative_mu}
    \frac{\partial\delta^{U \, \star}}{\partial\sigma}=\frac{\partial\delta^{L \, \star}}{\partial\sigma}=\frac{2\,\mu^{2}\,\sigma\left(4\,\pi-4\,\eta+\varepsilon\right)+4\,\sigma\left(1+\mu\right)\left(2\,\gamma+\mu^{2}\,\sigma^{2}\right)}{\left(4\,\pi-4\,\eta+\varepsilon\right)^{2}}>0\,,\qquad \forall\mu\in[-1,1]\,,
\end{equation}
shows that the optimal range is strictly increasing in the volatility $\sigma$ of the rate $Z,$ which one expects as increased activity  that exposes the position value to more PL, and increases the concentration risk.

\begin{figure}[!h]
\centering 
\includegraphics{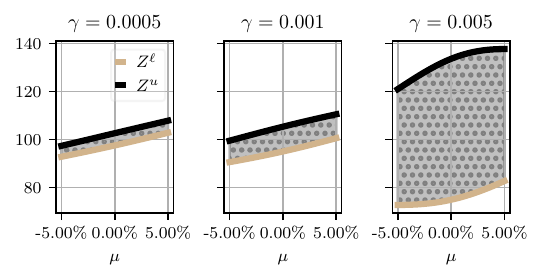}\\
\caption{Optimal LP position range $\left(Z^L, Z^U\right]$ as a function of the drift $\mu$ for different values of the concentration cost parameter $\gamma,$ when $Z = 100,$ $\pi = 0.02,$ and $\sigma=0.02.$}\label{fig:OLP_dyn1}
\end{figure}

\section{Performance of strategy\label{sec:num1}}

\subsection{Methodology}
In this section, we use Uniswap v3 data between 1 January and 18 August 2022 to study the performance of the strategy of Section \ref{sec:4}. We consider execution costs and discuss how gas fees and liquidity taking activity in the pool affect the performance of the strategy.\footnote{In practice, LPs pay gas fees when using the Ethereum network to deposit liquidity, withdraw liquidity, and adjust their holdings. Gas is paid in Ether, the native currency of the Ethereum network, and measures the computational effort of the LP operation; see \cite{cartea2022decentralised}.} 

Our strategy in Section \ref{sec:4} is solved in continuous time. In our performance study, we discretise the trading window in one-minute periods and the optimal spread is fixed at the beginning of each time-step. That is, let $t_i$ be the times where the LP interacts with the pool, where $i\in\{1,\dots,N\}$ and $t_{i+1} - t_i = 1$ minute.  For each time $t_i,$ the LP uses the optimal strategy in \eqref{eq:optimalspreadsul} based on information available at time $t_i,$ and she fixes the optimal spread of her position throughout the period $[t_i, t_{i+1});$ recall that the optimal spread is not a function of time.

To determine the optimal spread \eqref{eq:optimalspreadsul} of the LP's position at time $t_i,$ we use in-sample data $[t_i-\text{1 day}, t_i]$ to estimate the parameters. The volatility $\sigma$ of the rate $Z$ is given by the standard deviation of one-minute  log returns of the rate $Z,$ which is multiplied by $\sqrt{1440}$ to obtain a daily estimate. The pool fee rate $\pi_t$ is given by the total fee income generated by the pool during the in-sample period, divided by the pool size $2 \, \kappa \, \sqrt{Z_t}$ at time $t,$ where $\kappa$ is the active depth at time $t.$ {We run the linear regression described in Section \ref{sec:concentrationcost} to estimate} the concentration cost parameter, {which is set} to $\gamma = 5 \times 10^{-7}$.  
Finally, prediction of the future marginal rate $Z$ is out of the scope of this work, thus we set $\mu=0$ and $\rho = 0.5.$ 

To compute the LP's performance as a result of changes in the value of her holdings in the pool (position value), and as a result of fee income, we use out-of-sample data $[t_i, t_{i+1}].$ {We use  equation \eqref{eq:dynAlphaFirst} to determine the one-minute out-of-sample changes in the position value}. For fee income, we use  LT transactions in the pool at rates included in the range $\left(Z_t^L, Z_t^U \right]$ and equation \eqref{eq:remuneration_LP}. The income from fees accumulates in a separate account in units of $X$ with zero risk-free rate.\footnote{In practice, fees accumulate in both assets $X$ and $Y.$} At the end of the out-of-sample window, the LP withdraws her liquidity and collects the accumulated fees, and we repeat the in-sample estimation and out-of-sample liquidity provision described above. Thus, at times $t_i,$ where $i\in\{1,\dots,N-1\},$ the LP consecutively withdraws and deposits liquidity in different ranges. Between two consecutive operations (i.e., reposition liquidity provision), the LP may need to take liquidity in the pool to adjust her holdings in asset $X$ and $Y.$ In that case, we use results in \cite{cartea2022decentralised} to compute execution costs.\footnote{\cite{cartea2022decentralised} show that execution costs in the pool are a closed-form function of the rate $Z,$ the pool depth $\kappa,$ and the transaction size.} In particular, we consider execution costs when the LP trades asset $Y$ in the pool to adjust her holdings between two consecutive operations. More precisely, we consider that for every quantity $y$ of asset $Y$ bought or sold in the pool, a transaction cost $y \, Z_t^{3/2}/\kappa$ is incurred. We assume that the LP's taking activity does not impact the dynamics of the pool. Finally, we obtain 331,858 individual LP operations from 1 January to 18 August 2022.

\subsection{Benchmark}

We compare the performance of our strategy with the performance of LPs in the pool we consider. We select operation pairs that consist of first providing and then withdrawing the same depth of liquidity $\tilde \kappa$ at two different points in time by the same LP.\footnote{In blockchain data, every transaction is associated to a unique wallet address.} The operations that we  select represent approximately $66\%$ of all LP operations. Figure \ref{fig:LP1} shows the distribution of key variables that describe how LPs provide liquidity. The figure shows the distribution of the number of operations per LP, the changes in the position value, the length of time the position is held in the pool, and the position spread. 

\begin{figure}[!h]
\centering
\includegraphics{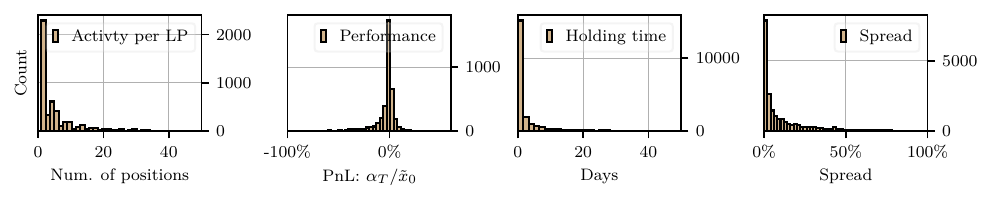}\\
\caption{From left to right: distribution of the number of operations per LP, changes in the holdings value as a percentage of initial wealth, position hold time, and position spread. ETH/USDC pool with selected operations from 5,156 LPs between 5 May 2021 and 18 August 2022.}\label{fig:LP1}
\end{figure}

Finally, Table \ref{table:dataLPdescr} shows the average and standard deviation of the distributions in Figure \ref{fig:LP1}. Notice that the bulk of liquidity is deposited in small ranges, and positions are held for short periods of time; $20\%$ of LP positions are held for less than five minutes and $30\%$ for less than one hour. Table \ref{table:dataLPdescr} also shows that, on average, the performance of the LP operations in the pool and the period we consider is $-1.49\,\%$ per operation.

{
\begin{table}[h]
\footnotesize{
\begin{center}
\begin{tabular}{c | r  r  r  r  r}
\hline 
 & Number of   & Position value & Fee income & Hold time & Spread \\ [0.5ex]
 & transactions per LP  & performance  ($\alpha_T / V_0 - 1$) & ($ p_T / V_0 - 1$) & & \\ [0.5ex]
\hline 
 Average  & $11.5$ & $-1.64\%$& $0.155\%$& $6.1$ days & $18.7\%$\\ [0.5ex]
Standard deviation & $40.2$ & $7.5\%\ \ $ & $0.274\%$ & $22.4$ days &  $43.2\%$\\
\hline 
\end{tabular}
\end{center}
\caption {LP operations statistics in the ETH/USDC pool using operation data of 5,156 different LPs between 5 May 2021 and 18 August 2022. Performance includes transaction fees and excludes gas fees. The position value performance and the fee income are not normalised by the hold time. }
\label{table:dataLPdescr}
}
\hfill
\end{table}}

\subsection{Performance results}
This subsection focuses on the performance of our strategy when gas fees are zero --- at the end of the section we discuss the profitability of the strategy when gas fees are included. Figure \ref{fig:DOLP_backtestgammas} shows the distribution of the optimal spread \eqref{eq:optimalspreadsul_symm} posted by the LP. The bulk of liquidity is deposited in ranges with a spread  $\delta$ below $1 \%$. Table \ref{table:dataLPdescr} compares the average performance of the components of  the optimal  strategy with  the performance of LP operations observed in the ETH/USDC pool.\footnote{In particular, performance is given for the selected operations shown in Figure \ref{fig:LP1}.} Table \ref{table:dataLPperf} suggests that the position of the LP loses value in the pool (on average) because of PL; however, the fee income would cover the loss, on average, if one assumes that gas fees are zero. Finally, the results show that our optimal strategy significantly improves the PnL of LP activity in the pool and the performance of the assets themselves.

\begin{figure}[!h]
\centering 
\includegraphics{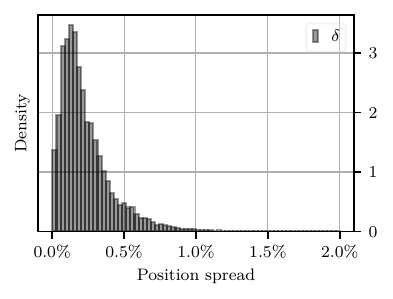}\\
\caption{Distribution of the position spread $\delta.$}\label{fig:DOLP_backtestgammas}
\end{figure}

{
\begin{table}[h]
\begin{center}
\begin{footnotesize}
\begin{tabular}{c || r  r  r} 
\hline 
 & Position value performance & Fee income & Total performance \\ 
  &  per operation   &  per operation & per operation  \\ 
    &     &  &  (with transaction costs, \\ 
      &     &  &  without gas fees) \\ 
  [0.5ex]
\hline
Optimal strategy & $-0.015 \%$  & $ 0.0197 \%$ & $0.0047 \% $ \\ [-0.5ex]
 &   $(0.0951 	 \%)$ & $( 0.005 \%)$ & $( 0.02 \%)$ \\ [1ex]
 
Market &  $ -0.0024 \%$  & $ 0.0017 \%$ & $ -0.00067 \% $ \\ [-0.5ex]
 &   $(0.02 \%)$ & $( 0.005 \%)$ & $( 0.02 \%)$ \\ [1ex]

Hold & n.a.   & n.a.  & $-0.00016 \% $ \\ [-0.5ex]
 &    &  & $( 0.08 \%)$ \\ [0.5ex]
\hline
\hline 
\end{tabular}
\end{footnotesize}
\end{center}
\caption {\textbf{Optimal strategy}: Mean and standard deviation of the one-minute  performance of the LP strategy \eqref{eq:optimalspreadsul_symm} and its components. \textbf{Market}: Mean and standard deviation of one-minute performance of LP activity in the ETH/USDC pool using data between 1 January and 18 August 2022.
\textbf{Hold}: Mean and standard deviation of the one-minute  performance of holding the assets. In all cases, the performance includes transaction costs (pool fee and execution cost), but does not include gas fees.}
\label{table:dataLPperf}
\hfill
\end{table}}

The results in Table \ref{table:dataLPperf} do not consider gas fees.   Gas cost is a flat fee, so it does not depend on the position spread or size of transaction. If the activity of the LP does not affect the pool and if the fees collected scale with the wealth that the LP deposits in the pool, then the LP should consider an initial amount of $X$ and $Y$ that would yield  enough fees to cover the flat gas fees. An estimate of the average gas cost gives an estimate of the minimum amount of initial wealth for a self-financing strategy to be profitable. Recall that, at any point in time $t,$ the LP withdraws her liquidity, adjusts her holdings, and then deposits new liquidity. In the data we consider, the average gas fee is $30.7\,$ USD to  provide liquidity, $24.5\,$ USD to  withdraw liquidity, and $29.6\,$ USD to take liquidity. Average gas costs are obtained from blockchain data which record the gas used for every transaction, and from historical gas prices. The LP pays a flat fee of  $84.8$ USD per operation when implementing the strategy in the pool we consider, so the LP strategy is profitable, on average, if the initial wealth deposited in the pool is greater than $1.8 \times 10^6$ USD. 

Fee income of the LP strategy is limited by the volume of liquidity taking activity in the pool, so one should not only consider increasing the initial wealth to make the strategy more profitable. There are $4.7$ LT transactions per minute and the average volume of LT transactions is $477,275$ USD per minute, so if the LP were to provide $100\%$ of the liquidity, the average fee income per operation would be 1,431 USD.

{
\section{Discussion: modelling assumptions \label{sec:modelling_assumptions}}

This section summarises our modelling assumptions and discusses their implications, strengths, and weaknesses.


\paragraph{Continuous trading.} Our model assumes continuous repositioning of the LP's position. However, when interacting with  blockchains, updates occur at the block validation frequency. For instance, the Ethereum network's blocks are validated every 13 seconds, on average. Moreover, within each block, the transactions form a (random) queue that determines the priority with  which they are executed. Our model can be extended to include delays inherent to blockchains and we refer the reader to the work in \cite{cartea2021shadow} and \cite{cartea2021latency}.

\paragraph{Rebalancing.} Continuous repositioning requires rebalancing the LP's assets which incurs costs as discussed in Section \ref{sec:costs}. To model this aspect, we assume that the LP pays costs that are proportional to the quantity of asset $Y$ held in the pool. In practice, the exact costs depend on variations in the holdings between two consecutive liquidity positions. Thus,  liquidity provision strategies should balance large variations in the holdings with fee revenue, PL, and concentration cost. However, the nonlinearity in the CL constant product formula complicates the mathematical modelling of this aspect of trading costs. 
Moreover, rebalancing costs depend on the cost structure of the trading venue where rebalancing trades are executed. 

\paragraph{Gas fees.} Our model assumes that gas fees paid by the LP to interact with the blockchain are flat and constant. In practice, gas fees are stochastic and depend on exogenous factors such as the price of electricity and network congestion. These could be included in our model by considering a stochastic gas fee that is observed by the LP.

\paragraph{Concentration costs.} The specific microstructure of CL markets features a new type of investment risk which we refer to as concentration risk. To capture the losses due to concentration risk in a continuous-time framework, we introduced an instantaneous cost which is inversely proportional to the spread of the LP's position, and we showed that it captures the losses due to concentration risk accurately. In practice, LPs must tailor the estimation of the concentration cost parameter $\gamma$ in \eqref{eq:feedyn_0} to the rebalancing frequency and to the volatility of the marginal rate $Z.$

\paragraph{Asymmetry.} Our model assumes a fixed relation between the asymmetry of the LP's position and the drift in the marginal rate. This relation fits observed data but also leads the LP to skew the position to capture more LT trades and to profit from expected rate movements. Future work will consider a richer characterization of the asymmetry because it may be desirable for LPs to adjust the asymmetry of their position as a function of other state variables or as a controlled variable

\paragraph{Fee dynamics.} 
Our model assumes that the distribution of fee revenue $\pi$   among LPs is stochastic and proportional to the size of the pool to reflect that large pools attract more trading flow because trading is cheaper. We also make the simplifying assumption that fees are uncorrelated to the price. In practice, the dynamics of fee revenue may be correlated to those of the volatility of the rate, which is also related to the concentration cost parameter $\gamma.$ Future work will consider more complex relations among these variables. 

\paragraph{Constant volatility.} At present, CFMMs with CL mainly serve as trading venues for crypto-assets which are better described with a diffusion model with stochastic volatility. It is straightforward to extend our strategy to this type of models.

\paragraph{Market impact.} Finally, our analysis does not take into account the impact of liquidity provision on liquidity taking activity, however, we expect liquidity provision in CPMMs with CL to be profitable in pools where the volatility of the marginal rate is low, where liquidity taking activity is high, and when the gas fee cost to interact with the liquidity pools is low. These conditions ensure that the fees paid to LPs in the pool, adjusted by gas fees and concentration cost, exceed PL so liquidity provision is viable.

\subsection{Discussion: related work}

Closest to this work are the strategic liquidity provision models proposed in \cite{fan2021strategic} and \cite{fan2022differential} which also consider CL markets. Both models allow LPs to compute liquidity positions over different intervals centered around the marginal rate within a given time horizon;   \cite{fan2021strategic} only consider static LP strategies which do not use of reallocations, and  \cite{fan2022differential} reposition liquidity whenever the price is outside of the position range. Both approaches only focus on maximising fee revenue and rely on approximations of the LP's objective and  use  Neural Networks to obtain context-aware approximate strategies. In contrast, our model leads to closed-form formulae that explicitly balance fee revenue with concentration risk, fee revenue, and rebalancing costs, while allowing LPs to use price signals (potentially based on exogenous information) to improve trading performance.

}

\section{Conclusions}

We studied the dynamics of the wealth of an LP in a CPMM with CL who implements a self-financing strategy that dynamically adjusts the range of liquidity. The wealth of the LP consists of the position value and fee revenue. We showed that the position value depreciates due to PL and the LP widens her liquidity range to minimise her exposure to PL. On the other hand, the fee revenue is higher for narrow ranges, but narrow ranges also increase concentration risk.

We derived the optimal strategy to provide liquidity in a CPMM with CL when the LP maximises expected utility of terminal wealth. This strategy is found in closed-form for log-utility of wealth, and it shows that liquidity provision is subject to a profitability condition. In particular, the potential gains from fees, net of gas fees and concentration costs, must exceed PL. Our model shows that the LP strategically adjusts the spread of her position around the reference exchange rate; the spread depends on various market features including tthe volatility of the rate, the liquidity taking activity in the pool, and the drift of the rate.
	
\chapter{Deep Liquidity Provision}\label{ch:deep}
\section{Introduction}

Constant function market makers (CFMMs) are trading venues where the takers and providers of liquidity interact in liquidity pools; liquidity providers (LPs) deposit their assets in the liquidity pool and liquidity takers (LTs) exchange assets directly with the pool. At present, constant product market makers (CPMMs) with concentrated liquidity (CL) are the most popular type of CFMM, with  Uniswap v3 as a prime example; see \cite{uniswap2021core}. In CPMMs with CL, LPs specify the rate ranges over which they deposit their assets, and this liquidity is counterparty to trades of LTs when the marginal exchange rate of the pool is within the liquidity range of the LPs. When LPs deposit liquidity, fees paid by LTs accrue and are paid to LPs when they withdraw their assets from the pool. The amount of fees accrued to LPs is proportional to the share of liquidity they hold in each liquidity range of the pool.

Current research characterises the losses of LPs and offers tools for strategic liquidity provision in a single liquidity pool with CL; see \cite{cartea2022decentralised2,cartea2023predictable}. In this chapter, we study strategic liquidity provision in multiple liquidity pools with CL. Here, each pool trades a risky asset and a common reference asset used by the LP to value her wealth.

We derive the discrete-time dynamics of the wealth of a strategic LP who dynamically provides liquidity in the various pools. At each time step, the LP specifies the proportion of her wealth to deposit in each pool and the range where she wants to provide liquidity. Her wealth consists of the position she holds in each pool and fee income. 

Here, we do not specify a model for the stochastic processes observed by the LP. Instead, we use a model-free approach to solve the liquidity provision problem. Specifically, we approximate the optimal liquidity provision strategy with a long short-term memory (LSTM) network, whose output at each time step determines how the LP spreads her wealth across the pools and the liquidity provision range in each pool. The LP use Sharpe Ratio and mean-variance criterion to train the LSTM strategy on simulated data and market data.

We test the LSTM strategy on simulated data to showcase the performance of the strategy under different market conditions. Firstly, we consider an LP who provides liquidity in a single liquidity pool with CL and we recover  results from the literature; see \cite{cartea2022decentralised2}. Specifically, when volatility is high, regardless of the performance criterion used, the LSTM strategy provides liquidity in wide ranges of the pool, as in see \cite{cartea2023predictable}. The optimal strategy increases the range posted by the LP as the fee tier decreases, and the LP skews her position to align provision of liquidity with the trend of the pools rate. Secondly, we consider an LP who provides liquidity in two liquidity pools with CL. The behaviour of the LSTM strategy in each individual pool is qualitatively the same as that in the single pool case, and the LP provides more liquidity in pools with lower volatility, higher fee tier, and where the rate between the risky asset and the reference asset trends upwards.

Finally, we use Uniswap v3 and Binance data to test the LSTM strategy. We use LP and LT data from the two trading venues  between 1 July 2021 and 30 September 2023 for the pair ETH/USDC and the pair BTC/USDC. ETH is the ticker of \emph{Ether}, BTC is the ticker of \textit{Bitcoin} and USDC is the ticker of \emph{USD coin}, which is a digital currency fully backed by U.S. Dollars (USD). The input to the LSTM strategy is a set of features which we build with data from both venues. Overall, we find that gas fees have a considerable impact on the performance of the strategy. Specifically, the average returns of the LSTM strategy before gas fees decrease with the LP's initial wealth and increase with the LP's trading frequency; in contrast, the average returns of the strategy after gas fees increase with the LP's initial wealth and decrease with the LP's trading frequency.

Early works on automated market makers (AMMs) include \cite{chiu2019blockchain, angeris2021analysis, lipton2021blockchain}. \cite{cartea2023predictable} study the losses faced by LPs in CFMMs and then introduce predictable loss, which measures the unhedgeable losses of LPs stemming from the depreciation of their holdings in the pool and from the opportunity costs from locking their assets in the pool. Predictable loss is similar to loss-versus-rebalancing in  \cite{milionis2022automated} which describes the unhedgeable losses of LPs in traditional CFMMs due to losses to arbitrageurs.\footnote{\cite{milionis2022automated} introduced loss-versus-rebalancing in August 2022 as a measure that quantifies the unhedgeable losses of LPs in CFMMs to arbitrageurs. Contemporaneously, \cite{cartea2023predictable} introduced predictable loss in November 2022 as a component in the wealth of strategic LPs that quantifies predictable loss due to the convexity of the trading function and due to opportunity costs in CFMMs and in CL markets.} Previous works on strategic liquidity provision include \cite{heimbach2022risks} which discusses the tradeoff between risks and returns that LPs face in Uniswap v3. \cite{cartea2022decentralised2} introduces a continuous-time model for optimal liquidity provision, \cite{milionis2023automated} study the impact of fees on the profits of arbitrageurs in CFMMs,  \cite{fukasawa2023model} study the hedging of the impermanent losses of LPs, and \cite{li2023yield} studies the economics of liquidity provision. Close to our work are the models in \cite{fan2021strategic} and \cite{fan2022differential} which focus on fee revenue and use approximation techniques to obtain dynamic strategies for liquidity provision. Finally, there is a growing literature on AMM design for fair competition between LPs and LTs. \cite{goyal2023finding} study an AMM with dynamic trading functions that incorporate beliefs of LPs, \cite{lommers2023:case} studies AMMs where the LP's strategy adjusts dynamically to market information, and  \cite{cartea2023automated} generalises CFMMs and propose AMM designs where LPs express their beliefs and risk preferences.

Our work is related to the algorithmic trading and optimal market making literature.  Early works on liquidity provision in traditional markets are \cite{ho1983dynamics}, \cite{biais1993price}, and  \cite{avellaneda2008high} with extensions in many directions; see \cite{cartea2014buy, cartea2017algorithmic, gueant2017optimal, bergault2021closed, drissi2022solvability}. We refer the reader to \cite{cartea2015book}, \cite{gueant2016book}, and \cite{donnelly2022optimal} for a comprehensive review of algorithmic trading models for takers and makers of liquidity in traditional markets. 

LSTM neural networks are a type of neural network typically employed for sequential data; see \cite{Hochreiter1997long}. In financial applications, LSTM networks are typically used for unsupervised learning tasks; see \cite{nelson2017stock,fischer2018deep,zhang2019at-lsts,moghar2020stock}. Here, we employ LSTM networks to approximate an optimal strategy which maximises a performance criterion. Close to our work are \cite{buehler2019deep,becker2020pricing,limmer2023robust} where the authors train neural networks on market data to hedge a basket of derivatives. Also, our work is related to the reinforcement learning and market making literature; see \cite{chan2001adaptive,balch2019how,gasperov2021market,jerome2023mbt}.

The remainder of the chapter proceeds as follows. Section \ref{DL:sec:1} describes CL pools and Section \ref{DL:sec:olp} studies the discrete-time dynamics of the wealth of LPs who provides liquidity in a CPMM with $N$ pools with CL. Finally, Section \ref{DL:sec:deep} introduces the framework for the model-free optimisation, Section \ref{DL:sec:fake} tests the optimal strategy on synthetic data, and Section \ref{DL:sec:market_data} tests the strategy on market data.
    
\section{Automated Market Making \label{DL:sec:1}}
In this section, we focus on CPMMs with CL and describe the revenue and the risks of liquidity provision in these venues. 
    
\subsection{Constant product market makers \label{DL:sec:CPMM}}
Traditional electronic exchanges are organised around limit order books  where liquidity taking and liquidity provision orders are matched. AMMs are a new type of trading venues where takers and providers of liquidity interact through liquidity pooling. LPs deposit assets in a common liquidity pool, and LTs trade directly with the pool. The liquidity pool consists of quantity $x>0$ of asset $X$ and quantity $y>0$ of asset $Y.$ Asset $Y$ is valued in terms of the reference asset $X$, and $Z$ denotes the (exchange) rate in the pool. CPMMs are a particular type of AMMs and they impose an LT trading condition and an LP trading condition, both of which define the state of the pool after an LT transaction and after an LP operation is completed. 
	
	\paragraph{\textbf{LT trading condition}} Let $(x,y)$ be the state of the pool before the arrival of an LT buy order for a quantity $\Delta y>0$ of asset $Y$. The quantity $\Delta x>0$ of asset $X$ that the LT pays to the pool, in exchange for $\Delta y>0$, is determined by the LT trading condition
	\begin{equation}\label{DL:eq:lt_condition_1}
		x\,y = \left(x+(1-\tau)\,\Delta x\right)\left(y-\Delta \Delta y\right) = \kappa^2\,,
	\end{equation}
	where $\kappa>0$ is the \emph{depth} of the pool and $\tau\geq0$ is the \textit{fee tier}. For a sell order, the LT trading condition is
	\begin{equation}\label{DL:eq:lt_condition}
		x\,y = \left(x-\Delta x\right)\left(y+(1-\tau)\,\Delta y\right) = \kappa^2\,.
	\end{equation}
	Chapter \ref{ch:paper} shows that the instantaneous exchange rate is given by
	\begin{align}\label{DL:eq:instantaneousprice}
		Z = \frac{x}{y}\,,
	\end{align}
	and that the execution rate for an LT trade of quantity $y$ is 
	\begin{equation}
		\tilde Z\left( \Delta y\right) = Z - Z^{3/2} \,  \Delta y \, / \, \kappa\,.
	\end{equation}
	
	\paragraph{\textbf{LP trading condition}} LPs deposit assets in the pool or withdraw assets from the pool.  The LP trading condition requires that LP operations do not impact the rate $Z$ of the pool. However, LP operations do change the depth $\kappa$ of the pool. More precisely, let $\left(x, y\right)$ be the state of the pool before an LP deposits liquidity $\left(x, y\right)$,\footnote{Here, $x\ge0$ and $y\ge0$\,.} so $\left(x+x, y+y\right)$ is the state of the pool after the deposit, which increases the depth of the pool because 
	\begin{equation}
		\left(x+\Delta x\right)\left(y+\Delta y\right) = \overline \kappa^2 > x\,y= \kappa^2\,.
	\end{equation}
	
    The LP trading condition requires that when an LP deposits the quantities $\left(\Delta x, \Delta y\right)$, then the following equality must hold
	\begin{align}\label{DL:eq:LPtradingCondition}
		Z = \frac{x}{y} = \frac{x+\Delta x}{y+\Delta y}\,.
	\end{align}
	
    In practice, LPs usually invest an initial wealth $V$ in terms of the reference asset $X\,,$ so the amounts $x$ and $y$ are obtained with the LP trading condition and the equality $V = \Delta x + Z \, \Delta y.$ In contrast, LTs usually specify a quantity $x$ to trade, and receive or pay a quantity $y$ which is determined by the LT trading condition.
	
    \subsection{Concentrated liquidity \label{DL:sec:CPMMwithCL}}
    In CPMMs without CL, liquidity in the pool is spread over all rates, and the LP trading condition ensures that the quantities $(x, y)$ deposited or withdrawn by LPs verify $(\Delta x, \Delta y) = (\rho\,x, \rho\,y)$ where $\rho \ge 0\,$; see \cite{cartea2022decentralised} and \cite{cartea2023execution}. In pools with CL, the space of rates where LPs provide liquidity is discretised into a set of values $\left\{\Zc(i)\right\}_{i\in\Z}$ called ticks.\footnote{In contrast, a tick in LOBs refers to the smallest possible price increment.} Each tick $\Zc(i)$ is uniquely identified by an integer $i\in\Z$ according to
    \begin{equation}\label{DL:eq:tick}
        \Zc(i) = 1.0001^i\,.
    \end{equation}
    The boundaries of an LP's position in the pool take values in the set of ticks. The range between two consecutive ticks defines the smallest range available for LPs. Every tick range $\left(\Zc(i), \Zc(i+1)\right]$ has its own depth which determines the execution cost of LT trades when the instantaneous rate $Z$ is in that tick range. Thus, the depth $\kappa$ can change as a result of the instantaneous rate $Z$ crossing the boundary of a tick range. 
	
    LPs specify a lower tick $\Zc(L)$ and an upper tick $\Zc(U)$ between which they provide liquidity that is evenly spread over all the tick ranges in the position range $\left(\Zc(L),\Zc(U)\right]$. The position of an LP who provides liquidity in a range $\left(\Zc(L),\Zc(U)\right]$ is characterised by the position depth $\tilde{\kappa}^{L,U}$. The value of $\tilde{\kappa}^{L,U}$ represents the depth of the LP's liquidity within every tick range in her position range, and $\tilde{\kappa}^{L,U} / \kappa^{i}$ represents the portion of liquidity that the LP holds in the tick range $\left(\Zc(i), \Zc(i+1)\right] \subset \left(\Zc(L),\Zc(U)\right]$. 
	
	The following key formulae define the LP trading condition in CPMMs with CL, i.e., they give the quantities $\left(x, y\right)$ that the LP must deposit in or withdraw from the range $\left(\Zc(L),\Zc(U)\right]$ so that the equality in \eqref{DL:eq:LPtradingCondition} holds
	\begin{equation}\label{DL:eq:position}
		\begin{cases}
			x = 0 \qquad\qquad\qquad\qquad\qquad\quad\,\,\,\,\, \text{and} \qquad y = \tilde{\kappa}^{L,U}\left(\frac{1}{\sqrt{\Zc(L)}}-\frac{1}{\sqrt{\Zc(U)}}\right) & \text{if}\quad Z<\Zc(L)\,,\\
			x = \tilde{\kappa}^{L,U}\left(\sqrt{Z}-\sqrt{\Zc(L)}\right) \,\qquad\,\,\, \text{and} \qquad y = \tilde{\kappa}^{L,U}\left(\frac{1}{\sqrt{Z}}-\frac{1}{\sqrt{\Zc(U)}}\right) & \text{if}\quad \Zc(L) \leq Z<\Zc(U) \,,\\
			x =  \tilde{\kappa}^{L,U}\left(\sqrt{\Zc(U)}-\sqrt{\Zc(L)}\right)\,\,\,\,\, \text{and} \qquad y =  0 & \text{if}\quad Z \geq \Zc(U) \,,\\
		\end{cases} 
	\end{equation}
	and the depth available in the tick range $\left(\Zc(i), \Zc(i+1)\right]$ is given by the sum of the depth of the positions of all LPs who provide liquidity in that tick range. We write
	\begin{align}
		\label{DL:eq:fromvirtualtorealreserves}
		\kappa^{i}=\sum_{j\in\{1,\cdots, M\}}\,\tilde{\kappa}^{L,U,j}\,\mathbbm{1}_{\{\Zc(L_{j})\leq \Zc(i)<\Zc(U_{j})\}}\,,
	\end{align}
	where $\tilde{\kappa}^{L,U,j}$ is the liquidity provided by the $j$th LP in the range $\left(\Zc(L_{j}), \Zc(U_{j})\right]$, and $\mathbbm{1}$ is the indicator function. 
	
    When the rate $Z$ is in a tick range, fee revenue is distributed to LPs who provide liquidity in that tick range, i.e., for a fee  $p $ paid by an LT, the $j$th LP with position depth  $\tilde{\kappa}^{L,U, j}$ earns\footnote{If the volume of the LT trade makes the rate $Z$ exit the tick range, then the transaction is subdivided into sub-transactions which get executed with different values of the depth.}  
    \begin{align}\label{DL:eq:remuneration_LP}
        \frac{\tilde{\kappa}^{L,U, j}}{ \kappa}\,p \,\mathbbm{1}_{\{\Zc(L_{j}) < Z \leq \Zc(U_{j})\}}\,.
    \end{align}
	
    \section{Optimal liquidity provision}\label{DL:sec:olp}
    
    In this section, we formalise the problem of an LP who provides liquidity in $n=1,\dots,N$ pools. Each pool provides liquidity in the reference asset $X$ and a risky asset $Y_{n}$. We denote by $Z^1,\dots,Z^N$ the marginal rates in these pools. 
	
    We fix a filtered probability space $\left(\Omega, \Fc, \P, \F=\left\{\Fc_{t}\right\}_{t=0,\dots,T} \right)$ satisfying the usual conditions. Here, $\F$ is the natural filtration generated by the $\R^D$-valued stochastic process $\left\{\mathfrak{I}_{t}\right\}_{t=0,\dots,T}$ that represents the information available to the LP at each time. Finally, we assume that each stochastic process defined below is adapted to $\F$.
	
    For each $n\in\{1,\dots,N\}$, the marginal rate of the pool to trade $X$ and $Y_{n}$ is the $\R_+$-valued stochastic process $\left\{Z_{t}^{n}\right\}_{t=0,\dots,T}$. The index of the active tick in the pool is the $\Z$-valued stochastic process $\left\{i_{t}^n\right\}_{t=0,\dots,T}$, where $n$ is the pool, $t$ is the time, and $i$ is the tick index of the pool as defined in \eqref{DL:eq:tick}. In particular, for each $t\in\{0,\dots,T\}$ the rate
	\begin{equation}
		Z_t^n\in \left(\Zc\left(i_{t}^n\right), \Zc\left(i_{t}^n+1\right)\right]\,.
	\end{equation}
	Similarly, the market depth is the $\R_+$-valued stochastic process $\left\{\kappa_{t}^n\right\}_{t=0,\dots,T}$, where $n$ is the pool, $t$ is the time and $\kappa$ is the market depth as defined in \eqref{DL:eq:lt_condition_1}.
	
    \subsection{Stochastic control problem}
    The LP provides liquidity in the $N$ pools between times 0 and $T$, and her wealth, in units of asset $X$, is the $\R_+$-valued stochastic process $\left\{V_{t}\right\}_{t=0,\dots,T}$. She provides liquidity in the CPMM by readjusting her position every $\Delta t$ time-steps, where $\Delta t\in\{1,2,3,\dots\}$ is fixed. We denote by $\T$ the set of times at which the LP readjusts her position, which is given by
    \begin{equation}
        \T \coloneqq \left\{t = k\,\Delta t \,\middle|\,k\in\{0,1,2,\dots\}\text{ and } k\,\Delta t < T \right\}\,.
    \end{equation}
    
    For each $t\in\T$, to optimally provide liquidity between $t$ and $t+\Delta t$, the LP chooses optimal controls $\bm{w}_t= \left(w^1_t,\dots,w^N_t\right) $ such that
	\begin{equation}\label{DL:eq:weights}
		\sum_{n=1}^{N}w^n_t =1\quad\text{ and }\quad w_t^n\geq 0\,,\,\,\forall\, n\in\left\{1,\dots,N\right\}.
	\end{equation}
Here, $w^n_t$ represents the proportion of wealth that she deposits into the pool that trades assets $X$ and $Y_{n}$. Simultaneously, the LP chooses optimal controls $\bm{\ell}_t = \left(\ell_{t}^1,\dots,\ell_{t}^N\right) \in\left\{0,\dots,\Tc\right\}^N$ and $\bm{u}_t = \left(u_{t}^1,\dots,u_{t}^N\right)\in\left\{0,\dots,\Tc\right\}^N$, where $\Tc$ is a fixed positive integer. For each $n$, the quantity $\ell_{t}^{n}$ is the \emph{lower spread}, and $\Zc(i^n_{t}-\ell^n_{t})$ is the lower tick of the LP's range. Similarly, the quantity $u_{t}^{n}$ is the \emph{upper spread}, and $\Zc(i^n_{t}+u^n_{t}+1)$ is the upper tick of the LP's range. Next, for each $n$ the LP deposits a fraction $w^n_t$ of her wealth into the range $\left(\Zc(i^n_{t}-\ell^n_{t}), \Zc(i_{t}^n+u_{t}^n+1)\right]$ of the pool that trades assets $X$ and $Y_{n}$. The quantity $\ell_{t}^n+u_{t}^n$ is the \textit{spread} of the strategy in the pool of assets $X$ and $Y_{n}$ at time $t$ and we have
\begin{equation}
    Z^n_t\in\left(\Zc(i^n_{t}-\ell^n_{t}), \Zc(i_{t}^n+u_{t}^n+1)\right]
\end{equation}
because $\ell^n_{t},\,u^n_{t}\geq0$. 

	A liquidity provision strategy is a sequence $\left\{\left(\bm{w}_t,  \bm\ell_{t}, \bm u_{t}\right)\right\}_{t\in\T} $ which specifies how the LP provides liquidity in the $N$ pools of the CPMM between time 0 and $T$. The set of admissible liquidity provision strategies is given by
	\begin{equation}
		\Ac = \left\{\left\{\left(\bm{w}_t,  \bm\ell_{t}, \bm u_{t}\right)\right\}_{t\in\T}\, \middle\vert 
		\begin{array}{c}
			\R^N_+\times\left\{0,\dots,\Tc\right\}^N \times\left\{0,\dots,\Tc\right\}^N \text{-valued},\\
			\F\text{-adapted, and $\bm{w_t}$ satisfies \eqref{DL:eq:weights}} 
		\end{array}\right\}\,.
	\end{equation}

    \subsection{Liquidity provision}
    At time $t=0$, the LP chooses a control $\bm{w}_0$ which determines the proportion of her wealth that she allocates to each pool. Simultaneously,  she chooses a pair of controls $\bm\ell_0, \bm u_0$ which specify the range where she provides liquidity between time 0 and $\Delta t$. For pool $n$, the quantities $w_{0}^{n}$,  $\Zc(i_0^n-\ell_0^n)$, and $  \Zc(i^n_0+u^n_0+1)$, determine the LP's market depth $\tilde{\kappa}^{n}_0$, which is given by
	\begin{equation}\label{DL:eq:kappa_tilde}
		\begin{split}
			\tilde{\kappa}^{n}_0 = &w_{0}^{n}\, V_{0}\times\left(2\,\sqrt{Z_0^n} - \sqrt{\Zc(i_0^n-\ell_0^n)} -\frac{Z_0^n}{\sqrt{\Zc(i_0^n+u_0^n+1)}}\right)^{-1}\,,
		\end{split}
	\end{equation}
	see \cite{uniswap2021core} for more details. The quantities $\left(\tilde x_0^n, \tilde y_0^n\right)$ of assets $X$ and $Y$ that the LP deposits in the liquidity pool are determined by $\Zc(i_0^n-\ell_0^n)$, $  \Zc(i_0^n+u_0^n+1)$, $Z_0^n$, and $\tilde{\kappa}^{n}_0$\,. Specifically, we have
	\begin{equation}\label{DL:eq:position0}
		\left\{
		\begin{aligned}
			\displaystyle &\tilde x_0^n = \tilde{\kappa}^{n}_0\left(\sqrt{Z_0^n}-\sqrt{\Zc(i_0^n-\ell_0^n)}\right)\,, \\
			\displaystyle &\tilde y_0^n = \tilde{\kappa}^{n}_0\left(\frac{1}{\sqrt{Z_0^n}}-\frac{1}{\sqrt{\Zc(i_0^n+u_0^n+1)}}\right)\,,
		\end{aligned}
		\right.
	\end{equation}
    where $\tilde x_0^n+\tilde y_0^n\, Z_0^n=w_0\, V_0\,.$ 
    
    For simplicity, we assume that before providing liquidity in the pool, the LP trades in an alternative trading venue to obtain the exact amounts of asset $X$ and assets $Y_{1},\dots,Y_{N}$ that she will deposit in the $N$ pools of the CPMM, we also assume that the LP does not pay transaction fees and trades at the instantaneous rate $Z^1_0,\dots,Z^N_0\,.$
	
    Between time 0 and $\Delta t$, the instantaneous rate in the $n$th pool moves from $Z_0^n$ to $Z_{\Delta t}^n$ and the LP's quantities of assets $X$ and $Y_{n}$ in the CPMM change. We denote by $\left(\mathring{x}_{\Delta t}^n, \mathring{y}_{\Delta t}^n\right)$ the quantities of assets $X$ and $Y_{n}$ held by the  LP in the pool at the beginning of time $\Delta t$. In this way, we distinguish these quantities from $\left(\tilde x_{\Delta t}^n,\tilde y_{\Delta t}^n\right),$ which denote the LP's holdings in the same pool after choosing new controls $w_{\Delta t}^n$, $\ell_{\Delta t}^n$, and $u_{\Delta t}^n$. In particular, the quantities $\mathring{x}_{\Delta t}^n$ and $\mathring{y}_{\Delta t}^n$ are given by
    \begin{equation}
        \begin{split}
            &\mathring{x}_{\Delta t}^n = \left\{
            \begin{aligned}
                &0 &\qquad\text{if}\quad Z_{\Delta t}^n<\Zc(i_0^n-\ell_0^n)\,,\\
                &\tilde{\kappa}^{n}_0\left(\sqrt{Z_{\Delta t}^n}-\sqrt{\Zc(i_0^n-\ell_0^n)}\right)&\qquad\text{if}\quad \Zc(i_0^n-\ell_0^n) \leq Z_{\Delta t}^n<\Zc(i_0^n+u_0^n+1) \,,\\
                &\tilde{\kappa}^{n}_0\left(\sqrt{\Zc(i_0^n+u_0^n+1)}-\sqrt{\Zc(i_0^n-\ell_0^n)}\right)& \qquad\text{if}\quad Z_{\Delta t}^n \geq \Zc(i_0^n+u_0^n+1) \,,\\
            \end{aligned}
            \right.\\
            &\mathring{y}_{\Delta t}^n = \left\{
            \begin{aligned}
                &\tilde{\kappa}^{n}_0\left(\frac{1}{\sqrt{\Zc(i_0^n-\ell_0^n)}}-\frac{1}{\sqrt{\Zc(i_0^n+u_0^n+1)}}\right) &\qquad\text{if}\quad Z_{\Delta t}^n<\Zc(i_0^n-\ell_0^n)\,,\\
                &\tilde{\kappa}^{n}_0\left(\frac{1}{\sqrt{Z_{\Delta t}^n}}-\frac{1}{\sqrt{\Zc(i_0^n+u_0^n+1)}}\right) &\qquad\text{if}\quad \Zc(i_0^n-\ell_0^n) \leq Z_{\Delta t}^n<\Zc(i_0^n+u_0^n+1) \,,\\
                &0& \qquad\text{if}\quad Z_{\Delta t}^n \geq \Zc(i_0^n+u_0^n+1) \,,\\
            \end{aligned}
            \right.\\
        \end{split}
    \end{equation}
    see \cite{uniswap2021core} for more details.

    Between time $0$ and $\Delta t$, the LP's liquidity deposited in the pool that trades assets $X$ and $Y_{n}$ generates $\Phi_{0,\Delta t}^{n}$, in units of $X$, and more generally, $\Phi_{s,t}^{n}$ denotes the amount of fees that the LP's position generates between time $s$ and $t$, also in units of $X$. Thus, the mark-to-market value of the LP's wealth at time $\Delta t$ is
    \begin{equation}
        V_{\Delta t} = \sum_{n=1}^N \,\mathring{x}_{\Delta t}^n + \mathring{y}_{\Delta t}^n\,Z_{\Delta t}^n+\Phi_{0,\Delta t}^{n}\,.
    \end{equation}
    Next, the LP withdraws her assets from the $N$ pools and reinvests all her wealth $V_{\Delta t}$ to provide liquidity between time $\Delta t$ and time $2\Delta t$. 
	
    In general, for each $t\in\T$ the mark-to-market value of the LP's wealth is equal to
    \begin{equation}
        V_t = \sum_{n=1}^N \,\mathring{x}_t^n + \mathring{y}_t^n\,Z_t^n+\Phi_{t-\Delta t,t}^{n}
    \end{equation}
    units of $X$. Here, $\mathring{x}^{n}_t + \mathring{y}^{n}_t\,Z_t^n$ is the mark-to-market value at the beginning of time $t$ of the LP's assets deposited in the pool that trades assets $X$ and $Y_{n}$, and $\Phi_{t-\Delta t,t}^n$ are the fees that the LP's position generates between $t-\Delta t$ and $t$. Then, at time $t$ she chooses optimal controls $\bm{w}_{t},\bm\ell_{t},\bm u_{t}$ to provide liquidity between $t$ and $t+\Delta t$. For each $n$, the quantities $w_{t}^{n}$,  $\Zc(i_t^n-\ell_t^n)$, and $  \Zc(i^n_t+u^n_t+1)$ determine the LP's liquidity depth $\tilde{\kappa}^{n}_t$, which is given by
	\begin{equation}\label{DL:eq:tilde_t}
		\begin{split}
			\tilde{\kappa}^{n}_t = & w_t^n \, V_{t}\times\left(2\,\sqrt{Z_t^n} - \sqrt{\Zc(i_{t}^n-\ell_{t}^n)} -\frac{Z_t^n}{\sqrt{\Zc(i_{t}^n+u_{t}^n+1)}}\right)^{-1}\,,
		\end{split}
	\end{equation}
	while $\Zc(i_{t}^n-\ell_{t}^n)$, $  \Zc(i_{t}^n+u_{t}^n+1)$, $Z_t^n$, and $\tilde{\kappa}^{n}_t$ determine the quantities  $\left(\tilde x_t^{n},\tilde y_t^{n}\right)$ of assets $X$ and $Y_{n}$ that the LP needs to deposit in the pool. Specifically, we have
	\begin{equation}\label{DL:eq:positionn}
		\left\{
		\begin{aligned}
			\displaystyle &x_t^{n} = \tilde{\kappa}^{n}_t\left(\sqrt{Z_t^n}-\sqrt{\Zc(i_{t}^n-\ell_{t}^n)}\right)\,, \\
			\displaystyle &y_t^{n} = \tilde{\kappa}^{n}_t\left(\frac{1}{\sqrt{Z_t^n}}-\frac{1}{\sqrt{\Zc(i_{t}^n+u_{t}^n+1)}}\right)\,.
		\end{aligned}
		\right.
	\end{equation}

    \subsection{Performance criterion}
    The LP uses a performance criterion $\bm{J}\colon\Ac\to\R$ to evaluate a strategy $\left(\bm w, \bm \ell, \bm u\right)\in\Ac$. We study two optimal criteria: \textbf{Sharpe Ratio}
    \begin{equation}
        J\left(\bm w, \bm \ell, \bm u\right)=\frac{\E\left[V_{T}\right]-V_0}{\sqrt{\V\left[V_{T}\right]}}\,,
    \end{equation}
    and \textbf{mean-variance} 
    \begin{equation}\label{DL:eq:mv}
        J\left(\bm w, \bm \ell, \bm u\right)=\E\left[V_{T}\right]-\gamma\,\V\left[V_{T}\right]\,,
    \end{equation}
    where $\gamma$ is a risk aversion parameter, $\E[\,\cdot\,]$ and $\V[\,\cdot\,]$ are the expectation and variance operators, respectively. The goal of the LP is to find a liquidity provision strategy $\left(\bm w^\star, \bm \ell^\star, \bm u^\star\right)$ such that
	\begin{equation}\label{DL:eq:optim}
		\left(\bm w^\star, \bm \ell^\star, \bm u^\star\right) = \argmax_{\left(\bm w, \bm \ell, \bm u\right)\in\Ac} J\left(\bm w, \bm \ell, \bm u\right)\,.
	\end{equation}

    \section{Model-free optimisation}\label{DL:sec:deep}
    In this section, we use LSTM networks to solve the optimisation problem in \eqref{DL:eq:optim}. LSTM is a type of neural network employed for sequential data; see \cite{Hochreiter1997long}. Specifically, LSTM networks are a particular type of recurrent neural network (RNN). 
    
    RNN networks take an $\R^p$-valued input sequence $\bm x =\left\{x_t\right\}_{t=0,\dots,T}$ and use information from the past of the sequence to produce an $\R^q$-valued output sequence $\left\{h^\theta_t\right\}_{t=0,\dots,T}$ through a self loop. Typically, an RNN is implemented as a feedforward neural network whose input at time $t$ is the vector $x_{t}$ and the output $h_{t-1}^\theta$ of the RNN at the previous time step. This self loop allows information from the input sequence to persist over time and makes RNN networks particularly suited for sequential data.

    However, RNN networks suffer from the vanishing gradient problem which prevents them from learning long term dependencies; see \cite{bengio1994learning}. LSTM networks solve the problem of long term dependencies by introducing a custom layer, called LSTM layer, which controls the flow of information over time. Typically, LSTM networks are implemented as one or more stacked LSTM layers which take an $\R^p$-valued input sequence $\left\{x_t\right\}_{t=0,\dots,T}$ and sequentially produce the $\R^q$-valued output sequence $\left\{h^\theta_t\right\}_{t=0,\dots,T}$ and the $\R^r$-valued hidden \textit{cell} sequence $\left\{c^\theta_t\right\}_{t=0,\dots,T}$. At time $t$, the inputs to the LSTM network are $x_t$, the output at the previous time step $h^\theta_{t-1}$, and the hidden cell $c^\theta_{t-1}$; then, the outputs are $h^\theta_{t}$ and $c^\theta_{t}$. These two self loops and the LSTM layer structure allow long and short term memory to persist over time. Specifically, the output $h^\theta_{t}$ carries short-term information while the hidden cell $c^\theta_{t}$ stores long term information. The amount of information retained by the hidden cell at each time is determined by the \textit{forget gate} inside each LSTM layer; see \cite{Goodfellow2016deep}.

    Learning long-term dependencies is particularly relevant in financial applications because information in financial time series tends to persist over time. For instance, in \cite{cartea2022decentralised} and \cite{cartea2023execution} the authors show that the past values of the instantaneous rate in Binance can be used to forecast the instantaneous rate in Uniswap v3.

    \subsection{Implementation}
    Next, we provide details about the neural network architecture we use to approximate the solution to the optimisation problem in \eqref{DL:eq:optim}. Specifically, we consider a class of LSTM networks with input dimension $\R^D$ and output dimension $\R^{N}\times \R^{(\Tc+1) N}\times \R^{(\Tc+1) N}$.
	
    At time $t$, the input to the LSTM network is the $D$-dimensional vector $\mathfrak{I}_t$ representing the information available to the LP up to time $t$. The network's output is the vector $ \bm{h}_t  = (\bm{h}_t^{\theta,w}, \bm{h}_t^{\theta,\ell}, \bm{h}_t^{\theta,u}) $, where $\theta$ are the network parameters, $\bm{h}_t^{\theta,w}\in\R^N,$ $ \bm{h}_t^{\theta,\ell}\in\R^{(\Tc+1) N}, $ and $ \bm{h}_t^{\theta,u}\in\R^{(\Tc+1) N}$. Next, to obtain the controls $\bm w_t^\theta, \bm \ell_t^\theta, \bm u_t^\theta$, we further apply a custom layer to the output vector $\bm{h}_t^\theta$. 
	\paragraph{\textbf{Allocation weights}} The allocation weights $\bm w_t^\theta$ are given by
	\begin{equation}\label{DL:eq:soft_weights}
		\bm w_t^\theta = \softmax \left(\bm{h}_t^{\theta,w}\right),
	\end{equation}
	and
	\begin{equation}\label{DL:eq:soft_weights_n}
		\left[\bm w_t^\theta\right]_n = \left[\softmax \left(\bm{h}_t^{\theta,w}\right)\right]_n =  \frac{\exp\left(\left[\bm{h}_t^{\theta,w}\right]_n\right)}{\sum_{k=1}^N \exp\left(\left[\bm{h}_t^{\theta,w}\right]_k\right)}\,,
	\end{equation}
    for each $n\in\{0,\dots,N\}$, where $[\,\cdot\,]_n$ is the $n$th entry of the input vector. The $\softmax$ activation function enforces the condition in \eqref{DL:eq:weights}.
	
	\paragraph{\textbf{Spreads}} To obtain the spreads $\bm \ell_t^\theta $ and $ \bm u_t^\theta$, we split the $(\Tc+1) N$-dimensional output vectors $\bm{h}_t^{\theta,\ell}$ and $\bm{h}_t^{\theta,u}$ into $N$ vectors in $\R^{\Tc+1}$ 
	\begin{equation}\label{DL:eq:split}
		\bm{h}_t^{\theta,\ell} = \left(\bm{h}_t^{1,\theta,\ell},\dots,\bm{h}_t^{N,\theta,\ell}\right)\quad\text{ and }\quad \bm{h}_t^{\theta,u} = \left(\bm{h}_t^{1,\theta,u},\dots,\bm{h}_t^{N,\theta,u}\right)\,,
	\end{equation}
	and then convert the $\Tc+1$-dimensional vectors $\bm{h}_t^{n,\theta,\ell}$ and $\bm{h}_t^{n,\theta,u}$ into the controls $\ell^{n,\theta}_t$ and $u^{n,\theta}_t$, which are given by
	\begin{equation}\label{DL:eq:tick_weights}
		\ell_t^{n,\theta} = \argmax_{i\in\{0,\dots,\Tc\}}\left\{\left[\bm{h}_t^{n,\theta,\ell}\right]_i\right\}\,\quad\text{ and } \quad u_t^{n,\theta}\argmax_{i\in\{0,\dots,\Tc\}}\left\{\left[\bm{h}_t^{n,\theta,u}\right]_i\right\}\,.
	\end{equation}
	
    Finally, for each set of parameters $\theta$, we obtain an LSTM liquidity provision strategy $\left(\bm w^\theta, \bm \ell^\theta,\bm u^\theta\right)$. 
    
    Next, the LP trains the LSTM network using the performance criterion $J$ to approximate the optimal liquidity provision strategy in \eqref{DL:eq:optim} with $\left(\bm w^{\theta^\star}, \bm \ell^{\theta^\star},\bm u^{\theta^\star}\right)$, where
	\begin{equation}\label{DL:eq:theta_star}
		\theta^\star = \argmax_{\theta\in\Theta} J\left(\bm w^\theta, \bm \ell^\theta,\bm u^\theta\right)\,.
	\end{equation}
    \section{LSTM strategy with simulated data}\label{DL:sec:fake}
    In this section, we test the LSTM strategy on simulated data to showcase the performance of the strategy under various market conditions and in Section \ref{DL:sec:market_data} we test the LSTM strategy on market data. Here, for each experiment, we train an LSTM network of five LSTM layers with $24$ neurons each. The dropout rate is 0.1 and we use the Adam optimiser for training; see \cite{kingma2015adam}. For each experiment we draw 256,000 samples and use $60\%$ of them for training, $20\%$ for validation, and $20\%$ for testing. We use the mean-variance criterion in \eqref{DL:eq:mv} to train our model. 

    \subsection{Model}
    We introduce a continuous-time model for the $N$ liquidity pools where the LP provides liquidity. We use this model to generate samples to train the LSTM network, and obtain the LSTM liquidity provision strategy.

    We fix a filtered probability space $\left(\Omega, \Fc, \P, \tilde\F=\left\{\Fc_{t}\right\}_{t\in[0,T]} \right)$ satisfying the usual conditions. Here, $\tilde\F$ is the natural filtration generated by the $\R^D$-valued stochastic process $\left\{\mathfrak{I}_{t}\right\}_{t\in[0,T]}$ which represents the information available to the LP at each time. Note that $\F\subset\tilde\F$, where $\F=\left\{\Fc_{t}\right\}_{t=0,\dots,T}$ is the natural filtration introduced in Section \ref{DL:sec:olp}.

    When a liquidity taking order arrives in the pool, the instantaneous rate is updated according to the LT condition in Section \ref{DL:sec:1}, and the volume of the order is distributed among LPs according to their contribution to market depth and on the fee tier $\tau_n$. Here, we distinguish between the market depth of the LP and the aggregated market depth of the other LPs. The market depth of the LP at time $t$ in pool $n$ is denoted by $\tilde\kappa^n_t$ and it is given by \eqref{DL:eq:tilde_t}, and the aggregated market depth of the other LPs in the pool is the $\R_{++}$-valued stochastic process $\kappa_t^n$. We assume that during the trading horizon there is no liquidity provision activity from other LPs and we further assume that their market depth $\kappa^n_t$ is constant across different ticks. Therefore, $\kappa_t^n$ is constant during the entire trading horizon. We write $\kappa_n$ to denote the constant aggregated market depth of the other LPs, so that the total market depth of the pool at time $t$ is $\kappa_n+\tilde\kappa^n_t$. 

    Liquidity taking orders arrive in the pool according to a Poisson process $N_t^n$ with parameter $\lambda_n$. A liquidity taking order is a buy order with probability $p_n\in(0,1)$ and its size, in units of asset $X$, is normally distributed with mean $\mu_n>0$ and standard deviation $\xi_n>0.$ Let $\{B_t^n\}_{t\in[0,T]}$ be i.i.d. Bernoulli random variables with parameter $p_n$, and let $\{S_t^n\}_{t\in[0,T]}$ be i.i.d. normal random variables with mean $\mu_n$ and standard deviation $\xi_n$, independent of $\{B_t^n\}_{t\in[0,T]}$. Then, the liquidity taking order flow in units of asset $X$ of buy and sell orders during the trading horizon is given by
    \begin{equation}
        \begin{split}
            dO_t^{n,\text{buy}}    &= B_t^n \, S_t^n\,dN_t^n\,,\\
            dO_t^{n,\text{sell}}   &= (1-B_t^n) \, S_t^n\,dN_t^n\,.
        \end{split}
    \end{equation}
    Here, we assume that the arrival of liquidity taking orders is independent between pools, so the Poisson processes $N^1_t,\dots,N^N_t$ are independent.
    
    In our experiments, we simulate the Poisson process $N_t^n$, the random variables $\{B_t^n\}_{t\in[0,T]}$, and $\{S_t^n\}_{t\in[0,T]}$ to obtain the order flow processes $O_t^{n,\text{buy}}$ and $O_t^{n,\text{sell}}$. We update the instantaneous rate $Z_t^n$ and distribute trading fees among LPs according to the LT condition in Section \ref{DL:sec:1}.
    
    \subsection{Single pool}
    Here, we consider an LP who provides liquidity in one liquidity pool, i.e., $N=1$, that trades the reference asset $X$ and the risky asset $Y$, and the instantaneous rate at time $t=0$ is $Z_0=\,$2,200. The LP provides liquidity for $24$ hours and readjusts her position every $30$ minutes, so $T=\,$1,440 minutes and $\Delta t=30$ minutes. The LP's initial wealth is 500,000 units of asset $X$ and the maximum value for the lower and upper spreads that the LP considers is $\Tc=\,$500.  At each time $t\in\T$ the input to the LSTM strategy is $\mathfrak{I}_{t} = \left(w^\theta_{t-\Delta t}, \ell^\theta_{t-\Delta t},u^\theta_{t-\Delta t}, Z_t, \Phi_{t-\Delta t,t}\right)$\,. 
    
    Liquidity taking orders arrive in the pool at a Poisson rate of $\lambda\in\R_{++}$ and the distribution of the order size is normal with mean $\mu_{\text{size}}=\,$132,030 and standard deviation $\sigma_{\text{size}}=\,$20,000 units of asset $X$. These choices of parameters are based on market data from the ETH/USDC 0.3\% pool in Uniswap v3 between 1 July 2021 and 30 September 2023; see Table \ref{DL:table:datadescr}.
    
    \subsubsection{Trading frequency}\label{DL:sec:one_frequency}
    
    Here, we train the LSTM network for various values of the arrival rate $\lambda$ of liquidity taking orders, while keeping the remaining parameters fixed. The market depth is $\kappa=\,$15,000,000, the fee tier $\tau=\,$0.3\%, and $p=0.5$ is the probability that the liquidity taking order is a buy. When the value of $\lambda$ is large, orders arrive more frequently to the pool and the volatility of the instantaneous rate increases; Table \ref{DL:table:volatility} shows the volatility of 30-minute returns as a function of the arrival rate $\lambda$.

    \begin{table}[H]\fontsize{9.0}{12.0}\selectfont
        \centering
        \begin{tabular}{c||c c c c c}
            \textbf{Arrival rate $\lambda$} & 1 $\text{min}^{-1}$ & $\displaystyle\frac{1}{3}$ $\text{min}^{-1}$ & $\displaystyle\frac{1}{5}$ $\text{min}^{-1}$ & $\displaystyle\frac{1}{10}$ $\text{min}^{-1}$ & $\displaystyle\frac{1}{20}$ $\text{min}^{-1}$ \\ [1.5ex]\hline
            \textbf{Volatility $\sigma_{\text{30min}}$} & 7.70\% & 4.64\% & 3.49\%& 2.51\%& 1.77\%
        \end{tabular}
        \caption{Volatility of 30-minute returns $\sigma_{\text{30min}}$ as a function of the arrival rate $\lambda$.  Market depth $\kappa=\,$15,000,000, fee tier $\tau=\,$0.3\%, and $p=\,$0.5 is the probability that a liquidity taking order is a buy order.}
        \label{DL:table:volatility}
    \end{table}
    
    Table \ref{DL:table:one_frequency} reports the performance of the LSTM strategy on the test set with simulated data. The table shows the average lower and upper spreads posted by the strategy, and it shows the Sharpe Ratio and the average returns of the strategy, both before and after gas fees; whenever the LP adjusts her position she pays $73.3$ units of asset $X$ for gas fees, which is the average USD amount of gas fees paid by LPs for withdrawing or depositing liquidity from the ETH/USDC 0.3\% pool in Uniswap v3 between 1 July 2021 and 30 September 2023.

    \begin{table}[H]\fontsize{9.0}{12.0}\selectfont
        \centering
        \begin{tabular}{|c|c|c|c|c|}
            \hline
            \textbf{Arrival rate $\lambda$} & \textbf{\begin{tabular}[c]{@{}c@{}}Sharpe Ratio\\ (after gas fees)\end{tabular}} & \textbf{\begin{tabular}[c]{@{}c@{}}Returns\\ (after gas fees)\end{tabular}} & \textbf{\begin{tabular}[c]{@{}c@{}}Average\\ lower spread\end{tabular}} & \textbf{\begin{tabular}[c]{@{}c@{}}Average\\ upper spread\end{tabular}} \\ \hline
            1  &  \begin{tabular}[c]{@{}c@{}}0.22\\ ($-$13.80)\end{tabular}                & \begin{tabular}[c]{@{}c@{}}0.01\%\\ ($-$0.71\%)\end{tabular}                    & 496.00    & 8.00\\ \hline
            1/3  &  \begin{tabular}[c]{@{}c@{}}3.74\\ ($-$1.26)\end{tabular}                 & \begin{tabular}[c]{@{}c@{}}0.54\%\\ ($-$0.18\%)\end{tabular}                    & 88.95     & 7.75 \\ \hline
            1/5  &  \begin{tabular}[c]{@{}c@{}}4.63\\ (2.41)\end{tabular}                  & \begin{tabular}[c]{@{}c@{}}1.51\%\\ (0.78\%)\end{tabular}                     & 24.76     & 9.20 \\ \hline
            1/10 &  \begin{tabular}[c]{@{}c@{}}6.14\\ (4.15)\end{tabular}                  & \begin{tabular}[c]{@{}c@{}}2.23\%\\ (1.51\%)\end{tabular}                     & 15.00     & 4.00 \\ \hline
            1/20 &  \begin{tabular}[c]{@{}c@{}}\textbf{6.37}\\ (\textbf{4.80})\end{tabular}& \begin{tabular}[c]{@{}c@{}}\textbf{2.93\%}\\ (\textbf{2.21\%})\end{tabular}   & 13.43     & 4.00 \\ \hline
        \end{tabular}
        \caption{Performance of the LSTM strategy on test set of simulated data. The LP trades in a CPMM with one liquidity pool with CL. We assume market depth $\kappa=\,$15,000,000, fee tier $\tau=\,$0.3\%, and $p=\,$0.5 is the probability that a liquidity taking order is a buy order.}
        \label{DL:table:one_frequency}
    \end{table}

    The position of the LP generates fees only when the instantaneous rate is within the range where she provides liquidity. When volatility is high, the spread of the LSTM strategy is high to ensure that the instantaneous rate remains within the range where the LP is providing liquidity. On the other hand, when volatility is low, the strategy posts a narrow spread, which increases the value of the LP's market depth $\tilde\kappa$, to increase the amount of fees that she earns; see \eqref{DL:eq:kappa_tilde}. This result aligns with the finding in \cite{cartea2023predictable}. They show that LPs widen (resp. tighten) the range where they provide liquidity when volatility is higher (resp. lower).

    If the value of the lower spread $\ell$ is higher than the value of the upper spread $u$, then the majority of the LP's wealth is allocated to the reference asset $X$; see \eqref{DL:eq:position}. On the other hand, when $u>\ell$ the LP allocates most of her wealth to the risky asset $Y$. If the LP allocates most of her liquidity to asset $Y$, then she is more exposed to the fluctuations of the instantaneous rate; the value of the position in $Y$ depends on the instantaneous rate. Table \ref{DL:table:one_frequency} shows that when the volatility of the instantaneous rate is high (i.e., when the value of the arrival rate $\lambda$ of liquidity taking orders is high), the mean-variance criterion pushes the LP to provide liquidity in a range which is skewed towards the reference asset $X$ to reduce the standard deviation of the LP's returns. 

    Figure \ref{DL:fig:ticks} shows the LP's liquidity provision spread as a function of the volatility. Specifically, the red line is the upper spread $u$ and the blue line is the negative lower spread $-\ell$. The distance between the red and the blue line is the average spread of the LSTM strategy. As volatility increases, the LP widens the range where she is providing liquidity and skews her position towards the reference asset $X$.

    \begin{figure}[H]\centering
	\includegraphics{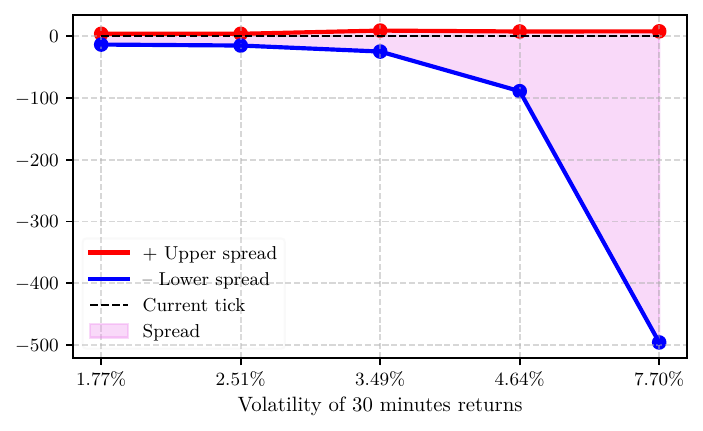}\\
	\caption{Average upper and lower spreads posted by the LSTM strategy as functions of the volatility $\sigma_{\text{30min}}$.}
	\label{DL:fig:ticks}
    \end{figure}

    Average returns and Sharpe Ratio, both before and after gas fees, increases when the arrival rate $\lambda$ decreases. When $\lambda=1/20$, the average returns of the LSTM strategy are $2.93\%$ before gas fees and $2.21\%$ after gas fees. The Sharpe Ratio of the strategy is $6.37$ before gas fees and $4.80$ after gas fees. When $\lambda=1$, the average returns of the LSTM strategy are $0.01\%$ before gas fees and $-0.71\%$ after gas fees, and Sharpe Ratio of the strategy is $0.22$ before gas fees and $-13.80$ after gas fees.

    \subsubsection{Fee tier}\label{DL:sec:one_fee_tier}
    Next, we train the LSTM network for various values of the fee tier $\tau$, while keeping the remaining parameters fixed, arrival rate $\lambda=\,$1/3 $\text{min}^{-1}$, market depth $\kappa=\,$15,000,000, and probability of liquidity taking order to be a buy order $p=\,$0.5. 
    
    Table \ref{DL:table:one_fee_rate} reports the performance of the LSTM strategy on the test set. It shows the average lower and upper spreads posted by the strategy and shows the Sharpe Ratio and the average returns of the strategy, both before and after gas fees. As above, whenever the LP adjusts her position she pays $73.3$ units of asset $X$ for gas fees.

    \begin{table}[H]\fontsize{9.0}{12.0}\selectfont
        \centering
        \begin{tabular}{|c|c|c|c|c|}
        
            \hline
            \multicolumn{1}{|l|}{\textbf{Fee Tier $\tau$}}  & \textbf{\begin{tabular}[c]{@{}c@{}}Sharpe Ratio \\ (after gas fees)\end{tabular}} & \textbf{\begin{tabular}[c]{@{}c@{}}Returns\\ (after gas fees)\end{tabular}} & \textbf{\begin{tabular}[c]{@{}c@{}}Average\\ lower spread\end{tabular}} & \textbf{\begin{tabular}[c]{@{}c@{}}Average\\ upper spread\end{tabular}} \\ \hline
            0.1\% & \begin{tabular}[c]{@{}c@{}}0.09\\ ($-$32.84)\end{tabular}& \begin{tabular}[c]{@{}c@{}}0.00\%\\ ($-$0.72\%)\end{tabular} & 382.88    & 12.16 \\ \hline
            0.3\% & \begin{tabular}[c]{@{}c@{}}3.74\\ ($-$1.26)\end{tabular}& \begin{tabular}[c]{@{}c@{}}0.54\%\\ ($-$0.18\%)\end{tabular} & 88.95& 7.75  \\ \hline
            1.0\% & \begin{tabular}[c]{@{}c@{}}\textbf{8.60} \\ (\textbf{7.32})\end{tabular}& \begin{tabular}[c]{@{}c@{}}\textbf{4.83\%}\\ (\textbf{4.10\%})\end{tabular}  & 12.83& 3.92  \\ \hline
        \end{tabular}
        \caption{Performance of the LSTM strategy on test set of simulated data. We assume arrival rate is $\lambda=\,$1/3 $\text{min}^{-1}$, market depth is $\kappa=\,$15,000,000\,, and the probability of a buy liquidity taking order is $p=\,$0.5.}
        \label{DL:table:one_fee_rate}
    \end{table}
    
    Figure \ref{DL:fig:fee} shows the LP's range as a function of the fee tier $\tau$. As the fee tier increases, we see that the LP narrows the range where she provides liquidity and increases her exposure to the risky asset $Y$. This result aligns with the finding of \cite{cartea2023predictable}, who show that LPs widen the range where they provide liquidity when the average fee income over time is lower. 

    Figure \ref{DL:fig:fee} also shows that as the value of the fee tier decreases, the LP skews her position more towards the reference asset $X$. Indeed, the exposure to the risks associated with $Y$ benefits the LP only if the fee yield is high enough. 

    Average returns of the LSTM strategy are highest when $\tau=\,$1.0\%;  they are equal to 4.83\% and 4.10\%, respectively, before and after gas fees. When $\tau=\,$1.0\%, the LP provides liquidity in a narrow range around the instantaneous rate such that the average lower and upper spreads  $\ell^\theta=\,$12.83 and $u^\theta=\,$3.92. Returns are lowest when $\tau=\,$0.1\% and they are $0.00\%$ before gas fees and $-0.72\%$ after gas fees. The LP skews her position towards the reference asset $X$ and the average lower and upper spreads are $\ell^\theta=\,$382.88 and $u^\theta=\,$12.16.
    
    \begin{figure}[H]\centering
	\includegraphics{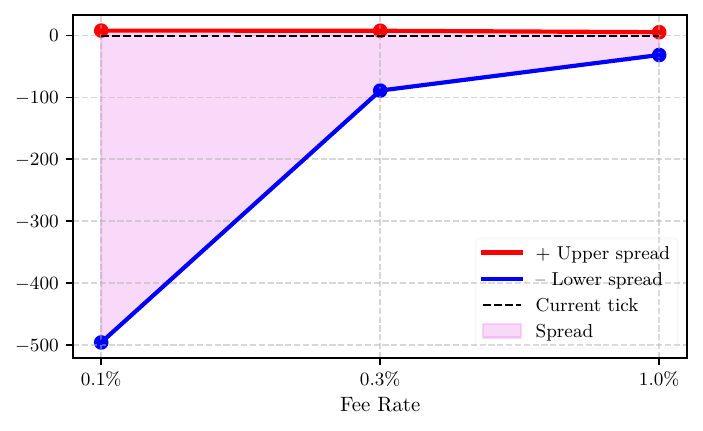}\\
	\caption{Average upper and lower spread posted by the LSTM strategy as functions of the fee tier $\tau$.}
	\label{DL:fig:fee}
    \end{figure}

\subsubsection{Trend}\label{DL:sec:one_trend}

Here, we train the LSTM network for various values of the probability $p$ of a buy liquidity taking, while keeping the remaining parameters fixed. The arrival rate is $\lambda=\,$1/3 $\text{min}^{-1}$, market depth is $\kappa=\,$15,000,000, and fee tier is $\tau=\,$0.3\%. When $p<1/2$, sell orders arrive more frequently than buy orders, so the pool's marginal rate trends downwards, on average. Similarly, when $p>1/2$, the pool marginal rate moves upwards, on average. Table \ref{DL:table:trend} shows how the average 30-minute returns $\widehat\mu_{\text{30min}}$ changes as a function of the probability $p$.

\begin{table}[H]\fontsize{9.0}{12.0}\selectfont
    \centering
    \begin{tabular}{c||c c c c c}
        \textbf{Probability of buy order $p$}           & 0.10          & 0.25          & 0.50      & 0.75      & 0.90 \\ \hline
        \textbf{Average Returns $\widehat\mu_{\text{30min}}$}   & $-$0.29\%     & $-$0.18\%     & 0.00\%    & 0.19\%    & 0.31\%
    \end{tabular}
    \caption{Average 30-minute returns $\widehat\mu_{\text{30min}}$ as a function of the probability $p$ of a buy liquidity taking order.  Market depth is $\kappa=\,$15,000,000, fee tier is $\tau=\,$0.3\%, and arrival rate is $\lambda=\,$1/3 $\text{min}^{-1}$.}
    \label{DL:table:trend}
\end{table}

    Table \ref{DL:table:one_trend} reports the performance of the LSTM strategy on the test set. It shows the average lower and upper spreads posted by the strategy, and shows the Sharpe Ratio and the average returns of the strategy, both before and after gas fees. As above, gas fees are $73.3$ units of asset $X$.

\begin{table}[H]\fontsize{9.0}{12.0}\selectfont
        \centering
        \begin{tabular}{|c|c|c|c|c|}
        
            \hline
            \multicolumn{1}{|l|}{\textbf{Trend $p$}}  & \textbf{\begin{tabular}[c]{@{}c@{}}Sharpe Ratio \\ (after gas fees)\end{tabular}} & \textbf{\begin{tabular}[c]{@{}c@{}}Returns\\ (after gas fees)\end{tabular}} & \textbf{\begin{tabular}[c]{@{}c@{}}Average\\ lower spread\end{tabular}} & \textbf{\begin{tabular}[c]{@{}c@{}}Average\\ upper spread\end{tabular}} \\ \hline
            0.10 & \begin{tabular}[c]{@{}c@{}}$-$7.86\\ ($-$21.18)\end{tabular}& \begin{tabular}[c]{@{}c@{}}$-$0.43\%\\ ($-$1.15\%)\end{tabular} & 495.00    & 2.00 \\ \hline
            0.25 & \begin{tabular}[c]{@{}c@{}}$-$3.40\\ ($-$21.11)\end{tabular}  & \begin{tabular}[c]{@{}c@{}}$-$0.14\%\\ ($-$0.86\%)\end{tabular} & 492.00& 2.00  \\ \hline
            0.50 & \begin{tabular}[c]{@{}c@{}}3.74\\ ($-$1.26)\end{tabular}  & \begin{tabular}[c]{@{}c@{}}0.54\%\\ ($-$0.18\%)\end{tabular} & 88.95& 7.75  \\ \hline
            0.75 & \begin{tabular}[c]{@{}c@{}}14.44 \\ (13.35)\end{tabular}  & \begin{tabular}[c]{@{}c@{}}9.59\%\\ (8.87\%)\end{tabular}  & 54.85& 170.41 \\ \hline
            0.90 & \begin{tabular}[c]{@{}c@{}}\textbf{25.40} \\ (\textbf{24.30})\end{tabular}  & \begin{tabular}[c]{@{}c@{}}\textbf{16.70\%}\\ (\textbf{15.97\%})\end{tabular}  & 69.14& 217.83  \\ \hline
        \end{tabular}
        \caption{Performance of the LSTM strategy on test set of simulated data. The LP trades in a CPMM with one liquidity pool with CL. Arrival rate is $\lambda=\,$1/3 $\text{min}^{-1}$, market depth is $\kappa=\,$15,000,000, and fee tier is $\tau=\,$0.3\%.}
        \label{DL:table:one_trend}
    \end{table}

Figure \ref{DL:fig:trend} shows the LP's range as a function of the average 30-minute returns $\widehat\mu_{\text{30min}}$. When $\widehat\mu_{\text{30min}}<0$, the LP skews her position towards the reference asset $X$ to reduce her losses associated to the drop of the instantaneous rate. Conversely, when $\widehat\mu_{\text{30min}}>0$, the LP skews her position towards the risky asset $Y$ to take advantage of the upward trend in the instantaneous rate $Z$. This result aligns with the finding of \cite{cartea2023predictable}, who show LPs skew their position in the direction of the rate's trend.

\begin{figure}[H]\centering
	\includegraphics{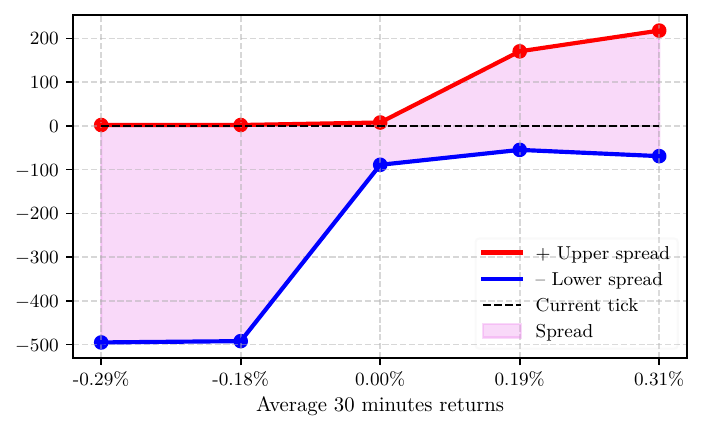}\\
	\caption{Average upper and lower spread posted by the LSTM strategy as functions of the average 30-minute returns $\widehat\mu_{\text{30min}}$.}
	\label{DL:fig:trend}
\end{figure}

When there is a negative trend, the LP keeps very little inventory in the risky asset $Y$, while when there is a positive trend, she keeps a considerable part of her inventory in the reference asset $X$; see Figure \ref{DL:fig:trend}. This is a consequence of using the mean-variance criterion, where it is optimal to increase this LP's investment in the risk-free reference asset $X$.

Table \ref{DL:table:one_trend} shows that the returns of the LSTM strategy are highest when $p=0.9$; returns are $16.70\%$ and $15.97\%$, before and after gas fees, respectively. When $p=0.9$ the LP skews her position towards the risky asset $Y$ to capitalise on the upward trend. The average lower and upper spreads are $\ell^\theta=69.14$ and $u^\theta=217.83$. Average returns are lowest when $p=\,$0.1 and the LP skews her position towards the reference asset $X$ to reduce losses related to the negative trend in the value of $Y$. Average returns are equal to $-0.43$\% and $-1.15$\% before and after gas fees, respectively. The average lower and upper spreads are $\ell^\theta=\,$495.00 and $u^\theta=\,$2.00.

\subsection{Two pools}

Next, consider an LP who provides liquidity in two liquidity pools, i.e., $N=2$. One trades the pair $X$ and $Y_1$, and the other trades the pair $X$ and $Y_2$, and the instantaneous rates at time $t=0$ are $Z_0^1=Z_0^2=\,$2,200. The LP provides liquidity for $24$ hours and readjusts her position every $30$ minutes, i.e., $T=$1,440 minutes and $\Delta t=30$ minutes. The LP's initial wealth is 500,000 units of asset $X$ and the maximum value for the lower and upper spreads that the LP considers is $\Tc=\,$500.  At each time $t\in\T$ the input to the LSTM strategy is $\mathfrak{I}_{t} = \left(\bm w^\theta_{t-\Delta t}, \bm \ell^\theta_{t-\Delta t},\bm u^\theta_{t-\Delta t}, \bm Z_t, \bm \Phi_{t-\Delta t,t}\right)$\,.
    
Liquidity taking orders arrive in the pools at Poisson rates of $\lambda_1,\lambda_2\in\R_{++}$ and the distribution of the order size in both pools is normal with mean $\mu_{\text{size}}=\,$132,030 and standard deviation $\sigma_{\text{size}}=\,$20,000 units of asset $X$. These choices of parameters are based on market data from the ETH/USDC 0.3\% pool in Uniswap v3 between 1 July 2021 and 30 September 2023; see Table \ref{DL:table:datadescr}.
    
    \subsubsection{Trading frequency}
    Here, we train the LSTM network for various values of the arrival rates $\lambda_1$ and $\lambda_2$, while keeping the remaining parameters fixed. Market depths are $\kappa_1=\kappa_2=\,$15,000,000, the fee tiers are $\tau_1=\tau_2=\,$0.3\%, and $p_1=p_2=\,$0.5 are the probabilities of a buy liquidity taking order in each pool.

    Table \ref{DL:table:two_frequency} reports the performance of the LSTM strategy on the test set, measured by returns and Sharpe Ratio, both before and after gas fees; at each time step the LP pays 73.3 units of asset $X$ for each pool where she adjusts her position. The table also shows the average lower and upper spreads posted by the strategy and the average proportion of wealth that the LP provides in each pool.

    Overall, Table \ref{DL:table:two_frequency} reports similar results to those of the experiments in Section \ref{DL:sec:one_frequency}. Indeed, the behaviour of the LSTM strategy in each individual pool is qualitatively the same as that of the single pool presented in Table \ref{DL:table:one_frequency}. Specifically, the LP provides liquidity in wide ranges when the value of $\lambda$ is big and in narrow ranges when the value of $\lambda$ is small. Moreover, when the value of $\lambda$ is big, the LP skews her position towards the reference asset $X$, and when the value of $\lambda$ is small, she increases her exposure to the risky asset $Y$. Note that the LP skews her position because the mean-variance criterion pushes the LP to reduce her exposure to the risky asset $Y$ to reduce the standard deviation of the strategy returns.
    
    Table \ref{DL:table:two_frequency} shows that the LSTM strategy provides more liquidity to the pool where volatility is lower. For instance, when the arrival rate $\lambda_2=\,$1  and $\lambda_1<\lambda_2$, then the LP deposits more than $99\%$ of her wealth in  the first pool. 

\begin{table}[H]\fontsize{7.5}{12.0}\selectfont
    \centering
\begin{tabular}{|cc|c|c|cc|cc|cc|}
\hline
\multicolumn{2}{|c|}{\textbf{Arrival rate $\lambda$}}& \multirow{2}{*}{\textbf{\begin{tabular}[c]{@{}c@{}}Sharpe Ratio \\ (after gas fees)\end{tabular}}} & \multirow{2}{*}{\textbf{\begin{tabular}[c]{@{}c@{}}Returns\\ (after gas fees)\end{tabular}}} & \multicolumn{2}{c|}{\textbf{\begin{tabular}[c]{@{}c@{}}Average\\ lower spread\end{tabular}}} & \multicolumn{2}{c|}{\textbf{\begin{tabular}[c]{@{}c@{}}Average\\ upper spread\end{tabular}}} & \multicolumn{2}{c|}{\textbf{Weights}}  \\ \cline{1-2} \cline{5-10} 
\multicolumn{1}{|c|}{\textbf{Pool 1}} & \textbf{Pool 2} & & & \multicolumn{1}{c|}{\textbf{Pool 1}}   & \textbf{Pool 2}   & \multicolumn{1}{c|}{\textbf{Pool 1}}   & \textbf{Pool 2}   & \multicolumn{1}{c|}{\textbf{Pool 1}} & \textbf{Pool 2} \\ \hline
\multicolumn{1}{|c|}{1}& 1& \begin{tabular}[c]{@{}c@{}}0.35\\ ($-$35.34)\end{tabular}& \begin{tabular}[c]{@{}c@{}}0.01\%\\ ($-$1.43\%)\end{tabular}  & \multicolumn{1}{c|}{494.00} & 480.62 & \multicolumn{1}{c|}{12.73}  & 7.75   & \multicolumn{1}{c|}{43.69\%}    & 56.31\% \\ \hline
\multicolumn{1}{|c|}{1/3}& 1 & \begin{tabular}[c]{@{}c@{}}3.57 \\ ($-$2.77)\end{tabular}& \begin{tabular}[c]{@{}c@{}}0.81\%\\ ($-$0.63\%)\end{tabular}  & \multicolumn{1}{c|}{51.89}  & 480.65 & \multicolumn{1}{c|}{8.51}   & 34.72  & \multicolumn{1}{c|}{99.43\%}    & 0.57\%  \\ \hline
\multicolumn{1}{|c|}{1/5}& 1 & \begin{tabular}[c]{@{}c@{}}4.60 \\ ($-$1.03)\end{tabular}& \begin{tabular}[c]{@{}c@{}}1.18\%\\ ($-$0.26\%)\end{tabular}  & \multicolumn{1}{c|}{37.14}  & 480.60 & \multicolumn{1}{c|}{6.68}   & 75.63  & \multicolumn{1}{c|}{97.66\%}    & 2.34\%  \\ \hline
\multicolumn{1}{|c|}{1/10}& 1 & \begin{tabular}[c]{@{}c@{}}6.98 \\ (1.52)\end{tabular} & \begin{tabular}[c]{@{}c@{}}1.85\%\\ (0.40\%)\end{tabular}   & \multicolumn{1}{c|}{19.40}  & 497.00 & \multicolumn{1}{c|}{6.98}   & 1.06   & \multicolumn{1}{c|}{99.15\%}    & 0.85\%  \\ \hline
\multicolumn{1}{|c|}{1/20}& 1 & \begin{tabular}[c]{@{}c@{}}9.64 \\ (4.58)\end{tabular} & \begin{tabular}[c]{@{}c@{}}2.75\%\\ (1.31\%)\end{tabular}   & \multicolumn{1}{c|}{74.53}  & 492.00 & \multicolumn{1}{c|}{5.00}   & 425.25 & \multicolumn{1}{c|}{\textbf{99.85\%}}    & 0.15\%  \\ \hline
\multicolumn{1}{|c|}{1/3}& 1/3& \begin{tabular}[c]{@{}c@{}}5.42 \\ (0.60)\end{tabular} & \begin{tabular}[c]{@{}c@{}}1.63\%\\ (0.18\%)\end{tabular}   & \multicolumn{1}{c|}{15.08}  & 24.00  & \multicolumn{1}{c|}{6.37}   & 7.00   & \multicolumn{1}{c|}{42.50\%}    & 57.50\% \\ \hline
\multicolumn{1}{|c|}{1/5}& 1/3& \begin{tabular}[c]{@{}c@{}}6.38\\ (2.03)\end{tabular}  & \begin{tabular}[c]{@{}c@{}}2.12\%\\ (0.67\%)\end{tabular}   & \multicolumn{1}{c|}{14.08}  & 9.20   & \multicolumn{1}{c|}{2.96}   & 3.86   & \multicolumn{1}{c|}{67.28\%}    & 32.72\% \\ \hline
\multicolumn{1}{|c|}{1/10}& 1/3& \begin{tabular}[c]{@{}c@{}}6.44\\ (2.16)\end{tabular}  & \begin{tabular}[c]{@{}c@{}}2.18\%\\ (0.73\%)\end{tabular}   & \multicolumn{1}{c|}{45.50}  & 471.00 & \multicolumn{1}{c|}{5.33}   & 5.15   & \multicolumn{1}{c|}{94.05\%}    & 5.95\%  \\ \hline
\multicolumn{1}{|c|}{1/20}& 1/3& \begin{tabular}[c]{@{}c@{}}\textbf{11.46}\\ (\textbf{6.62})\end{tabular} & \begin{tabular}[c]{@{}c@{}}\textbf{3.42\%}\\ (\textbf{1.97\%})\end{tabular}   & \multicolumn{1}{c|}{26.94}  & 93.29  & \multicolumn{1}{c|}{7.66}   & 3.43   & \multicolumn{1}{c|}{81.97\%}    & 18.03\% \\ \hline
\multicolumn{1}{|c|}{1/5}& 1/5& \begin{tabular}[c]{@{}c@{}}5.20\\ ($-$0.50)\end{tabular} & \begin{tabular}[c]{@{}c@{}}1.32\%\\ ($-$0.13\%)\end{tabular}  & \multicolumn{1}{c|}{28.51}  & 175.90 & \multicolumn{1}{c|}{5.37}   & 431.18 & \multicolumn{1}{c|}{98.97\%}    & 1.03\%  \\ \hline
\multicolumn{1}{|c|}{1/10}& 1/5& \begin{tabular}[c]{@{}c@{}}5.81\\ (0.67)\end{tabular}  & \begin{tabular}[c]{@{}c@{}}1.64\%\\ (0.19\%)\end{tabular}   & \multicolumn{1}{c|}{167.41} & 9.34   & \multicolumn{1}{c|}{21.14}  & 5.03   & \multicolumn{1}{c|}{51.69\%}    & 48.31\% \\ \hline
\multicolumn{1}{|c|}{1/20}& 1/5& \begin{tabular}[c]{@{}c@{}}9.46\\ (4.79)\end{tabular}  & \begin{tabular}[c]{@{}c@{}}2.93\%\\ (1.48\%)\end{tabular}   & \multicolumn{1}{c|}{48.05}  & 234.57 & \multicolumn{1}{c|}{7.35}   & 5.48   & \multicolumn{1}{c|}{74.30\%}    & 25.70\% \\ \hline
\multicolumn{1}{|c|}{1/10}& 1/10& \begin{tabular}[c]{@{}c@{}}7.63\\ (2.30)\end{tabular}  & \begin{tabular}[c]{@{}c@{}}2.07\%\\ (0.62\%)\end{tabular}   & \multicolumn{1}{c|}{209.04} & 13.55  & \multicolumn{1}{c|}{9.98}   & 4.49   & \multicolumn{1}{c|}{23.63\%}    & \textbf{76.37\%} \\ \hline
\multicolumn{1}{|c|}{1/20}& 1/10& \begin{tabular}[c]{@{}c@{}}8.12\\ (4.51)\end{tabular}  & \begin{tabular}[c]{@{}c@{}}3.25\%\\ (1.81\%)\end{tabular}   & \multicolumn{1}{c|}{42.26}  & 32.19  & \multicolumn{1}{c|}{3.85}   & 4.00   & \multicolumn{1}{c|}{62.30\%}    & 37.70\% \\ \hline
\multicolumn{1}{|c|}{1/20}& 1/20& \begin{tabular}[c]{@{}c@{}}9.49\\ (4.10)\end{tabular}  & \begin{tabular}[c]{@{}c@{}}2.55\%\\ (1.10\%)\end{tabular}   & \multicolumn{1}{c|}{79.80}  & 351.76 & \multicolumn{1}{c|}{4.24}   & 296.61 & \multicolumn{1}{c|}{74.33\%}    & 25.67\% \\ \hline
\end{tabular}
\caption{Performance of the LSTM strategy on test set of simulated data, in a CPMM with two liquidity pools with CL. Market depths are $\kappa_1=\kappa_2=\,$15,000,000, fee tiers are $\tau_1=\tau_2=\,$0.3\%, and the probabilities of a buy order are $p_1=p_2=\,$0.5.}
\label{DL:table:two_frequency}

\end{table}

\subsubsection{Fee tier}

Here, we train the LSTM network for various values of the fee tiers $\tau_1$ and $\tau_2$, while keeping the remaining parameters fixed. Arrival rates are $\lambda_1=\lambda_2=\,$1/3, market depths are $\kappa_1=\kappa_2=\,$15,000,000, and $p_1=p_2=\,$0.5 are the probabilities that of buy orders.

Table \ref{DL:table:two_fee_tier} reports the performance of the LSTM strategy on the test set, measured by returns and Sharpe Ratio, both before and after gas fees; as above, she pays 73.3 units of asset $X$ in each pool for every adjustment of liquidity. The table also shows the average lower and upper spreads posted by the strategy and the average proportion of wealth that the LP provides in each pool.

Overall, Table \ref{DL:table:two_fee_tier} reports similar results to those of the experiments in Section \ref{DL:sec:one_fee_tier}. Indeed, the behaviour of the LSTM strategy in each individual pool is qualitatively the same as that of the single pool presented in Table \ref{DL:table:one_fee_rate}. Specifically, the LP provides liquidity in wide ranges when the value of the fee tier $\tau$ is small and in narrow ranges when the value of $\tau$ is big. Moreover, when the value of the fee tier is low, the LP skews her position towards the reference asset $X$.

Table \ref{DL:table:two_fee_tier} shows that the LSTM strategy provides more liquidity to the pool where the fee tier $\tau$ is higher. For instance, when  $\tau_1=\,$1.0\% and $\tau_2=\,$0.1\%, the LP provides all her liquidity to the first pool. 

The returns of the LSTM strategy are highest when $\tau_1=\tau_2=\,$1.0\%  and the LP splits her wealth almost equally between the two pools. Returns are equal to 17.67\% before gas fees and to 16.22\% after gas fees.

\begin{table}[H]\fontsize{7.5}{12.0}\selectfont
    \centering
\begin{tabular}{|cc|c|c|cc|cc|cc|}
\hline
\multicolumn{2}{|c|}{\textbf{Fee   Tier $\tau$}}        & \multirow{2}{*}{\textbf{\begin{tabular}[c]{@{}c@{}}Sharpe Ratio\\ (after gas fees)\end{tabular}}} & \multirow{2}{*}{\textbf{\begin{tabular}[c]{@{}c@{}}Returns\\ (after gas fees)\end{tabular}}} & \multicolumn{2}{c|}{\textbf{\begin{tabular}[c]{@{}c@{}}Average   \\ lower spread\end{tabular}}} & \multicolumn{2}{c|}{\textbf{\begin{tabular}[c]{@{}c@{}}Average   \\ upper spread\end{tabular}}} & \multicolumn{2}{c|}{\textbf{Weights}}                  \\ \cline{1-2} \cline{5-10} 
\multicolumn{1}{|c|}{\textbf{Pool 1}} & \textbf{Pool 2} &                                                                                             &                                                                                        & \multicolumn{1}{c|}{\textbf{Pool 1}}                     & \textbf{Pool 2}                    & \multicolumn{1}{c|}{\textbf{Pool 1}}                     & \textbf{Pool 2}                    & \multicolumn{1}{c|}{\textbf{Pool 1}} & \textbf{Pool 2} \\ \hline
\multicolumn{1}{|c|}{0.1\%}          & 0.1\%          & \begin{tabular}[c]{@{}c@{}}0.12\\ ($-$60.00)\end{tabular}                                     & \begin{tabular}[c]{@{}c@{}}$-$0.00\%\\ ($-$1.44\%)\end{tabular}                            & \multicolumn{1}{c|}{480.75}                              & 481.28                             & \multicolumn{1}{c|}{14.00}                               & 19.15                              & \multicolumn{1}{c|}{57.13\%}         & 42.87\%         \\ \hline
\multicolumn{1}{|c|}{0.3\%}          & 0.1\%          & \begin{tabular}[c]{@{}c@{}}3.19\\ ($-$2.61)\end{tabular}                                      & \begin{tabular}[c]{@{}c@{}}0.80\%\\ ($-$0.65\%)\end{tabular}                             & \multicolumn{1}{c|}{84.86}                               & 480.32                             & \multicolumn{1}{c|}{4.00}                                & 33.86                              & \multicolumn{1}{c|}{96.40\%}         & 3.60\%          \\ \hline
\multicolumn{1}{|c|}{1.0\%}          & 0.1\%          & \begin{tabular}[c]{@{}c@{}}7.78\\ (4.04)\end{tabular}                                       & \begin{tabular}[c]{@{}c@{}}3.01\%\\ (1.57\%)\end{tabular}                              & \multicolumn{1}{c|}{70.66}                               & 491.00                             & \multicolumn{1}{c|}{5.46}                                & 397.69                             & \multicolumn{1}{c|}{\textbf{99.49\%}}         & 0.51\%          \\ \hline
\multicolumn{1}{|c|}{0.3\%}          & 0.3\%          & \begin{tabular}[c]{@{}c@{}}5.42\\ (0.60)\end{tabular}                                       & \begin{tabular}[c]{@{}c@{}}1.63\%\\ (0.18\%)\end{tabular}                              & \multicolumn{1}{c|}{15.08}                               & 24.00                              & \multicolumn{1}{c|}{6.37}                                & 7.00                               & \multicolumn{1}{c|}{42.50\%}         & \textbf{57.50\%}      \\ \hline
\multicolumn{1}{|c|}{1.0\%}          & 0.3\%          & \begin{tabular}[c]{@{}c@{}}8.39\\ (4.26)\end{tabular}                                       & \begin{tabular}[c]{@{}c@{}}2.94\%\\ (1.49\%)\end{tabular}                              & \multicolumn{1}{c|}{71.32}                               & 21.55                              & \multicolumn{1}{c|}{6.78}                               & 13.96                              & \multicolumn{1}{c|}{82.31\%}         &   17.69\%       \\ \hline
\multicolumn{1}{|c|}{1.0\%}          & 1.0\%          & \begin{tabular}[c]{@{}c@{}}\textbf{11.05}\\ (\textbf{8.51})\end{tabular}                                      & \begin{tabular}[c]{@{}c@{}}\textbf{6.29\%}\\ (\textbf{4.84\%})\end{tabular}                              & \multicolumn{1}{c|}{42.31}                               & 23.58                              & \multicolumn{1}{c|}{11.95}                               & 11.78                              & \multicolumn{1}{c|}{53.73\%}         & 46.27\%         \\ \hline
\end{tabular}
\caption{Performance of the LSTM strategy on test set of simulated data. The LP trades in a CPMM with two liquidity pools with CL. Arrival rates are $\lambda_1=\lambda_2=\,\,$1/3, market depths are $\kappa_1=\kappa_2=\,\,$15,000,000, and the probabilities of a buy order are $p_1=p_2=\,\,$0.5.}
\label{DL:table:two_fee_tier}
\end{table}

\subsubsection{Trend}
Here, we train the LSTM network for various values of the probabilities $p_1$ and $p_2$ of liquidity taking orders to be buy orders in each pool, while keeping the remaining parameters fixed. Arrival rates are $\lambda_1=\lambda_2=\,$1/3 $\text{min}^{-1}$, market depths are $\kappa_1=\kappa_2=\,$15,000,000, and fee tiers are $\tau_1=\tau_2=\,$0.3\%.

Table \ref{DL:table:two_trend} reports the performance of the LSTM strategy on the test set, measured by returns and Sharpe Ratio, both before and after gas fees; as above, she pays 73.3 units of asset $X$ in each pool for every adjustment of liquidity. The table also shows the average lower and upper spreads posted by the strategy and the average proportion of wealth that the LP provides in each pool.

Overall, Table \ref{DL:table:two_trend} reports similar results to those of the experiments in Section \ref{DL:sec:one_trend}. Indeed, the behaviour of the LSTM strategy in each individual pool is qualitatively the same as that of presented in Table \ref{DL:table:one_trend}. Specifically, when there is a downward trend in the instantaneous rate, the LP skews her position towards the reference asset $X$. Conversely, when the instantaneous rate trends upwards, the LP skews her position towards the risky asset. Similarly to the results of Section \ref{DL:sec:one_trend}, the mean-variance criterion pushes the LP to keep a considerable position in the reference asset even when the instantaneous rate has a strong upwards trend.

Table \ref{DL:table:two_frequency} shows that the LSTM strategy provides more liquidity to the pool when the probability $p$ of a buy order is higher. For instance, when  $p_1=\,$0.25 and $p_2=\,$0.1, the LP provides all of her liquidity to the first pool. 

The returns of the LSTM strategy are highest when $p_1=p_2=\,$0.9 and there is a strong upward trend in both pools. Returns are equal to 17.67\% before gas fees and to 16.22\% after gas fees.

\begin{table}[H]\fontsize{7.5}{12.0}\selectfont
    \centering
\begin{tabular}{|cc|c|c|cc|cc|cc|}
\hline
\multicolumn{2}{|c|}{\textbf{Trend $p$}}                & \multirow{2}{*}{\textbf{\begin{tabular}[c]{@{}c@{}}Sharpe Ratio\\ (after gas fees)\end{tabular}}} & \multirow{2}{*}{\textbf{\begin{tabular}[c]{@{}c@{}}Returns\\ (after gas fees)\end{tabular}}} & \multicolumn{2}{c|}{\textbf{\begin{tabular}[c]{@{}c@{}}Average   \\ lower spread\end{tabular}}} & \multicolumn{2}{c|}{\textbf{\begin{tabular}[c]{@{}c@{}}Average   \\ upper spread\end{tabular}}} & \multicolumn{2}{c|}{\textbf{Weights}}                  \\ \cline{1-2} \cline{5-10} 
\multicolumn{1}{|c|}{\textbf{Pool 1}} & \textbf{Pool 2} &                                                                                             &                                                                                        & \multicolumn{1}{c|}{\textbf{Pool 1}}                     & \textbf{Pool 2}                    & \multicolumn{1}{c|}{\textbf{Pool 1}}                     & \textbf{Pool 2}                    & \multicolumn{1}{c|}{\textbf{Pool 1}} & \textbf{Pool 2} \\ \hline
\multicolumn{1}{|c|}{10\%}            & 10\%            & \begin{tabular}[c]{@{}c@{}}$-$7.48 \\ ($-$34.91)\end{tabular}                                   & \begin{tabular}[c]{@{}c@{}}$-$0.39\%\\ ($-$1.84\%)\end{tabular}                            & \multicolumn{1}{c|}{500.00}                              & 500.00                             & \multicolumn{1}{c|}{1.00}                                & 1.00                               & \multicolumn{1}{c|}{0.00\%}          & \textbf{100.00\%}       \\ \hline
\multicolumn{1}{|c|}{25\%}            & 10\%            & \begin{tabular}[c]{@{}c@{}}$-$3.12\\ ($-$40.16)\end{tabular}                                    & \begin{tabular}[c]{@{}c@{}}$-$0.12\%\\ ($-$1.57\%)\end{tabular}                            & \multicolumn{1}{c|}{500.00}                              & 500.00                             & \multicolumn{1}{c|}{1.00}                                & 1.00                               & \multicolumn{1}{c|}{\textbf{100.00\%}}        & 0.00\%          \\ \hline
\multicolumn{1}{|c|}{50\%}            & 10\%            & \begin{tabular}[c]{@{}c@{}}3.71\\ ($-$5.43)\end{tabular}                                      & \begin{tabular}[c]{@{}c@{}}0.59\%\\ ($-$0.86\%)\end{tabular}                             & \multicolumn{1}{c|}{62.01}                               & 425.47                             & \multicolumn{1}{c|}{4.02}                                & 16.41                              & \multicolumn{1}{c|}{98.24\%}         & 1.76\%          \\ \hline
\multicolumn{1}{|c|}{75\%}            & 10\%            & \begin{tabular}[c]{@{}c@{}}20.21 \\ (17.13)\end{tabular}                                    & \begin{tabular}[c]{@{}c@{}}9.49\%\\ (8.04\%)\end{tabular}                              & \multicolumn{1}{c|}{55.86}                               & 483.00                             & \multicolumn{1}{c|}{71.55}                               & 5.00                               & \multicolumn{1}{c|}{99.95\%}         & 0.05\%          \\ \hline
\multicolumn{1}{|c|}{90\%}            & 10\%            & \begin{tabular}[c]{@{}c@{}}34.65\\ (31.63)\end{tabular}                                     & \begin{tabular}[c]{@{}c@{}}16.57\%\\ (15.12\%)\end{tabular}                            & \multicolumn{1}{c|}{1.00}                                & 483.00                             & \multicolumn{1}{c|}{261.14}                              & 2.00                               & \multicolumn{1}{c|}{93.04\%}         & 6.96\%          \\ \hline
\multicolumn{1}{|c|}{25\%}            & 25\%            & \begin{tabular}[c]{@{}c@{}}$-$3.12\\ ($-$40.48)\end{tabular}                                    & \begin{tabular}[c]{@{}c@{}}$-$0.12\%\\ ($-$1.57\%)\end{tabular}                            & \multicolumn{1}{c|}{500.00}                              & 500.00                             & \multicolumn{1}{c|}{1.00}                                & 1.00                               & \multicolumn{1}{c|}{0.00\%}          & 100.00\%        \\ \hline
\multicolumn{1}{|c|}{50\%}            & 25\%            & \begin{tabular}[c]{@{}c@{}}3.21\\ ($-$1.73)\end{tabular}                                      & \begin{tabular}[c]{@{}c@{}}0.94\%\\ ($-$0.51\%)\end{tabular}                             & \multicolumn{1}{c|}{32.47}                               & 499.10                             & \multicolumn{1}{c|}{6.15}                                & 16.00                              & \multicolumn{1}{c|}{98.56\%}         & 1.44\%          \\ \hline
\multicolumn{1}{|c|}{75\%}            & 25\%            & \begin{tabular}[c]{@{}c@{}}16.50\\ (13.95)\end{tabular}                                     & \begin{tabular}[c]{@{}c@{}}9.35\%\\ (7.91\%)\end{tabular}                              & \multicolumn{1}{c|}{5.62}                                & 478.12                             & \multicolumn{1}{c|}{67.00}                               & 17.78                              & \multicolumn{1}{c|}{92.16\%}         & 7.84\%          \\ \hline
\multicolumn{1}{|c|}{90\%}            & 25\%            & \begin{tabular}[c]{@{}c@{}}\textbf{42.15}\\ (\textbf{38.52})\end{tabular}                                     & \begin{tabular}[c]{@{}c@{}}16.80\%\\ (15.35\%)\end{tabular}                            & \multicolumn{1}{c|}{8.48}                                & 483.00                             & \multicolumn{1}{c|}{64.80}                               & 4.00                               & \multicolumn{1}{c|}{89.26\%}         & 10.74\%         \\ \hline
\multicolumn{1}{|c|}{50\%}            & 50\%            & \begin{tabular}[c]{@{}c@{}}5.42\\ (0.60)\end{tabular}                                       & \begin{tabular}[c]{@{}c@{}}1.63\%\\ (0.18\%)\end{tabular}                              & \multicolumn{1}{c|}{24.00}                               & 15.08                              & \multicolumn{1}{c|}{7.00}                                & 6.37                               & \multicolumn{1}{c|}{57.50\%}         & 42.50\%         \\ \hline
\multicolumn{1}{|c|}{75\%}            & 50\%            & \begin{tabular}[c]{@{}c@{}}15.97\\ (13.57)\end{tabular}                                     & \begin{tabular}[c]{@{}c@{}}9.64\%\\ (8.20\%)\end{tabular}                              & \multicolumn{1}{c|}{26.81}                               & 50.90                              & \multicolumn{1}{c|}{139.64}                              & 38.95                              & \multicolumn{1}{c|}{96.91\%}         & 3.09\%          \\ \hline
\multicolumn{1}{|c|}{90\%}            & 50\%            & \begin{tabular}[c]{@{}c@{}}33.24\\ (30.23)\end{tabular}                                     & \begin{tabular}[c]{@{}c@{}}16.00\%\\ (14.55\%)\end{tabular}                            & \multicolumn{1}{c|}{43.19}                               & 45.54                              & \multicolumn{1}{c|}{289.96}                              & 356.38                             & \multicolumn{1}{c|}{95.72\%}         & 4.28\%          \\ \hline
\multicolumn{1}{|c|}{75\%}            & 75\%            & \begin{tabular}[c]{@{}c@{}}22.70\\ (19.36)\end{tabular}                                     & \begin{tabular}[c]{@{}c@{}}9.84\%\\ (8.39\%)\end{tabular}                              & \multicolumn{1}{c|}{7.42}                                & 33.99                              & \multicolumn{1}{c|}{65.05}                               & 82.86                              & \multicolumn{1}{c|}{46.88\%}         & 53.12\%         \\ \hline
\multicolumn{1}{|c|}{90\%}            & 75\%            & \begin{tabular}[c]{@{}c@{}}28.23\\ (25.81)\end{tabular}                                     & \begin{tabular}[c]{@{}c@{}}16.87\%\\ (15.43\%)\end{tabular}                            & \multicolumn{1}{c|}{4.38}                                & 83.83                              & \multicolumn{1}{c|}{106.30}                              & 500.00                             & \multicolumn{1}{c|}{87.94\%}         & 12.06\%         \\ \hline
\multicolumn{1}{|c|}{90\%}            & 90\%            & \begin{tabular}[c]{@{}c@{}}25.31\\ (23.24)\end{tabular}                                     & \begin{tabular}[c]{@{}c@{}}\textbf{17.67\%}\\ (\textbf{16.22\%})\end{tabular}                            & \multicolumn{1}{c|}{26.44}                               & 9.85                               & \multicolumn{1}{c|}{332.94}                              & 176.75                             & \multicolumn{1}{c|}{47.70\%}         & 52.30\%         \\ \hline
\end{tabular}
\caption{Performance of the LSTM strategy on test set of simulated data. The LP trades in a CPMM with two liquidity pools with CL. Arrival rates are $\lambda_1=\lambda_2=\,$1/3, market depths are $\kappa_1=\kappa_2=\,$15,000,000, and fee tiers are $\tau_1=\tau_2=\,$0.3\%.}
\label{DL:table:two_trend}
\end{table}

\section{Numerical experiments with market data}\label{DL:sec:market_data}
    \subsection{Data}
    \subsubsection{Uniswap v3}
	
	
    We use LP and LT data from the Uniswap v3 pools for the pair ETH/USDC and the pair BTC/USDC between 1 July 2021 and 30 September 2023 to train the LSTM strategy introduced in Section \ref{DL:sec:deep}. The fee tier in both pools is $0.3\%$. 

    {\begin{table}[H]
    \centering
    \begin{tabular}{c || r  r || r  r } 
        \hline 
   & \multicolumn{2}{c||}{ETH/USDC $0.3\%$}   & \multicolumn{2}{c}{BTC/USDC $0.3\%$} \\ [0.5ex]
        \hline 
   & LT  & LP & LT & LP\\ [0.5ex]
        \hline
        Number of instructions &  380,162 & 115,170 & 136,365 & 16,307 \\ [0.8ex]
        Average daily number &  & \\ [-0.5ex] 
        of instructions &  462 & 140 & 166 & 20 \\ [0.8ex]
        Total USD volume & $\approx$ \$ 50.2$\times 10^9$ & $\approx$ \$ 160.9 $\times 10^9$ & $\approx$ \$ 10.3 $\times 10^9$ & $\approx$ \$ 29.9 $\times 10^9$\\[1ex] 
        Average daily USD volume & \$ 61,061,745 & \$ 195,700,068 & \$ 12,560,091 & \$ 36,361,862 \\ [0.8ex]
        Average LT transaction &  & \\[-0.5ex] 
        or LP operation size & \$ 132,030 & \$ 1,396,765 & \$ 75,711 & \$ 1,832,921\\ [0.8ex]
        
        Average trading frequency & 3 minutes & 10 minutes & 9 minutes & 72 minutes \\ [0.8ex]
        \hline
        
        Average pool depth &  \multicolumn{2}{c||}{14,443,050 $\sqrt{\textrm{ETH}\cdot\textrm{USDC}}$}   & \multicolumn{2}{c}{1,342,009 $\sqrt{\textrm{BTC}\cdot\textrm{USDC}}$}  \\ [0.8ex]
        \hline
        \hline 
    \end{tabular}
    \caption {LT and LP activity details in the ETH/USDC $0.3\%$ and  ETH/USDC $0.3\%$ pools between 1 July 2021 and 30 September 2023. }
    \label{DL:table:datadescr}
    \hfill
    \end{table}}

    Table \ref{DL:table:datadescr} shows statistics for both pools. Overall, Table \ref{DL:table:datadescr} shows that there is more trading activity in the ETH/USDC than in the BTC/USDC pool. The total number of both LT and LP instructions in the ETH/USDC  pool is significantly higher than in the BTC/USDC pool. Here, an LT instruction is a buy or a sell liquidity taking order, while an LP instruction is to deposit or to withdraw liquidity into and from the pool, respectively. The frequency of arrival of LT and LP instructions is significantly higher in the ETH/USDC pool compared with the BTC/USDC pool. Specifically, on average in the ETH/USDC pool there is an LT instruction every 3 minutes, compared to one every 9 minutes in the BTC/USDC pool. Similarly, on average in the ETH/USDC pool there is an LP instruction every 10 minutes, while in the BTC/USDC pool there is one every 72 minutes.

    The total USD volume traded in the two pools is significantly different. in the ETH/USDC pool, the total USD volume traded by LTs is around \$ $50.2\times 10^9$, which is roughly five times greater than the total USD volume traded by LTs in the BTC/USDC pool, which is approximately \$ $10.3\times 10^9$. 

    \subsubsection{Binance}
    \cite{cartea2022decentralised} solves the problem of an LT who liquidates a large position in a CPMM with CL. The authors backtest their strategy on market data from Uniswap v3 and use the instantaneous rate from Binance as market signals to enhance the performance of their strategy. At present, Binance is the largest centralised exchange for cryptocurrencies and it is organised as a limit order book market. Here, alongside data from Uniswap v3, we use data from Binance to train the LSTM strategy introduced in Section \ref{DL:sec:deep}. Specifically, we use volume and rate data for the pair ETH/USDC and the pair BTC/USDC between 1 July 2021 and 30 September 2023; Table \ref{DL:table:binance} shows statistics for both pairs. 

    A comparison between Table \ref{DL:table:datadescr} and \ref{DL:table:binance} highlights that in Uniswap v3 there is significantly more activity in the ETH/USDC pair than in BTC/USDC, while in Binance there is more activity in BTC/USDC than in ETH/USDC. Specifically, the total USD volume traded by LTs in the BTC/USDC pair is approximately \$ 1,386.2$\times10^9$, which is roughly 7 times higher than the total USD volume traded in the ETH/USDC pair, equal to \$ 201.6$\times10^9$. 

    
    In Binance, the average trading sizes of ETH/USDC and BTC/USDC are \$ 1,079 and \$ 953, respectively. These values are significantly lower than the average trading sizes of the same pairs in Uniswap v3, which are equal to \$ 132,030 and \$ 75,711; see Table \ref{DL:table:datadescr}. Gas fees are the main reason behind the difference in the average trading sizes between the two trading venues. Indeed, in Uniswap v3 users pay a high flat fee, which prevents them from sending multiple small orders to the pool and to opt instead for sending individual large orders.

    {\begin{table}[H]
    \centering
    \begin{tabular}{c || c || c } 
        \hline
        & ETH/USDC   & BTC/USDC \\ \hline
        &&\\[-1.2ex]
        Total USD volume                & $\approx$ 201.6 $\times10^9$ & $\approx $ 1,386.2 $\times10^9$ \\ [0.8ex]
        Average daily USD volume        & 245,365,037             & 1,687,095,461  \\[0.8ex]
        Average daily number of trades  & 227,385            & 1,770,836    \\[0.8ex]
        Average USD trade size          & \$ 1,079                 & \$ 953  \\[0.8ex]
        Average trading frequency       & 0.38 seconds            & 0.05 seconds \\[0.8ex]\hline\hline
        
    \end{tabular}
    \caption {Binance activity in ETH/USDC and BTC/USDC between 1 July 2021 and 30 September 2023. }
    \label{DL:table:binance}
    \hfill
    \end{table}}

    \subsection{Features}\label{DL:sec:feat}
    
    We train the LSTM strategy from Section \ref{DL:sec:deep} with market data from Uniswap v3 and Binance between 1 July 2021 and 30 September 2023. At each time $t$ the input to the strategy $\mathfrak{I}_t$ uses the following features from the two trading venues:
    \begin{itemize}
        \item \textbf{Pool current status:} pool rate $Z_t$, tick $\mathcal{Z}(i)$, market depth $\kappa$, and gas price.
        \item \textbf{Past returns, volatility, and fees:} percentage change between the current price and prices 5, 10, and 30 minutes before, and the volatility of the time series of these returns over the past 180, 360, and 720 minutes. Total fees earned by LPs and the pool fee rate of the past 180, 360, and 720 minutes.
        \item \textbf{LP activity:}  total USD amount deposited and withdrawn by LPs over the past 10 and 30 minutes. Total number of LPs transactions and the net USD amount over the same time horizons.
        \item \textbf{Binance activity:} total USD volume of buy and sell orders and the number of liquidity taking orders in Binance over the past 30 minutes. Total USD amount of limit orders posted in Binance and current spread between Binance and Uniswap v3.
        \item \textbf{Signals:} trading flow imbalance in Uniswap v3 and Binance. Moving average convergence/divergence (MACD) signal and exponential moving average of the MACD signal with periods 9, 14, and 26; see \ref{DL:ax:macd}.
    \end{itemize}
    
    \subsection{Results}\label{DL:sec:results_real}

    Next, we use market data from Uniswap v3 and Binance to showcase the performance of the LSTM strategy of Section \ref{DL:sec:deep}. We train an LSTM network of five LSTM layers with 120 neurons each.  The dropout rate is 0.1 and we use the Adam optimiser for training. Note that in Section \ref{DL:sec:fake}, each LSTM layer has 24 neurons. Here, we increase the number of neurons per layer because we use more input features and because market data are more noisy than the simulated data, so the LSTM network needs more parameters to learn the optimal strategy.
    
    We split our dataset into non-overlapping daily time series and further split our dataset between train, validation, and test data. Then, to avoid look ahead bias, we use the first 30\% of daily time series to train the model, the next 30\% to validate the model, and the remaining 40\% to test the model. In Section \ref{DL:sec:fake}, we use 60\% of the data for training, 20\% for validation, and 20\% for testing. In this section, a higher percentage of the dataset is used for testing because the amount of market data available is low and we need to increase the number of datapoints in the test set to ensure that there is enough variability among them. Indeed, we want the empirical distribution of the test set to be as close as possible to the empirical distribution of the entire dataset to ensure that the performance of the LSTM strategy on the test set is unbiased.
    
    In our experiments, we consider three values for the initial wealth of the LP and five values for her trading frequency $\Delta t$, and use  Sharpe Ratio and the mean-variance as criteria to train the LSTM network. We execute the LSTM strategy for five choices of trading frequency and three initial wealth positions. Figure \ref{DL:fig:returns} shows the returns of the LSTM strategy, before and after gas fees, and Table \ref{DL:table:real} reports the results.

    \begin{figure}[H]\centering
	\includegraphics{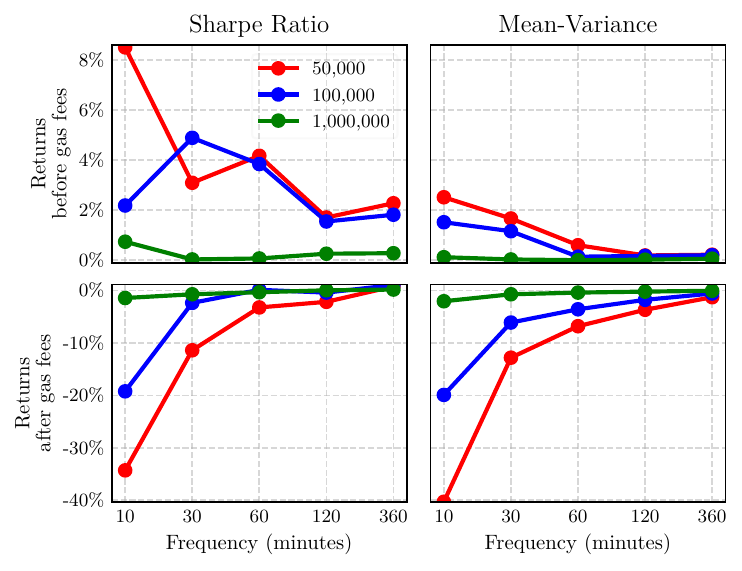}\\
	\caption{\textbf{Left}: average returns of the LSTM strategy before (top) and after (bottom) gas fees using Sharpe Ratio to train the LSTM network. \textbf{Right}: average returns of the LSTM strategy before (top) and after (bottom) gas fees using the mean-variance criterion to train the LSTM network.}
	\label{DL:fig:returns}
    \end{figure}

    Overall, Table \ref{DL:table:real} shows that gas fees have a considerable impact on the performance of the strategy. The average returns of the LSTM strategy before gas fees are negative only for one of our experiments. In contrast, the average returns after gas fees are positive only in 5 of the 30 experiments. Specifically, returns after gas fees are negative for all the experiments where we use the mean-variance criterion to train the LSTM network. 
    
    Figure \ref{DL:fig:returns} shows that average returns are a decreasing function of the LP's initial wealth. This follows from the LP's remuneration rule in \eqref{DL:eq:remuneration_LP} and the definition of the LP's market depth $\tilde{\kappa}$, which is a linear function of the LP's wealth; see \eqref{DL:eq:kappa_tilde}. The remuneration rule in \eqref{DL:eq:remuneration_LP} says that for each fee $p $ paid by LTs, the LP earns a percentage of $p $ equal to 
    \begin{equation}\label{DL:eq:linear}
        \frac{\tilde{\kappa}}{\kappa+\tilde{\kappa}}\,,
    \end{equation}
    where $\kappa$ is the aggregated market depth of the other LPs. Note that \eqref{DL:eq:linear} is an increasing function of $\tilde{\kappa}$ which is bounded by 1 and grows sub-linearly. Therefore, from \eqref{DL:eq:remuneration_LP}, \eqref{DL:eq:kappa_tilde}, \eqref{DL:eq:linear}, it follows that the LP's absolute returns per trade are bounded by the amount of fees paid by LTs and percentage returns are a decreasing function of the initial wealth.

    In contrast, Figure \ref{DL:fig:returns} shows that average returns after gas fees are an increasing function of the LP's initial wealth because gas fees are a flat fee and thus they impact small values of the initial wealth more. In

    Average returns before gas fees are highest and equal to 8.50\% when the LP's initial wealth is 50,000 USDC, the trading frequency is $\Delta t=10$ minutes and Sharpe Ratio is used as performance criterion. Here, 50,000 USDC and $\Delta t=10$ are, respectively, the lowest and the highest values for the initial wealth and the trading frequency. After gas fees, average returns drop to $-$34.31\%, which is the lowest value for returns after gas fees when the LP maximises Sharpe Ratio. This noticeable drop in the performance of the strategy is due to the relatively low initial wealth and the high frequency at which she readjusts her positions; when $\Delta t=10$ minutes, the LP adjusts her position too frequently and thus pays more gas fees. The relative impact of gas fees is maximised because her initial wealth is the at the minimum value considered in our experiments.

    Sharpe Ratio is highest and equal to 0.34 when the LP's initial wealth is 100,000 USDC, the trading frequency is $\Delta t=360$ minutes, and Sharpe Ratio is used as performance criterion. This shows that gas fees are a bottleneck for dynamic liquidity provision in Uniswap v3.

\begin{table}[H]\fontsize{7.5}{12.0}\selectfont
    \centering
    \begin{tabular}{|c|c|c|c|c|}
\hline

\textbf{\begin{tabular}[c]{@{}c@{}}Performance \\ Criterion\end{tabular}}  & \textbf{Frequency} & \textbf{Initial Wealth (USDC)}   & \textbf{Returns (after gas fees)}               & \textbf{Sharpe Ratio (after gas fees)}      \\ \hline
\multirow{15}{*}{Sharpe Ratio}  & 10min              & 50,000                           & \textbf{8.50\%} \qquad($-$34.31\%)        & 1.24 \qquad($-$5.01)                  \\ \cline{2-5} 
                                & 10min              & 100,000                          & 2.17\% \qquad($-$19.23\%)                 & 0.91 \qquad($-$8.07)                  \\ \cline{2-5} 
                                & 10min              & 1,000,000                        & 0.72\% \qquad($-$1.42\%)                  & 1.01 \qquad($-$1.99)                  \\ \cline{2-5} 
                                & 30min              & 50,000                           & 3.08\% \qquad($-$11.39\%)                 & 0.92 \qquad($-$3.39)                  \\ \cline{2-5} 
                                & 30min              & 100,000                          & 4.88\% \qquad($-$2.35\%)                  & \textbf{1.37} \qquad($-$0.66)         \\ \cline{2-5} 
                                & 30min              & 1,000,000                        & 0.01\% \qquad($-$0.71\%)                  & 0.01 \qquad($-$0.51)                  \\ \cline{2-5} 
                                & 60min              & 50,000                           & 4.16\% \qquad($-$3.22\%)                  & 1.06 \qquad($-$0.82)                  \\ \cline{2-5} 
                                & 60min              & 100,000                          & 3.83\% \qquad(0.14\%)                     & 1.19 \qquad(0.04)                     \\ \cline{2-5} 
                                & 60min              & 1,000,000                        & 0.05\% \qquad($-$0.32\%)                  & 0.03 \qquad($-$0.19)                  \\ \cline{2-5} 
                                & 120min             & 50,000                           & 1.69\% \qquad($-$2.15\%)                  & 0.47 \qquad($-$0.61)                  \\ \cline{2-5} 
                                & 120min             & 100,000                          & 1.53\% \qquad($-$0.39\%)                  & 0.51 \qquad($-$0.13)                  \\ \cline{2-5} 
                                & 120min             & 1,000,000                        & 0.24\% \qquad(0.05\%)                     & 0.17 \qquad(0.03)                     \\ \cline{2-5} 
                                & 360min             & 50,000                           & 2.26\% \qquad(0.79\%)                     & 0.60 \qquad(0.21)                     \\ \cline{2-5} 
                                & 360min             & 100,000                          & 1.80\% \qquad(\textbf{1.06\%})            & 0.58 \qquad(\textbf{0.34})            \\ \cline{2-5} 
                                & 360min             & 1,000,000                        & 0.26\% \qquad(0.19\%)                     & 0.18 \qquad(0.13)                     \\ \hline
\multirow{15}{*}{Mean-Variance} & 10min              & 50,000                           & 2.50\% \qquad($-$40.30\%)                 & 0.67 \qquad($-$10.82)                 \\ \cline{2-5} 
                                & 10min              & 100,000                          & 1.50\% \qquad($-$19.91\%)                 & 0.82 \qquad($-$10.88)                 \\ \cline{2-5} 
                                & 10min              & 1,000,000                        & 0.10\% \qquad($-$2.04\%)                  & 0.33 \qquad($-$6.72)                  \\ \cline{2-5} 
                                & 30min              & 50,000                           & 1.65\% \qquad($-$12.81\%)                 & 0.81 \qquad($-$6.31)                  \\ \cline{2-5} 
                                & 30min              & 100,000                          & 1.14\% \qquad($-$6.10\%)                  & 1.16 \qquad($-$6.24)                  \\ \cline{2-5} 
                                & 30min              & 1,000,000                        & 0.01\% \qquad($-$0.71\%)                  & 0.04 \qquad($-$2.14)                  \\ \cline{2-5} 
                                & 60min              & 50,000                           & 0.58\% \qquad($-$6.80\%)                  & 0.41 \qquad($-$4.80)                  \\ \cline{2-5} 
                                & 60min              & 100,000                          & 0.12\% \qquad($-$3.57\%)                  & 0.24 \qquad($-$7.10)                  \\ \cline{2-5} 
                                & 60min              & 1,000,000                        & $-$0.02\% \qquad($-$0.39\%)               & $-$0.06 \qquad($-$1.03)               \\ \cline{2-5} 
                                & 120min             & 50,000                           & 0.17\%  \qquad($-$3.67\%)                 & 0.21 \qquad($-$4.56)                  \\ \cline{2-5} 
                                & 120min             & 100,000                          & 0.15\%  \qquad($-$1.77\%)                 & 0.22 \qquad($-$2.66)                  \\ \cline{2-5} 
                                & 120min             & 1,000,000                        & $-$0.00\%  \qquad($-$0.20\%)              & $-$0.01 \qquad($-$0.45)               \\ \cline{2-5} 
                                & 360min             & 50,000                           & 0.20\%  \qquad($-$1.27\%)                 & 0.21 \qquad($-$1.34)                  \\ \cline{2-5} 
                                & 360min             & 100,000                          & 0.18\%  \qquad($-$0.56\%)                 & 0.20 \qquad($-$0.62)                  \\ \cline{2-5} 
                                & 360min             & 1,000,000                        & 0.04\%  \qquad($-$0.04\%)                 & 0.06 \qquad($-$0.07)                  \\ \hline
\end{tabular}
    \caption{Performance of the LSTM strategy on test set of market data.}
    \label{DL:table:real}
\end{table}

\appendix
\chapter{Appendix for Chapter \ref{ch:paper}}
\section{Proof of Theorem \ref{thm}}\label{sec:annex_proof}

Recall that for each fixed values of  $N$ and $j$
\begin{equation}
    \nu^{\star,j,N}\left(t,\tilde{y},Z,S\right)=-\frac{1}{\eta\,\zeta_{N}^{j}}A_{j,N}(t)\,\tilde y+\frac{1}{2\,\eta\,\zeta_{N}^{j}}B_{j,N}(t)(S-Z)\ ,
\end{equation}
where
\begin{equation}
    \zeta_{N}^{j}\coloneqq \frac{1}{\kappa}\left(Z^N_j\right)^{3/2}\,,
\end{equation}
and
\begin{equation}
    \begin{split}
        &A_{j,N}(t)\coloneqq A_{\zeta_j^N}(t)=\sqrt{\phi\,\eta\,\zeta_{N}^{j}}\tanh\left(\frac{\sqrt{\phi}}{\sqrt{\eta\,\zeta_{N}^{j}}}t+\arctanh\left(-\frac{\alpha}{\sqrt{\phi\,\eta\,\zeta_{N}^{j}}}\right)\right)\ ,\\
        &B_{j,N}(t)\coloneqq B_{\zeta_j^N}(t) = \int_{0}^{t}\beta \exp\left(\int_{s}^{t} \left(\beta-\frac{1}{\eta\,\zeta_{N}^{j}}A_{j,N}(u)\right)du\right)ds\ .
    \end{split}
\end{equation}

Moreover, recall that 
\begin{equation}
    \begin{split}
        A(t,Z) =& \sqrt{\frac{\phi\,\eta\,Z^{3/2}}{\kappa}}\tanh\left(\frac{\sqrt{\phi\,\kappa}}{\sqrt{\eta\,Z^{3/2}}}t+\arctanh\left(-\frac{\alpha\,\sqrt{\kappa}}{\sqrt{\phi\,\eta\,Z^{3/2}}}\right)\right)\ ,\\
        B(t,Z) =& \int_{0}^{t}\beta \exp\left(\int_{s}^{t} \left(\beta-\frac{\kappa}{\eta\,Z^{3/2}}A(u,Z)\right)du\right)ds\ .\\
    \end{split}
\end{equation}

To prove \eqref{eq:inequality}\,, take $(t,\tilde{y},S)$ and write
\begin{equation}
    \begin{split}
        \big|\nu^{\star,j,N}\big(t,\tilde{y}&,Z^N_{j+1},S\big)-\nu^{\star,j+1,N}\big(t,\tilde{y},Z^N_{j+1},S\big)\big| \\ =&\left|-\frac{1}{\eta\,\zeta_{N}^{j}}A_{j,N}(t)\,\tilde y
        + \frac{1}{\eta\,\zeta_{N}^{j+1}}A_{j+1,N}(t)\,\tilde y + (S-Z^N_{j+1})\left(\frac{1}{2\,\eta\,\zeta_{N}^{j}}B_{j,N}(t)-\frac{1}{2\,\eta\,\zeta_{N}^{j+1}}B_{j+1,N}(t)\right)\right|\\
        \leq&\,\frac{|\tilde{y}|}{\eta}\left|-\frac{1}{\zeta_{N}^{j}}A_{j,N}(t)
        + \frac{1}{\zeta_{N}^{j+1}}A_{j+1,N}(t) \right|+\frac{\left|S\right|+\overline{Z}}{\eta}\, \left| \frac{1}{2\,\zeta_{N}^{j}}B_{j,N}(t)-\frac{1}{2\,\zeta_{N}^{j+1}}B_{j+1,N}(t)\right|\\
        =&\,\frac{|\tilde{y}|}{\eta}\left|-\frac{\kappa}{\left(Z^N_j\right)^{3/2}}A\left(t,Z^N_j\right)
        + \frac{\kappa}{\left(Z^N_{j+1}\right)^{3/2}}A\left(t,Z^N_{j+1}\right) \right|\\
        &+\frac{\left|S\right|+\overline{Z}}{2\,\eta}\, \left| \frac{\kappa}{\left(Z^N_j\right)^{3/2}}B\left(t,Z^N_j\right)-\frac{\kappa}{\left(Z^N_{j+1}\right)^{3/2}}B\left(t,Z^N_{j+1}\right)\right|\,.
    \end{split}
\end{equation}

Observe that for a fixed $t\in[0,T]$ the functions
\begin{equation}
    \begin{split}
        Z\mapsto\frac{\kappa}{Z^{3/2}}A\left(t,Z\right)\, \ \ \text{and} \ \ 
        Z\mapsto\frac{\kappa}{Z^{3/2}}B\left(t,Z\right)
    \end{split}
\end{equation}
are uniformly continuous on $\left[\underline{Z},\overline{Z}\right]$ because they are both compositions of continuous functions defined over a closed interval. By definition of the partition in \eqref{eq:partitionedConvexity}, $\left|Z^N_{j} - Z^N_{j+1}\right| = 1/N $ so for each $\varepsilon>0$ there exists $N\in\N$ such that
\begin{equation}
    \max_{j=1,\dots,N}\big|\nu^{\star,j,N}\big(t,\tilde{y},Z^N_{j+1},S\big)-\nu^{\star,j+1,N}\big(t,\tilde{y},Z^N_{j+1},S\big)\big|\leq\varepsilon\,.
\end{equation}

To prove that $\{\tilde{\nu}^{\star,N}\}$ converges uniformly to $\tilde{\nu}^{\star}$ in $ [0,T]\times\R\times\left[\underline{Z},\overline{Z}\right]\times\R $\,, take $(t,\tilde{y},Z,S)\in[0,T]\times\R^{2}\times\left[\underline{Z},\overline{Z}\right]\times\R\,, $ and $N\in\N\,,$ and observe that there exists $\overline{j}\in\{1,\dots,N\}$ such that $Z\in\left[Z_{\overline{j}}^N,Z^N_{\overline{j}+1}\right)$ and thus
\begin{equation}
    \begin{split}
        \big|\tilde{\nu}^{*}_{N}\big(t,\tilde{y}&,Z,S\big)-\tilde{\nu}^{*}\big(t,\tilde{y},Z,S\big)\big| \\
        =&\big|\tilde{\nu}^{*}_{\overline{j},N}\big(t,\tilde{y},Z,S\big)-\tilde{\nu}^{*}\big(t,\tilde{y},Z,S\big)\big|\\
        =&\left|-\frac{1}{\eta\,\zeta_{N}^{\overline{j}}}A_{\overline{j},N}(t)\,\tilde y
        + \frac{\kappa}{\eta\,Z^{3/2}}A(t,Z)\,\tilde y + (S-Z)\left(\frac{1}{2\,\eta\,\zeta_{N}^{\overline{j}}}B_{\overline{j},N}(t)-\frac{\kappa}{2\,\eta\,Z^{3/2}}B(t,Z)\right)\right|\\
        \leq&\,\frac{|\tilde{y}|\,\kappa} {\eta\,\underline{Z}^{3/2}}\left|-A\left(t,Z^N_{\overline{j}}\right)
        + A\left(t,Z\right) \right|+ \frac{\left|S\right|+\overline{Z}} {2\,\eta\,\underline{Z}^{3/2}}\,\kappa\, \left| B\left(t,Z^N_{\overline{j}}\right)-B\left(t,Z\right)\right|\,.
    \end{split}
\end{equation}

The uniform convergence of $\{\tilde{\nu}^{*}_{N}\}$ to $\tilde{\nu}^*$ follows from the uniform continuity of $A(t,Z)$ and $B(t,Z)$ on $[0,T]\times\left[\underline{Z},\overline{Z}\right]\,.$
\qed

\section{Example for the liquidation strategy \label{Example one run liquidation}}
In this appendix, we describe the parameters and strategy performance for a specific run of the liquidation strategy. Assume the LT will start trading at noon on 16 March 2022, so she uses the data between noon 15 March 2022 and noon 16 March to estimate model parameters. 

For the 24 hours before noon 16 March 202, there are, on average,  one liquidity taking order every 13 seconds in the liquid pool and one every 360 seconds in the illiquid pool; i.e.,  the time steps in the regressions \eqref{eq: discrete price dyn} are $\Delta t = 13$ for ETH/USDC and $\Delta t = 360$ for ETH/DAI. Table \ref{table:params1} shows parameter estimates. 

{\footnotesize
    \begin{table}[H] 
        \begin{center}
            \begin{tabular}{c  c  c} 
                \hline 
                & ETH/USDC & ETH/DAI \\ [0.5ex] 
                \hline
                $\hat \sigma$ & $0.045\ \textrm{day}^{- 1/2}$ & $0.053\ \textrm{day}^{- 1/2}$ \\ [0.5ex] 
                $\hat\gamma$ & $0.034\ \textrm{day}^{- 1/2}$& $0.027\ \textrm{day}^{- 1/2}$ \\ [0.5ex]
                $\hat\beta$ & $657.9\ \textrm{day}^{-1}$ & $14.78\ \textrm{day}^{-1}$  \\  [0.5ex]
                \hline 
            \end{tabular}
        \end{center}
        \caption {Parameter estimates for dynamics of $Z$ and $S$ with data between noon 15 March 2022 and noon 16 March 2022.}
        \label{table:params1}
    \end{table}
}

The parameter $\eta$ of the execution costs in \eqref{eq:execratesApproxCPMM_nu} is also set to $13\, \textrm{seconds} \, = \, 17.3\times 10^{-5} \,\textrm{days}$ and $360\, \textrm{seconds} \,=\, 41 \times 10^{-4}\, \textrm{days}$ for the liquid and illiquid pool, respectively. The number of observed transactions in the in-sample data is approximately 238,039 ETH and 4,031 ETH in the liquid and illiquid pool, respectively. Thus the LTs' target is to liquidate 14,877 and 1,007 units of ETH within $2$ and $12$ hours in the ETH/USDC and ETH/DAI pools, respectively. Table \ref{table:params2} summarises all the parameters used to run our strategy.

{\footnotesize
    \begin{table}[H]
        \begin{center}
            \begin{tabular}{c  r  r} 
                \hline 
                & ETH/USDC & ETH/DAI \\ [0.5ex] 
                \hline
                $\kappa_0$ & 22,561,783 & 1,666,175 \\[0.5ex] 
                $\tilde y_0$ & 14,877 ETH & 1,007 ETH\\ [0.5ex]
                $\tilde S_0$ & 2,689.2$\  \textrm{USDC}$ & 2,686.09$\  \textrm{DAI}$ \\ [0.5ex]
                $\tilde Z_0$ & 2,690.77$\  \textrm{USDC}$ & 2,694.04$\  \textrm{DAI}$ \\ [0.5ex]
                $\eta$ & $17.3\times 10^{-5}$ $\ \textrm{days}$ & $41 \times 10^{-4}$ $\ \textrm{days}$ \\
                \hline 
            \end{tabular}
        \end{center}
        \caption {Values of model parameters.}
        \label{table:params2}
    \end{table}
}

Figure \ref{fig:backtest_USDC_DAI} shows the instantaneous and oracle rates and the inventories of the strategies during the execution, for both ETH/USDC and ETH/DAI. Figure \ref{fig:backtest_USDC_DAI} clearly showcases the difference between the strategies. In particular, the liquidation strategy is adaptive and trades on the difference between the two rates $S$ and $Z$ during the liquidation programme. Figure \ref{fig:backtest_USDC_DAI_speed} shows how the difference $S_t - Z_t$ drives the trading speed $\nu_t.$ The oracle rate is used as a predictive signal for future moves of the instantaneous rate.
\begin{figure}
    \centering
    \subfloat[\textbf{Top}: Out-of-sample instantaneous and oracle rates for the pair ETH/USDC. \textbf{Bottom}: Inventory process $\tilde y$ for the optimal, TWAP, and single order strategies.]  {{\includegraphics{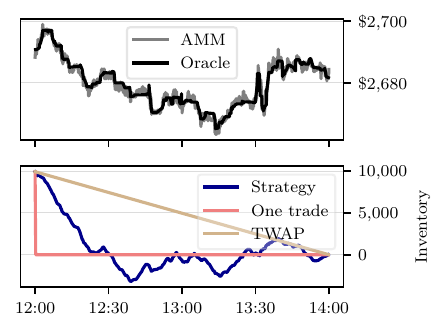} }}%
    \qquad
    \subfloat[\textbf{Top}: Out-of-sample instantaneous and oracle rates for the pair ETH/DAI. \textbf{Bottom}: Inventory process $\tilde y$ for the optimal, TWAP, and single order strategies.] {{\includegraphics{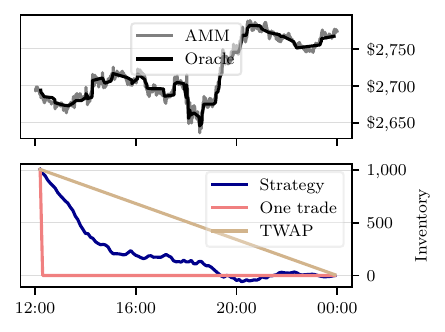} }}%
    \caption{Liquidation strategies starting at noon on 16 March 2022.}%
    \label{fig:backtest_USDC_DAI}%
\end{figure}

\begin{figure}
    \centering
    \subfloat[\textbf{Top}: Difference between the oracle and the instantaneous rate for ETH/USDC. \textbf{Bottom}: $-\nu_t.$] {{\includegraphics{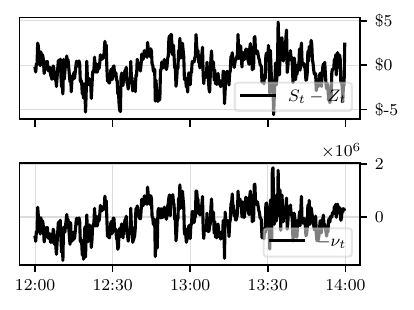} }}%
    \qquad
    \subfloat[\textbf{Top}: Difference between the oracle and the instantaneous rate for ETH/DAI. \textbf{Bottom}: $-\nu_t.$] {{\includegraphics{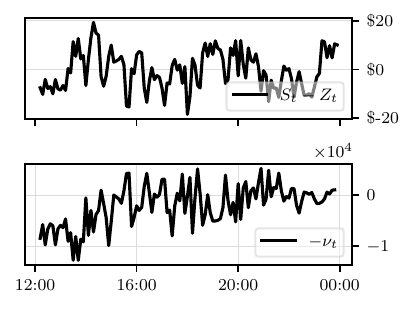} }}%
    \caption{Trading speed.}%
    \label{fig:backtest_USDC_DAI_speed}%
\end{figure}

\section{Model I with stochastic convexity cost \label{sec:APX:model}}


We consider the same problem as that of Section \ref{sec:model:optimal} where an LT  exchanges a large position in asset $Y$ into asset $X$ in a CPMM or executes a statistical arbitrage over a period of time $[0,T]$, where $T>0\,$. She uses the rate $Z$ in \eqref{eq:ZProcess} from the pool and the rate $S$ in \eqref{eq:SProcess} from  another more liquid exchange. The depth $\kappa$ of the pool is constant and the LT liquidates a position $\tilde{y}_0 \in \R$ in asset $Y$. Her wealth is valued in terms of asset $X$ and has dynamics \eqref{eq:wealth}, where the execution cost is stochastic and its dynamics are known. During the trading programme, the LT trades at the speed $(\nu_t)_{t \in [0,T]}$, so the inventory $(\tilde y_t)_{t \in [0,T]}$ evolves as in \eqref{eq:ytildeProcess_modelI}, where we do not restrict the speed to be positive, and for simplicity, trading fees are set to zero. 

The LT maximises her expected terminal wealth in units of $X$ while penalising inventory. The set of admissible strategies is defined in \eqref{def:admissibleset_t_modelI} and the LT's  performance criterion is a function $u^{\nu}\colon[0,T] \times \R \times \R \times \R_{++} \times\R_{++} \rightarrow \R\ $ defined in  \eqref{eq:perfcriteria_modelI}.  The value function $u:[0,T] \times \R \times \R \times \R_{++} \times\R_{++} \rightarrow \R$ of the LT is given by
\begin{align}
    \label{eq:ANX:valuefunc}
    u(t,\tilde{x},\tilde{y},Z,S)=\underset{\nu\in\mathcal{A}}{\sup}\{u^{\nu}(t,\tilde{x},\tilde{y},Z,S)\}\ .
\end{align}

The value function solves the Hamilton--Jacobi--Bellman (HJB) equation\begin{equation}\label{eq:ANX:hjbu}
    \begin{split}
        0=\,&\partial_{t}w-\phi\,\tilde{y}^{2}+\beta\,\left(S-Z\right)\,\partial_{Z}w+\frac{1}{2}\,\gamma^{2}\,Z^{2}\,\partial_{ZZ}w+\frac{1}{2}\,\sigma^{2}\,S^{2}\,\partial_{SS}w\\
        &+\sup_{\nu\in\R}\Bigg(\left(\nu\, Z-\frac{\eta}{\kappa}\,Z^{3/2}\,\nu^{2}\right)\partial_{\tilde{x}}w-\nu\,\partial_{\tilde{y}}w\Bigg)\ ,
    \end{split}
\end{equation}
with terminal condition \begin{align} \label{eq:ANX:termcondu}
    w(T,\tilde{x},\tilde{y},Z,S)=\tilde{x}+\tilde{y}\,Z-\alpha\,\tilde{y}^{2}\ .
\end{align}

The form of the terminal condition \eqref{eq:ANX:termcondu} suggests the ansatz
\begin{align}
    \label{eq:ANX:ansatz1}
    w(T,\tilde{x},\tilde{y},Z,S)=\tilde{x}+\tilde{y}\,Z+\theta(t,\tilde{y},Z,S)\ .
\end{align}
Now replace \eqref{eq:ANX:ansatz1} into \eqref{eq:ANX:hjbu} to obtain the PDE
\begin{equation}\label{eq:ANX:hjbtheta1}
    \begin{split}
        0\,=\,&\partial_{t}\theta-\phi\,\tilde{y}^{2}+\beta\,\left(S-Z\right)\,\left(\tilde{y}+\partial_{Z}\theta\right)+\frac{1}{2}\,\gamma^{2}\,Z^{2}\,\partial_{ZZ}\theta+\frac{1}{2}\,\sigma^{2}\,S^{2}\,\partial_{SS}\theta\\
        &+\sup_{\nu\in\R}\Bigg(-\frac{\eta}{\kappa}\,Z^{3/2}\,\nu^{2}-\nu\,\partial_{\tilde{y}}\theta\Bigg)\ ,
    \end{split}
\end{equation}
defined over $[0,T)\times \R \times \R^\star \times \R$, with terminal condition
\begin{equation}\label{eq:ANX:terminal1}
    \theta(T,\tilde{y},Z,S)= -\alpha\,\tilde{y}^2\ .
\end{equation}

The first two terms on the right-hand side of \eqref{eq:ANX:ansatz1} are the mark-to-market value of the LT's holdings and the last term is the additional value that the LT obtains by following the optimal strategy.  Next, solve the first order condition in \eqref{eq:ANX:hjbtheta1} to obtain the optimal trading speed in feedback form
\begin{equation}\label{eq:ANX:optimalspeed}
    \nu^{\star}=-\frac{\kappa}{2\,\eta}\,Z^{-3/2}\,\partial_{\tilde{y}}\theta\ ,
\end{equation}
and substitute \eqref{eq:ANX:optimalspeed} into \eqref{eq:ANX:hjbtheta1} to write
\begin{equation}\label{eq:ANX:hjbtheta2}
    \partial_{t}\theta=-\phi\,\tilde{y}^{2}+\beta\,\left(S-Z\right)\,\left(\tilde{y}+\partial_{Z}\theta\right)+\frac{1}{2}\,\gamma^{2}\,Z^{2}\,\partial_{ZZ}\theta+\frac{1}{2}\,\sigma^{2}\,S^{2}\,\partial_{SS}\theta+\frac{\kappa}{4\,\eta}\,Z^{-3/2}\,\partial_{\tilde{y}}\theta^{2}\ .
\end{equation}
Finally, we propose the ansatz
\begin{equation}
    \label{eq:ANX:ansatz2}
    \theta(t,\tilde{y},Z,S)=\theta_{2}(t,Z,S)\,\tilde{y}^{2}+\theta_{1}(t,Z,S)\,\tilde{y}+\theta_{0}(t,Z,S)\ .
\end{equation}
Now substitute \eqref{eq:ANX:ansatz2} into \eqref{eq:ANX:hjbtheta2} to obtain the system of PDEs
\begin{equation}\label{eq:ANX:system}
    \left\{
    \begin{aligned}
        &-\partial_{t}\theta_{2}= -\phi+\beta\,\left(S-Z\right)\,\partial_{Z}\theta_{2}+\frac{1}{2}\,\gamma^{2}\,Z^{2}\,\partial_{ZZ}\theta_{2}+\frac{1}{2}\,\sigma^{2}\,S^{2}\,\partial_{SS}\theta_{2}+\frac{\kappa}{\eta}\,Z^{-3/2}\,\theta_{2}^{2}\ ,\\
        &-\partial_{t}\theta_{1}=  \beta\,\left(S-Z\right)\left(1+\partial_{Z}\theta_{1}\right)+\frac{1}{2}\,\gamma^{2}\,Z^{2}\,\partial_{ZZ}\theta_{1}+\frac{1}{2}\,\sigma^{2}\,S^{2}\,\partial_{SS}\theta_{1}+\frac{\kappa}{\eta}\,Z^{-3/2}\,\theta_{2}\,\theta_{1}\ ,\\
        &-\partial_{t}\theta_{0}=  \beta\,\left(S-Z\right)\,\partial_{Z}\theta_{0}+\frac{1}{2}\,\gamma^{2}\,Z^{2}\,\partial_{ZZ}\theta_{0}+\frac{1}{2}\,\sigma^{2}\,S^{2}\,\partial_{SS}\theta_{0}+\frac{\kappa}{4\,\eta}\,Z^{-3/2}\,\theta_{1}^{2}\ ,
    \end{aligned}
    \right.
\end{equation}
defined over $[0,T)\times\R^\star \times \R $\ , with terminal condition 
\begin{equation}\label{eq:ANX:system_TC}
    \theta_{2}(T,Z,S)=-\alpha\ ,\qquad\theta_{1}(T,Z,S)=0\ ,\qquad\theta_{0}(T,Z,S)=0\ .
\end{equation}

The optimal strategy in feedback form \eqref{eq:ANX:optimalspeed} is given by
\begin{equation}\label{eq:ANX:optimalspeed2}
    \nu^{\star}=-\frac{\kappa}{2\,\eta}\,Z^{-3/2}\,\left(2\,\theta_2\,\tilde y+\theta_1\right)\ .
\end{equation}

The next result provides a-priori bounds for the solution to the system of PDEs \eqref{eq:ANX:system}.
\begin{thm} \label{thm:ANX:existence}
    Assume there exist $\theta_0\in C^{1,2,2}([0,T]\times\R_{++}\times\R_{++})$, $\theta_1\in C^{1,2,2}([0,T]\times\R_{++}\times\R_{++})$, and $\theta_2\in C^{1,2,2}([0,T]\times\R_{++}\times\R_{++})$ solution to the system of PDEs \eqref{eq:ANX:system} with terminal condition \eqref{eq:ANX:system_TC}. Define $\theta$ in \eqref{eq:ANX:ansatz2} and define $w$ in \eqref{eq:ANX:ansatz1}. Assume that for all  $\left(t,\tilde{x},\tilde{y},Z,\kappa \right) \in [0,T] \times \R \times \R \times \R_{++} \times \R_{++}$ we have $u^{\nu}(t,\tilde{x},\tilde{y},Z,\kappa ) \leq w\left(t,\tilde{x},\tilde{y},Z,\kappa\right),$ where $u^\nu$ is defined in \eqref{eq:perfcriteria_modelI}, and that equality is obtained for the optimal control $\left(\nu^{\star}\right)_{t\in[0, T]}$ in feedback form in \eqref{eq:ANX:optimalspeed2}. Then $\theta_0$, $\theta_1$, and $\theta_2$ have the following bounds for all $\left(t,Z,S\right) \in [0,T]\times\R_{++}\times\R_{++}$:
    \begin{equation}
        \begin{split}
            &\begin{cases}
                \theta_{0}\left(t,Z,S\right) \geq 0\\
                \theta_{0}\left(t,Z,S\right)\leq A(t)\,S^{2}+\frac{1}{2}B(t)\,S\,Z+C(t) \,Z^{2}-Z+\E\left[Z_{T}\right]\,,
            \end{cases}\\
            &\begin{cases}
                \theta_{1}\left(t,Z,S\right) \geq -\left(\alpha+\phi\left(T-t\right)\right)-A(t)\,S^{2}-\frac{1}{2}B(t)\,S\,Z-C(t)\,Z^{2}\\
                \theta_{1}\left(t,Z,S\right)\leq A(t)\,S^{2}+\frac{1}{2}B(t)\,S\,Z+C(t)\,Z^{2}+\alpha+\phi\left(T-t\right)\,,
            \end{cases}\\
            &\begin{cases}
                \theta_{2}\left(t,Z,S\right) \geq -\left(\alpha+\phi\left(T-t\right)\right)\\
                \theta_{2}\left(t,Z,S\right)\leq A(t)\,S^{2}+\frac{1}{2}B(t)\,S\,Z+C(t)\,Z^{2}\,,
            \end{cases}
        \end{split}
    \end{equation}
    where $\E\left[Z_{T}\right]$ is in \eqref{eq:ANX:EZT} and $A,$ $B$, and $C$ are in \eqref{eq:ANX:proof_odesystem_merton_sol}.
\end{thm}

\begin{proof}
    We provide a-priori bounds for the solutions $\theta_2,$ $\theta_1,$ $\theta_0$ of \eqref{eq:ANX:system}.

    \paragraph{\underline{Lower bound for $\theta$}} By definition of our control problem in \eqref{eq:ANX:valuefunc} we know that for any control $\underline \nu\in\mathcal{A}$, we have
    \begin{align*}
        u^{\underline \nu}(t,\tilde{x},\tilde{y},Z,S) \leq \sup_{\nu\in\mathcal{A}}u^{\nu}(t,\tilde{x},\tilde{y},Z,S)\ .
    \end{align*}
    Use the definition of $\theta$ in \eqref{eq:ANX:ansatz1} and the form of the performance criterion to write
    \begin{align*}
        \E\left[\tilde{x}_{T}^{\underline \nu}+\tilde{y}_{T}^{\underline\nu}\,Z_T-\alpha\left(\tilde{y}_{T}^{\underline\nu}\right)^{2}-\phi\int_{t}^{T}\left(\tilde{y}_{s}^{\underline\nu}\right)^{2}ds\right] \leq \tilde{x}+\tilde{y}\,Z+\theta(t,\tilde{y},Z,S)\ .
    \end{align*}
    
    Next, consider a sub-optimal strategy consisting in keeping a constant inventory with no trading, i.e., a strategy defined by $\underline\nu_{s}=0$ for each $s\geq t\ .$ Then the inventory process is $\left(\tilde{y}^{\underline \nu}_{s}\right)_{s\geq t}=\tilde{y}.$ Use the inequality above to write
    \begin{equation}
        \tilde{x}+\tilde{y}\,Z+\theta(t,\tilde{y},Z,S) \geq \E\left[\tilde{x}+\tilde{y}\,Z-\tilde{y}\left(Z-Z_{T}\right)-\alpha\,\tilde{y}^{2}-\phi\,\tilde{y}^{2}\left(T-t\right)\right],
    \end{equation}
    and conclude that
    \begin{equation}
        \label{eq:ANX:proof_lowerbound}
        \theta(t,\tilde{y},Z,S)\geq-\tilde{y}\,Z-\tilde{y}\E\left[Z_{T}\right]-\alpha\,\tilde{y}^{2}-\phi\,\tilde{y}^{2}\left(T-t\right)>-\infty\,,
    \end{equation}
    which provides a lower bound for $\theta$\ . The last inequality in \eqref{eq:ANX:proof_lowerbound} is a consequence of
    \begin{equation}
        \E
        \begin{bmatrix}
            S_T\\ Z_T 
        \end{bmatrix} = 
        \begin{pmatrix}
            S\\ Z
        \end{pmatrix}e^{A(T-t)}\,, \quad \text{where} \quad A = 
        \begin{pmatrix}
            0 & 0 \\ -\theta & \theta 
        \end{pmatrix}\,,
    \end{equation}
    which simplifies to
    \begin{equation} \label{eq:ANX:EZT}
        \E\left[Z_T\right] = Z\,e^{-\theta(T-t)}+\theta\,S\,\int_t^Te^{-\theta(T-s)}ds = Z\,e^{-\theta(T-t)}+S\,\left(1-e^{-\theta\,(T-t)}\right)\,.
    \end{equation}
    
    \paragraph{\underline{Upper bound for $\theta$}} Use integration by parts to write, for all $(t, V, \tilde y, Z, S),$
    \begin{equation}
        \label{eq:ANX:proof_upperbound_0}
        \begin{split}
            \tilde{x}+\tilde{y}\,Z+\theta(t,\tilde{y},Z,S) = & \underset{\nu\in\mathcal{A}_{t}}{\sup}\E\left[\tilde{x}_{T}+\tilde{y}_{T}Z-\alpha\,\tilde{y}_{T}^{2}-\phi\,\int_{t}^{T}\tilde{y}_{s}^{2}\,ds\right] \\
            = & \tilde{x}+\tilde{y}\,Z+\underset{\nu\in\mathcal{A}_t}{\sup}\E\left[\int_{t}^{T}\tilde{y}_{s}\,dZ_{s}-\frac{\eta}{\kappa}\,\int_{t}^{T}Z_{s}^{3/2}\,v_{s}^{2}\,ds-\alpha\,\tilde{y}_{T}^{2}-\phi\,\int_{t}^{T}\tilde{y}_{s}^{2}\,ds\right] \\
            \leq & \tilde{x}+\tilde{y}\,Z+\underset{\nu\in\mathcal{A}_t}{\sup}\E\left[\int_{t}^{T}\tilde{y}_{s}\,dZ_{s}-\phi\,\int_{t}^{T}\tilde{y}_{s}^{2}\,ds\right]\, ,
        \end{split}
    \end{equation}
    where the first equality is by definition of our control problem and $\theta,$ the second equality uses integration by parts, and the last inequality holds because the terms $\frac{\eta}{\kappa}\int_{t}^{T}Z_{s}^{3/2}\,v_{s}^{2}\,ds$ and $\alpha\,\tilde{y}_{T}^{2}$ are always positive.
    
    Let $t \,\in\, [0,T] . \ $ Define the set 
    \begin{equation}\label{def:ANX:proof_admissibleset_t_model}
        \mathcal A^{M}_t = \left\{ \left(\tilde y_s^{M}\right)_{s \in [t,T]},\ \R\textrm{--valued},\ \F\textrm{--adapted and such that } \int_t^T \left|\tilde y_u^{M}\right|^2 du < +\infty, \ \ \P\textrm{--a.s.}  \right\}.
    \end{equation}
    Clearly, one has
    \begin{align*}
        \underset{\nu\in\mathcal{A}_t}{\sup}\E\left[\int_{t}^{T}\tilde{y}_{s}\,dZ_{s}-\phi\int_{t}^{T}\tilde{y}_{s}^{2}ds\right] \leq \underset{\tilde y^M\in\mathcal A^{M}_t}{\sup}\E\left[\int_{t}^{T}\tilde{y}_{s}^M\,dZ_{s}-\phi\int_{t}^{T}\left(\tilde y_s^{M}\right)\,ds\right].
    \end{align*}

    The term on the right-hand side of the inequality is a classical Merton problem, where the rate $Z$ follows the dynamics in \eqref{eq:ZProcess}\,. We introduce the wealth process $\left(\tilde{x}_{t}^M\right)_{t\in[0,T]}$ with dynamics
    \begin{align} \label{ANX:proof_merton_problem_dynamics}
        d\tilde{x}_{t}^M&=\tilde{y}_{t}^M\,dZ_{t}\, ,
    \end{align}
    In the Merton problem, the LT controls her inventory $\tilde{y}^M$\ , and her strategy solves the optimization problem
    \begin{equation}\label{eq:ANX:proof_merton_problem}
        \begin{split}
            \overline u (t,\tilde{x},Z,S) =\sup_{\tilde{y}\in\mathcal{A}_t^M}\E\left[\tilde{x}^M_T-\phi\,\int_{t}^{T}\left(\tilde{y}_{s}^M\right)^{2}\,ds\right]    =\underset{\tilde{y}\in A^{M}_t}{\sup}\E\left[\int_{t}^{T}\tilde{y}_{s}^M\,dZ_{s}-\phi\,\int_{t}^{T}\,\left(\tilde{y}_{s}^M\right)^{2}\,ds\right]\,.
        \end{split}
    \end{equation}
    We introduce the  value function $\overline w$ associated with the optimisation problem \eqref{eq:ANX:proof_merton_problem} and which solves the  HJB
    \begin{align} \label{def:ANX:proof_HJB_u}
        0 =\ & \partial_{t}\overline w+\beta\,\left(S-Z\right)\,\partial_{Z}\overline w+\frac{1}{2}\,\gamma^{2}\,Z^{2}\,\partial_{ZZ}\overline w+\frac{1}{2}\,\sigma^{2}\,S^{2}\,\partial_{SS}\overline w \\
        & + \underset{\tilde{y}^M\in\R}{\sup}\Bigg(\beta\,\left(S-Z\right)\,\tilde{y}^M\,\partial_{\tilde{x}}\overline w+\frac{1}{2}\,\gamma^{2}\,Z^{2}\,\left(\tilde{y}^M\right)^{2}\,\partial_{\tilde{x}\tilde{x}}\overline w+\gamma^{2}\,Z^{2}\tilde{y}\,\partial_{\tilde{x}Z}\overline w-\phi\,\left(\tilde{y}^M\right)^{2}\Bigg)\,,
    \end{align}
    with terminal condition $\overline w\left(T,\tilde{x}^M,Z,S\right)=\tilde{x}^M.$ Use the ansatz $\overline w(t,\tilde{x}^M,Z,S)=\tilde{x}^M+\overline{\theta}(t,Z,S)$ to obtain the PDE 
    \begin{align} \label{ANX:proof_hjb_theta_merton}
        0=\partial_{t}\overline{\theta}+\beta\,\left(S-Z\right)\,\partial_{Z}\overline{\theta}+\frac{1}{2}\,\gamma^{2}\,Z^{2}\,\partial_{ZZ}\overline{\theta}+\frac{1}{2}\,\sigma^{2}\,S^{2}\,\partial_{SS}\overline{\theta}+\frac{\beta^{2}\left(S-Z\right)^{2}}{4\phi}\,,
    \end{align}
    with terminal condition $\overline{\theta}(T,Z,S)=0\ .$ The optimal control is
    \begin{align}
        \label{eq:ANX:proof_optimal_control_merton}
        \tilde{y}^{M\star}=\frac{\beta\,\left(S-Z\right)}{2\,\phi}\,.
    \end{align}
    Use the ansatz
    \begin{align}
        \label{eq:ANX:proof_ansatz2_merton}
        \overline{\theta}(t,Z,S)=&\,A(t)\,S^{2}+\frac{1}{2}\,B(t)\,S\,Z+C(t)\,Z^{2}+D(t)\,S+E(t)\,Z\,,
    \end{align}
    to obtain the following system of ODEs:
    \begin{equation}
        \label{eq:ANX:proof_odesystem_merton}
        \left\{
        \begin{aligned}
            -A'(t)&= \sigma^{2}\,A(t)+\frac{\beta^{2}}{4\,\phi}+\frac{1}{2}\,\beta\, B(t)\,,\\
            -B'(t)&= -\beta\, B(t)-\frac{\beta^{2}}{\phi}+4\,\beta \,C(t)\,,\\
            -C'(t)&=  \left(\gamma^{2}-2\,\beta\right)\,C(t)+\frac{\beta^{2}}{4\,\phi}\,,\\
            -D'(t)&= \beta\, E(t)\,,\\
            -E'(t) &= -\beta\, E(t)\,,
        \end{aligned}
        \right.
    \end{equation}
    with terminal conditions $A(T) = B(T) = C(T) = D(T) = E(T) = 0.$ The solution is
    \begin{equation}
        \label{eq:ANX:proof_odesystem_merton_sol}
        \left\{
        \begin{aligned}
            A(t)&= \int_{t}^{T}e^{\sigma^{2}\,\left(u-t\right)}\,(\frac{1}{2}\,\beta B(u)+\frac{\beta^{2}}{4\,\phi})\,du\,,\\
            B(t)&= \int_{t}^{T}e^{-\beta\,\left(u-t\right)}(4\,\beta\, C(u)-\frac{\beta^{2}}{\phi})\,du\,,\\
            C(t)&= \frac{-\beta^{2}}{4\,\phi\,\left(\gamma^{2}-2\,\beta\right)}\left(1-e^{\left(\gamma^{2}-2\,\beta\right)\,\left(T-t\right)}\right)\,,\\
            D(t)&=0\,,\\
            E(t)&=0\,.\\
        \end{aligned}
        \right.
    \end{equation}
    
    Note that we obtain a classical solution to the Merton problem so standard results apply. In particular, one only needs to verify that the optimal control is admissible, which is clear from the optimal control \eqref{eq:ANX:proof_optimal_control_merton} and the admissible set \eqref{def:ANX:proof_admissibleset_t_model}. Finally, write the inequalities in \eqref{eq:ANX:proof_lowerbound} and \eqref{eq:ANX:proof_upperbound_0} as
    \begin{equation}
        \label{eq:ANX:proof_bounds}
        A(t)\,S^{2}+\frac{1}{2}\,B(t)\,S\,Z+C(t)\,Z^{2}\geq\theta(t,\tilde{y},Z,S)\geq-\tilde{y}\,\left(-Z+\E\left[Z_{T}\right]\right)-\left(\alpha+\phi\,\left(T-t\right)\right)\,\tilde{y}^{2}\,.
    \end{equation}

    \paragraph{\underline{Lower and upper bounds for $\theta_0,$ $\theta_1,$ and $\theta_2$}} From the à-priori bounds on $\theta$ given by the original control problem, we deduce à-priori bounds on the solutions $\theta_0,$ $\theta_1,$ and $\theta_2$ of the system \eqref{eq:ANX:system}. Recall that
    \begin{equation}\label{eq:ANX:temp1}
        \theta(t,\tilde{y},Z,S)=\theta_{2}(t,Z,S)\,\tilde{y}^{2}+\theta_{1}(t,Z,S)\,\tilde{y}+\theta_{0}(t,Z,S)\, .
    \end{equation} 
    
    First, notice that the lower bound is in the form of a quadratic polynomial in $\tilde{y}$, so $\theta_{2}\left(t,Z,S\right)\geq-\left(\alpha+\phi\left(T-t\right)\right).$ Next, for the upper bound, see that it holds for any value of the inventory $\tilde{y}$. In particular, observe that for $\tilde{y}=0$ we find that
    \begin{equation}\label{eq:ANX:temp2}
        A(t)\,S^{2}+\frac{1}{2}B(t)\,S Z+C(t)\,Z^2\geq\theta_{0}\left(t,Z,S\right)\geq0 \,.
    \end{equation}
    Moreover, if we let $\tilde{y}=1$ and $\tilde{y}=-1$ in \eqref{eq:ANX:proof_bounds} one obtains
    \begin{equation}
        2A(t)\,S^{2}+B(t)\,S\,Z+2C(t)\,Z^{2}\geq2\theta_{0}\left(t,Z,S\right)+2\theta_{2}(t,Z,S)\geq -2\left(\alpha+\phi\,\left(T-t\right)\right)\,.
    \end{equation}
    Finally, use \eqref{eq:ANX:temp2} to write
    \begin{equation}
        \label{eq:ANX:boundstheta2}
        A(t)\,S^{2}+\frac12 B(t)\,S\,Z+C(t)\,Z^2 \geq\theta_{2}\left(t,Z,S\right)\geq-\left(\alpha+\phi\left(T-t\right)\right).
    \end{equation}
    Next, substitute $\tilde{y}=1$ into \eqref{eq:ANX:proof_bounds} to write
    \begin{equation}
        \begin{split}
            \theta_{1}(t,Z,S)\geq&-Z+\E\left[Z_{T}\right]-\left(\alpha+\phi\left(T-t\right)\right)-\theta_{2}\left(t,Z,S\right)-\theta_{0}(t,Z,S)\\
            \geq &-Z+\E\left[Z_{T}\right]-\left(\alpha+\phi\left(T-t\right)\right)-\left(A(t)\,S^{2}+\frac{1}{2}B(t)\,S\,Z+C(t)\,Z^{2}\right)\ ,
        \end{split}
    \end{equation}
    and
    \begin{equation}
        \begin{split}
            \theta_{1}(t,Z,S)\leq&A(t)\,S^{2}+\frac{1}{2}B(t)\,S\,Z+C(t)\,Z^{2}-\theta_{2}\left(t,Z,S\right)-\theta_{0}(t,Z,S)\\
            \leq&A(t)\,S^{2}+\frac{1}{2}B(t)\,S\,Z+C(t)\,Z^{2}+\alpha+\phi\left(T-t\right)\ .
        \end{split}
    \end{equation}
    
    Thus, provided $\theta$ exists, then $\theta_{0}$, $\theta_{1}$, $\theta_{2}$ have à-priori upper and lower bounds. The bounds are all at most linear in time, and quadratic in $\{S, \ Z\}.$
    
\end{proof}

\chapter{Appendix to Chapter \ref{ch:pl} }

\section{Proof of Proposition \ref{prop:1} \label{sec:proofs:hjb1}}

To solve the problem \eqref{LP:eq:valuefunc}, we introduce an equivalent control problem. First, define the process $\left(\tilde \pi_t\right)_{t\in[0,T]} = \left( \pi_t - \eta_t\right)_{t\in[0,T]}$ with dynamics $$d\tilde{\pi}_{t}=\Gamma\,\left(\overline{\pi}-\tilde{\pi}_{t}\right)\diff t+\psi\,\sqrt{\tilde{\pi}_{t}}\,\diff B_{t}\,,$$ where $\tilde \pi_0 = \pi_0 - \eta_0$ and $\eta$ is in \eqref{eq:eta stochastic}. 

We introduce the performance criterion $\tilde u^{\delta}\colon[0,T] \times \R^4 \rightarrow \R$ given by
\begin{align}
\label{eq:perfcriteria_mu_tilde}
    \tilde u^{\delta}(t,v,z, \pi,\mu)=\mathbbm{E}_{t,v,z, \tilde \pi,\mu}\left[\log\left(V_{T}^{\delta}\right)\right]\, ,
\end{align}
and the value function $\tilde u:[0,T] \times \R^4 \rightarrow \R$ given by
\begin{equation}
\label{eq:valuefunc_mu_tilde}
    \tilde u(t,v,z,\tilde \pi,\mu)=\underset{\delta\in\mathcal{A}}{\sup} \, u^{\delta}(t,v,z,\pi,\mu)\,.
\end{equation}
Clearly, the problems \eqref{eq:valuefunc_mu_tilde} and \eqref{LP:eq:valuefunc} are equivalent, and the value functions satisfy $u(t,v,z,\pi,\mu) = \tilde u(t,v,z,\tilde \pi,\mu)$ for all $\left(t,v, z, \pi, \mu\right)\in[0,T] \times \R^4$ and for all $\tilde \pi = \pi - \eta \in \R$, where $\eta = \frac{\sigma^{2}}{8}-\frac{\mu}{4}\left(\mu-\frac{\sigma^{2}}{2}\right)+\frac{\varepsilon}{4}\,.$

The value function in \eqref{eq:valuefunc_mu_tilde} admits the dynamic programming principle, so it satisfies the HJB equation
\begin{align}
\label{eq:HJB_mu}
0=&\,\partial_{t}w+\frac{1}{2}\sigma^{2}\,z^{2}\,\partial_{zz}w+\mu\,Z\,\partial_{z}w+\Gamma\left(\overline{\pi}-\tilde{\pi}\right)\partial_{\tilde{\pi}}w+\frac{1}{2}\,\psi^{2}\,\tilde{\pi}\,\partial_{\tilde{\pi}\tilde{\pi}}w+\mathcal{L}^{\mu}w\\
&
+\underset{\delta\in\mathbbm{R}^{+}}{\sup}\Bigg(\frac{1}{\delta}\left(4\,\tilde{\pi}+4\,\eta-\frac{\sigma^{2}}{2}\right)v\,\partial_{v}w+\mu\,\rho\left(\delta,\mu\right)\,v\,\partial_{v}w+\frac{1}{2}\,\sigma^{2}\,\rho\left(\delta,\mu\right)^{2}\,v^{2}\,\partial_{vv}w
\\
&\qquad\quad\quad\quad-\frac{\gamma}{\delta^{2}}\,v\,\partial_{v}w+\sigma^{2}\,\rho\left(\delta,\mu\right)\,v\,z\,\partial_{vz}w\Bigg)\,,
\end{align}
with terminal condition 
\begin{align}
\label{eq:tchjb_mu}
w(T,v,z,\tilde{\pi},\mu)=\log\left(v\right), \quad \forall \left(v,z,\tilde{\pi},\mu\right) \in \R^4\,,
\end{align}
where $\mathcal L^\mu$ is the infinitesimal generator of $\mu$.

To study the HJB in \eqref{eq:HJB_mu}, use the ansatz
\begin{align}
\label{LP:eq:ansatz1}
w\left(t,v,z,\tilde{\pi},\mu\right)=\log\left(v\right)+\theta\left(t,z,\tilde{\pi},\mu\right)\,,
\end{align}
to obtain the HJB
\begin{align}
\label{eq:HJBtheta_mu}
0=&\,\partial_{t}\theta+\frac{1}{2}\sigma^{2}\,z^{2}\,\partial_{zz}\theta+\mu\,Z\,\partial_{z}\theta+\Gamma\left(\overline{\pi}-\tilde{\pi}\right)\partial_{\tilde{\pi}}w+\frac{1}{2}\,\psi^{2}\,\tilde{\pi}\,\partial_{\tilde{\pi}\tilde{\pi}}w+\frac{\mu}{2}-\frac{1}{8}\,\sigma^{2}
\\&+\mathcal{L}^{\mu}\theta+\underset{\delta\in\mathbbm{R}^{+}}{\sup}\Bigg(\frac{1}{\delta}\left(4\,\tilde{\pi}+4\,\eta-\frac{\sigma^{2}}{2}\right)+\frac{\mu^{2}}{\delta}-\frac{1}{2}\,\sigma^{2}\,\left(\frac{\mu^{2}}{\delta^{2}}+\frac{\mu}{\delta}\right)-\frac{\gamma}{\delta^{2}}\Bigg)\,,
\end{align}
with terminal condition 
\begin{align}
\label{eq:hjbthetatc_mu}
\theta\left(T, z, \tilde \pi, \mu\right) = 0\,, \quad \forall\, (z, \tilde \pi, \mu)\in\R^3\,.
\end{align}

The supremum in the HJB \eqref{eq:HJBtheta_mu} is attained at 
$$\delta^{\star}=\frac{2\,\gamma+\mu^{2}\,\sigma^{2}}{4\,\tilde{\pi}+4\,\eta-\frac{\sigma^{2}}{2}+\mu\left(\mu-\frac{\sigma^{2}}{2}\right)}=\frac{2\,\gamma+\mu^{2}\,\sigma^{2}}{4\,\tilde{\pi}+\varepsilon}\,.$$ 
Thus, \eqref{eq:HJBtheta_mu} becomes
\begin{align}
\label{eq:HJB2_mu}
0=&\,\partial_{t}\theta+\frac{1}{2}\sigma^{2}\,z^{2}\,\partial_{zz}\theta+\mu\,Z\,\partial_{z}\theta+\Gamma\left(\overline{\pi}-\tilde{\pi}\right)\partial_{\tilde{\pi}}\theta+\frac{1}{2}\,\psi^{2}\,\tilde{\pi}\,\partial_{\tilde{\pi}\tilde{\pi}}\theta+\frac{\mu}{2}-\frac{\sigma^{2}}{8}+\mathcal{L}^{\mu}\theta+\frac{1}{2}\frac{\left(4\,\tilde{\pi}+\varepsilon\right)^{2}}{2\,\gamma+\mu^{2}\,\sigma^{2}}\,.
\end{align}

Next, substitute the ansatz
\begin{align}
\label{eq:ansatz2_mu}
\theta\left(t,z,\tilde{\pi},\mu\right)=&\,A\left(t,\mu\right)z^{2}+B\left(t,\mu\right)\tilde{\pi}\,z+C\left(t,\mu\right)\tilde{\pi}^{2}\\&+D\left(t,\mu\right)z+E\left(t,\mu\right)\tilde{\pi}+F\left(t,\mu\right)\,,
\end{align}
in \eqref{eq:HJBtheta_mu}, collect the terms in $Z$ and $\tilde \pi$, and write the following system of PDEs:
\begin{equation*}
\left\{
\begin{aligned}
    \left(\partial_{t}+\mathcal{L}^{\mu}\right)A\left(t,\mu\right)= & -\sigma^{2}\,A\left(t,\mu\right)-2\,\mu\,A\left(t,\mu\right)\:,\\
\left(\partial_{t}+\mathcal{L}^{\mu}\right)B\left(t,\mu\right)= & -\mu\,B\left(t,\mu\right)+\Gamma B\left(t,\mu\right)\:,\\
\left(\partial_{t}+\mathcal{L}^{\mu}\right)C\left(t,\mu\right)= & 2\,C\left(t,\mu\right)\Gamma-\frac{8}{2\,\gamma+\mu^{2}\,\sigma^{2}}\:,\\
\left(\partial_{t}+\mathcal{L}^{\mu}\right)D\left(t,\mu\right)= & -\mu\,D\left(t,\mu\right)-\Gamma\,\overline{\pi}\,B\left(t,\mu\right)\:,\\
\left(\partial_{t}+\mathcal{L}^{\mu}\right)E\left(t,\mu\right)= & -2\,\Gamma\,\overline{\pi}\,C\left(t,\mu\right)-\psi^{2}\,C\left(t,\mu\right)+\Gamma\,E\left(t,\mu\right)-\frac{4\,\varepsilon}{2\,\gamma+\sigma^{2}\,\mu^{2}}\:,\\
\left(\partial_{t}+\mathcal{L}^{\mu}\right)F\left(t,\mu\right)= & -\Gamma\,\overline{\pi}\,E\left(t,\mu\right)+\psi^{2}\,\eta\,C\left(t,\mu\right)-\frac{1}{2}\frac{\varepsilon^{2}}{2\,\gamma+\sigma^{2}\,\mu^{2}}-\frac{\mu}{2}+\frac{\sigma^{2}}{8}\:,
\end{aligned}
\right.
\end{equation*}
with terminal conditions $A(T,\mu)=B(T,\mu)=C(T,\mu)=D(T,\mu)=E(T,\mu)=F(T,\mu)=0$ for all $\mu \in \R\,.$

First, note that the PDEs in $A$, $B$, and $D$ admit the unique solutions $A = B = D = 0\,.$ Next, we solve the PDE in $C\,.$ Use Itô's lemma to write
$$C\left(T,\mu_{T}\right)=C\left(t,\mu_{t}\right)+\int_{t}^{T}\left(\partial_{t}+\mathcal{L}^{\mu}\right)C\left(s,\mu_{s}\right)\diff s\,.$$
Next, replace $\left(\partial_{t}+\mathcal{L}^{\mu}\right)C\left(s,\mu_{s}\right)$ with $2\,C\left(s,\mu_{s}\right)\Gamma-\frac{8}{2\,\gamma+\mu_{s}^{2}\,\sigma^{2}}$ to obtain $$C\left(T,\mu_{T}\right)=C\left(t,\mu_{t}\right)+2\,\Gamma\int_{t}^{T}C\left(s,\mu_{s}\right)\diff s-\int_{t}^{T}\frac{8}{2\,\gamma+\mu_{s}^{2}\,\sigma^{2}}\,\diff s\,.$$
Take expectations to get the equation $$C\left(t,\mu_{t}\right)=	\E_{t,\mu}\left[-2\,\Gamma\int_{t}^{T}C\left(s,\mu_{s}\right)\diff s+\int_{t}^{T}\frac{8}{2\,\gamma+\mu_{s}^{2}\,\sigma^{2}}\,\diff s\right]\,.$$
Now consider the candidate solution function $$\hat{C}\left(t,\mu_{t}\right)=\E_{t,\mu}\left[\,\int_{t}^{T}\frac{8}{2\,\gamma+\mu_{s}^{2}\,\sigma^{2}}\exp\left(-2\,\Gamma\left(s-t\right)\right)\,\diff s\right]$$ and write \begin{align*}
    &\E_{t,\mu}\left[-2\,\Gamma\int_{t}^{T}\hat{C}\left(s,\mu_{s}\right)\diff s+\int_{t}^{T}\frac{8}{2\,\gamma+\mu_{s}^{2}\,\sigma^{2}}\,\diff s\right]\\=&\E_{t,\mu}\left[-2\,\Gamma\int_{t}^{T}\E_{s,\mu}\left[\,\int_{s}^{T}\frac{8}{2\,\gamma+\mu_{u}^{2}\,\sigma^{2}}\exp\left(-2\,\Gamma\left(u-s\right)\right)\,\diff u\right]\diff s+\int_{t}^{T}\frac{8}{2\,\gamma+\mu_{s}^{2}\,\sigma^{2}}\,\diff s\right]\\=&\E_{t,\mu}\left[\,\int_{t}^{T}\frac{8}{2\,\gamma+\mu_{s}^{2}\,\sigma^{2}}\exp\left(-2\,\Gamma\left(s-t\right)\right)\,\diff s\right]\,.
\end{align*}
Thus, $\hat C$ is a solution to the equation in $C$ and by uniqueness of solutions, we conclude that $C = \hat C\,.$

Follow the same steps as above to obtain the solution $$E\left(t,\mu\right)=\E_{t,\mu}\left[\,\int_{t}^{T}\left(\left(2\,\Gamma\,\overline{\pi}+\psi^{2}\right)C\left(s,\mu\right)+\frac{4\,\varepsilon}{2\,\gamma+\sigma^{2}\,\mu_{s}^{2}}\right)\exp\left(-\Gamma\left(s-t\right)\right)\,\diff s\right]$$ to the PDE in $E$, and the solution $$F\left(t,\mu\right)=\E_{t,\mu}\left[\,\int_{t}^{T}\left(\Gamma\,\overline{\pi}\,E\left(s,\mu_{s}\right)+\psi^{2}\,\eta_s\,C\left(s,\mu_{s}\right)-\frac{1}{2}\frac{\varepsilon^{2}}{2\,\gamma+\sigma^{2}\,\mu_{s}^{2}}-\frac{\mu_{s}}{2}+\frac{\sigma^{2}}{8}\right)\diff s\right]\,$$ to the PDE in $F\,,$ where $\eta_s = \frac{\sigma^{2}}{8}-\frac{\mu_s}{4}\left(\mu_s-\frac{\sigma^{2}}{2}\right)+\frac{\varepsilon}{4}\,,$ which proves the result.  \qed

\section{Proof of Theorem \ref{thm:verif} \label{sec:proofs:hjb2} }
Proposition \ref{prop:1} provides a classical solution to \eqref{eq:HJB_mu}. Therefore, standard results apply and  showing that \eqref{eq:optimalspeed} is an admissible control is enough to prove that \eqref{eq:hjbsol} is the value function \eqref{LP:eq:valuefunc}. Specifically, use the form of the optimal control $\delta^{\star}$ in \eqref{eq:optimalspeed} {to obtain
\begin{equation}
    0<\frac{1}{\delta_{s}^{\star}}=\frac{\tilde\pi_s+\varepsilon}{\sigma^{2}\mu_s^{2}+2\,\gamma}\leq\frac{\tilde\pi_s+\varepsilon}{2\,\gamma}\,,\qquad \forall s\in[t,T]\,,
\end{equation}
where $\tilde \pi_s = \pi_s - \eta_s,$ thus $\delta^{\star}$ is an admissible control.} \qed

\chapter{Appendix for Chapter \ref{ch:deep}}
\section{Moving average convergence/divergence}\label{DL:ax:macd}
In this section, we formalise the \textit{moving average convergence/divergence} (MACD) signal, which is one of the features that we mention in Section \ref{DL:sec:feat} and use to train the LSTM network in Section \ref{DL:sec:results_real}. The MACD signal is a trading indicator commonly used in technical analysis; see \cite{baz2015dissecting,cartea2023bandits}. 

First, we recall the definition of exponential moving average. Let $\left\{x_t\right\}_{t\in\N}$ be a real valued time series and let $s$ be a positive natural number. The \textit{exponential moving average} $\left\{ E^{(s)}\left(x_t\right) \right\}_{t\in\N}$ of $\left\{x_t\right\}_{t\in\N}$ with span $s$ is given by
\begin{equation}
    \left\{
    \begin{aligned}
        E^{(s)}\left(x_0\right) &\coloneqq x_0\,, \\
        E^{(s)}\left(x_t\right) &\coloneqq \frac{s-1}{s+1}\,E^{(s)}\left(x_{t-1}\right) + \frac{2}{s+1}\,x_t\,.
    \end{aligned}
    \right.
\end{equation}
Then, we define the MACD signal $\left\{ I\left(x_t\right) \right\}_{t\in\N}$ as
\begin{equation}
    I\left(x_t\right) \coloneqq E^{(12)}\left(x_t\right) - E^{(26)}\left(x_t\right)\,.
\end{equation}
Finally, for each $n\in\N$, the MACD $I^{(n)}\left(x_t\right)$ with period $n$ is
\begin{equation}
    I^{(n)}\left(x_t\right) \coloneqq E^{(n)}\left(I\left(x_t\right)\right)\,.
\end{equation}

\addcontentsline{toc}{chapter}{Bibliography}
\bibliography{references}        
\bibliographystyle{elsarticle-harv}  
	
\end{document}

%% file: preamble.tex
\usepackage{amssymb}
\usepackage{amsfonts,amsmath,amssymb,lscape,float,graphicx,tabularx,url,amsthm,color,natbib,afterpage,rotating,hyperref,xcolor,tikz,moresize,multirow,enumerate,colortbl,mathtools,bm,bbm}
\usepackage{enumerate}
\usepackage{enumitem}
\usepackage{bbold}
\usepackage{bbm}
\usepackage{subfig}
\usepackage{verbatim}
\usepackage{afterpage}
\newtheorem{thm}{Theorem}

\newtheorem{prop}{Proposition}

\newtheorem{assume}{Assumption}

\mathtoolsset{showonlyrefs} 

\usepackage{appendix}

\newcommand{\Ac}{\mathcal{A}}
\newcommand{\Fc}{\mathcal{F}}
\newcommand{\Tc}{\mathcal{T}}
\newcommand{\Zc}{\mathcal{Z}}

\newcommand{\E}{\mathbbm{E}}
\newcommand{\F}{\mathbbm{F}}
\newcommand{\N}{\mathbbm{N}}
\renewcommand{\P}{\mathbbm{P}}
\newcommand{\R}{\mathbbm{R}}
\newcommand{\T}{\mathbbm{T}}
\newcommand{\V}{\mathbbm{V}}
\newcommand{\Z}{\mathbbm{Z}}

\renewcommand{\epsilon}{\varepsilon}

\newcommand{\Rk}{\mathfrak{R}}

\DeclareMathOperator*{\arctanh}{arctanh}
\DeclareMathOperator*{\argmax}{arg\,max}

\DeclareMathOperator*{\softmax}{softmax}
\DeclareMathOperator*{\Tr}{Tr}

\newcommand\blankpage{%
    \null
    \thispagestyle{empty}%
    \addtocounter{page}{-1}%
    \newpage}

\definecolor{dbblue}{RGB}{10,65,155}				
\hypersetup{colorlinks,urlcolor=dbblue,citebordercolor=dbblue,linkbordercolor=dbblue,filecolor=dbblue,filebordercolor=dbblue,citecolor=dbblue,linkcolor=dbblue,colorlinks=true}

\graphicspath{{images/}{../images/}}
\newcommand{\diff}{d}